\documentclass[a4paper,reqno]{amsart}
\usepackage{amsmath, amssymb, amsthm}
\usepackage{graphicx,color}
\usepackage[dvipsnames]{xcolor} 
\usepackage{caption}
\usepackage{framed} 
\usepackage[utf8]{inputenc} 
\usepackage[T1]{fontenc} 
\usepackage{lmodern} 
\usepackage{a4wide} 
\usepackage[colorlinks,breaklinks]{hyperref} 
\usepackage{bm}
\usepackage{mathrsfs}  
\usepackage{xargs}
\usepackage{float}  
\usepackage{paralist} 
\usepackage{blkarray}  
\usepackage{multirow} 
\usepackage{longtable}
\usepackage[shortlabels,inline]{enumitem} 

\usepackage{todonotes}

\makeatletter
\renewcommand\part{\@startsection{part}{0}
  \z@{\linespacing\@plus\linespacing}{.5\linespacing}%
  {\Large\bfseries\raggedright}}
\makeatother

\usepackage{tikz}
\usepackage{tikz-cd} 
\usetikzlibrary{calc,decorations.pathreplacing,decorations.markings,arrows,decorations.pathmorphing}

\definecolor{orange}{rgb}{0.898, 0.621, 0.0}
\definecolor{skyblue}{rgb}{0.336, 0.703, 0.910}
\definecolor{bluishgreen}{rgb}{0, 0.617, 0.449}
\definecolor{yellow}{rgb}{0.937, 0.890, 0.258}
\definecolor{blue}{rgb}{0, 0.445, 0.695}
\definecolor{red}{rgb}{0.832, 0.367, 0}
\definecolor{purple}{rgb}{0.797, 0.473, 0.652}
\definecolor{forestgreen}{rgb}{0.07,0.35,0.10}

\newtheorem{theorem}{Theorem}[section] 
\newtheorem{mainthm}{Theorem}
\newtheorem{proposition}[theorem]{Proposition} 
\newtheorem{lemma}[theorem]{Lemma}

\newtheorem{corollary}[theorem]{Corollary}
\newtheorem*{theorem*}{Theorem}
\newtheorem*{nrproblem*}{$N$-Representability Problem}
\newtheorem*{enrproblem*}{Ensemble $N$-Representability Problem}

\theoremstyle{definition}
\newtheorem{definition}[theorem]{Definition} 
\newtheorem{example}[theorem]{Example}
\newtheorem{convention}[theorem]{Convention} 
\newtheorem{remark}[theorem]{Remark}
\newtheorem{algorithmplain}[mainthm]{Algorithm}
\newtheorem{question}[theorem]{Question}
\newtheorem{conjecture}[theorem]{Conjecture}
\newtheorem*{conjecture*}{Conjecture}

\definecolor{darkblue}{rgb}{0,0,0.7} 
\newcommand{\darkblue}{\color{darkblue}} 
\newcommand{\defn}[1]{\emph{\darkblue #1}} 

\DeclareMathAlphabet{\mathdutchcal}{U}{dutchcal}{m}{n}

\newcommand\ba{\bm{a}}
\newcommand\bA{\mathbf{A}}
\newcommand\bb{\bm{b}}
\newcommand\bc{\bm{c}}
\newcommand\be{\bm{e}}
\newcommand\bef{\bm{f}}
\newcommand\bi{\bm{i}}
\newcommand\bj{\bm{j}}
\newcommand\bk{\bm{k}}
\newcommand\bo{\bm{o}}

\newcommand\bv{\bm{v}}
\newcommand\bV{\mathbf{V}}
\newcommand\bu{\bm{u}}
\newcommand\bw{\bm{w}}
\newcommand\bx{\bm{x}}
\newcommand\by{\bm{y}}
\newcommand\bz{\bm{z}}
\newcommand\bchi{\bm{\chi}}

\newcommand\cB{\mathcal{B}}
\newcommand\cH{\mathcal{H}}
\newcommand\J{\mathcal{J}}
\newcommand\cV{\mathcal{V}}

\newcommand\fO{\mathscr{O}}
\newcommand\Lineups{\mathscr{L}}

\newcommand\NN{\mathbb{N}}
\newcommand\RR{\mathbb{R}}
\newcommand\CC{\mathbb{C}}

\newcommand\multiset[2]{\ensuremath{\left(\kern-.3em\left(\genfrac{}{}{0pt}{}{#1}{#2}\right)\kern-.3em\right)}}
\newcommandx{\Sym}[1]{\mathfrak{S}_{#1}} 
\newcommand{\np}[1]{\langle#1\rangle}  
\DeclareMathOperator{\projec}{pr} 
\DeclareMathOperator{\Projec}{Pr} 
\newcommandx{\proj}[1]{\projec_{#1}}
\newcommandx{\Proj}[1]{\Projec_{#1}}
\newcommandx\lu[1][1=\bV]{\Lineups_r(#1)} 
\newcommandx\non[4][1=f,2=r,3=N,4=d]{\mathbf{O}_#2^\mathrm{#1\downarrow}(#3,#4)} 
\newcommandx\Rg[2][1=r,2=\bV]{\mathcal{R}_{#1}(#2)}
\newcommand\Thres[1]{\mathcal{T}(#1)}

\newcommandx{\Lineup}[2][1={r},2=\bV]{\pol[L]_{#1}(#2)}
\newcommandx{\FerPC}[2][1=N,2=d]{\textbf{Fer}(#1,#2)}
\newcommandx{\BosPC}[2][1=N,2=d]{\textbf{Bos}(#1,#2)}
\newcommandx{\FerPoset}[2][1=N,2=d]{\mathdutchcal{Fer}(#1,#2)}
\newcommandx{\BosPoset}[2][1=N,2=d]{\mathdutchcal{Bos}(#1,#2)}
\newcommandx{\FerFan}[3][1=r,2=N,3=d]{\fan[R]_{#1}^{\mathrm{f}}(#2,#3)}
\newcommandx{\BosFan}[3][1=r,2=N,3=d]{\fan[R]_{#1}^{\mathrm{b}}(#2,#3)}

\newcommandx\simplexconfig[1][1=m]{\mathbf{E}_{#1}} 
\newcommandx{\ksetpol}[2][1=k,2=\bV]{\pol[P]_{#1}(#2)}

\newcommand\Unitary{\mathrm{U}}
\newcommand\Torus{\mathrm{T}}
\newcommand\Herm{\mathscr{B}}
\newcommand{\Density}{\mathscr{D}}
\newcommand{\OccVec}{\mathbf{o}}

\newcommandx{\pol}[1][1=P]{\mathsf{#1}} 
\newcommandx{\fan}[1][1=F]{\mathcal{#1}} 
\newcommandx{\FerPoly}[2][1={},2={N,d}]{\mathsf{\Sigma}^{\mathrm{f}}_{#1}(#2)}
\newcommandx{\BosPoly}[2][1={},2={N,d}]{\mathsf{\Sigma}^{\mathrm{b}}_{#1}(#2)}
\newcommand{\klypo}{\mathsf{\Pi}(\bw,N,d)}
\newcommand{\Pauli}{\pol[\Delta]}  
\DeclareMathOperator{\Perm}{Perm} 

\DeclareMathOperator{\aff}{aff} 
\DeclareMathOperator{\spa}{span} 
\DeclareMathOperator{\conv}{conv} 
\DeclareMathOperator{\cone}{cone} 
\DeclareMathOperator{\ncone}{ncone} 
\DeclareMathOperator{\Trace}{Tr}  
\DeclareMathOperator{\spec}{spec}
\DeclareMathOperator{\diag}{diag}

\DeclareMathOperator{\SymP}{Sym}

\DeclareMathOperator{\Id}{Id} 

\title{An effective solution to convex $1$-body $N$-representability}

\thanks{This research was supported by the grant ANR-17-CE40-0018 of the French National Research Agency ANR (project CAPPS) (A.P. and E.P.) and by the German Research Foundation (Grant SCHI 1476/1-1) (J.L. and C.S.) and by the UK Engineering and Physical Sciences Research Council (Grant EP/P007155/1) (C.S.).}

\author[F.~Castillo]{Federico Castillo}
\address[F.~Castillo]{Max Planck Institute for Mathematics in the Sciences, Inselstraße 22, 04103 Leipzig, Germany}
\email{efecastillo.math@gmail.com}
\urladdr{\url{https://sites.google.com/view/fcastillo}}

\author[J.-P.~Labbé]{Jean-Philippe Labbé}
\address[J.-P. Labb\'e]{Électron libre}
\email{labbe@math.fu-berlin.de}
\urladdr{\url{http://page.mi.fu-berlin.de/labbe}}

\author[J.~Liebert]{Julia Liebert}
\address[J.~Liebert, C.~Schilling]{Department of Physics, Arnold Sommerfeld Center for Theoretical Physics, Ludwig-Maximilians-Universität München, Theresienstrasse 37, 80333 München, Germany\newline
Munich Center for Quantum Science and Technology, Schellingstrasse 4, 80799 München, Germany}
\email{julia.liebert@physik.uni-muenchen.de}
\email{c.schilling@physik.uni-muenchen.de}
\urladdr{\url{https://www.theorie.physik.uni-muenchen.de/lsschollwoeck/schilling_group}}

\author[A.~Padrol]{Arnau Padrol}
\address[A.~Padrol]
{Sorbonne Université and Université de Paris, CNRS, IMJ-PRG, Paris, France.}
\email{arnau.padrol@imj-prg.fr.}
\urladdr{\url{https://webusers.imj-prg.fr/~arnau.padrol/}}

\author[E.~Philippe]{Eva Philippe}
\address[E.~Philippe]{DMA \'Ecole Normale Sup\'erieure de Paris and Universit\'e de Paris
, France}
\email{eva.philippe@ens.fr}

\author[C.~Schilling]{Christian Schilling}

\keywords{Reduced density matrices, $N$-representability, Pauli exclusion principle, convex geometry, Schur-Horn theorem, permutation-invariant polytopes}
\subjclass[2010]{Primary 52B12; Secondary 81V45, 81R05, 81V73, 81V74, 81C99, 47L07, 52A27, 90C25}

\begin{document}
\sloppy
\maketitle

\begin{abstract}
From a geometric point of view, Pauli’s exclusion principle defines a hypersimplex.
This convex polytope describes the compatibility of $1$-fermion and $N$-fermion density matrices, therefore it coincides with the convex hull of the pure $N$-representable $1$-fermion density matrices.
Consequently, the description of ground state physics through $1$-fermion density matrices may not necessitate the intricate pure state generalized Pauli constraints.
In this article, we study the generalization of the $1$-body $N$-representability problem to ensemble states with fixed spectrum~$\bw$, in order to describe finite-temperature states and distinctive mixtures of excited states. 
By employing ideas from convex analysis and combinatorics, we present a comprehensive solution to the corresponding convex relaxation, thus circumventing the complexity of generalized Pauli constraints.
In particular, we adapt and further develop tools such as symmetric polytopes, sweep polytopes, and Gale order.
For both fermions and bosons, generalized exclusion principles are discovered, which we determine for any number of particles and dimension of the $1$-particle Hilbert space.  
These exclusion principles are expressed as linear inequalities satisfying hierarchies determined by the non-zero entries of $\bw$.
The two families of polytopes resulting from these inequalities are part of the new class of so-called \emph{lineup polytopes}.
\end{abstract}

\section*{Introduction}
The notion of convexity has been part of quantum mechanics since its inception almost 100~years ago.
Without a doubt, the interpretation of Pauli's exclusion principle via reduced density matrices is among the most prominent manifestations of convexity in quantum mechanics.
It may be stated as follows: If $\rho$ is the density operator corresponding to the state of a system of $N$ fermions, then the eigenvalues of the 1-body reduced density matrix of $\rho$ (normalized to have trace $N$) should be contained in the real interval $[0,1]$.
The fact that the eigenvalues are bounded above by $1$ illustrates the fact that---according to Pauli's exclusion principle---\emph{no two fermions can occupy at the same time the same $1$-particle quantum state}.

Pauli formulated his principle in 1925 \cite{pauli_zusammenhang_1925}, density matrices were introduced by Dirac a few years later \cite{dirac_note_1930}, and Husimi introduced reduced density matrices some time after that \cite[Section~4]{husimi_formal_1940}.
Almost simultaneously, the above necessary spectral conditions appeared in an article of Watanabe, written in 1939 while he visited Heisenberg in Leipzig \cite{watanabe_anwendung_1939}.
Because the trace of the $1$-body reduced density matrix is normalized to $N$ and each eigenvalue lies between $0$ and $1$, the set of possible spectra is included in a convex polytope: the hypersimplex \cite{gabrielov_combinatorial_1975}\cite{gelfand_combinatorial_1987}\cite[Exercise~4.8.16, see also p.69a]{grunbaum_convex_2003}.
Around 20 years later, during an IBM--Princeton Research Project, Kuhn (that possibly heard of the potential sufficiency of these conditions from Watanabe who did work at IBM around that time) showed that these conditions on the spectra were sufficient to represent an $N$-fermion density operator \cite{kuhn_linear_1960}.
As a consequence, hypersimplices characterize the spectra of $1$-body reduced density matrices.
This intimate relationship between a well-known polytope and spectra of density operators seems to have escaped the attention of geometers since then.
We remark, however, that hypersimplices appear in the recent study of scattering amplitudes via the combinatorics of positive Grassmannians~\cite{arkani_grassmannian_2016,arkani_positive_2021}\cite{lukowski_positive_2020,parisi_amplituhedron_2021}.

\medskip
\textbf{$N$-representability problem.}
During the 1950's, convexity appeared in physics through the \emph{$N$-representability problem}, which asks for a computationally efficient method of characterizing the $M$-body reduced density matrices of an $N$-particle system satisfying the corresponding fermionic or bosonic particle exchange symmetry, for $1\leq M<N$.
Coleman recalls about hearing the problem for the first time: ``at that moment, for the first time as far as I am aware, the powerful modern mathematical theory of convexity appeared \emph{explicitly} in physics'' \cite[Introduction]{coleman_convex_1977}.
Since then, this problem has had a life of its own and motivated computational chemists, quantum physicists, and mathematicians alike to obtain a solution \cite{coleman_sympo_1987,coleman_reduced_2000,cioslowski_many_2000,coleman_reduced_2001,mazziotti_reduced_2007}.
The interest in this problem lies in the conceptual dream of eliminating the wave function of a system from all practical applications, replacing it with the $1$- and $2$-body reduced density matrices, referred to as \emph{Coulson's vision} \cite{coulson_present_1960,coleman_reduced_2001}.
Most significant are the cases $M=1$ and $M=2$, since the particles of common quantum systems interact by $1$- and $2$-particle forces only.
In their report \cite[Section~4, p.48 and p.59]{national_1995}, the American National Research Council described it as one of the most important mathematical challenges in theoretical and computational chemistry:
\begin{quote}
``Because the energy is a linear function of the one- and two-body distribution functions, the variational minimum will lie on the boundary determined by the N-representability conditions. Unfortunately, only an incomplete set of necessary conditions are known, but these are already so complex that further work in this area has been abandoned by chemists.'' (\cite[p.50]{national_1995})
\end{quote}

An $N$-particle density operator $\rho$ is called \emph{pure} if its spectrum is $\{1,0,\dots\}$ and \emph{mixed} otherwise.
A mixed density operator represents an ensemble of states.
The difficulty of the general $N$-representability problem is emphasized by the fact that the $N$-representability problem is QMA-complete while its restriction to the pure case is conceivably harder~\cite{liu_quantum_2007}.
The case $M=1$ was solved for the \emph{unspecified mixed} case by several authors around the 1960's, see \cite{chen_comment_2012} for a historical survey.
In the context of pure states, a few cases with small parameters $N$ and $d$ have been determined, and then the situation stagnated for decades \cite{coleman_structure_1963,borland_conditions_1972,ruskai_2007}.
However, at the turn of the century, the situation changed dramatically due to Klyachko's solution to Weyl's problem \cite{klyachko_stable_1998}. 
Subsequently, a series of outstanding papers finally led to a complete theoretical solution of the pure $1$-body representability problem \cite{berenstein_coadjoint_2000}\cite{klyachko_2006,altunbulak_pauli_2008}.
The solution involves sophisticated tools and results from symplectic geometry, representation theory, and cohomology of flag varieties.
For larger $M$'s, as noted by Coleman \cite{coleman_kummer_2002}, the $N$-representability problem has a complete solution in terms of the Kummer variety from \cite{kummer_repr_1967}, ``which unfortunately is more of theoretical than practical interest''.
Further fundamental work has been done to make the conditions more practical, see \cite{mazziotti_anti_2006,mazziotti_structure_2012,mazziotti_pure_2016} and the book \cite{mazziotti_reduced_2007}.

Besides, the $N$-representability problem is an instance of a much larger class of problem called Quantum Marginal Problems, see e.g. \cite{daftuar_quantum_2005}\cite{christandl_spectra_2006}\cite{walter_multipartite_2014}\cite{schilling_quantum_2015}\cite{maciazek_quantum_2007}.
Indeed, interpreting density matrices as probability density functions, the reduced density matrices are called marginals.
The $N$-representability problem then asks to characterize marginals from a system of indistinguishable particles. 
Quantum marginal problems are more general as they are usually phrased with distinguishable particles and the marginals may be multivariate or even overlapping. 

\medskip
\textbf{No-win scenario for Coulson's vision?}
Unfortunately, although Klyachko's solution to Weyl's problem and the description of pure $1$-body $N$-representability constraints constitute landmark mathematical achievements, they do not fulfill Coulson's vision for replacing the wave function.
For most applications, the knowledge of a few of the lowest eigenenergies of a Hamiltonian is sought \cite{coleman_kummer_2002}.
Furthermore, the American National Research Council also advertised the $N$-representability problem in classical equilibrium statistical  mechanics \cite[p.59]{national_1995} stating: ``For most cases of interest, one focuses on the infinite-system limit, where the container size and~$N$ diverge, while temperature, number density, and container shape are held constant'', see also \cite[Section~2]{ayres_variational_1958} and \cite[p.18]{coleman_sympo_1987}.
Most of the restrictions to the spectra in the solution given by Klyachko depend strongly on the container size (i.e. the dimension of the underlying $1$-particle Hilbert space) and the number of particles~$N$.
For instance, the application of these restrictions is possible for systems of at most five electrons and ten orbitals.
\emph{These facts combined with the sheer number of inequalities involved in the solution and the difficulty involved in obtaining them motivate a reformulation of the problem whose solution could lead to practical large scale applications.}

\medskip
\textbf{Convexity Ansatz.}
The present article draws its motivation from this practical consideration, and the notion of convexity plays a key role:
Convexity is essential in the foundation of reduced density matrix functional theory for excited states \cite{schilling_ensemble_2021}\cite{liebert_foundation_2021}.
Herein, we provide an effective solution to a \emph{convex relaxation} of the constrained 1-body $N$-representability problem on which reduced density matrix functional theory relies, see Figure~\ref{fig:single_flower}.
That is, given a non-increasing sequence of non-negative weights $\bw=(w_1,w_2,\dots)$ summing to $1$, we characterize the convex hull of the set of 1-body reduced density matrices $\overline{\Density}{}^1_N(\bw)$ that stem from $\bw$-ensemble of states (Theorem~\ref{thm:relaxed_gok}).
Furthermore, we give a characterization of the spectra of such matrices as points in a polytope, that we call \emph{fermionic} (or \emph{bosonic}) \emph{spectral polytope}.
First, we describe this polytope as the convex hull of finitely many points (Theorem~\ref{thm:spectral}).
Then, we describe in Algorithm~\ref{algo:lineup}, a procedure generating a minimal description of this polytope using linear inequalities, therefore allowing to solve the membership test efficiently.
The crucial step to handle arbitrarily large number of particles $N$ and dimension~$d$ (the dimension of the $1$-particle Hilbert space) is to restrict the number of non-zero entries in $\bw$ to a fixed value, i.e. the number of states involved in the ensemble.
Subsequently, by invoking a mixture of well-established and state-of-the-art results from convex geometry and combinatorics, we obtain necessary conditions for the $N$-representability problem of $1$-reduced density matrices that are tight and valid for arbitrarily large $N$ and $d$.
In turn, these inequalities contain the necessary essence to apply $1$-particle reduced density matrix functional theory also known as RDMFT.

\medskip
\textbf{Summary of results.}
Using a variational principle of Gross--Oliviera--Kohn \cite{gross_rayleigh_1988}, we compute the support function for $\overline{\Density}{}^1_N(\bw)$ in Theorem~\ref{thm:relaxed_gok}.
This leads to the spectral characterization in Theorem~\ref{thm:spectral}, where $\FerPoly[][\bw,N,d]:=\spec(\overline{\Density}{}^1_N(\bw))$ is shown to be a polytope, the fermionic spectral polytope.
In fact, the polytopality also follows from general theorems on moment maps (see Section~\ref{ssec:lie}).
The polytope $\FerPoly[][\bw,N,d]$ has the following important features.

\begin{compactenum}[(\arabic{enumi})]
\item It is a \textbf{spectral polytope} by definition.
In \cite{sanyal_spectral_2020} Sanyal and Saunderson studied polyhedra arising from spectra of a set of symmetric matrices. 
In contrast, the polytope $\FerPoly[][\bw,N,d]$ describes spectra of Hermitian matrices.
\item It is a \textbf{moment polytope} (see Proposition \ref{prop:moment}) for a moment map on the coadjoint orbit $\Density^N(\bw)$, that is, the set of density operators on the $N$-particle Hilbert space with spectrum $\bw$ (see \cite[Section 1]{kirillov_lectures_2004}).
This connection opens the door to an analysis of its facial structure, as done by Brion in \cite{brion_general_1999}.
\item It is a \textbf{symmetric polytope}.
The polytope $\FerPoly[][\bw,N,d]$ is stabilized by the action of the symmetric group which permutes coordinates. 
Symmetric polytopes whose vertices form a single orbit are objects of intense study; they are known as permutohedra.
In contrast, the polytope $\FerPoly[][\bw,N,d]$ is the convex hull of \emph{multiple} orbits.
We provide an original approach to determine the combinatorial structure of these types of polytopes, see Theorem~\ref{thm:dimension}.
\item It is a \textbf{sweep polytope}. In \cite{PadrolPhilippe2021} Padrol and Philippe introduce sweep polytopes.
Face lattices of sweep polytopes encode total orderings of point configurations coming from linear functionals. \label{res:swe}
As a consequence, the polytope $\FerPoly[][\bw,N,d]$ is also a fiber polytope.
\end{compactenum}

We then introduce an extra parameter $r$ specifying the number of non-zero entries of $\bw$.
This parameter allows to better understand the situation for a fixed number of eigenvalues.
The upshot is the possibility to let the dimension $d$ go to infinity, while $r$ remains fixed for any number of particles $N$.
These polytopes are denoted by $\FerPoly[r][\bw,N,d]$ and they satisfy properties (1),(2), and~(3).
Moreover they satisfy a refined version of property (4) extending the definition of sweep polytopes to a new class, that we name \emph{lineup polytopes}, see Theorem~\ref{thm:main}.
We provide a $V$-representation of lineup polytopes in Theorem~\ref{thm:vrepr_lineup}.
We prove the following additional properties, explaining the influence of the defining parameters.

\begin{compactenum}[(\arabic{enumi})]
\setcounter{enumi}{4}
\item The polytope $\FerPoly[r][\bw,N,d]$ is a \textbf{combinatorial polytope}.
The orbits of vertices of $\FerPoly[r][\bw,N,d]$ correspond to certain shellings of pure threshold complexes of dimension $N-1$ with $r$ facets on at most $d$ vertices (Theorem~\ref{thm:spectral} and Lemma~\ref{lem:Fthreshold}).
Threshold complexes are studied in \cite{klivans_threshold_2007}\cite{edelman_simplicial_2013}.
In \cite{heaton_dual_2020}, Heaton and Samper study these shellings in the more general context of matroid polytopes.
\item For any $r,N$, and $d$, if $\bw$ and $\bw'$ have pairwise distinct entries, $\FerPoly[r][\bw,N,d]$ and $\FerPoly[r][\bw',N,d]$ have the same normal fan (Corollary \ref{cor:dep_w}). \label{res:norm}
When $\bw$ has repeated entries some inequalities may become redundant (Proposition \ref{prop:construction3}).
\item For $r=1$, the polytope $\FerPoly[1][(1),N,d]$ is the $(d-1)$-dimensional hypersimplex $\pol[H](N,d)$. \label{res:hyper}
\item For each $r\geq 1$, we establish a \textbf{Stability Theorem} (Theorem~\ref{thm:stability} for fermions and \ref{thm:stability_bosonic} for bosons) stating that, up to symmetry, when $N>r-1$ no new inequality arises for fermions if $d>N+r-1$ and for bosons if $d\geq r$. \label{res:stab}
This allows the ``container size'' and number of particle~$N$ to diverge as asked by the $N$-representability problem in classical equilibrium statistical mechanics.
\item As $r$ increases, the inequalities for $r-1$ remain facet-defining and new inequalities are added, creating a hierarchy of restrictions (see Proposition \ref{prop:hierarchy}). \label{res:hier}
Consequently, $\FerPoly[r][\bw,N,d]\subseteq \pol[H](N,d)$.
This can be interpreted as a \textbf{hierarchy of generalized exclusion principles}:
They are linear inequalities on the spectra of $1$-reduced density matrices that are necessary for $N$-representability, which are valid as $d\to\infty$.
By property (8), they are essentially independent of the values of $N$ and $d$, extending the linear inequalities interpretation of Pauli's exclusion principle given by $\pol[H](N,d)$.
\item The ``$\mathrm{f}$'' on the notation stands for fermions. \label{res:bos}
There is a parallel story for bosons obtaining a polytope denoted $\BosPoly[r][\bw,N,d]$.
\end{compactenum}

To test membership, an $H$-representation of $\FerPoly[][\bw,N,d]$ of the form $\bA\bx\leq \bb$, where $\bA$ is a matrix and $\bb$ a vector, is essential.
We obtain one in two steps.
First, for a fixed $r$, there is a minimal choice of parameters $(N,d)=(r-1,2r-2)$ for which the $V$- and $H$-representations of $\FerPoly[r][\bw,N,d]$ should be obtained.
Second, using the Stability Theorem~\ref{thm:stability}, it is possible to describe the matrix $\bA$ for arbitrarily large values of $N$ and $d$ as long as $d-N>r-1$.
The analog result for bosons is given in Theorem~\ref{thm:stability_bosonic}.
While the result (6) shows that $\bA$ does not depend on $\bw$, as long as it contains distinct entries, the vector $\bb$ is determined from $\bw$ as explained in Proposition~\ref{prop:rhs}.
The minimal polytope in the above first step is obtained as follows.
In Theorem~\ref{thm:spectral}, the polytope $\FerPoly[][\bw,N,d]$ is represented as the convex hull of finitely many points and we seek a representation using finitely many linear inequalities.
This is a well-studied problem in discrete geometry and many general algorithms exists.
In the present case, we provide the original Algorithm~\ref{algo:lineup} and successfully implemented it in \texttt{SageMath} \cite{sagemath}, see Section~\ref{ssec:recursive}.
This lead to a complete solution for $r\leq 13$ without much computational efforts, and we provide the outcome for $r\leq 8$ in the Appendices.
Furthermore, without resorting to computational tools, we present the complete solution in the case when $r=4$, by carrying out Algorithm \ref{algo:lineup} by hand in Section \ref{ssec:ex4}.
In particular, we showcase a complete inequality descriptions of $\FerPoly[4][\bw,N,d]$ for all $(N,d)$ with $N\geq 3, d-N\geq 3$ and $\FerPoly[5][\bw,N,d]$ for all $(N,d)$ with $N\geq 4, {d-N\geq 4}$.
In the bosonic case, there is no analogue of Pauli's exclusion principle.
Despite that fact, through the result (8) and (10), we uncover such exclusion principles for mixed bosonic states.

The spectra of operators in $\overline{\Density}{}^1_N(\bw)$ retain a physical relevance to study $1$-particle reduced density matrix functional theory for excited states, see Section~\ref{ssec:relaxed} and \cite{schilling_ensemble_2021}\cite{liebert_foundation_2021}.
This description refines and solves a constrained $1$-body $N$-representability problem, where the weights of the ensemble can be specified, compare \cite[Theorem~3]{altunbulak_pauli_2008}.
Indeed, the results allow to progressively decide from which ensemble a certain $1$-body reduced density matrix may come from by adding more and more necessary conditions.
This is possible since the inequalities in the $H$-representation do not intrinsically depend on the dimension of the $1$-particle state space (the dimension of the basic Hilbert space at play) nor the number of particles $N$.
This allows for a novel understanding of the effects of the number of weights and their relative sizes on possible spectra. 
The inequalities obtained are valid necessary conditions for arbitrarily large~$N$ and dimension of the Hilbert spaces, see Theorems~\ref{thm:r4} and~\ref{thm:r5}.
The techniques underlying the procedure to obtain the necessary conditions rely on elementary combinatorial structures and minimally on convex representation translations.
Hence, these results may be used as the stepping stone to further investigate convex relaxation of $N$-representability problems.

In Section \ref{ssec:klyachko}, we compare the inequalities obtained by Klyachko to describe the polytope $\klypo$, the polytope of spectra (ordered decreasingly) of $\Density^1_N(\bw)$, with those obtained here.
The linear inequalities defining $\klypo$ are referred to as generalized Pauli constraints, whereas we refer to the inequalities of $\FerPoly$ as \emph{generalized exclusion principles}.
The present methods have the advantage of simplicity.
They allow to recover the description of the matrix describing the left-hand side of $\klypo$, and interpret it through sweep polytopes (and therefore leading to a combinatorial interpretation through threshold complexes) using less sophisticated machinery.
Furthermore, it seems possible to study further Klyachko's polytope through $\FerPoly[][\bw,N,d]$.
We adventure to conjecture the following.
\begin{conjecture*}
If the entries of $\bw$ are distinct, then the polytope $\klypo$ is a Minkowski summand of the polytope $\FerPoly[][\bw,N,d]$.
\end{conjecture*}

\noindent
Equivalently, we conjecture the existence of some polytope $Q$ such that $\klypo+Q=\FerPoly[][\bw,N,d]$.
This would imply that the normal fan of $\klypo$ is a coarsening of that of $\FerPoly[][\bw,N,d]$.

\medskip
\textbf{The tools.}
The relevant methods to study fermionic spectral polytopes have well-established theories that we gather and bring \emph{au goût du jour} in Part~\ref{part2}.

\emph{Tool 1: Symmetric polytopes.}
The study of symmetric polytopes is more than millennial.
Well-known instances include regular polytopes, uniform polytopes, and polytopes with reflection symmetries, see \cite{schulte_symmetry_1997} for a detailed survey.
It is common to describe symmetries of polytopes via the action of a finite group of suitable transformations of the ambient space.
A myriad of variations exist and many---if not all---seem to have some degree of \emph{transitivity} in common: faces or flags of faces should be contained in a \emph{single} orbit under the group action \cite{babai_symmetry_1977}\cite{sanyal_orbitopes_2011}\cite{matteo_combinatorially_2016}.
The convex hull of the orbit of a point under a finite group action is called an \emph{orbit polytope}.
The uniqueness of the orbit has a structuring effect on the geometry of orbit polytopes, facilitating their study.
These polytopes and their symmetries have been studied extensively, especially via their relations to the many generalizations of the permutohedron, see e.g. \cite{rado_inequality_1952}\cite{onn_geometry_1993}\cite{sanyal_orbitopes_2011}\cite{friese_affine_2016} and the references therein. 
Be that as it may, the extension of orbit polytopes to convex hulls of orbits of more than one points (in other words, convex hull of unions of orbit polytopes) seems to have received far less attention.
In Section~\ref{sec:perm_inv_poly}, we remedy this and initiate a study of the facial incidences of these polytopes.
In particular, we determine the faces of symmetric polytopes (see Theorem~\ref{thm:dimension}) thereby extending the known results for a single orbit.

\emph{Tool 2: Sweep polytopes.}
The polytopes $\FerPoly[][\bw,N,d]$ and $\BosPoly[][\bw,N,d]$ are defined via their normal fans.
Since a lineup of length $r$ is obtained from a sweep by taking its first $r$ entries, we adapt the machinery to construct sweep polytopes from \cite{PadrolPhilippe2021} to construct lineup polytopes.
One distinction between sweep polytopes and lineup polytopes is that the fans of lineup polytopes are not necessarily hyperplane arrangements.
Expanding on \cite{PadrolPhilippe2021}, we present four different constructions for lineup polytopes: as a convex hull of a finite set of points, as a projection of a partial permutahedron, as a Minkowski sum of $k$-set polytopes, and as a monotone path polytope of a truncated zonotope.
The construction as a Minkowski sum of $k$-set polytopes (Proposition \ref{prop:construction3}) gives a conceptual explanation for the role of $\bw$, even when it has repeated entries.
One important consequence of this section is that the facets of the polytope correspond to orderings that do not refine a non-trivial ordering.

\emph{Tool 3: Gale order.}
In addition to being related to sweep polytopes and being symmetric, the facial structure of fermionic and bosonic spectral polytopes is determined via the notion of Gale order, which is important in matroid theory~\cite{gale_optimal_1968}.
It is the combination of these three features that enables the proof of the Stability Theorem~\ref{thm:stability}.
In fact, the vertices of the spectral polytopes are related to certain simplicial complexes called \emph{threshold complexes}.
The typical case is related to fermion configurations, while boson configurations give rise to a different notion of threshold complexes related to multisubsets rather than subsets, see Sections~\ref{sec:gale} and~\ref{ssec:fer_vs_bos}.
These \emph{bosonic} threshold complexes seem original and hence could be further studied.

\medskip
\textbf{Structure of the article.}
Aside from providing a thorough description of the solution to the convex ensemble $1$-body $N$-representability problem, the overarching goal of the present article is to increase interactions between the convex geometry and quantum mechanics communities.

In this regard, Part~\ref{part1} is aimed to introduce the background theory from quantum physics in order to state appropriately the convex problem, which is aimed at readers with a more mathematical background.
In Section~\ref{sec:fund_concepts}, we review the fundamental concepts from quantum physics necessary to define and study the $N$-representability problem.
In Section~\ref{sec:convex_repr}, we state the classical $N$-representability problem and its convex relaxation for fermions and bosons.
Furthermore, we obtain a first description of the sought spectral polytopes using $V$-representations.
In Section~\ref{sec:spec_char}, we give a spectral characterization of the set of $\bw$-ensemble 1-reduced density matrices.
The spectra form a convex polytope that we then study thoroughly using combinatorial and geometrical tools.

Part~\ref{part2} is concerned with a self-contained introduction to the convex geometry tools along with the presentation of new results necessary to solve the problem.
In Section~\ref{sec:V_to_H}, we describe the technique privileged here to translate a $V$-representation to an $H$-representation.
In Section~\ref{sec:perm_inv_poly}, we provide the necessary notions relative to fundamental domains of symmetric polytopes and give a criterion mixing geometry and combinatorics to determine the dimension of their faces.
In this case, it is particularly important to determine when a linear functional determines a facet.
In Section~\ref{sec:lineups}, we adapt the recent study of sweep polytopes to the current situation.
In Section~\ref{sec:gale}, we study the combinatorial construction that dictates much of the symmetry and facial structure of spectral polytopes.
Namely, the Gale order is reviewed and extended to multisubsets to deal with bosonic systems.

Finally, Part~\ref{part3} presents a procedure to obtain the desired inequalities along with a description of several theoretical consequences stemming from the proofs.
In Section~\ref{sec:Hrepr_spectral}, we describe the algorithm, provide some analysis of the first output, and prove the stability properties of the produced inequalities, see Theorems~\ref{thm:stability} and~\ref{thm:stability_bosonic}.
In Section~\ref{sec:outro}, we compare our results to the solution of the pure $1$-body $N$-representability problem obtained in \cite{klyachko_2006,altunbulak_pauli_2008}, we exhibit two key properties of symmetric polytopes, and provide extremal examples that distinguish fermionic and bosonic systems.
Finally, in the appendices, we provide some output produced by the algorithm.

\medskip
\textbf{Acknowledgements.}
The authors are thankful to Pauline Gagnon, Fulvio Gesmundo, Allen Knutson, Fu Liu, Georgoudis Panagiotis, Nicholas Proudfoot, Raman Sanyal and Lauren Williams for valuable discussions.
The authors express their gratitude to Manfred Lehn and G\"unter M. Ziegler for sparking this fruitful collaboration.

\newpage
\tableofcontents

\newpage
\part{The Problem}
\label{part1}

The objective of this part is to introduce the problem motivating the present article.
In Section~\ref{sec:fund_concepts}, we review the fundamental concepts from quantum theory.
In Section~\ref{sec:convex_repr}, we present the various $N$-representability problems that lead to the convex version.
In Section~\ref{sec:spec_char}, we give a spectral characterization of the convex hull of $1$-body reduced density matrices.
Furthermore, we present how this characterization fits into a Lie theory perspective.
Finally, we formulate precisely the challenge within polyhedral combinatorics that must be solved in order to obtain a practical solution.
This challenge is addressed in Part~\ref{part3}, after developing the necessary tools in Part~\ref{part2}.

\medskip
\textbf{Notations and conventions.}
We adopt the following conventions: ${d\in\NN\setminus\!\{0\}}$, $[d]:=\{1,2,\dots,d\}$. 
Let $\RR^d$ be the $d$-dimensional Euclidean space with elementary basis $\{\be_i~:~i\in[d]\}$ and inner product given by $\np{\be_i,\be_j}=\delta_{i,j}$ for $i,j\in[d]$.
Vectors in $\RR^d$ and tuples of numbers are denoted using bold letters such as~$\mathbf{e,v,x}$, and scalars using normal script such as~$x$.
The symmetric group $\Sym{d}$ acts linearly on $\RR^d$ by $\pi\cdot \be_i=\be_{\pi(i)}$ for $i=1,\ldots,d$.
A table is provided on p.\pageref{notations} to gather most of the notations used throughout this article.

\section{Fundamental concepts}
\label{sec:fund_concepts}

Among the abundant literature on quantum theory, we refer here to the books \cite{hall_quantum_2013}\cite{landsman_foundations_2017} for a direct quantum mechanic approach, \cite[Chapter~2]{nielsen_quantum_2000} for an introduction from the quantum information point of view and the article \cite{cassam_hopf_2003} for a recent combinatorial formalism of many-particle spaces, and suggest them to readers who want to learn more about the mathematical aspects of quantum mechanics.

\subsection{Density operators on Hilbert spaces}
\label{ssec:hilbert}

Let $\cH$ be a  $D$-dimensional complex Hilbert space with a positive definite sesquilinear form (also called \emph{Hermitian inner product}) $\np{\cdot,\cdot}:\cH\times \cH \to \mathbb{C}$. 
We identify the vector space of linear endomorphisms on $\cH$, called \emph{operators}, with the tensor product $\cH\otimes \cH^*$; for a unit vector $\bv\in\cH$, the operator $\bv\otimes\bv^*$ is the orthogonal projection onto the subspace spanned by $\bv$.
An operator $H$ is \emph{Hermitian} if it satisfies $\np{H\bx, \by}=\np{\bx, H\by}$ for any $\bx, \by \in \cH$. 
The real vector space of Hermitian operators is denoted by $\Herm(\cH)$ and its inner product by $\np{A,B}=\Trace(AB)$.
Hermitian operators are diagonalizable with real eigenvalues, and may therefore be expressed as
\begin{equation}\label{eq:diagonalizable}
H = \sum_{i=1}^D \lambda_i \bv_i\otimes\bv_i^*,\quad \text{where } \lambda_i\in\RR, \text{ and }\lambda_1\geq \lambda_2\geq\dots\geq\lambda_D, 
\end{equation}
and $\{\bv_i\}_{i\in[D]}$ is an orthonormal basis of $\cH$, see \cite[Ch.~XV, Thm.~6.4]{serge_lang_2002}. 
Each real scalar $\lambda_i$ is the eigenvalue of $H$ with eigenvector $\bv_i$.
We define \defn{$\spec^{\downarrow}(H)$} as the vector $(\lambda_1,\dots,\lambda_D)\in\RR^D$ and $\spec(H)$ as the orbit $\Sym{D}\cdot\spec^{\downarrow}(H)\subset \RR^D$.
An Hermitian operator $H\in \Herm(\cH)$ is positive semidefinite if its eigenvalues are non-negative, equivalently if $\lambda_D\geq 0$.

\begin{definition}[Density operators, $\Density(\cH)$]\label{def:density}
A \defn{density operator} $\rho$ is a positive semidefinite Hermitian operator with $\Trace(\rho)=1$.
The set of density operators on $\cH$ is denoted $\Density(\cH)$.
\end{definition}

\begin{remark}[State vs Density operator]\label{rem:states}
In quantum physics, a \emph{state} on $\Herm(\cH)$ is a linear map $\omega:\Herm(\cH)\to\RR$, such that $\omega(H^2)\geq0$ for each $H\in\Herm(\cH)$ and $\omega(\Id)=1$, see \cite[Definition~2.5]{landsman_foundations_2017}.
When $\cH$ is finite-dimensional, the states $\omega$ on $\Herm(\cH)$ are in bijection with density operators $\rho$ on $\cH$ through the equality $\omega(H)=\np{\rho,H}=\Trace(\rho H)$ \cite[Theorem~2.7]{landsman_foundations_2017}.
The two normalizations $1=\omega(\Id)=\Trace(\rho)=1$ fit together, allowing a duality between states and density operators, through the inner product on $\Herm(\cH)$.
This permits to treat states and density operators interchangeably. 
\defn{Pure state} refers to a density operator $\bv\otimes\bv^*$ for a certain unit vector $\bv \in \cH$. 
The corresponding state is $\omega : H \in \Herm(\cH) \mapsto \np{H\bv, \bv}$. 
Equation~\eqref{eq:diagonalizable} shows that in general, density operators or states are given by a sum of pure states $\bv_i\otimes\bv_i^*$, each weighted with probability $\lambda_i$. 
Indeed, for density operators the eigenvalues are positive and sum to $1$. 
This is the setting of \defn{ensemble of states}.
\end{remark}

\subsection{Many-particle state spaces}\label{ssec:manyparticle}
We denote by $\cH_1$ the \emph{one-particle Hilbert space} of dimension~$d$.
The name refers to the interpretation of states on $\Herm(\cH_1)$ as the probability distributions over orbitals in which a single particle may be found, and the possible orbitals are spanned by $d$ basis orbitals.
The $1$-particle Hilbert space may have some additional substructure corresponding to different degrees of freedom such as orbitals and spin, see e.g.~\cite{cassam_hopf_2003}\cite{altunbulak_pauli_2008}, but we do not consider spin here.
By taking the $N$-th tensor product $\bigotimes_{i=1}^N\cH_1$ we may represent states of $N$-particle systems.
The subspace given by the $N$-th exterior power $\cH_N:=\bigwedge^N\cH_1$ represents states of $N$ fermions, or \emph{$N$-fermion states} making $\cH_N$ the \emph{$N$-fermion Hilbert space}.
Elements in these spaces have the antisymmetry and indistinguishability of fermions built-in.
Furthermore, the antisymmetry enforces Pauli's exclusion principle: two fermions can not occupy the same orbital.

Let $\FerPoset:= \{\bi=(i_1,\dots,i_N)\in [d]^N~:~1\leq i_1<\dots< i_N\leq d\}$ denote the set of all \defn{$N$-fermion configurations}.
If $\cB=\{\bb_i\}_{i=1}^d$ is an orthonormal basis for $\cH_1$, then $\cB^N:=\{\bb_{\bi}\}_{\bi\in\FerPoset }$ is an orthonormal basis for~$\cH_N$, where
\[
\bb_{\bi}:= \bb_{i_1}\wedge \bb_{i_2}\wedge\dots\wedge \bb_{i_N}.
\]
It follows that the dimension of $\cH_N$ is $D:=\binom{d}{N}$.
For the rest of this section and Section~\ref{sec:convex_repr}, we restrict ourselves to $N$-fermion Hilbert spaces.
In Section~\ref{ssec:repr_boson}, we describe how to adapt the discussion to the bosonic case.

\begin{definition}[{Many-particle density operators, $\Density^1$ and $\Density^N$}]
\label{def:many}
The set of density operators in $\Herm(\cH_1)$ and in $\Herm(\cH_N)$ are denoted by $\Density^1$ and $\Density^N$ respectively.
\end{definition}

\begin{remark}
On physical and mathematical grounds, it is also desirable to consider infinite-dimensional one-particle Hilbert spaces.
Kummer provided a rigorous framework for the infinite case in the article \cite{kummer_repr_1967}.
Nevertheless, in the case at hand, the space may be assumed to be finite without losing much physical or mathematical relevance, see e.g. \cite{loewdin_quantum_1955} and \cite{schilling_pinning_2013} and in particular Appendix~B of the latter.
For this reason we always assume that~$\cH_1$ is finite-dimensional, while seeking to minimize the dependence of the results on this dimension.
\end{remark}

\subsection{Variational principle for ensemble states}

The unitary group $\Unitary(\cH)$ acts on $\Herm(\cH)$ by conjugation: $U\cdot H:=U HU^{-1}\in\Herm(\cH)$.
Unitary conjugation is equivalent to a change of orthonormal basis so it does not change the spectrum.
Therefore the set $\Density(\cH)$ is partitioned into $\Unitary(\cH)$-orbits indexed by points in the \defn{Pauli simplex}
\begin{equation}\label{eq:pauli_simplex}
\Pauli_{D-1}:=\left\{\bw\in\RR^D:1\geq w_1\geq w_2\geq \cdots\geq w_D \geq 0,\quad \sum_{i=1}^D w_i=1\right\}.
\end{equation}
of possible spectra (with values ordered decreasingly) of density operators.
For each $\bw\in\Pauli_{D-1}$, let
\[
\Density(\bw):=\{\rho\in\Density(\cH)~:~\spec^\downarrow(\rho)=\bw\}.
\]
Thus, $\Density(\cH)=\bigsqcup_{\bw\in\Pauli_{D-1}}\Density(\bw)$.
We shall use the following variational principle.

\begin{theorem}[{Gross--Oliviera--Kohn \cite[Section II]{gross_rayleigh_1988}}]
\label{thm:gok}
Let $H$ be a Hermitian operator on a Hilbert space $\cH$ with eigenvalues $\lambda_1\geq \lambda_2\geq\dots\geq \lambda_D$ and $\bw \in \Pauli_{D-1}$.
We have
\[
\lambda_{\bw}:=\sum_{i=1}^D w_i \lambda_i=\max_{\rho\in \Density(\bw)} \Trace(\rho H).
\]
A maximizer $\rho_{\bw}$ on the right-hand side is given by $\sum_{i=1}^D w_i \bv_i\otimes \bv_i^*$, where $\bv_i$ is a unit eigenvector of $H$ for the eigenvalue $\lambda_i$.
\end{theorem}

In \cite{gross_rayleigh_1988}, this is stated with minima but to keep our conventions we use maxima.
In quantum mechanics, $\Trace(\rho H)$ gives the expectation value of observable $H$ for a quantum state described by a density operator $\rho$, see e.g. \cite[Sections~19.3-4]{hall_quantum_2013}\cite[Section~IIIE]{cassam_hopf_2003}.
For example if $H$ is the Hamiltonian that describes the energy of the system, then the eigenvalues $\lambda_i$ can be thought of as energy levels.
The value of $\Trace(\rho H)$ gives the expectation value of observable~$H$ for the ensemble of states described by $\rho$.

\section{Convex formalism of the $N$-representability problem}
\label{sec:convex_repr}

For basic notions on convex analysis, we refer the reader to the books \cite{grunbaum_convex_2003}\cite{ziegler_lectures_1995}\cite{rockafellar_convex_1997} and to Section~\ref{sec:V_to_H}.
A subset $\pol[K]$ of $\RR^d$ is \emph{convex} if $\bx,\by\in\pol[K]$ implies $\alpha\bx+(1-\alpha)\by\in\pol[K]$ for any $0\leq\alpha\leq1$.
The \emph{convex hull} $\conv(\pol[K])$ of a set $\pol[K]$ is the intersection of all convex sets containing it.
A \emph{polytope} is the convex hull of a finite set of points.
 The \emph{support function} of a convex set $\pol[K]$ is defined as $\mathrm{supp}_{\pol[K]}(\by):=\max_{\bx\in\pol[K]}\np{\by,\bx}$.
 For any compact subset $\pol[K]$ of $\RR^d$, we have 
 $\conv(\pol[K])=\{\bx\in \RR^d~:~\np{\by,\bx}\leq \mathrm{supp}_{\pol[K]}(\by) \text{ for all } \by\in \RR^d\}$.
 Consequently, a convex set and its support function $\mathrm{supp}_K$ uniquely determine each other.
Using Equation \eqref{eq:diagonalizable} we see that the sets $\Density^1$ and $\Density^N$ of Definition \ref{def:many} are convex in $\Herm(\cH_1)$ and $\Herm(\cH_N)$, respectively.

\subsection{$N$-representability problem}

We may now proceed to phrase the $N$-represen\-ta\-bi\-li\-ty problem.
It is formulated via the important concept of \emph{reduced} density matrices which are defined using the partial trace operation.

\begin{definition}[Partial trace]\label{def:partial_trace}
Let $\cV=\{\bv_i\}_{i=1}^d$ be an orthonormal basis of $\cH_1$ and $M,N$ be positive integers such that $M<N$.
The \defn{partial trace} (sometimes called \emph{contraction}) of a Hermitian operator $H_N$ on $\cH_N$ is the operator $H_M=L^N_M(H_N)$ on $\cH_M$ uniquely defined by the equations
\begin{equation}\label{eq:partial_trace}
\np{\bv_{\bi},H_M(\bv_{\bj})}:=\sum_{\bk\in\FerPoset[N-M][d]} \left\langle\bv_{\bi}\wedge\bv_{\bk},H_N\left(\bv_{\bj}\wedge\bv_{\bk}\right)\right\rangle,
\end{equation}
for all $\bi,\bj\in \FerPoset[M]$.
The partial trace $L^N_M$ induces a real linear map between $\Herm(\cH_N)$ and $\Herm(\cH_M)$.
Equation \eqref{eq:partial_trace} provides the entries of a matrix representation of $H_M$ with respect to the basis $\cB^M$.
\end{definition}

\begin{remark}\label{rem:partial trace}
For a proof that the partial trace is a well-defined linear map see \cite[Theorem 3]{kummer_repr_1967}.
For a more general discussion of the partial trace in quantum information see \cite[Section~2.4.3]{nielsen_quantum_2000}. 
For a Lie theory interpretation see Section~\ref{ssec:lie} on page~\pageref{ssec:lie}.
For combinatorially inclined reader, it is possible to phrase the partial trace completely in Hopf algebraic terms, see e.g. \cite[Section~3]{cassam_hopf_2003}.
\end{remark}

A straightforward computation shows that partial traces preserve positive semi-definiteness and that $\Trace(H_M)=\binom{N}{M}\Trace(H_N)$.
So after a suitable normalization factor, the partial trace of a density operator in $\Density^N$ is again a density operator.

\begin{definition}[Reduced density matrices, $\Density^M_N$]
Let $M,N$ be positive integers such that ${M<N}$.
The operators in $L^N_M(\Density^N)$ are called ($M$-)\defn{reduced density matrices}.
The set of $M$-reduced density matrices on $\cH_M$ is denoted $\Density^M_N$.
\end{definition}

In the literature, an $M$-reduced density matrix is sometimes abbreviated as being an $M$-RDM. 
The set $\Density^1_N$ is convex and compact in the set of scaled density operators $N\Density^1$, see Figure~\ref{fig:partial_trace} for a schematic illustration.
Not every extreme points of $\Density^N$ is extreme in the partial trace image~$\Density^1_N$: the extreme points of $\Density^1_N$ are thickened in Figure~\ref{fig:partial_trace}.

\begin{figure}[!h]
\begin{center}
\begin{tikzpicture}[auto]

\node (step1) at (-4,0) {\begin{tikzpicture}
	[vertex/.style={inner sep=0.25pt,circle,draw=black,fill=black,thick},
	ebound/.style={blue,very thick,fill=blue!20!white},
	baseline=(center)]

\draw[ebound] (0,0) circle (2cm);
\node (center) at (0,0) {};

\node[blue] at (0,0) {\Large$\Density^N$};

\end{tikzpicture}};
\node (step2) at  (4,0) {\begin{tikzpicture}
	[scale=1,
	vertex/.style={inner sep=0.25pt,circle,draw=black,fill=black,thick},
	grassbound/.style={black,line width=2.5pt},
	ebound1/.style={blue,very thick},
	ebound/.style={red,very thick, dashed},
	pint/.style={blue!20!white},
	baseline=(center)]

\node (center) at (0,0) {};

\draw[ebound] (0,0) circle (2cm);
\node (center) at (0,0) {};

\node[red] at (1.8,1.8) {\Large$N\Density^1$};

\def\radius{1.4}
\def\run{1.3}
\def\rdeux{1.2}
\def\aun{17.5}
\def\adeux{22.5}
\def\rend{0.8}

\foreach \a in {0,45,...,360}{%

\coordinate (1) at (\a+10:\radius);
\coordinate (2) at (\a:\radius+0.05);
\coordinate (3) at (\a-10:\radius);
\coordinate (next) at (\a-35:\radius);
\coordinate (prev) at (\a+35:\radius);

\fill[pint] (0,0) -- ($(prev)!0.5!(1)$) -- (1) .. controls (2) .. (3) -- ($(next)!0.5!(3)$) -- (0,0);

\draw[ebound1] ($(prev)!0.5!(1)$) -- (1);
\draw[grassbound] (1) .. controls (2) .. (3);
\draw[ebound1] (3) -- ($(next)!0.5!(3)$);

\node[vertex] at (1) {};
\node[vertex] at (3) {};
\node[vertex] at (next) {};
}

\node[blue] at (center) {$\Density^1_N$};

\end{tikzpicture}};

\draw[->,thick,shorten <=0.25cm,shorten >=0.25cm] (step1) to node[out=180,in=80,swap,text width=3cm,align=center] {partial trace\\ $L^N_1$}  (step2);

\end{tikzpicture}
\end{center}
\caption{Schematic representation of the image of the partial trace map}
\label{fig:partial_trace}
\end{figure}
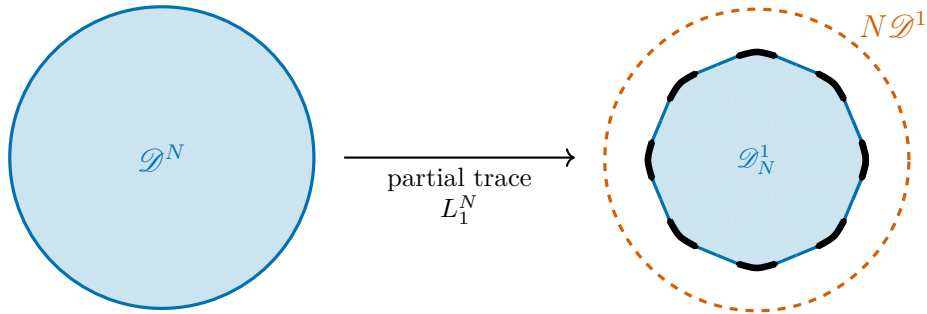

The extreme points of $\Density^1_N$ are indexed by the elements of a Grassmannian \cite[Corollary~9.1B]{coleman_structure_1963}.
Due to dimensional restrictions, Figure~\ref{fig:partial_trace} fails to represents an honest connected Grassmannian.
We may now formulate the representability problem formulated around the 1950's \cite{loewdin_quantum_1955}\cite[Introduction]{coleman_structure_1963}, see also \cite{coleman_necessary_1972} and the references therein.

\begin{nrproblem*}
Let $\rho\in \binom{N}{M}\Density^M$.
Give necessary and sufficient conditions for $\rho$ to be in $\Density^M_N$.
In order words, characterize when $\rho$ is the partial trace of an operator $\rho_N\in\Density^N$.
\end{nrproblem*}

\noindent
Already at the time of its formulation, some necessary conditions were known.
Moreover, it was also noted that the conditions derived from the antisymmetry had deeper consequences on the eigenvalues of reduced density matrices than Pauli's exclusion principle, see e.g. \cite{watanabe_anwendung_1939}\cite{loewdin_quantum_1955}\cite{kuhn_linear_1960}.
The cases $M\in\{1,2\}$ have attracted most attention due to the aspiration in quantum theory to replace the wave function by the $1$- and $2$-reduced density matrices: ``Charles Coulson made clear his expectation that all, or most, of the properties of matter with which chemistry and physics are concerned can be discussed using the 2-matrix as our sole tool with no explicit recourse to the $N$-particle wave function.'' (\cite[Conclusion, p.7]{coleman_kummer_2002}).
Coleman remarked that from the outset, the importance was attributed not to obtaining a solution, but to obtaining simple rules---that are computationally efficient---allowing the application of variational principles \cite{coleman_convex_1977,coleman_kummer_2002}. 
We refer the readers to \cite{mazziotti_reduced_2007,mazziotti_structure_2012,mazziotti_pure_2016}\cite{schilling_generalized_2018} and \cite{bach_orthogonalization_2019} to get more details about recent developments on the general $N$-representability problem.

\subsection{$N$-representability of $1$-reduced matrices}

From now on, we focus on the case $M=1$.
The following lemma describes how unitary orbits behave with respect to the partial trace.

\begin{lemma}[{\cite{coleman_structure_1963}}]\label{lem:unitary}
If $u\in\Unitary(\cH_1)$ and $U=\wedge^N u \in \Unitary(\cH_N)$, then ${L^{N}_1(U\rho U^{-1}) = uL^N_1(\rho)u^{-1}}$.
\end{lemma}

This lemma implies that $\Density^1_N$ is unitary invariant, a property that does not hold for $\Density^M_N$ with $M>1$, see e.g. \cite[Section~5]{coleman_structure_1963} or \cite[Section 1.6]{coleman_reduced_2000}.
Being unitary invariant means that the description of $\Density^1_N$ depends only on the spectrum of the operators involved and the number of particles $N$.
The $N$-representability problem for $M=1$ has the following solution.

\begin{theorem}[$\approx1950/60$]
\label{thm:E1N}
The set $\Density^1_N$ consists of all operators $\rho\in N\Density^1$ that have trace equal to~$N$ and each eigenvalue $\lambda$ satisfies $0\leq \lambda\leq 1$.
\end{theorem}

The above theorem appeared in writing in \cite[Theorem~9.3]{coleman_structure_1963}, see also \cite[Theorem 2.18]{coleman_reduced_2000}.
The necessity of such a bound was already observed in relation with Pauli's exclusion principle or using second quantization, see e.g. \cite[Behauptung I, Equation~(16)]{watanabe_anwendung_1939}\cite[Equation~(75)]{loewdin_quantum_1955}.
Other proofs of this result have been obtained around the same period \cite{kuhn_linear_1960}\cite{yang_concept_1962}\cite{garrod_reduction_1964}, although the results in the latter have been rigorously established in \cite{kummer_repr_1967}.
We refer to the surveys \cite{coleman_convex_1977,coleman_reduced_2001} and the comment \cite{chen_comment_2012} for further details.
Let us highlight two important facts related to this theorem.
First, Theorem \ref{thm:E1N} expresses a bound on the eigenvalues of an $N$-representable density operator in $\Density^1_N$ independently of the dimension $d=\dim(\cH_1)$, which makes it scalable when $d\to\infty$.
Second, since the trace of an $N$-representable density operator is~$N$, then 
\begin{equation}\label{eq:def_hypersimplex}
\spec(\Density^1_N)=\left\{\bx\in\RR^d~:~\sum_{i=1}^d x_i = N,\quad 0\leq x_i\leq 1\text{ for all }i=1,\dots,d\right\}=:\pol[H](N,d),
\end{equation}
which is a polytope $\pol[H](N,d)$ called hypersimplex, see Figure~\ref{fig:ordered_spec} and Examples~\ref{ex/def:hypersimplex} and~\ref{ex:hypersimplex}.

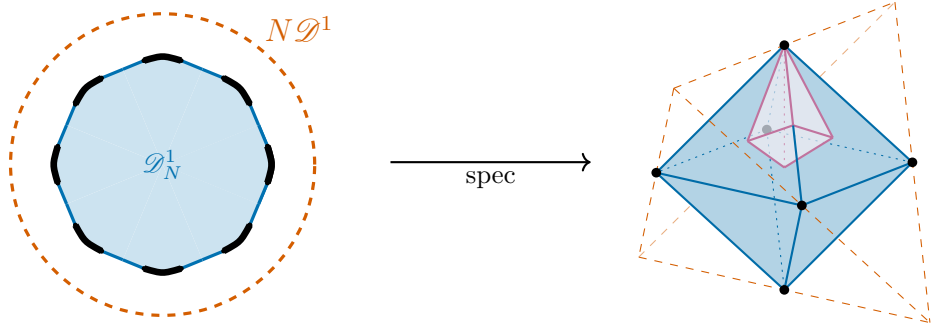
\begin{figure}[!h]
\begin{center}
\begin{tikzpicture}[auto]

\node (step1) at (-4,0) {\begin{tikzpicture}
	[scale=1,
	vertex/.style={inner sep=0.25pt,circle,draw=black,fill=black,thick},
	grassbound/.style={black,line width=2.5pt},
	ebound1/.style={blue,very thick},
	ebound/.style={red,very thick, dashed},
	pint/.style={blue!20!white},
	baseline=(center)]

\node (center) at (0,0) {};

\draw[ebound] (0,0) circle (2cm);
\node (center) at (0,0) {};

\node[red] at (1.8,1.8) {\Large$N\Density^1$};

\def\radius{1.4}
\def\run{1.3}
\def\rdeux{1.2}
\def\aun{17.5}
\def\adeux{22.5}
\def\rend{0.8}

\foreach \a in {0,45,...,360}{%

\coordinate (1) at (\a+10:\radius);
\coordinate (2) at (\a:\radius+0.05);
\coordinate (3) at (\a-10:\radius);
\coordinate (next) at (\a-35:\radius);
\coordinate (prev) at (\a+35:\radius);

\fill[pint] (0,0) -- ($(prev)!0.5!(1)$) -- (1) .. controls (2) .. (3) -- ($(next)!0.5!(3)$) -- (0,0);

\draw[ebound1] ($(prev)!0.5!(1)$) -- (1);
\draw[grassbound] (1) .. controls (2) .. (3);
\draw[ebound1] (3) -- ($(next)!0.5!(3)$);

\node[vertex] at (1) {};
\node[vertex] at (3) {};
\node[vertex] at (next) {};
}

\node[blue] at (center) {$\Density^1_N$};

\end{tikzpicture}};
\node (step2) at  (4,0) {\begin{tikzpicture}%
	[x={(-0.990780cm, -0.040379cm)},
	y={(0.135484cm, -0.295838cm)},
	z={(0.000078cm, 0.954384cm)},
	scale=1.70000,
	back/.style={dotted, thin},
	edge/.style={color=blue!95!black, thick},
	facet/.style={fill=blue!95!black,fill opacity=0.300000},
	edge2/.style={color=purple!95!black, thick},
	facet2/.style={fill=purple!20!white,fill opacity=0.300000},
	edge3/.style={color=red,dashed},
	vertex/.style={inner sep=1pt,circle,draw=black,fill=black,thick}]
%
%

\coordinate (-1.00000, 0.00000, 0.00000) at (-1.00000, 0.00000, 0.00000);
\coordinate (0.00000, -1.00000, 0.00000) at (0.00000, -1.00000, 0.00000);
\coordinate (0.00000, 0.00000, -1.00000) at (0.00000, 0.00000, -1.00000);
\coordinate (0.00000, 0.00000, 1.00000) at (0.00000, 0.00000, 1.00000);
\coordinate (0.00000, 1.00000, 0.00000) at (0.00000, 1.00000, 0.00000);
\coordinate (1.00000, 0.00000, 0.00000) at (1.00000, 0.00000, 0.00000);
\draw[edge3,opacity=0.5] (1.00000, -1.00000, -1.00000) -- (-1.00000, -1.00000, 1.00000);
\fill[white,opacity=1] (0.00000, 1.00000, 0.00000) -- (-1.00000, 0.00000, 0.00000) -- (0.00000, 0.00000, 1.00000) -- cycle {};
\fill[white,opacity=1] (0.00000, 1.00000, 0.00000) -- (-1.00000, 0.00000, 0.00000) -- (0.00000, 0.00000, -1.00000) -- cycle {};
\fill[white,opacity=1] (1.00000, 0.00000, 0.00000) -- (0.00000, 0.00000, -1.00000) -- (0.00000, 1.00000, 0.00000) -- cycle {};
\fill[white,opacity=1] (1.00000, 0.00000, 0.00000) -- (0.00000, 0.00000, 1.00000) -- (0.00000, 1.00000, 0.00000) -- cycle {};

\draw[edge,back] (-1.00000, 0.00000, 0.00000) -- (0.00000, -1.00000, 0.00000);
\draw[edge,back] (0.00000, -1.00000, 0.00000) -- (0.00000, 0.00000, -1.00000);
\draw[edge,back] (0.00000, -1.00000, 0.00000) -- (0.00000, 0.00000, 1.00000);
\draw[edge,back] (0.00000, -1.00000, 0.00000) -- (1.00000, 0.00000, 0.00000);
\node[vertex] at (0.00000, -1.00000, 0.00000)     {};
\fill[facet] (0.00000, 1.00000, 0.00000) -- (-1.00000, 0.00000, 0.00000) -- (0.00000, 0.00000, 1.00000) -- cycle {};
\fill[facet] (0.00000, 1.00000, 0.00000) -- (-1.00000, 0.00000, 0.00000) -- (0.00000, 0.00000, -1.00000) -- cycle {};
\fill[facet] (1.00000, 0.00000, 0.00000) -- (0.00000, 0.00000, -1.00000) -- (0.00000, 1.00000, 0.00000) -- cycle {};
\fill[facet] (1.00000, 0.00000, 0.00000) -- (0.00000, 0.00000, 1.00000) -- (0.00000, 1.00000, 0.00000) -- cycle {};
\fill[white,opacity=0.5] (0.33333, 0.33333, 0.33333) -- (0.00000, 0.00000, 0.00000) -- (-0.33333, 0.33333, 0.33333) -- (0.00000, 0.50000, 0.50000) -- cycle {};
\fill[white,opacity=0.5] (0.00000, 0.50000, 0.50000) -- (-0.33333, 0.33333, 0.33333) -- (0.00000, 0.00000, 1.00000) -- cycle {};
\fill[white,opacity=0.5] (0.33333, 0.33333, 0.33333) -- (0.00000, 0.00000, 1.00000) -- (0.00000, 0.50000, 0.50000) -- cycle {};
\draw[edge] (-1.00000, 0.00000, 0.00000) -- (0.00000, 0.00000, -1.00000);
\draw[edge] (-1.00000, 0.00000, 0.00000) -- (0.00000, 0.00000, 1.00000);
\draw[edge] (-1.00000, 0.00000, 0.00000) -- (0.00000, 1.00000, 0.00000);
\draw[edge] (0.00000, 0.00000, -1.00000) -- (0.00000, 1.00000, 0.00000);
\draw[edge] (0.00000, 0.00000, -1.00000) -- (1.00000, 0.00000, 0.00000);
\draw[edge] (0.00000, 0.00000, 1.00000) -- (1.00000, 0.00000, 0.00000);
\draw[edge] (0.00000, 1.00000, 0.00000) -- (1.00000, 0.00000, 0.00000);

\coordinate (-0.33333, 0.33333, 0.33333) at (-0.33333, 0.33333, 0.33333);
\coordinate (0.00000, 0.00000, 0.00000) at (0.00000, 0.00000, 0.00000);
\coordinate (0.00000, 0.00000, 1.00000) at (0.00000, 0.00000, 1.00000);
\coordinate (0.00000, 0.50000, 0.50000) at (0.00000, 0.50000, 0.50000);
\coordinate (0.33333, 0.33333, 0.33333) at (0.33333, 0.33333, 0.33333);
\draw[edge2,back] (0.00000, 0.00000, 0.00000) -- (0.00000, 0.00000, 1.00000);
\fill[facet2] (0.33333, 0.33333, 0.33333) -- (0.00000, 0.00000, 0.00000) -- (-0.33333, 0.33333, 0.33333) -- (0.00000, 0.50000, 0.50000) -- cycle {};
\fill[facet2] (0.00000, 0.50000, 0.50000) -- (-0.33333, 0.33333, 0.33333) -- (0.00000, 0.00000, 1.00000) -- cycle {};
\fill[facet2] (0.33333, 0.33333, 0.33333) -- (0.00000, 0.00000, 1.00000) -- (0.00000, 0.50000, 0.50000) -- cycle {};
\draw[edge2] (-0.33333, 0.33333, 0.33333) -- (0.00000, 0.00000, 0.00000);
\draw[edge2] (-0.33333, 0.33333, 0.33333) -- (0.00000, 0.00000, 1.00000);
\draw[edge2] (-0.33333, 0.33333, 0.33333) -- (0.00000, 0.50000, 0.50000);
\draw[edge2] (0.00000, 0.00000, 0.00000) -- (0.33333, 0.33333, 0.33333);
\draw[edge2] (0.00000, 0.00000, 1.00000) -- (0.33333, 0.33333, 0.33333);
\draw[edge2] (0.00000, 0.50000, 0.50000) -- (0.33333, 0.33333, 0.33333);
\draw[edge] (0.00000, 0.50000, 0.50000) -- (0.00000, 1.00000, 0.00000);
\draw[edge2] (0.00000, 0.00000, 1.00000) -- (0.00000, 0.50000, 0.50000);

\coordinate (1.00000, 1.00000, 1.00000) at (1.00000, 1.00000, 1.00000);
\coordinate (1.00000, -1.00000, -1.00000) at (1.00000, -1.00000, -1.00000);
\coordinate (-1.00000, -1.00000, 1.00000) at (-1.00000, -1.00000, 1.00000);
\coordinate (-1.00000, 1.00000, -1.00000) at (-1.00000, 1.00000, -1.00000);
\draw[edge3] (1.00000, 1.00000, 1.00000) -- (1.00000, -1.00000, -1.00000);
\draw[edge3] (1.00000, 1.00000, 1.00000) -- (-1.00000, -1.00000, 1.00000);
\draw[edge3] (1.00000, 1.00000, 1.00000) -- (-1.00000, 1.00000, -1.00000);
\draw[edge3] (1.00000, -1.00000, -1.00000) -- (-1.00000, 1.00000, -1.00000);
\draw[edge3] (-1.00000, -1.00000, 1.00000) -- (-1.00000, 1.00000, -1.00000);

\node[vertex] at (-1.00000, 0.00000, 0.00000)     {};
\node[vertex] at (0.00000, 0.00000, -1.00000)     {};
\node[vertex] at (0.00000, 0.00000, 1.00000)     {};
\node[vertex] at (0.00000, 1.00000, 0.00000)     {};
\node[vertex] at (1.00000, 0.00000, 0.00000)     {};

\end{tikzpicture}};

\draw[->,thick,shorten <=0.5cm,shorten >=0.5cm] (step1) to node[out=180,in=80,swap,text width=3cm,align=center] {$\spec$}  (step2);

\end{tikzpicture}
\end{center}

\caption{On the right, the dashed tetrahedron (a simplex in general) represents spectra of density operators in $\spec\left(N\Density^1\right)$. 
Only the spectra in the octahedron (a hypersimplex in general) may occur as a spectra of a $1$-reduced density matrix in $\spec\left(\Density^1_N\right)$.
The smaller polytope included in the octahedron represents the vectors in $\spec^{\downarrow}\left(\Density^1_N\right)$.
The dotted vertices of the hypersimplex represent all orderings of the spectra $(1,\dots,1,0,\dots)$ with $N$ occurrences of~$1$.}
\label{fig:ordered_spec}
\end{figure}

The hypersimplex has been used to describe these spectra in \cite{watanabe_anwendung_1939}\cite{loewdin_quantum_1955}\cite{kuhn_linear_1960}\cite{coleman_structure_1963}, and this result seems to have been already folklore for physicists when the term hypersimplex was coined in the articles \cite[Section~1.6]{gabrielov_combinatorial_1975} and later studied in \cite{gelfand_combinatorial_1987}, see e.g. \cite[Example~0.11]{ziegler_lectures_1995}.
Coleman proved that the extreme points of $\Density^1_N$ are indexed by the elements of the corresponding Grassmannian \cite[Corollary~9.1B]{coleman_structure_1963}.
Recently, Lukowski, Parisi, Sherman-Bennett and Williams found explicit---yet not completely understood connections---between scattering amplitudes (the \emph{amplituhedron}) and the hypersimplex, see \cite{lukowski_positive_2020,parisi_amplituhedron_2021}.

\subsection{Ensemble $N$-representability}
\label{ssec:ens_nrepr}

We seek a refined version of Theorem \ref{thm:E1N}, where we consider $\bw$-ensemble states, i.e density operators $\rho\in\Density^N$ with a prescribed spectrum $\bw\in\Pauli_{D-1}$. 
The vector $\bw=(w_1, \ldots, w_D)$ is called the \emph{weight vector} of the ensemble. 
When $\bw=(1,0,\dots)$, the ensemble is called \emph{pure}.
The set of \emph{$\bw$-ensemble $N$-representable 1-reduced density matrices} is 
\begin{equation}\label{eq:w-reduced}
\Density^1_N(\bw):=L^N_1(\Density^N(\bw)),
\end{equation}
which is also easily seen to be unitary invariant by Lemma \ref{lem:unitary}, see Figure~\ref{fig:reduced_ensemble}.

\begin{figure}[!h]
\begin{center}
\begin{tikzpicture}[auto]

\node (step1) at (-4,0) {\begin{tikzpicture}
	[vertex/.style={inner sep=0.25pt,circle,draw=black,fill=black,thick},
	ebound/.style={blue,very thick,fill=blue!20!white},
	pint/.style={blue!20!white},
	baseline=(center)]

\draw[ebound] (0,0) circle (2cm);
\node (center) at (0,0) {};

\node[blue] at (60:2.33) {$\Density^N$};

\draw[decorate, decoration={snake, segment length=2.4mm, amplitude=0.75mm},forestgreen, thick] (0,-0.5) circle (0.75cm);

\node[forestgreen] at (1.2,0.45) {$\Density^N(\bw)$};

\end{tikzpicture}};
\node (step2) at  (4,0) {\begin{tikzpicture}
	[scale=0.7,
	vertex/.style={inner sep=0.25pt,circle,draw=black,fill=black,thick},
	grassbound/.style={black,very thick},
	pint/.style={blue!20!white},
	wint/.style={white},
	wbound/.style={forestgreen,thick},
	ebound1/.style={blue,very thick},
	baseline=(center)]

\node (center) at (0,0) {};

\node[blue] at (2.6,2.6) {\Large$\Density^1_N$};

\def\radiusout{2.75}
\def\radius{2}
\def\run{1.75}
\def\rdeux{1.5}
\def\aun{17.5}
\def\adeux{22.5}
\def\rend{0.8}

\foreach \a in {0,45,...,360}{%

\coordinate (1out) at (\a+10:\radiusout);
\coordinate (2out) at (\a:\radiusout+0.05);
\coordinate (3out) at (\a-10:\radiusout);
\coordinate (nextout) at (\a-35:\radiusout);
\coordinate (prevout) at (\a+35:\radiusout);

\fill[pint] (0,0) -- ($(prevout)!0.5!(1out)$) -- (1out) .. controls (2out) .. (3out) -- ($(nextout)!0.5!(3out)$) -- (0,0);

\draw[ebound1] ($(prevout)!0.5!(1out)$) -- (1out);
\draw[ebound1] (1out) .. controls (2out) .. (3out);
\draw[ebound1] (3out) -- ($(nextout)!0.5!(3out)$);


\coordinate (1) at (\a+10:\radius);
\coordinate (2) at (\a:\radius+0.05);
\coordinate (3) at (\a-10:\radius);
\coordinate (next) at (\a-35:\radius);
\coordinate (prev) at (\a+35:\radius);

\coordinate (c1) at (\a+\aun:\run);
\coordinate (c2) at (\a+\adeux:\rdeux);
\coordinate (d1) at (\a-\aun:\run);
\coordinate (d2) at (\a-\adeux:\rdeux);
\coordinate (endc) at (\a+22.5:\rend);
\coordinate (endd) at (\a-22.5:\rend);


\draw[wbound] (1) .. controls (c1) and (c2).. (endc);
\draw[wbound] (3) .. controls (d1) and (d2).. (endd);
\draw[wbound] (1) .. controls (2) .. (3);

}

\node[forestgreen] at (60:2.25) {$\Density^1_N(\bw)$};

\end{tikzpicture}};

\draw[->,thick,shorten <=0.25cm,shorten >=0.25cm] (step1) to node[out=180,in=80,swap,text width=3cm,align=center] {partial trace\\ $L^N_1$}  (step2);

\end{tikzpicture}
\end{center}
\caption{Schematic representation of the non-convex set of density operators $\Density^N(\bw)$ (represented as a snake curve) with spectrum $\bw$.
It reduces to a non-convex set (a rosette curve) $\Density^1_N(\bw)$ of $1$-reduced density matrices.}
\label{fig:reduced_ensemble}
\end{figure}
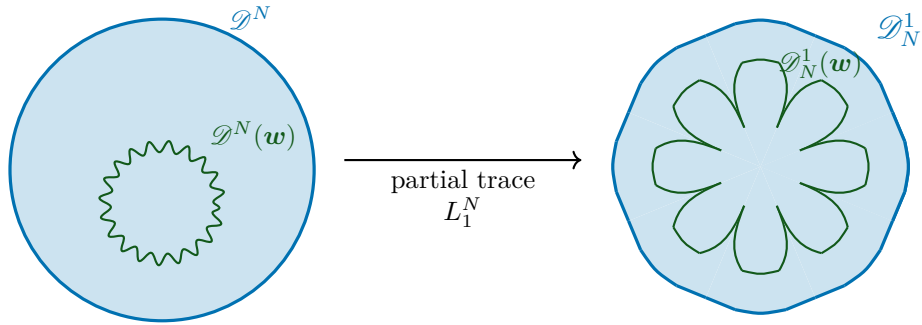

\begin{enrproblem*}
Let $\bw\in\Pauli_{D-1}$ and $\rho\in \Density^1$.
Give necessary and sufficient conditions for $\rho$ to be in $\Density^1_N(\bw)$.
In order words, characterize when $\rho$ is the partial trace of an operator $\rho_N\in\Density^N(\bw)$.
\end{enrproblem*}

In a series of ground-breaking articles \cite{klyachko_stable_1998,klyachko_2006,altunbulak_pauli_2008,klyachko_pauli_2009}, Klyachko et al. gave a spectral characterization of $\Density^1_N(\bw)$.
Let
\begin{equation}
\Lambda(\bw,N,d):=\spec(\Density^1_N(\bw)) \label{eq:rosette_klypo} \hspace{1cm}\text{ and }\hspace{1cm} \klypo:=\spec^{\downarrow}(\Density^1_N(\bw)).
\end{equation}
Klyachko gave a description of $\klypo$ as a convex set defined by \emph{finitely many} inequalities, i.e. a polytope. 
We briefly discuss how his approach compares to ours in Section \ref{ssec:lie} and Section~\ref{ssec:klyachko}.
The set $\Lambda(\bw,N,d)$ is by definition the symmetrization of $\klypo$ and it is not convex (see Figure~\ref{fig:single_flower} on page~\pageref{fig:single_flower}), it is sometimes referred to as a \emph{rosette}, see Figure~\ref{fig:reduced_ensemble}.
Furthermore, in \cite{klyachko_2006} they provided a procedure using representation theory to obtain all inequalities for $\bw=(1,0,\dots)$.
The general inequalities are described in \cite[Theorem~4.16]{klyachko_2006} and \cite[Theorem~2]{altunbulak_pauli_2008}.
In Equation~\eqref{eq:w-reduced}, the parameter $d$ does not appear:
Recall that in Theorem~\ref{thm:E1N}, the parameter $d$ is not directly involved in the restrictions on the eigenvalues.
However, in the solution given by Klyachko, there is a strong dependence on the value of~$d$, see for instance the complete inequalities for the pure case when $N=3$ and $d=6,7,8$ \cite[pages~83-84]{klyachko_2006}.
Unfortunately, if $d>10$ these inequalities are out of reach for practical applications in quantum chemistry.
In Section~\ref{ssec:klyachko}, we compare these inequalities to those obtained in this article.

\subsection{Relaxed $N$-representability}
\label{ssec:relaxed}

The set $\spec(\Density^1_N(\bw))$ is generally not convex, see Figure~\ref{fig:reduced_ensemble}.
We refer the reader to Section~\ref{ssec:lie} where this set is described from a Lie theory perspective.
Characterizing $\spec^{\downarrow}(\Density^1_N(\bw))$ calls for a description of a \emph{petal} whose orbit forms the rosette.
Facing the practical difficulties involved in the solution of the $N$-representability of $1$-reduced density matrices for ensemble states, we now consider a convex relaxation of the ensemble $1$-body $N$-representability problem.

We briefly would like to emphasize the great physical significance of this convex relaxation, particularly also relative to its non-relaxed version. 
For this, we first observe that Coulson’s vision to replace the $N$-particle wave function by the $2$-particle reduced density matrix can even be taken to another level: Since each subfield of physics restricts typically to a \emph{fixed} pair interaction $W$ (e.g., Coulomb interaction in quantum chemistry and contact-interaction in the field of ultracold gases), the class of relevant Hamiltonians is merely parameterized by the one-particle Hamiltonian~$h$, $H_W(h)=h+W$.  
Exploiting this in the context of the Rayleigh--Ritz variational principle applied to $H_W(h)$ leads directly to a universal $1$-reduced density matrix functional $\mathcal{F}_W$ \cite{Levy79,Lieb83}.
Minimizing then $\mathcal{F}_W+\Trace[h \cdot]$ over the non-convex space $\Density^1_N((1,0,\dots))$ of pure $N$-representable $1$-reduced density matrices would yield the energy and $1$-reduced density matrix of the ground state of $H_W(h)$, \emph{for any} $h$. 
Unfortunately, the task of describing the set $\Density^1_N((1,0,\dots))$ for realistic system sizes is a hopeless endeavor. 
The corresponding $1$-reduced density matrix functional theory (RDMFT) became a feasible method only after Valone \cite{V80} understood that these fundamental obstacles can be avoided if the Rayleigh--Ritz variational principle is applied in the context of ensemble rather than pure states.
This relaxes the functional $\mathcal{F}_W$ to its lower convex envelop, and more importantly its domain from $\Density^1_N((1,0,\dots))$ to its easy-to-describe convex hull $\overline{\Density}{}^1_N((1,0,\dots))$. 
By resorting to Gross--Oliviera--Kohn variational principle in Theorem~\ref{thm:gok}, we have recently generalized RDMFT to excited states where $\bw$ is used to fix the weights of carefully chosen low-lying excited states~\cite{schilling_ensemble_2021}\cite{liebert_foundation_2021}. 
Application of exact convex relaxation turns $\bw$-ensemble RDMFT into a practically feasible method, provided a compact halfspace-representation of $\overline{\Density}{}^1_N(\bw)$ is found, which is nothing else than the ambition of the present work.

\begin{framed}\label{pb:relaxedNrepr}
\noindent
\textbf{Convex $1$-body Ensemble $N$-representability Problem}\hfill\\
Let $\bw\in\Pauli_{D-1}$ and $\rho\in N\Density^1$. 
Give necessary and sufficient conditions for $\rho$ to belong to the convex set
\[
\overline{\Density}{}^1_N(d,\bw):=\conv\left\{L^N_1(\tau)~:~\tau\in\Density^N(\bw)\right\},
\]
where the dimension of the one-particle Hilbert space is $\dim(\cH_1)=d$.
\end{framed}

The convex hull $\overline{\Density}{}^1_N(d,\bw)$ is taken inside $N\Density^1$, see Figure~\ref{fig:relaxed_ensemble} for an illustration.
We added a ``$d$'' in the notation here to emphasize the \emph{a priori} dependence on $d$.

\begin{figure}[!h]
\begin{center}
\begin{tikzpicture}[auto]

\node (step1) at (-4,0) {\begin{tikzpicture}
	[scale=0.7,
	vertex/.style={inner sep=0.25pt,circle,draw=black,fill=black,thick},
	grassbound/.style={black,very thick},
	convbound/.style={red,thick},
	pint/.style={white},
	wint/.style={red!30!white},
	convint/.style={red!30!white},
	wbound/.style={forestgreen,very thick},
	ebound1/.style={blue,very thick,dashed},
	baseline=(center)]

\node (center) at (0,0) {};

\node[blue] at (2.6,2.6) {\Large$\Density^1_N$};

\def\radiusout{3}
\def\radius{2}
\def\run{1.75}
\def\rdeux{1.5}
\def\aun{17.5}
\def\adeux{22.5}
\def\rend{0.8}

\foreach \a in {0,45,...,360}{%

\coordinate (1out) at (\a+10:\radiusout);
\coordinate (2out) at (\a:\radiusout+0.05);
\coordinate (3out) at (\a-10:\radiusout);
\coordinate (nextout) at (\a-35:\radiusout);
\coordinate (prevout) at (\a+35:\radiusout);

\fill[pint] (0,0) -- ($(prevout)!0.5!(1out)$) -- (1out) .. controls (2out) .. (3out) -- ($(nextout)!0.5!(3out)$) -- (0,0);

\draw[ebound1] ($(prevout)!0.5!(1out)$) -- (1out);
\draw[ebound1] (1out) .. controls (2out) .. (3out);
\draw[ebound1] (3out) -- ($(nextout)!0.5!(3out)$);


\coordinate (1) at (\a+10:\radius);
\coordinate (2) at (\a:\radius+0.05);
\coordinate (3) at (\a-10:\radius);
\coordinate (next) at (\a-35:\radius);
\coordinate (prev) at (\a+35:\radius);

\coordinate (c1) at (\a+\aun:\run);
\coordinate (c2) at (\a+\adeux:\rdeux);
\coordinate (d1) at (\a-\aun:\run);
\coordinate (d2) at (\a-\adeux:\rdeux);
\coordinate (endc) at (\a+22.5:\rend);
\coordinate (endd) at (\a-22.5:\rend);

\filldraw[wint] (0,0) -- (endc) -- (1) -- (3) -- (endd) -- (0,0);
\filldraw[wint] (1) .. controls (2) .. (3);
\filldraw[wint] (1) .. controls (c1) and (c2).. (endc);
\filldraw[wint] (3) .. controls (d1) and (d2).. (endd);
\filldraw[convint] (3) .. controls (d1) and (d2).. (endd) -- ($(next)!0.5!(3)$) -- cycle;
\filldraw[convint] (1) .. controls (c1) and (c2).. (endc) -- ($(prev)!0.5!(1)$) -- cycle;

\draw[wbound] (1) .. controls (c1) and (c2).. (endc);
\draw[wbound] (3) .. controls (d1) and (d2).. (endd);
\draw[convbound] (1) .. controls (2) .. (3);
\draw[convbound] (1) -- (prev);
\draw[convbound] (3) -- (next);

}

\node[red] at (90:\radius+0.5) {$\overline{\Density}{}^1_N(\bw)$};

\end{tikzpicture}};
\node (step2) at  (4,0) {\input{ordered_spec423}};

\draw[->,thick,shorten <=0.25cm,shorten >=0.25cm] (step1) to node[out=180,in=80,swap,text width=3cm,align=center] {$\spec$}  (step2);

\end{tikzpicture}
\end{center}
\caption{Schematic representation of the set $\spec\left(\overline{\Density}{}^1_N(\bw)\right)$ (equivalently the convex hull of $\spec(\Density^1_N(\bw))$, see Proposition~\ref{prop:rotation}).
It is a symmetric polytope contained in the convex hull of $\spec(\Density^1_N)=\pol[H](N,d)$ (an octahedron $\pol[H](2,4)$ here).
On the right, the illustration of $\spec\left(\overline{\Density}{}^1_2(\frac{1}{2},\frac{1}{3},\frac{1}{6},0,0,0)\right)$ with $d=4$.
The polytope whose facets are shaded is $\spec^{\downarrow}\left(\Density^1_2(\frac{1}{2},\frac{1}{3},\frac{1}{6},0,0,0)\right)$, obtained from \cite[Section~4.2.3]{klyachko_2006}.}
\label{fig:relaxed_ensemble}
\end{figure}

By the linearity of the partial trace, the properties of convex hull and the conjugation action of $\Unitary$, the set $\overline{\Density}{}^1_N(d,\bw)$ is unitary invariant.
Given two weight vectors $\bw,\bw'\in\Pauli_{D-1}$, we say that $\bw$ \emph{majorizes} $\bw'$ \emph{weakly} if and only if, for all $k=1,2,\dots,D$
\[
\sum_{i=1}^kw'_i \leq \sum_{i=1}^kw_i,
\]
and write $\bw'\prec\bw$.
The following structural result describing the orbits in $\overline{\Density}{}^1_N(d,\bw)$ allows the emergence of a geometric and combinatorial approach to the convex ensemble $N$-representability problem.

\begin{theorem}[{see \cite[Theorem~12]{liebert_foundation_2021}}]
Let $d\geq N\geq 1$ and $\bw\in\Pauli_{D-1}$.
The convex set $\overline{\Density}{}^1_N(d,\bw)$ is the union of majorized orbits:
\[
\overline{\Density}{}^1_N(d,\bw)=\bigcup_{\bw'\prec \bw}\Density^1_N(d,\bw').
\]
\end{theorem}

The theorem indicates that the smaller the weight vector within the majorization order, the smaller the spectral convex body should be.

\begin{example}
Let $N=2$, $d=4$, $\bw_1=(\frac{1}{2},\frac{1}{3},\frac{1}{6},0,0,0)$ and $\bw_2=(\frac{2}{5},\frac{3}{10},\frac{1}{5},\frac{1}{10},0,0)$.
Therefore $\bw_2\prec\bw_1$.
The theorem implies that $\Density{}^1_N(d,\bw_2)\subseteq\overline{\Density}{}^1_N(d,\bw_1)$, see Figure~\ref{fig:majorization}.
\end{example}

\begin{figure}[!h]
\begin{center}
\begin{tikzpicture}[auto]

\node (step1) at (-4,0) {\begin{tikzpicture}
	[scale=0.75,
	vertex/.style={inner sep=0.25pt,circle,draw=black,fill=black,thick},
	grassbound/.style={black,very thick},
	pint/.style={fill=white},
	wint/.style={fill=red!20!white},
	wbound/.style={black,thick},
	wbound2/.style={red,thick},
	ebound1/.style={blue,very thick,dashed},
	baseline=(center)]

\node (center) at (0,0) {};

\node[blue] at (2.6,2.6) {\Large$\Density^1_N$};

\def\radiusout{2.75}
\def\radius{2}
\def\radiusdeux{1.5}
\def\run{1.75}
\def\rdeux{1.5}
\def\rtrois{1.2}
\def\rquatre{0.8}
\def\aun{17.5}
\def\adeux{22.5}
\def\rend{0.8}
\def\rendtwo{0.4}

\foreach \a in {0,45,...,360}{%

\coordinate (1out) at (\a+10:\radiusout);
\coordinate (2out) at (\a:\radiusout+0.05);
\coordinate (3out) at (\a-10:\radiusout);
\coordinate (nextout) at (\a-35:\radiusout);
\coordinate (prevout) at (\a+35:\radiusout);

\fill[pint] (0,0) -- ($(prevout)!0.5!(1out)$) -- (1out) .. controls (2out) .. (3out) -- ($(nextout)!0.5!(3out)$) -- (0,0);

\draw[ebound1] ($(prevout)!0.5!(1out)$) -- (1out);
\draw[ebound1] (1out) .. controls (2out) .. (3out);
\draw[ebound1] (3out) -- ($(nextout)!0.5!(3out)$);

\coordinate (1) at (\a+10:\radius);
\coordinate (2) at (\a:\radius+0.05);
\coordinate (3) at (\a-10:\radius);
\coordinate (next) at (\a-35:\radius);
\coordinate (prev) at (\a+35:\radius);

\coordinate (c1) at (\a+\aun:\run);
\coordinate (c2) at (\a+\adeux:\rdeux);
\coordinate (d1) at (\a-\aun:\run);
\coordinate (d2) at (\a-\adeux:\rdeux);
\coordinate (endc) at (\a+22.5:\rend);
\coordinate (endd) at (\a-22.5:\rend);

\coordinate (1p) at (\a+10:\radiusdeux);
\coordinate (2p) at (\a:\radiusdeux+0.05);
\coordinate (3p) at (\a-10:\radiusdeux);
\coordinate (nextp) at (\a-35:\radiusdeux);
\coordinate (prevp) at (\a+35:\radiusdeux);

\coordinate (c1p) at (\a+\aun:\rtrois);
\coordinate (c2p) at (\a+\adeux:\rquatre);
\coordinate (d1p) at (\a-\aun:\rtrois);
\coordinate (d2p) at (\a-\adeux:\rquatre);
\coordinate (endcp) at (\a+22.5:\rendtwo);
\coordinate (enddp) at (\a-22.5:\rendtwo);


\draw[wbound] (1) .. controls (c1) and (c2).. (endc);
\draw[wbound] (3) .. controls (d1) and (d2).. (endd);
\draw[wbound] (1) .. controls (2) .. (3);


\draw[wbound2] (1p) .. controls (c1p) and (c2p).. (endcp);
\draw[wbound2] (3p) .. controls (d1p) and (d2p).. (enddp);
\draw[wbound2] (1p) .. controls (2p) .. (3p);

}

\node[red] at (0,0) {\footnotesize$\Density^1_N(\bw_2)$};
\node[black] at (90:2.35) {$\Density^1_N(\bw_1)$};

\end{tikzpicture}};
\node (step2) at  (4,0) {\input{majorize}};

\draw[->,thick,shorten <=0.25cm,shorten >=0.25cm] (step1) to node[out=180,in=80,swap,text width=3cm,align=center] {$\spec$}  (step2);

\end{tikzpicture}
\end{center}
\caption{On the left, the two curves represent $\Density^1_N(\bw_1)$ and $\Density^1_N(\bw_2)$ such that $\bw_2\prec\bw_1$, and therefore $\overline{\Density}{}^1_N(\bw_1)$ contains $\Density^1_N(\bw_2)$. 
On the right, the images $\spec\left(\overline{\Density}{}^1_2(\frac{1}{2},\frac{1}{3},\frac{1}{6},0,0,0)\right)$ (the solid skeleton) and $\spec\left(\overline{\Density}{}^1_N(\frac{2}{5},\frac{3}{10},\frac{1}{5},\frac{1}{10},0,0)\right)$ (the polytope whose facets are shaded) for $d=4$.
Both of them are contained in the hypersimplex $\spec(\Density^1_N)=\spec(\Density^1_N(\bw_3))$, with $\bw_3=(1,0,0,0,0,0)$.}
\label{fig:majorization}
\end{figure}

Thus this suggests the following approach: start with the largest weight vector (i.e. $(1,0,\dots)$ and $r_1=1$) and proceed to solve the convex ensemble $N$-representability problem iteratively increasing the length of the weight vector to $r_2=r_1+1$.
This indicates a \emph{hierarchy} between weight vectors to be revealed later.
Remarkably, the role of $d$ in this approach becomes secondary.
Furthermore, as long as the weights are kept distinct, their specific values become less relevant while majorization becomes more important.

The convex ensemble $N$-representability problem is an example of the classical \emph{membership problem} in convex geometry:
Given a convex set $\pol[K]$, determine whether an element belongs to~$\pol[K]$.
This question can be answered very efficiently when $\pol[K]$ is described using linear inequalities, say using a non-redundant list of supporting hyperplanes, this shall be done in Part~\ref{part3}.
In the next section, we give a spectral characterization of $\overline{\Density}{}^1_N(d,\bw)$ as the convex hull of finitely many points.

\section{Spectral characterization}
\label{sec:spec_char}

In this section, as a first step to solve the convex ensemble $N$-representability problem we obtain an internal representation of $\overline{\Density}{}^1_N(\bw)$.
The idea behind this approach is based on lifting the objects from $\cH_1$ to $\cH_N$ and then using Gross--Oliviera--Kohn's Theorem~\ref{thm:gok}.

\subsection{Expansion and convex relaxation}
\label{ssec:expansion}

\begin{definition}[Expansion of a Hermitian operator]
\label{def:lift}
Let $h\in\Herm(\cH_1)$ be a Hermitian operator, $\{\lambda_i\}_{i\in[d]}$ its eigenvalues, and $\{\bv_i\}_{i\in[d]}$ the corresponding orthonormal eigenbasis of $\cH_1$.
We can lift ${h=\sum_{i=1}^d \lambda_i \bv_i\otimes\bv_i^*}$ to a Hermitian operator $\Gamma^1_N(h)\in\Herm(\cH_N)$ that acts diagonally on the basis $\{\bv_{\bi}\}_{\bi\in\FerPoset}$ as follows
\begin{equation}\label{eq:lift}
\Gamma^1_N(h):=\sum_{\bi\in\FerPoset}\left(\sum_{i\in\bi}\lambda_i\right)\bv_{\bi}\otimes\bv_{\bi}^*.
\end{equation}
The operator $\Gamma^1_N(h)$ is called the \defn{expansion} of $h$.
The expansion induces a real linear map between $\Herm(\cH_1)$ and $\Herm(\cH_N)$.
\end{definition}

Alternatively, this expansion can be constructed using the second quantization and the annihilator/creator operators \cite[Section 1.4]{coleman_reduced_2000}.
The linearity of the expansion operator can be seen as a consequence of the next lemma which shows that the expansion operation may be defined as the adjoint (see e.g. \cite[Section~IIB]{groetsch_functional_2003} for a definition of adjoint) of the partial trace.

\begin{lemma}[{\cite[Theorem~4]{kummer_repr_1967}}]
\label{lem:lift}
Let $N\geq 1$, $h\in\Herm(\cH_1)$ and $G\in\Herm(\cH_N)$.
The expansion operator is the adjoint of the partial trace, i.e.
\[
\np{L^N_1(G),h}_1=\np{G,\Gamma^1_N(h)}_N,
\]
where $\np{A,B}_1$ and $\np{A,B}_N$ are the inner products on Hermitian operators defined as $\Trace(AB)$.
\end{lemma}

\begin{proof}
This follows from the definitions of partial trace and expansion and computing them with respect to an eigenbasis of $h$.
\end{proof}

To use Theorem \ref{thm:gok}, we introduce certain vectors that turn out to be key in the following sections.
Let $\bw\in\Pauli_{D-1}$.
Furthermore, let $h\in\Herm(\cH_1)$, $\cV=\{\bv_i\}_{i=1}^d$ be an eigenbasis, $\bm{\lambda}:=(\lambda_1,\lambda_2,\dots,\lambda_d)$ be the corresponding vector of eigenvalues, and $H$ be the corresponding expansion to $\cH_N$. 
For $\bi\in\FerPoset$, let $\bchi(\bi):=\sum_{j\in\bi} \be_j \in \RR^d$.
The spectrum of $H$ is determined by using $\{\lambda_i\}_{i=1}^d$ and Equation \eqref{eq:lift}: for each index $\bi\in\FerPoset$ we have the eigenvalue $\lambda(\bi):=\sum_{i\in \bi} \lambda_i$.
For a given $\bm{\lambda}\in\RR^d$ we say that a linear ordering of $\FerPoset$ is \emph{compatible with $\bm{\lambda}$} if  $\lambda(\bi_1)\geq \lambda(\bi_2)\geq \cdots \geq \lambda(\bi_D)$.
Therefore, the expansion operation and a choice of a compatible linear order for $\bm{\lambda}$ lead to a function $\ell:\RR^d\to\Sym{\bchi(\FerPoset)}$, such that $\ell(\bm{\lambda})=(\bchi(\bi_j))_{j=1}^D$ is an ordered point configuration.
This allows to associate an \defn{occupation vector} to each $\ell(\bm{\lambda})$\label{eq:NON}:
\[
\OccVec_{\bw}(\ell(\bm{\lambda})):=\sum_{j=1}^D w_j \bchi(\bi_j)\in\RR^d.
\]
Let $\Lineups(\bw):= \{\OccVec_{\bw}(\ell(\bm{\lambda}))~:~\bm{\lambda}\in\RR^d\}$ be the set of occupation vectors for all $\bm{\lambda}\in\RR^d$ and all choices of compatible orders.
Observe that this set is finite for any choice of $\bw\in\Pauli_{D-1}$.

\begin{mainthm}\label{thm:relaxed_gok}
Let $N>1$, $\bw\in\Pauli_{D-1}$, and $\fO(\cH_1)$ be the set of all orthonormal bases of $\cH_1$.
The convex set $\overline{\Density}{}^1_N(\bw)$ satisfies 
\[
\overline{\Density}{}^1_N(\bw)=
\conv\left(\sum_{i=1}^d o_i\bv_i\otimes \bv_i^*~:~\OccVec\in\Lineups(\bw),~ \{\bv_i\}_{i=1}^d \in \fO(\cH_1)\right)\subset\Herm(\cH_1).
\]
\end{mainthm}

\begin{proof}
We show that both convex sets have the same support function.
Let $h\in\Herm(\cH_1)$, $\cV=\{\bv_i\}_{i=1}^d$ be an eigenbasis, $\bm{\lambda}:=(\lambda_1,\lambda_2,\dots,\lambda_d)$ be the corresponding vector of eigenvalues, and $H$ be its expansion to $\cH_N$. 
By Theorem~\ref{thm:gok} and Lemma~\ref{lem:lift}, we have
\[
\textrm{supp}_{\overline{\Density}{}^1_N(\bw)}(h)=
\max_{\rho\in\overline{\Density}{}^1_N(\bw)} \np{h,\rho}_1=\max_{\rho\in\Density^1_N(\bw)} \np{h,\rho}_1=\max_{\rho\in\Density^N(\bw)} \np{H,\rho}_N= \sum_{j\in[D]}w_j\lambda(\bi_j),
\]
where $\bi_1,\bi_2,\dots,\bi_D$ is the permutation of the $N$-fermion configurations chosen for $\bm{\lambda}$ and the subscripts on the inner products indicate that we take trace in $\cH_1$ and $\cH_N$ respectively.
By the second part of Theorem \ref{thm:gok}, a maximizer in $\cH_N$ is 
\begin{equation}\label{eq:maximizer}
\sum_{j\in[D]} w_{j} \bv_{\bi_j}\otimes \bv_{\bi_j}^*\in\Density^N(\bw).
\end{equation}
By Equation~\eqref{eq:partial_trace}, using the basis~$\cV$, this operator reduces to
\begin{equation}\label{eq:maximizer_reduced}
\sum_{i=1}^d (\OccVec_{\bw}(\ell(\bm{\lambda})))_i\bv_i\otimes \bv_i^* \in \overline{\Density}{}^1_N(\bw).
\end{equation}
It follows that the sets $\left\{\sum_{i=1}^d o_i\bv_i\otimes \bv_i^*~:~\OccVec\in\Lineups(\bw),~ \{\bv_i\}_{i=1}^d \in \fO(\cH_1)\right\}\subseteq\overline{\Density}{}^1_N(\bw)$ have the same support function hence the sought equality follows.
\end{proof}

\begin{remark}
Let $\rho\in L^N_1(\Density^N)$ be a $1$-reduced density matrix.
The eigenvectors $\{\bv_i\}_{i=1}^d$ of $\rho$ are referred to as the \emph{natural orbitals} of $\rho$.
Similarly, its eigenvalues $\{n_i\}_{i=1}^d$ such that $n_1\geq n_2\geq \cdots \geq n_d$ and $\sum_{i=1}^dn_i=N$ are its \emph{occupation numbers}.
The term occupation number refers to the fact that they represent the average number of particles in each one of the natural orbitals.
Recall from Theorem \ref{thm:E1N} that $0\leq n_i\leq 1$ for $1\leq i\leq d$.
For more details on natural orbitals and occupation numbers, we refer the reader to \cite{schilling_pinning_2013} and the original article \cite[Section~4]{loewdin_quantum_1955}.
\end{remark}

\subsection{Fermionic spectral polytope}

We can interpret Theorem \ref{thm:relaxed_gok} as taking a single polytope and moving it according to the conjugation action of the unitary group and taking the convex hull of the orbit.
To be more precise, we fix an orthonormal basis $\cV=\{\bv_i\}_{i=1}^d$ of $\cH_1$ and define
\[
\pol[P](\cV):=\conv\left(\sum_{i=1}^d o_i\bv_i\otimes \bv_i^*~:~\text{where }\OccVec\in\Lineups(\bw) \right).
\]
It is a polytope in $\Herm(\cH_1)$ since the set $\Lineups(\bw)$ is finite.

\begin{lemma}\label{lem:aux_origin}
Let $\lambda\Id$ be a scalar multiple of the identity operator in $\cH_1$ and assume $\lambda\Id\in \overline{\Density}{}^1_N(\bw)$.
Then $\lambda=d^{-1}$ and $d^{-1}\Id\in \pol[P](\cV)$ for every orthonormal basis $\cV$.
\end{lemma}

\begin{proof}
Since density operators have trace equal to 1 by definition, we must have $\lambda=d^{-1}$. 
For the second part, since the set $\Lineups(\bw)$ is symmetric-invariant, the average of its elements must be too.
Therefore, for every orthonormal basis $\cV=\{\bv_i\}_{i\in[d]}$ of $\cH_1$, we have
\[
\pol(\cV)\ni\frac{1}{|\Lineups(\bw)|}\sum_{\OccVec \in\Lineups(\bw)} \left( \sum_{i=1}^d o_i\bv_i\otimes \bv_i^* \right)  = d^{-1}\sum_{i=1}^d \bv_i\otimes \bv_i^* = d^{-1} \Id.\qedhere
\]
\end{proof}

\begin{proposition}\label{prop:cover}
The convex set $\overline{\Density}{}^1_N(\bw)$ is equal to the union of all polytopes in the orbit of~$\pol[P](\cV)$ under the conjugacy action of the unitary group. 
In other words,
\[
\overline{\Density}{}^1_N(\bw) = \Unitary(\cH_1)\cdot\pol[P](\cV) = \bigcup_{u\in \Unitary(\cH_1)} u\pol[P](\cV)u^{-1}=\bigcup_{\cV\in\fO(\cH_1)}\pol[P](\cV).
\]
\end{proposition}

\begin{proof}
We argue by contradiction.
Suppose that $\rho\in \overline{\Density}{}^1_N(\bw)$ and $\rho\notin u\pol[P](\cV)u^{-1}$ for every $u\in\Unitary(\cH_1)$.
Then the same is true for its orbit $\Unitary(\cH_1)\cdot \rho=\{u\rho u^{-1}~:~u\in\Unitary(\cH_1)\}$.
Even more, the same is true for the average of the orbit:
\[
\rho'=\int_{\Unitary(\cH_1)} u\rho u^{-1} \mathrm{d}u \in \conv(\Unitary(\cH_1)\cdot \rho)\subset \left(\overline{\Density}{}^1_N(\bw)\mathbin{\Big\backslash}  \bigcup_{u\in \Unitary(\cH_1)} u\pol[P](\cV)u^{-1}\right),
\]
where the integral is taken with the unitary invariant Haar measure \cite[Section 4]{serre_1977} on $\Unitary(\cH_1)$.
Since $\rho'$ is the average, it is unitary invariant.
So $\rho'$ must be a scalar multiple of the identity operator.
This contradicts Lemma~\ref{lem:aux_origin}, thus $\rho$ has to belong to $u\pol(\cV)u^{-1}$, for some $u\in\Unitary(\cH_1)$.
\end{proof}

Finally, we define the central geometric object of our investigation.

\begin{definition}[Fermionic spectral polytope]
\label{def:ferm_specpo}
Let $\bw\in\Pauli_{D-1}$.
The polytope
\[
\FerPoly[][\bw,N,d]:= \conv(\Lineups(\bw))
\]
is the \defn{fermionic spectral polytope} of the $N$-fermion Hilbert space on $d$ orbitals with weight~$\bw$. 
\end{definition}

The following theorem provides a spectral characterization of $\overline{\Density}{}^1_N(\bw)$.
\begin{mainthm}\label{thm:spectral}
	Let $d\geq N\geq 1$, $\bw\in\Pauli_{D-1}$.
	A Hermitian operator $\rho\in\Herm(\cH_1)$ is in $\overline{\Density}{}^1_N(\bw)$ if and only if $\spec(\rho)\subset\FerPoly[][\bw,N,d]$.
	Equivalently, we have $\FerPoly[][\bw,N,d]=\spec\left(\overline{\Density}{}^1_N(\bw)\right)$.
\end{mainthm}

\begin{proof}
	By Proposition~\ref{prop:cover}, $\rho$ must be in $\pol[P](\cV)$ for some orthonormal basis $\cV$.
	Then $\rho$ is a diagonal matrix on the basis $\cV$ and its vector $(x_1,\dots,x_d)\in\RR^d$ of diagonal entries must be in $\FerPoly[][\bw,N,d]$.
	Since the elements in its diagonal correspond to its spectrum in some order, we are done.
\end{proof}

Using a basis, a finite Hilbert space $\cH$ may be identified with $\CC^D$, and operators with square matrices.
In this case, the set $\Density(\bw)$ is the set of Hermitian \emph{matrices} with spectrum $\bw$ (as a multiset).
We have a linear map $\diag:\textrm{Mat}_{D\times D}(\CC)\to\CC^D$ which takes any matrix $[m_{ij}]_{i,j=1}^D$ to its vector of diagonal entries $[m_{ii}]_{i=1}^D$.
Since Hermitian matrices have real entries in their diagonal, restricting the diagonal map to Hermitian matrices leads to a linear map from $\Herm(\cH)$ to $\RR^D$.
The Schur--Horn theorem \cite[Chapter 9B]{marshall_inequalities_2011} states that $\diag(\Density(\bw))=\Perm(\bw) :=\conv\{\pi\cdot\bw~:~\pi\in\Sym{D}\}\subset \RR^D$.
Using this theorem, we can reinterpret the polytope $\FerPoly[][\bw,N,d]$.

\begin{proposition}\label{prop:rotation}
Let $\cV$ be a basis of $\cH_1$, and consider operators in $\cH_1$ as matrices written in that basis.
The fermionic spectral polytope is the convex hull of the rosette of $\Lambda(\bw,N,d)$ (defined in Equation \eqref{eq:rosette_klypo}).
Furthermore
\[
\FerPoly[][\bw,N,d]=\diag\left(\overline{\Density}{}^1_N(\bw)\right)=\diag(\Density^1_N(\bw)).
\]
\end{proposition}
\begin{proof}
The first part follows from the proof of Theorem \ref{thm:relaxed_gok}:
the maximizers in Equation \eqref{eq:maximizer} are clearly in $\Density^N(\bw)$ so their partial traces given by Equation~\eqref{eq:maximizer_reduced} are in $\Density^1_N(\bw)$.
It follows that $\FerPoly[][\bw,N,d]\subseteq\conv(\Lambda(\bw,N,d))$.
The reverse inclusion follows from the fact that $\overline{\Density}^1_N(\bw)\supset \Density^1_N(\bw)$ and $\Lambda(\bw,N,d)=\spec(\Density^1_N(\bw))$.
By the Schur--Horn theorem, we have that
\[
\begin{split}
\diag\left(\overline{\Density}{}^1_N(\bw)\right)&=\conv(\spec(\overline{\Density}{}^1_N(\bw))=\FerPoly[][\bw,N,d], \text{ and}\\
\diag\left(\Density{}^1_N(\bw)\right)&=\conv(\spec(\Density{}^1_N(\bw))=\conv(\Lambda(\bw,N,d))=\FerPoly[][\bw,N,d].
\end{split}
\]
The first row follows by Theorem \ref{thm:spectral}.
The second row follows by the first part of the proposition.
\end{proof}

By symmetrizing, we obtain an object which is less complex than $\klypo$, a petal of the rosette. 
In Figure \ref{fig:single_flower}, we illustrate the difference between $\klypo$ and $\FerPoly[][\bw,N,d]$.
In the figure, we also highlight the polytope $\pol[\Sigma]^{\downarrow}(\bw,N,d):=\FerPoly[][\bw,N,d]\cap \Pauli_{D-1}$ which is a closer approximation of $\klypo$.

\begin{figure}[ht]
\begin{tikzpicture}
	[scale=0.8, rotate=300,
	vertex/.style={inner sep=0.25pt,circle,draw=black,fill=black,thick},
	pbound/.style={yellow,thick},
	ebound/.style={red,thick},
	pint/.style={yellow!20!white},
	eint/.style={red!20!white},
	kly/.style={draw=forestgreen,thick,fill=forestgreen!30!white},
	ourint/.style={blue!50!white},
	ourbound/.style={blue,thick},
	baseline=(center)]

\node (center) at (0,0) {};

\def\radius{3}
\def\run{2.75}
\def\rdeux{2.5}
\def\aun{17.5}
\def\adeux{22.5}
\def\rend{1}

\foreach \a in {0,60,...,360}{%

\coordinate (1) at (\a+10:\radius);
\coordinate (2) at (\a:\radius);
\coordinate (3) at (\a-10:\radius);
\coordinate (next) at (\a-50:\radius);
\coordinate (prev) at (\a+50:\radius);

\coordinate (c1) at (\a+\aun:\run);
\coordinate (c2) at (\a+\adeux:\rdeux);
\coordinate (d1) at (\a-\aun:\run);
\coordinate (d2) at (\a-\adeux:\rdeux);
\coordinate (endc) at (\a+30:\rend);
\coordinate (endd) at (\a-30:\rend);

\fill[pint] (0,0) -- (endc) -- (1) -- (2) -- (3) -- (endd) -- (0,0);
\fill[pint] (1) .. controls (c1) and (c2).. (endc);
\fill[pint] (3) .. controls (d1) and (d2).. (endd);
\fill[eint] (1) .. controls (c1) and (c2).. (endc) -- ($(prev)!0.5!(1)$) -- (1);
\fill[eint] (3) .. controls (d1) and (d2).. (endd) -- ($(next)!0.5!(3)$) -- (3);

\draw[pbound] (1) .. controls (c1) and (c2).. (endc);
\draw[pbound] (3) .. controls (d1) and (d2).. (endd);
\draw[ebound] ($(prev)!0.5!(1)$) -- (1) -- (2) -- (3) -- ($(next)!0.5!(3)$);

\node[vertex] at (1) {};
\node[vertex] at (2) {};
\node[vertex] at (3) {};
\node[vertex] at (next) {};
}

\coordinate (v-3) at (90:\rend);
\coordinate (v-1) at (70:3);
\coordinate (v0) at (110:3);
\coordinate (v-2) at ($(v0)!0.5!(v-1)$);
\coordinate (v1) at (120:3);
\coordinate (v2) at (130:3);
\coordinate (v3) at (170:3);
\coordinate (v4) at ($(v2)!0.5!(v3)$);
\coordinate (v5) at (150:\rend);

\coordinate (k1) at (120+\aun:\run);
\coordinate (k2) at (120+\adeux:\rdeux);
\coordinate (k3) at (120-\aun:\run);
\coordinate (k4) at (120-\adeux:\rdeux);

\fill[ourint] (0,0) -- (v4) -- (v2) -- (v1) -- (v0) -- (v-2) -- (0,0);
\fill[white] (0,0) -- (v-3) .. controls (k4) and (k3) .. (v0) -- (v1) -- (v2) .. controls (k1) and (k2) .. (v5) -- (0,0);
\fill[kly] (0,0) -- (v-3) .. controls (k4) and (k3) .. (v0) -- (v1) -- (v2) .. controls (k1) and (k2) .. (v5) -- (0,0);
\draw[ourbound] (0,0) -- (v-3) -- (v-2)-- (v0) -- (v1) -- (v2) -- (v4) -- (0,0);


\node[vertex] at (v1) {};
\node[vertex] at (v2) {};
\node[vertex] at (v0) {};

\node[red] at (30:\radius+1) {$\FerPoly[][\bw,N,d]$};
\node[forestgreen,rotate=60] at (120:1.75) {$\klypo$};
\node[blue] at (120:\radius+0.6) {$\pol[\Sigma]^{\downarrow}(\bw,N,d)$};

\end{tikzpicture}
\caption{Schematic comparison between $\klypo$, $\FerPoly[][\bw,N,d]$ and $\pol[\Sigma]^{\downarrow}(\bw,N,d)$, to be compared with \cite[Figure 3]{berenstein_coadjoint_2000}.}
\label{fig:single_flower}
\end{figure}
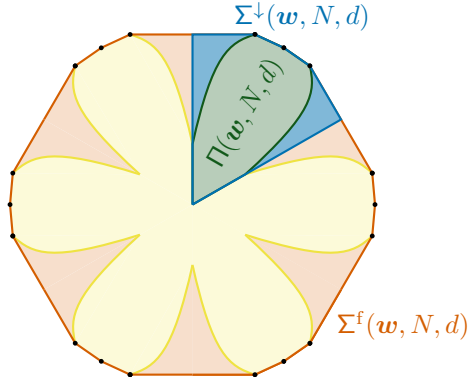

The Schur--Horn theorem can be seen as an example of a more general convexity result involving moment maps.
Proposition \ref{prop:rotation} allows us to think of $\FerPoly[][\bw,N,d]$ as a moment polytope, as explained in Section~\ref{ssec:lie}.

\subsection{Lie theory perspective}
\label{ssec:lie}
In this section, we outline the Lie theory point of view that Klyachko \cite{klyachko_pauli_2009,klyachko_2006} and Berenstein--Sjaamar \cite{berenstein_coadjoint_2000} employ.
For any finite Hilbert space~$\cH$, the unitary group $\Unitary(\cH)$ is a compact Lie group with Lie algebra $\mathfrak{u}(\cH)$ consisting of skew-Hermitian operators.
We can identify the set $\Herm(\cH)$ of Hermitian operators with its dual Lie algebra $\mathfrak{u}^*(\cH)$ via the pairing $H\mapsto \Trace(iH\cdot)$.
Therefore the conjugation action of $\Unitary(\cH)$ in $\Herm(\cH)$ is the coadjoint action.
The set $\Density(\bw)\subset \mathfrak{u}^*(\cH)$ of density operators with spectrum given by $\bw$ is then a coadjoint orbit.

We refer the reader to \cite{knutson_symplectic_2000} for a nice survey providing a definition of moment map.
Here we only use the basic ingredients.
Let $M$ be a connected symplectic manifold together with an action of a compact Lie group $\textrm{G}$.
A moment map is a $\textrm{G}$ equivariant map $\Phi:M\to\mathfrak{g}^*$ with one extra property involving its relation to the symplectic structure of $M$.
Here are two important convexity results about moment maps.
\begin{compactitem}
\item When $\mathrm{G}$ is a torus, the convexity theorems of Atiyah \cite{atiyah_convexity_1982} and Guillemin--Sternberg \cite{guillemin_convexity_1982} state that the image $\Phi(M)$ is a polytope.
The case relevant to us, when $M$ is a coadjoint orbit, was proved earlier by Kostant \cite{kostant_convexity_1973}.
\item When $\mathrm{G}$ is a general compact group, there is a different convexity result of Kirwan \cite{kirwan_convexity_1984}, resolving a conjecture of Guillemin--Sternberg from \cite{guillemin_convexity_1982}, stating that the image restricted to the positive Weyl chamber (for a choice of maximal torus) is a polytope.
\end{compactitem}
For our purposes we interpret Kirwan's theorem as follows. Let $\mathrm{G}=\Unitary(\cH)$ be a unitary group then the image of the map $\Phi:M\to \mathfrak{u}^*(\cH)$ composed with the map $\spec^{\downarrow}:\mathfrak{u}^*(\cH)\to\RR^d$ is a polytope.

\begin{example}[Schur--Horn's theorem]
\label{ex:schur-horn}
There is a canonical symplectic structure on $\Density(\bw)$ that makes the inclusion into $\mathfrak{u}^*(\cH)$ a moment map  for the action of $\Unitary(\cH)$, see \cite[Chapter 1.2]{kirillov_lectures_2004} for details.
Let $\Torus(\cH)\subset\Unitary(\cH)$ be some maximal torus which consist of diagonal matrices in a basis $\cV$.
The basis $\cV$ allows us to consider all operators as matrices.
Composing the inclusion $i:\Density(\bw)\hookrightarrow\mathfrak{u}^*(\cH)$ with the projection to the diagonal, we get a moment map $\diag:\Density(\bw)\to\RR^D$ for the $\Torus(\cH)$-action on $\Density(\bw)$.
In this context, the first convexity theorem implies Schur--Horn's theorem.
\end{example}

Let $g:\Unitary(\cH_1)\to \Unitary(\cH_N)$ be the group homomorphism defined by $g(u)=u\wedge u\wedge\dots\wedge u$ for every unitary $u\in\Unitary(\cH_1)$.
The map $g$ allows to identify $\Unitary(\cH_1)$ with a connected subgroup of~$\Unitary(\cH_N)$.
The induced map on cotangent spaces $g^*:\mathfrak{u}^*(\cH_N)\to\mathfrak{u}^*(\cH_1)$ is the partial trace of Equation~\eqref{eq:partial_trace} and the restriction $g^*:\Density^N(\bw)\to\Density^1_N(\bw)\subset \mathfrak{u}^*(\cH_1)$ is a moment map for the action of $\Unitary(\cH_1)$ on $\Density^N(\bw)$.
Kirwan's theorem guarantees a priori that the set $\spec^{\downarrow}(\Density^1_N(\bw))$ is a polytope.
Berenstein and Sjaamar \cite{berenstein_coadjoint_2000} give an inequality representation for such polytopes in the general setup where~$g$ is any inclusion of compact Lie groups and $\Density(\bw)$ is any coadjoint orbit.
In \cite{ressayre_geometric_2010}, Ressayre describes a minimal subset of the inequalities from \cite{berenstein_coadjoint_2000} which suffices to describe the polytope.
In \cite{klyachko_2006}, Klyachko builds upon the results in \cite{berenstein_coadjoint_2000} (which themselves build upon previous work of Klyachko \cite{klyachko_stable_1998} on Weyl's problem) to describe inequalities in the concrete case of $\spec^{\downarrow}(\Density^1_N(\bw))$.
Now we contextualize the fermionic spectral polytope within this framework.

\begin{proposition}
\label{prop:moment}
Fix a basis of $\cH_1$ and let $\Torus$ be the maximal torus of diagonal matrices.
The composition $\diag\circ g^*:\Density^N(\bw)\to \FerPoly[][\bw,N,d]$ is a moment map for the action of $\Torus$.
\begin{center}
\begin{tikzcd}[column sep=small]
& \mathfrak{u}^*(\cH_1) \arrow[dr, "\diag"] & \\
\Density^N(\bw) \arrow[ur,"g^*"] \arrow[rr,"\Phi"] & & \FerPoly[][\bw,N,d]
\end{tikzcd}
\end{center}
\end{proposition}

\begin{proof}
Through the map $g$, $\Torus$ acts on the coadjoint orbit $\Density^N(\bw)$.
For the chosen basis of $\cH_1$, let $\diag: \mathfrak{u}^*(\cH_1)\to \mathfrak{t}^*(\cH_1)$ be the projection of the matrix to its diagonal.
The map $g^*$ is a moment map for the action of $\Unitary(\cH_1)$ on $\Density^N(\bw)$. 
By basic properties of the moment map \cite[Section~2.2]{knutson_symplectic_2000}, a moment map $\Phi$ for this torus action on $\Density^N(\bw)$ is given by composing $g^*:\Density^N(\bw)\to\mathfrak{u}^*(\cH_1)$ with $\diag$.
The image $g^*(\Density^N(\bw))$ is $\Density^1_N(\bw)$ and by Proposition \ref{prop:rotation} the image of the composition is $\FerPoly[][\bw,N,d]$.
\end{proof}

The Kostant convexity theorem not only states that the image of the moment map is a polytope, it also states that the image is the convex hull of the images of $\Torus$-fixed points.
Using this idea combined with Proposition \ref{prop:moment} leads to an alternative derivation of the convex hull representation of $\FerPoly[][\bw,N,d]$.

\begin{proposition}\label{prop:kostant}
The polytope $\FerPoly[][\bw,N,d]$ is the convex hull 
\begin{equation}\label{eq:kostant}
\FerPoly[][\bw,N,d]=\conv\left(\left\{\sum_{j=1}^D w_j \bchi(\bi_j)~:~\text{for any total ordering }\bi_1,\dots,\bi_D\text{ of } \FerPoset\right\}\right).
\end{equation}
\end{proposition}

\begin{proof}
Fix a basis $\cV$ of $\cH_1$ and let $\Torus$ be the maximal torus consisting of diagonal matrices.
By Proposition~\ref{prop:moment} and Kostant's theorem \cite{kostant_convexity_1973}, we know that $\FerPoly[][\bw,N,d]$ is the convex hull of the image of the $\Torus$-fixed points of  $\Density^N(\bw)$.
We first determine the $\Torus$-fixed points.
We use the basis induced by $\cV$, if $\rho$ is $\Torus$-fixed then
\[
\begin{split}
t\rho t^{-1} &= \rho,\\
t\left(\sum_{\bi,\bj} c_{\bi,\bj}\bv_{\bi}\otimes \bv_{\bj}^*\right)t^{-1} &= \sum_{\bi,\bj} c_{\bi,\bj}\bv_{\bi}\otimes \bv_{\bj}^*,\\
\sum_{\bi,\bj} c_{\bi,\bj}\left(\prod_{i\in\bi} t_i\right)\left(\prod_{j\in\bj} t_j^{-1}\right) \bv_{\bi}\otimes \bv_{\bj}^* &= \sum_{\bi,\bj} c_{\bi,\bj}\bv_{\bi}\otimes \bv_{\bj}^*,
\end{split}
\]
for any $t\in\Torus$, since by definition $t v_i = t_i$.
It is well-known that the $\Torus(\cH_N)$ fixed points are diagonal matrices, however since $\Torus$ is a subgroup coming from the injective map $g:\Torus(\cH_1)\to\Torus(\cH_N)$, we make sure that the set of fixed points is not larger.
By comparing coefficients, we see that if $\bi\neq \bj$ we must have $c_{\bi,\bj}=0$, hence $\rho$ is indeed diagonal in the basis induced by $\bV$.
This implies that the $\Torus$-fixed points of  $\Density^N(\bw)$ are of the form
\[
\sum_{j\in[D]} w_{j} \bv_{\bi_j}\otimes \bv_{\bi_j}^*\in\Density^N(\bw),
\]
for any total order of the indices.
The partial traces of these operators are diagonal in $\cV$ and their diagonals are the elements of Equation \eqref{eq:kostant}.
\end{proof}

The $V$-representation presented in Proposition~\ref{prop:kostant} is non-minimal because some of the $\Torus$-fixed points do not induce vertices.

\begin{example}
We have considered four compact Lie groups acting on $\Density(\bw)$: $\Torus(\cH_1),\Torus(\cH_N),\Unitary(\cH_1)$ and $\Unitary(\cH_N)$, see the top of Figure \ref{fig:diagram}.

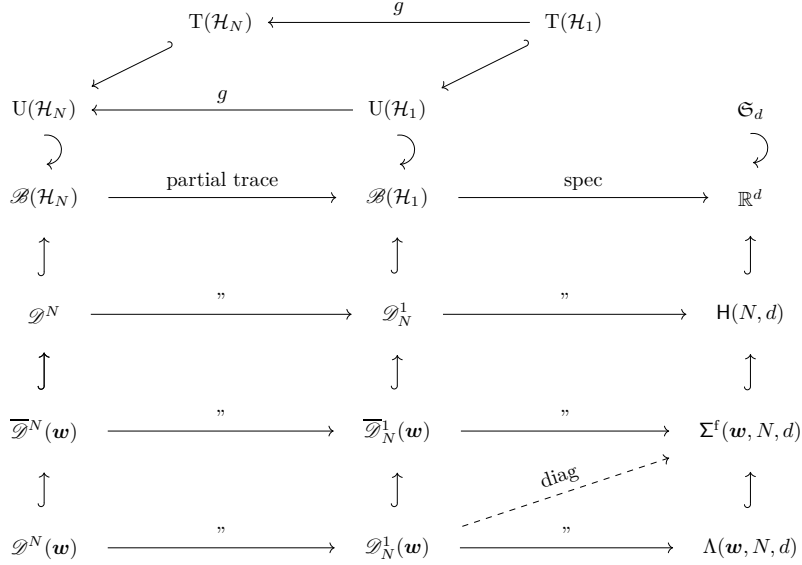
\begin{figure}[!ht]
\resizebox{0.75\textwidth}{!}{$
\begin{tikzpicture}[scale=1]
\def\x{6}
\def\y{2}

\node[inner sep=0.5cm] (DNw) at (0,0) {$\Density^N(\bw)$};
\node[inner sep=0.5cm] (D1Nw) at (\x,0) {$\Density^1_N(\bw)$}  edge [<-] node[above] {''} (DNw) ;
\node[inner sep=0.5cm] (LwND) at (2*\x,0) {$\Lambda(\bw,N,d)$} edge [<-] node[above] {''}       (D1Nw);

\node[inner sep=0.5cm] (DbarNw) at (0,\y) {$\overline{\Density}{}^N(\bw)$}     edge [<-right hook] (DNw) ;
\node[inner sep=0.5cm] (D1Nbarw) at (\x,\y) {$\overline{\Density}{}^1_N(\bw)$} edge [<-] node[above] {''} (DbarNw)
                                                                               edge [<-right hook] (D1Nw);
\node[inner sep=0.5cm] (SwND) at (2*\x,\y) {$\FerPoly[][\bw,N,d]$}                        edge [<-] node[above] {''}       (D1Nbarw)
                                                                               edge [<-right hook] (LwND)
									       edge [<-,dashed] node[above,sloped] {$\diag$} (D1Nw);

\node[inner sep=0.5cm] (Dw) at (0,2*\y) {$\Density^N$} edge [<-right hook] (DbarNw) edge [<-right hook] (DbarNw);
\node[inner sep=0.5cm] (D1N) at (\x,2*\y) {$\Density^1_N$}   edge [<-] node[above] {''} (Dw) 
                                                             edge [<-right hook] (D1Nbarw);
\node[inner sep=0.5cm] (HND) at (2*\x,2*\y) {$\pol[H](N,d)$} edge [<-] node[above] {''}       (D1N)
                                                             edge [<-right hook] (SwND);

\node[inner sep=0.5cm] (BHN) at (0,3*\y) {$\Herm(\cH_N)$} edge [<-right hook] (Dw);
\node[inner sep=0.5cm] (BH1) at (\x,3*\y) {$\Herm(\cH_1)$} edge [<-] node[above] {partial trace} (BHN)
                                                           edge [<-right hook] (D1N);
\node[inner sep=0.5cm] (RD) at (2*\x,3*\y) {$\RR^d$}       edge [<-] node[above] {$\spec$}       (BH1)
                                                           edge [<-right hook] (HND);

\node[inner sep=0.25cm] (UHN) at (0,3.75*\y) {$\Unitary(\cH_N)$};
\node[inner sep=0.25cm] (UH1) at (\x,3.75*\y) {$\Unitary(\cH_1)$} edge [->] node[above] {$g$} (UHN);
\node[inner sep=0.25cm] (SD) at (2*\x,3.75*\y) {$\Sym{d}$};

\draw[->] (UHN.south) arc (100:-90:0.25);
\draw[->] (UH1.south) arc (100:-90:0.25);
\draw[->] (SD.south) arc (100:-90:0.25);

\node[inner sep=0.25cm] (THN) at (0.5*\x,4.5*\y) {$\Torus(\cH_N)$} edge [left hook->] (UHN);
\node[inner sep=0.25cm] (TH1) at (1.5*\x,4.5*\y) {$\Torus(\cH_1)$} edge [left hook->] (UH1)
                                                                    edge [->] node[above] {$g$} (THN);

\end{tikzpicture}$}
\caption{The various torus actions on density matrices.}
\label{fig:diagram}
\end{figure}

\noindent
Each action leads to a moment map and therefore a moment polytope associated to it:
\begin{center}
\begin{tabular}{c@{\hspace{1.5cm}}c}
\begin{tabular}{c|c}
Lie Group           & Moment Polytope \\\hline
$\Unitary(\cH_N)$   & $\{\bw\}$ \\
$\Torus(\cH_N) $    & $\Perm(\bw)$ \\
\end{tabular} & 
\begin{tabular}{c|c}
Lie Group           & Moment Polytope \\\hline
$\Unitary(\cH_1)$   & $\klypo$ \\
$\Torus(\cH_1)$     & $\FerPoly[][\bw,N,d]$ \\
\end{tabular}
\end{tabular}
\end{center}
In Figure~\ref{fig:diagram}, each unitary group acts on the elements of its respective column, whereas the symmetric group acts on the last column. 
On the right of the figure, we have the hypersimplex $\pol[H](N,d)$ defined in Equation \eqref{eq:def_hypersimplex}, the polytope $\FerPoly[][\bw,N,d]$ defined in Definition \ref{def:ferm_specpo}, and the rosette $\Lambda(\bw,N,d)=\Sym{d}(\klypo)$ defined in Equation~\eqref{eq:rosette_klypo}.
\end{example}

\subsection{Bosonic spectral polytope}
\label{ssec:repr_boson}

Here, we describe how to adapt the situation to the $N$-boson Hilbert spaces.
The subspace given by the $N$-th symmetric power $\SymP^N\cH_1$ consists of \emph{$N$-boson states}:
\[
\bv_1\vee\cdots\vee\bv_N\in\SymP^N\cH_1.
\]
Let $\BosPoset:=\{\bi=(i_1,\dots,i_N)\in [d]^N~:~1\leq i_1\leq\dots\leq i_N\leq d\}$ denote the set of all \defn{$N$-boson configurations}.
If $\cB=\{\bb_i\}_{i=1}^d$ is an orthonormal basis for $\cH_1$, then $\cB^N:=\{\bb_{\bi}\}_{\bi\in\BosPoset }$ is an orthonormal basis for~$\cH_N$, where $\bb_{\bi}:= \bb_{i_1}\vee\dots\vee \bb_{i_N}$.
It follows that $D:=\dim(\SymP^N\cH_1)=\binom{d+N-1}{N}$.
An element $\bi\in\BosPoset$ can be considered as a multisubset $S$ of $[d]$ of cardinality $|S|=N$.
The multisubset $S$ may be represented as a function $S:[d]\to\NN$, where $S(i)=m_i$ is the number of times $i$ appears in $\bi$, i.e. the multiplicity of $i$ in $\bi$.
Given a vector $\bm{\lambda}=(\lambda_1,\dots,\lambda_d)\in\RR^d$, we define $\lambda(\bi)=\sum_{j=1}^N m_j\lambda_{j}$.
This induces a (not necessarily unique) compatible total order on $N$-boson configurations
$\bi_1,\bi_2,\dots,\bi_D$ where $\lambda(\bi_1)\geq \lambda(\bi_2)\geq \cdots \geq \lambda(\bi_D)$.
Fixing $\bw\in \Delta_{D-1}$ and $\bm{\lambda}\in\RR^d$, we associate the \emph{occupation vector}
\[
\OccVec_{\bw}(\bm{\lambda}):=\sum_{j=1}^D w_j \bchi(\bi_j)\in\RR^d, \text{ where }\bchi(\bi):=\sum_{j=1}^d m_j\be_j \in \RR^d,
\]
and $\bi_1,\dots,\bi_D$ is a compatible total order induced by $\bm{\lambda}$.
Let $\Lineups(\bw):= \{\OccVec_{\bw}(\bm{\lambda})~:~\bm{\lambda}\in\RR^d\}$ be the set of occupation vectors.

\begin{definition}[Bosonic spectral polytope]
	\label{def:bos_specpo}
	Let $\bw\in\Pauli_{D-1}$.
	The polytope
	\[
	\BosPoly[][\bw,N,d]:= \conv(\Lineups(\bw))
	\]
	is the \defn{bosonic spectral polytope} of the $N$-boson Hilbert space on $d$ orbitals with weight~$\bw$.
\end{definition}

\subsection{The Challenge}
\label{ssec:challenge}

Theorem \ref{thm:relaxed_gok} gives a spectral characterization of the set of convex ensemble $N$-representable density operators that should be converted into a set of linear inequalities to solve the membership problem effectively.
Furthermore, as explained earlier, it is also meaningful to restrict our attention to vectors $\bw\in\Delta_{D-1}$ with a few non-zero entries.
We thus introduce an extra parameter~$r$ that specifies the number of non-zero entries of $\bw$.
This leads to a problem within polyhedral combinatorics, namely, provide a non-redundant $H$-representation of the polytopes $\FerPoly[r][\bw,N,d]$ and $\BosPoly[r][\bw,N,d]$.
Two examples of $3$-dimensional spectral polytopes are illustrated in Figure~\ref{fig:spectral_poly}.

\begin{figure*}[!h]
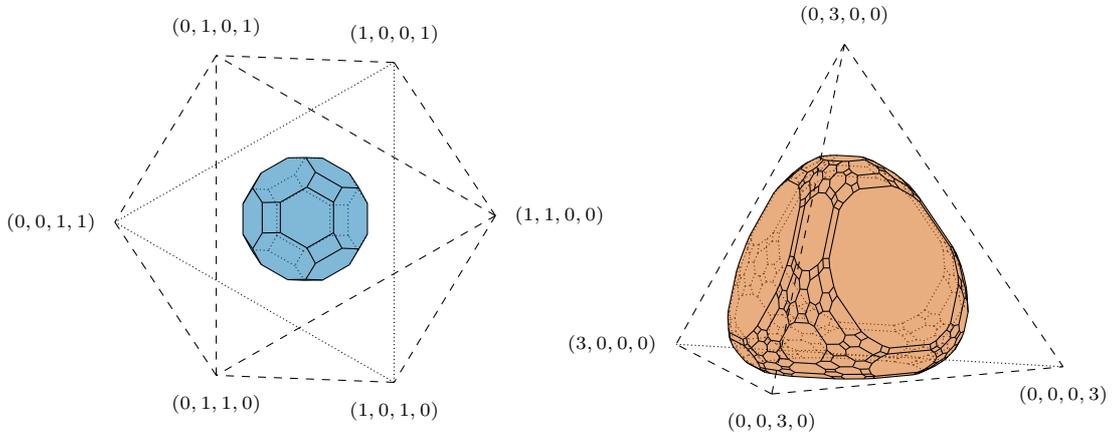

\begin{center}
\begin{tikzpicture}

\node  at (-3.5,0) {\input{ordered_spec426}};
\node  at (3.5,0) {\input{boson8}};

\end{tikzpicture}
\caption{On the left, the fermionic spectral polytope $\FerPoly[][\frac{1}{21}(6,5,4,3,2,1),2,4]$.
On the right, the bosonic spectral polytope $\BosPoly[][\bw,3,4]$, with $\bw=\frac{1}{36}(8,7,6,5,4,3,2,1,0,0,0,0,0,0,0,0,0,0,0,0)$.}
\label{fig:spectral_poly}
\end{center}
\end{figure*}

\noindent
Contrary to what the image on the left may suggest, fermionic spectral polytopes are not \emph{that} simple. 
The simplicity of the dimension $3$ case is due to the Pauli exclusion principle and to the coincidence that the hypersimplex $\pol[H](2,4)$ is a cross-polytope.

\part{Convex Geometry Tools}
\label{part2}

Before expressing the polytopes $\FerPoly[r][\bw,N,d]$ and $\BosPoly[r][\bw,N,d]$ as the intersection of finitely many halfspaces, we review necessary discrete geometry concepts, expand upon recent developments presented in \cite{PadrolPhilippe2021}, and whet the tools to be used in Part~\ref{part3} to provide an effective solution to the convex $1$-body $N$-representability problem.
To simplify the notation, in Part~\ref{part2} and~\ref{part3}, we omit the mention to $\bw$ in $\FerPoly[r][\bw,N,d]$ and denote it $\FerPoly[r]$, when $\bw$ is some fixed vector that is clear from the context and similarly for $\BosPoly[r]$.

\section{$V$- to $H$-representation translation}
\label{sec:V_to_H}

The problem presented in Section~\ref{ssec:challenge} relies on a classical procedure in discrete geometry which may be done in several ways.
Several algorithms exist and are also implemented in various computer algebra systems dealing with geometric computations, see e.g. \cite{sagemath}\cite{bruns_power_2016}\cite{polymake}\cite{avis_pivoting_1992,fukuda_exact_2008}\cite{bagnara_parma_2008}.
In this section, we describe a method using normal fans which is well suited to deal with symmetric polytopes.
We refer the reader to the reference books \cite{schrijver_theory_1986}\cite{ziegler_lectures_1995}\cite{grunbaum_convex_2003} for further background on polyhedral objects which are not described below.

\subsection{Polyhedra, cones, and polytopes}

A \emph{polyhedron} is the intersection of finitely many closed halfspaces: 
\begin{equation}\label{eq:H-rep}
	\pol[Q] := \left\{ \bx\in\RR^d: \bA\bx \leq \bb \right\},
\end{equation}
where $\bA$ is a matrix and $\bb$ is a vector.
The expression in \eqref{eq:H-rep} is a \defn{$H$-representation} of $\pol[Q]$.
A row of $\bA$ and its corresponding entry in $\bb$ gives a \emph{defining inequality} of $\pol[Q]$.
Given a row~$\ba_i$ of $\bA$, if the equation $\np{\ba_i,\bx}\leq b_i$ is a positive linear combination of other equations, it is not necessary to define $\pol[Q]$, and this $H$-representation is called \emph{redundant}.

A \emph{point configuration} in $\RR^d$ is an ordered set $\bV:=\{\bv_1,\dots,\bv_m\}$ of vectors in $\RR^d$.
We denote by $\aff(\bV), \cone(\bV)$, and $\conv(\bV)$ the affine, conical, and convex hull of $\bV$, respectively. 
The latter object is called a \emph{polytope}.
Furthermore, the elements in these sets are called \emph{affine}, \emph{conical} and \emph{convex combinations} of~$\bV$, respectively.
A cone is \emph{pointed} if it contains no lines. 
We refer to the elements of minimal generating sets of affine, conical and convex hulls as \emph{line generators}, \emph{ray generators}, and \emph{vertices}.
By the Minkowski--Weyl theorem, every polyhedron $\pol[Q]$ can be decomposed uniquely as the sum of an affine hull, a conical hull and a convex hull:
\begin{equation}\label{eq:structure}
\pol[Q]=\pol[L]+\pol[K]+\pol,
\end{equation}
where $\pol[L]$ is a linear subspace (called the \emph{lineality space} of $\pol[Q]$),~$\pol[K]$ is a pointed cone (called the \emph{recession cone} of $\pol[Q]$), $\pol$ is a polytope, and both $\pol[K]$ and $\pol[P]$ lie in the orthogonal complement of~$\pol[L]$.
The expression in Equation~\eqref{eq:structure} is a \defn{$V$-representation} of~$\pol[Q]$.
Thus, polytopes and cones are polyhedra: polytopes are \emph{bounded} polyhedra and cones are \emph{homogeneous} polyhedra that is, $\bb=(0,\dots,0)$ in Equation~\eqref{eq:H-rep}. 

Let $\binom{J}{k}$ denote the collection of $k$-elements subsets of the set $J$, and $\multiset{J}{k}$ denote the collection of $k$-elements multisubsets.
The cardinality of a set $J$ is denoted by~$|J|$.
Multisubsets and subsets of $[d]$ are regarded as functions $S:[d]\rightarrow\NN$, $S\in\NN^d$, where subsets are those $S$'s whose image is contained in $\{0,1\}$.
Let ${\bchi:\NN^{d}\rightarrow\RR^d}$ be the \defn{multiplicity function} sending a multiset $S\in\NN^d$ with support included in $[d]$ to $\bchi(S)=\sum_{j\in [d]}S(j)\be_j$.
In order to write multisubsets and subsets compactly, we write $113$ and $123$ to mean $\{1,1,3\}$ and $\{1,2,3\}$ respectively.
There are two point configurations that play a central role in the present article.

\begin{definition}[Fermionic and bosonic point configurations]
\label{def:FBPC}
The \defn{fermionic} and \defn{bosonic point configurations} are defined as 
\[
\FerPC :=\left\{\bchi(S)~:~S\in \binom{[d]}{N} \right\}\subset\RR^d,\quad\text{and}\quad
\BosPC :=\left\{\bchi(S)~:~S\in \multiset{[d]}{N} \right\}\subset\RR^d.
\]
\end{definition}

\noindent
These configurations and their convex hulls are well-known geometric objects.

\begin{example}[Hypersimplices]\label{ex/def:hypersimplex}
The convex hull $\conv(\FerPC)=\pol[H](N,d)$ is the hypersimplex.
This showcases the Minkowski--Weyl theorem: the elements of $\FerPC$ correspond to the vertices of the polyhedron defined by the linear inequalities in Equation~\eqref{eq:def_hypersimplex}, see e.g. \cite{kuhn_linear_1960}\cite{coleman_structure_1963}.
\end{example}

\begin{example}[Inflated simplex]\label{ex/def:inflated_simplex}
The convex hull $\conv(\BosPC)=N\cdot\pol[H](1,d)$ is the $N$-th dilation of the hypersimplex $\pol[H](1,d)$, i.e. the \emph{standard simplex}.
In this case, not every point in the configuration $\BosPC$ corresponds to a vertex of the dilated standard simplex, only those supported in a single coordinate.
In fact, $\BosPC$ is the set of integer points of $N\cdot\pol[H](1,d)$.
\end{example}

\subsection{Normal fans}\label{sec:nfans}
A \emph{fan} is a family $\fan=\{\pol[K]_1,\pol[K]_2,\dots,\pol[K]_m\}$ of non-empty cones such that every non-empty face of a cone in $\fan$ is also a cone in $\fan$, and the intersection of any two cones in~$\fan$ is a face of both, see e.g. \cite[Section~7.1]{ziegler_lectures_1995}.
The \emph{support} of $\fan$ is $\bigcup_{i\in[m]} \pol[K]_i$.
The $1$-dimensional cones of a fan are called \emph{rays}.
Let $\pol=\conv(\bV)$ be a polytope. 
A linear inequality satisfied by every point $\bx\in\pol$ is called \emph{valid}.
Recall that the \emph{support function} $\mathrm{supp}_{\pol}:\RR^d\to \RR\cup\{\infty\}$ of $\pol$ is defined as $\mathrm{supp}_{\pol}(\by):=\max_{\bx\in \pol} \np{\by,\bx}$.
Every vector $\by\in\RR^d$ induces a unique valid inequality on a polytope~$\pol$, according to $\np{\by,\bx}\leq \mathrm{supp}_{\pol}(\by)$.
The polytope $\pol^{\by} = \{\bx\in \pol~:~\np{\by,\bx}=\mathrm{supp}_{\pol}(\by) \}$ is referred to as a \emph{face} of~$\pol$.
Vertices of $\pol$ are $0$-dimensional faces and \emph{facets} of $\pol$ are codimension-$1$ faces.
Whereas for polytopes, $\mathrm{supp}_{\pol}$ is indeed finite, extending the definition of support function to general polyhedra $\pol[Q]$ requires $\infty$.
Vertices, faces and facets of polyhedra are defined similarly with the added $\infty$ value.
Given a face $\pol[F]$ of $\pol$, we define its \emph{relatively open} and \emph{closed normal cones}:
\begin{align}
\ncone_{\pol}(\pol[F])^\circ &:= \{\by\in\RR^d~:~\pol^{\by}=\pol[F]\} \text{ and}\nonumber \\
\ncone_{\pol}(\pol[F])       &:= \{\by\in\RR^d~:~\pol^{\by}\supseteq \pol[F]\}. \label{eq:ncone_closed}
\end{align}
The family $\fan[N](\pol):=\{\ncone_{\pol}(\pol[F])~:~\pol[F]\text{ a face of }\pol\}$ is the \emph{normal fan} of $\pol$.
The normal fan of a polytope is entirely recovered from the normal cones of the vertices, since all other normal cones are faces of them.
Using Equation~\eqref{eq:structure}, the normal cone of a face decomposes as $\ncone_{\pol}(\pol[F])=\pol[L]_{\pol}(\pol[F])+\pol[K]_{\pol}(\pol[F])$.
The lineality space $\pol[L]_{\pol}(\pol[F])$ is the orthogonal complement of $\aff(\pol)$, hence it does not depend on $\pol[F]$.
The pointed cone $\pol[K]_{\pol}(\pol[F])$ is called the \emph{essential cone} of $\pol[F]$ with respect to~$\pol$.
The collection $\widehat{\fan[N]}(\pol):=\{\pol[K]_{\pol}(\pol[F])~:~\pol[F]\text{ a face of }\pol\}$ is the \emph{essential fan} of $\pol$. 
Elements of $\widehat{\fan[N]}(\pol)_1 = \{\pol[K]_{\pol}(\pol[F])~:~\pol[F]\text{ a facet of }\pol\}$ are called \emph{essential rays}.
Choosing one generator for each essential ray together with a basis of $\pol[L]_{\pol}=\aff(\pol)^\top$ leads to a minimal $H$-representation of $\pol$ as in Equation~\eqref{eq:H-rep}.

By Equation~\eqref{eq:ncone_closed}, the normal cone of a face~$\pol[F]$ consists of all vectors $\by\in\RR^d$ whose linear functional is maximized on $\pol[F]$. 
Whence, for each vertex~$\bv$ of $\pol$ its normal cone has the following $H$-representation:
\begin{equation}\label{eq:ncone_vertex}
\ncone_{\pol}(\bv) = \{\by\in\RR^y : \np{\by,\bv-\bv'}\geq0, \text{ for }\bv'\in \bV\}.
\end{equation}
By translating the $H$-representation of $\ncone_{\pol}(\bv)$ into a $V$-representation and disregarding its lineality space, we obtain a set of ray generators for the pointed cones $\pol[K]_{\pol}(\bv)$.
By evaluating the support function at the ray generators for each $\pol[K]_{\pol}(\bv)$, we obtain a non-redundant $H$-representation of~$\pol$.

\section{Permutation invariant polytopes}
\label{sec:perm_inv_poly}

The fermionic and bosonic spectral polytopes $\FerPoly[r]$ and $\BosPoly[r]$ presented in Section~\ref{ssec:challenge} are $\Sym{d}$-invariant polytopes.
In this section, we give a general condition for a linear functional to determine a facet-defining inequality of $\Sym{d}$-invariant polytopes.
The combinatorial and geometric nature of this condition makes it very practical and opens the study of a larger family of $\Sym{d}$-invariant polytopes, where several orbits are involved.
In particular, the presented technique exploits a restriction to representatives of equivalence classes of the symmetric group action.

\subsection{Fundamental representatives}
A polytope $\pol$ in $\RR^d$ is \emph{$\Sym{d}$-invariant} if it is stabilized by the standard action of $\Sym{d}$ on $\RR^d$: $\pi\cdot (v_1,\ldots,v_d):= (v_{\pi(1)},\ldots,v_{\pi(d)})$.
Classical examples of $\Sym{d}$-invariant polytopes are the usual \emph{permutohedra}: 
They are defined as the convex hull of a single $\Sym{d}$-orbit, i.e. $\Perm(\bv):=\conv\{\pi\cdot \bv:\pi\in\Sym{d}\}$ for some $\bv\in\RR^d$, see \cite{postnikov_permutahedra_2009}. 
These polytopes are the subject of the Schur--Horn theorem as explained in Section~\ref{ssec:lie}.
For all $\bv\in\RR^d$ whose coordinates are pairwise distinct, the combinatorial type of~$\Perm(\bv)$ coincides, but for non-generic points one obtains many different combinatorial types. 
The fundamental basis of $\RR^d$ is given by the vectors $\bef_k := \sum_{1\leq i\leq k} \be_i$ for $k=1,\ldots,d$.
In particular, in view of Theorem~\ref{thm:E1N}, the particular case of hypersimplices is particularly relevant.

\begin{example}[Example~\ref{ex/def:hypersimplex} continued]
\label{ex:hypersimplex}
Let $d\geq N\geq 0$. 
The hypersimplex $\pol[H](N,d)$ is the $\Sym{d}$-invariant polytope $\Perm(\bef_N)$ whose vertices are the permutation of the vector $\bef_N$ composed of $N$ coordinates equal to $1$ followed by $d-N$ coordinates equal to $0$.
When $N\notin\{0,d\}$, the hypersimplex $\pol[H](N,d)$ has dimension $d-1$, otherwise it is a point. It is a simplex if $N\in\{1,d-1\}$, and has  $2d$ facets if $2\leq N\leq d-2$.
\end{example}

\noindent
General $\Sym{d}$-invariant polytopes are convex hulls of finitely many $\Sym{d}$-orbits. 

\begin{definition}[$\Sym{d}$-invariant polytope $\Perm(\bV)$]
\label{def:invariant_pol}
Let $\bV=\{\bv_1,\dots,\bv_m\}\subset\RR^d$.
The $\Sym{d}$-\defn{invariant polytope} $\Perm(\bV)$ is
\[
\Perm(\bV):=\conv\{\pi\cdot\bv_i~:~\pi\in\Sym{d}, \quad \bv_i\in\bV \}.
\]
The vectors $\bv_1,\dots,\bv_m$ are the \emph{generators} of $\Perm(\bV)$.
\end{definition}

As usual with symmetric objects, they can be described by restricting to the fundamental domain of the action of $\Sym{d}$ in $\RR^d$.

\begin{definition}[Fundamental chamber]
\label{def:fund_cham}
The  \defn{fundamental chamber} is the polyhedron
\begin{equation}
\Phi_d:=\{\by\in\RR^d~:~y_{1}\geq y_{2}\geq\cdots\geq y_{d}\}.
\end{equation}
This is a polyhedral cone that can be also expressed as $\cone\{\bef_1,\ldots,\bef_{d-1},\bef_d, -\bef_d\}$, where $\bef_1,\dots,\bef_d$ are the elements of the fundamental basis.
\end{definition}

In general, we use the adjective \defn{fundamental} preceding any object that is directly related to~$\Phi_d$, for instance a fundamental vector is any vector with decreasing coordinates.
\begin{definition}[Fundamental representative, $\bx^{\downarrow}$]
\label{def:orb_repr}
For a vector $\bx\in\RR^d$, we write
\[
\bx^\downarrow:=\left(x^\downarrow_1,x^\downarrow_2,\dots,x^\downarrow_d\right)\in\Phi_d
\]
for the vector obtained from $\bx$ by ordering its coordinates in decreasing order.
Equivalently, $\bx^\downarrow$ is the unique representative of the orbit $\Sym{d}\cdot\bx$ in the fundamental chamber $\Phi_d$.
The scalar $x_i^{\downarrow}$ is the $i$-th largest (with possible ties) coordinate of $\bx$.
\end{definition}

\begin{example}[Example~\ref{ex:hypersimplex} continued]
The $(d-1)$-dimensional hypersimplex $\pol[H](N,d)$ has the following $H$-representation using fundamental representatives.
\[
\pol[H](N,d) =
\left\{
\bx\in\RR^d~:~
\begin{array}{lcc} 
x_1^\downarrow                                                        & \leq & 1\\ 
x_1^\downarrow+x_2^\downarrow+\dots+x_{d-1}^\downarrow                & \leq & N\\ 
x_1^\downarrow+x_2^\downarrow+\dots+x_{d-1}^\downarrow+x_d^\downarrow & =    & N   
\end{array}\right\}.
\]
The second inequality together with the linear equality are equivalent to $x^\downarrow_d \geq 0$, so this representation is indeed equivalent to that of Equation \eqref{eq:def_hypersimplex}.
\end{example}

From here on, when writing a $\Sym{d}$-invariant polytope in terms of generators $\bv_1,\dots,\bv_m$, we always assume them to be fundamental representatives $\bv_1^{\downarrow},\dots,\bv_m^{\downarrow}$ in $\Phi_d$.

\begin{definition}[Fundamental fan]
Let $\bV=\{\bv_1,\dots,\bv_m\}\subset\RR^d$ and $\pol=\Perm(\bV)$.
The intersection of fans $\fan[N](\pol)\cap \Phi_d$ is called the \defn{fundamental fan} of $\pol$.
Its cones are called \defn{fundamental cones} and in particular its rays are \defn{fundamental rays}.
\end{definition}

A $\Sym{d}$-invariant polytope is called \defn{homogeneous} if the sums of the coordinates of its generators are all equal.

\begin{convention}\label{convention}
The polytopes $\FerPoly[r]$ and $\BosPoly[r]$ are homogeneous polytopes in $\RR^d$ of dimension $d-1$, and have the same lineality space: the $1$-dimensional subspace spanned by $\bef_d$.
This implies that neither their normal cones in $\RR^d$ nor their fundamental cones in $\Phi_d$ are pointed.
Rather that repeating ad nauseam the word \emph{essential}, we consider $\fan[N](\pol)$ and $\fan(\pol)$ as pointed fans by applying the projection on $\RR^d$ that fixes $\bef_i$ for $i\in[d-1]$ and sends $\bef_d$ to the origin. 
\end{convention}

\begin{remark}\label{rem:fundamental_normal_rays}
By symmetry, the $\Sym{d}$-orbit of each normal ray has a unique representative in~$\Phi_d$, which is a fundamental ray.
This is the motivation for focusing on the determination of the fundamental fan. 
However, it is not true that all rays of $\fan(\pol)$ are rays of $\fan[N](\pol)$ as we may create new rays when we restrict to the fundamental chamber.
For example, the vector $\bef_i$ spans a fundamental ray of the hypersimplex $\pol[H](N,d)$ for every $i\in[d-1]$, but it spans a normal ray (see Convention \ref{convention}) only if $i\in\{1,d-1\}$.
Proposition~\ref{prop:basic_cleanup_2} on page~\pageref{prop:basic_cleanup_2} illustrates this phenomenon on another polytope.
\end{remark}

The crucial tool left to discuss is a combinatorial criterion to certify that a valid inequality induced from an arbitrary ray in $\fan(\pol)$ is facet-defining on $\pol$.
For this, we characterize the dimension of faces of symmetric polytopes in the next section.

\subsection{Faces of symmetric polytopes}
\label{ssec:faces_sym}

The goal of this section is to determine the dimension of faces of $\Sym{d}$-invariant polytopes $\pol=\Perm(\bV)$.
By the same token, it makes it possible to determine the combinatorial type of $\Sym{d}$-invariant polytopes.
In view of the previous section, it is practical to decompose $\RR^d$ as $\spa\{\bef_d\}\oplus \pol[R]^{d-1}$, where $\pol[R]^{d-1}$ is the $(d-1)$-dimensional subspace of $\RR^d$ consisting of vectors whose coordinates sum is $0$.
A $\Sym{d}$-invariant polytope is homogeneous exactly when it is contained in a parallel copy of $\pol[R]^{d-1}$, i.e. the sum of the coordinates of its points is constant.
A vector in $\RR^d$ is \defn{inert} if it is fixed by the action of $\Sym{d}$, i.e. it is an element of $\spa\{\bef_d\}$.
By extension, a $\Sym{d}$-invariant polytope is inert if all its generators are inert.
The following lemma gives the dimension of $\Sym{d}$-invariant polytopes and its proof is self-evident from this decomposition.
\begin{lemma}[Dimension of $\Sym{d}$-invariant polytopes]\label{lem:homogeneous}
Let $\bV=\{\bv_1,\dots,\bv_m\}\subset\RR^d$.
\begin{compactenum}[i)]
\item Assume $\pol=\Perm(\bV)$ to be inert.
\subitem If $\bv_1=\cdots=\bv_m$ (i.e.\ $\pol$ is homogeneous), then $\dim(\pol)=0$.
\subitem Otherwise, $\dim(\pol)=1$.
\item Assume that $\pol=\Perm(\bv_1,\dots,\bv_m)$ is not inert. 
\subitem The dimension of~$\pol$ is $d-1$ if and only if $\pol$ is homogeneous.
\subitem Otherwise, $\dim(\pol)=d$.
\end{compactenum}
\end{lemma}

\begin{example}
The fermionic and bosonic spectral polytopes $\FerPoly[r]$ and $\BosPoly[r]$ are non-inert homogeneous polytopes.
Therefore they have codimension-$1$ in $\RR^d$ and following Convention~\ref{convention} we consider the normal and fundamental fan within the space $\pol[R]^{d-1}=\textrm{span}\{\bef_1,\dots,\bef_{d-1}\}$ to make them pointed.
\end{example}

\begin{definition}[{Labeling of faces of $\Phi_d$, see \cite[Section~1.15]{humphreys_reflection_1990}}]
Let $\pol[F]$ be a face of $\Phi_d$.
Furthermore, let $\Sym{\bc}=\prod_{i=1}^{k}\Sym{c_i}$ be the largest standard parabolic subgroup (i.e. Young subgroup) of $\Sym{d}$ that fixes $\pol[F]$, where $\bc=(c_1,\dots,c_k)\vDash d$ is a composition of $d$.
The face $\pol[F]$ is \emph{labeled} by the composition $\bc$ that describes $\Sym{\bc}$.
\end{definition}

Let $\by\in\Phi_d$.
In order to generalize Lemma~\ref{lem:homogeneous} to determine the dimension of the face $\pol^{\by}$ of~$\pol$, we first determine the largest subgroup of $\Sym{d}$ that stabilizes it.
Let $\bc_{\by}=(c_1,\dots,c_k)$ be the composition labeling the inclusion-minimal face $\pol[F]_{\by}$ of $\Phi_d$ that contains $\by$.
For example, the vector $\by=(3,2,2,2,1,1,0,0)$ leads to the composition $\bc_{\by}=(1,3,2,2)$ of $8$.
The Young subgroup $\Sym{\bc_{\by}}$ is the inclusion-maximal Young subgroup stabilizing the face $\pol^{\by}$ for any $\Sym{d}$-invariant polytope~$\pol$.
The following theorem provides the dimension of $\pol^{\by}$ by extracting the combinatorial information from a linear functional $\by$.
It reduces significantly the dimension of the convex hull computation, since it restricts all computations to the vectors in the fundamental chamber.

\begin{mainthm}
\label{thm:dimension}
Let $\bV=\{\bv_1,\dots,\bv_m\}\subset\Phi_d$, $\pol=\Perm(\bV)$ and $\by\in\Phi_d$.
The dimension of the face~$\pol^{\by}$ is
\begin{equation}
	\label{eq:dimension}
\dim \pol^{\by} =  \dim\left(\conv\{\Proj{\by}(\bv):\bv\in\bV_{\by}\}\right) + \sum_{i\in[k]\setminus\operatorname{Fix}(\by)} (c_i-1),
\end{equation}
where
\begin{itemize}[-]
\item $\bc_{\by}=(c_1,\dots,c_k)$ is the label of the inclusion-minimal face of $\Phi_d$ containing $\by$,
\item $\Proj{\by}$ is the projection 
$\RR^d\rightarrow \RR^k$ defined by
\[
\Proj{\by}(\bx):=\left(\sum_{i=1}^{c_1} x_i, \sum_{i=c_1+1}^{c_1+c_2} x_i, \dots, \sum_{i=c_{k-1}+1}^{d} x_i \right),
\]
\item $\bV_{\by}$ is the set of generators of $\pol$ maximized by the linear functional $\by$, and 
\item $\operatorname{Fix}(\by)$ is the set of indices $j\in[k]$ such that the subgroup $\Sym{c_j}$ of $\Sym{\bc_{\by}}$ fixes $\pol^{\by}$ pointwise.
\end{itemize}
\end{mainthm}

\begin{proof}
The face $\pol^{\by}$ can be expressed as
\[
\pol^{\by} = \conv\left\{ \pi\cdot \bv~:~\pi\in \Sym{\bc}, \bv\in \bV_{\by} \right\}.
\]
Indeed, for fundamental vectors $\by\in\Phi_d$, $\bv\in \bV_{\by}$, $\bz\in \bV$, and $\pi \in \Sym{d}$, the rearrangement inequality \cite[Theorem~368]{hardy_inequalities_1988} implies that in order to have
\[ 
\mathrm{supp}_{\pol}(\by)=\np{\by,\bv}=\np{\by,\pi\cdot \bz},
\]
one necessarily has on the one hand that $\np{\by,\bv}=\np{\by,\bz}$ (and therefore that $\bz\in \bV_{\by}$), and on the other hand that $y_{\pi(i)}= y_{i}$ whenever $w_{\pi(i)}\neq w_{i}$ (and therefore that $\pi\cdot \bz = \pi'\cdot \bz$ for some $\pi'\in \Sym{\bc}$).

Let $\pol[W]_{\by}$ be the $(d-k)$-dimensional kernel of $\Proj{\by}$. 
For example, if $\bc_{\by}=(d)$, then the image of $\Proj{\by}$ is $\spa\{\bef_d\}$  and $\pol[W]_{\by}=\pol[R]^{d-1}$. 
In general, we have the decomposition $\pol[W]_{\by}=\pol[W]^1_{\by}\oplus\cdots \oplus \pol[W]^k_{\by}$, where $\pol[W]^j_{\by}$ is the $(c_j-1)$-dimensional subspace spanned by the vectors $\be_i-\be_{i'}$ with $c_{j-1}< i,i'\leq c_j$.
We can decompose the computation of the dimension of $\pol^{\by}$ as the sum of the dimension of its image under $\Proj{\by}$ and the dimension of the largest linear subspace $\pol[W]_{\by}^{\pol}$ of $\pol[W]_{\by}$ that can be translated into the affine hull of~$\pol^{\by}$.
The first summand is $\dim\left(\conv\{\Proj{\by}(\bv):\bv\in\bV_{\by}\}\right)$, and therefore it remains to show that $\dim \pol[W]_{\by}^{\pol} = \sum_{i\in[k]\setminus\operatorname{Fix}(\by)} (c_i-1)$.

For each ${j\in[k]}$, let $\proj{j}:\RR^d\to \RR^{c_j}$ be projection onto the coordinates $\{c_{j-1}+1,\dots,c_j\}$ (setting $c_0=0$).
For a fixed $\bv$, the polytope $\conv\left\{ \pi\cdot \bv~:~\pi\in \Sym{\bc} \right\}$ is naturally isomorphic to $\Perm(\proj{1}(\bv))\times\dots\times\Perm(\proj{k}(\bv))$, which we call a $\bc$-Permutohedron. 
Hence, $\pol^{\by}$ is the convex hull of a union of $\bc$-Permutohedra, each lying in some translation of $\pol[W]_{\by}$. 
On the one hand, if $j\in\operatorname{Fix}(\by)$ and $\bv\in\bV_{\by}$, then $\Perm(\proj{j}(\bv))$ is a point and its dimension is~$0$.
In particular, $\pol^{\by}$ is entirely contained in the orthogonal complement of the linear subspace~$\pol[W]^j_{\by}$. 
Therefore $\pol[W]_{\by}^{\pol}\subseteq \bigoplus_{i\in[k]\setminus\operatorname{Fix}(\by)}  \pol[W]^i_{\by}$.
On the other hand, if $j\not\in\operatorname{Fix}(\by)$, then a certain vertex $\bv\in\bV_{\by}$ is not inert with respect to~$\Sym{c_j}$.
As we consider a translated copy of the kernel $\pol[W]_{\by}$, $\Perm(\proj{j}(\bv))$ is homogeneous and by Lemma~\ref{lem:homogeneous}, it has dimension $c_j-1$,  which is the dimension of $\pol[W]^j_{\by}$. This shows that  $\pol[W]^j_{\by}\subseteq \pol[W]_{\by}^{\pol}$. We conclude that
\[
\pol[W]_{\by}^{\pol}= \bigoplus_{i\in[k]\setminus\operatorname{Fix}(\by)}  \pol[W]^i_{\by},
\]
providing the sought dimension.
\end{proof}

\begin{remark}[Determination of combinatorial type]
Let $\pol=\Perm(\bv_1,\dots,\bv_m)$ be an homogeneous $\Sym{d}$-invariant polytope, and $\pol^{\by}$, for some $\by\in\Phi_d$, be a facet.
By evaluating ${b_{\by}=\max_{i\in[m]}\np{\by,\bv_i}}$, and keeping track of which generators achieve the maximum $b_{\by}$, we obtain that the generators of~$\pol$ incident to the facet $\pol^{\by}$ are exactly those in $\bV_{\by}$.
To obtain all other vertices of $\pol$ incident to~$\pol^{\by}$, it suffices to permute each vertex in $\bV_{\by}$ according to $\Sym{\bc}$.
Every other facet-vertex incidence is obtained by acting uniformly on the facet and vertices according to the permutation action of $\Sym{d}$.
\end{remark}

\section{Lineup polytopes}
\label{sec:lineups}

The vertices of the fermionic and bosonic spectral polytopes $\FerPoly[r]$ and $\BosPoly[r]$ presented in Section~\ref{ssec:challenge} are occupation vectors.
These occupation vectors are weighted sums of points in the point configurations $\FerPC$ and $\BosPC$ of Definition~\ref{def:FBPC}.
In this section, we determine the normal cones of the occupation vectors and clarify the role of $\bw$ in $\FerPoly[r][\bw,N,d]$ and $\BosPoly[r][\bw,N,d]$, see Corollary~\ref{cor:dep_w} and Proposition~\ref{prop:construction3}.
The starting point lies in the fact that $\FerPoly[][\bw,N,d]$ and $\BosPoly[][\bw,N,d]$ are sweep polytopes in the sense of \cite{PadrolPhilippe2021}.
Here we gather useful facts from \cite{PadrolPhilippe2021} adapted to the present setting, which amounts to consider prefixes of sweeps, that we call \emph{lineups}.
In this section, we let~$\bV=\{\bv_1,\dots,\bv_{m}\}\subset \RR^d$ be an ordered configuration of $m$ distinct points and $r\in[m]$. 

\subsection{Lineups and their normal cones}
Let $\by\in \RR^d$. 
If the linear functional $\np{\by,\cdot}$ is injective on $\bV$, we call the vector~$\by$ \defn{generic} with respect to $\bV$.
The linear functional $\np{\by,\cdot}$ lines up the points in $\bV$ on the line spanned by $\by$ through its projection; that is, the linear functional $\np{\by,\cdot}$ totally orders the elements of $\bV$ according to the values given by $\np{\by,\cdot}$, say from the maximum to minimum. 

\begin{definition}[Lineups of a point configuration, $\lu$]
Let $\by$ be generic with respect to $\bV$.
The $r$-tuple $\ell_{\bV,r}(\by)\in\bV^r$ given by the largest $r$ elements in decreasing order with respect to the values given by $\np{\by,\cdot}$ is a \defn{lineup of length~$r$} or \emph{$r$-lineup} of $\bV$.
The set of $r$-lineups of $\bV$ is denoted $\lu$.  
\end{definition}

When $\bV$ and~$r$ are clear from context, we simply write $\ell(\by)$.
Equivalently, an $r$-tuple $\ell=(\bv_{i_1},\dots,\bv_{i_r})\in\bV^r$ is an $r$-lineup if there exists a linear functional $\by\in\RR^d$ such that 
\begin{equation}\label{def:lineup}
	\np{\by,\bv_{i_1}}>\np{\by,\bv_{i_2}}>\cdots>\np{\by,\bv_{i_r}}>\np{\by,\bb},\text{ for all }\mathbf{b}\in\bV\setminus\{\bv_{i_1},\dots,\bv_{i_r}\}.
\end{equation}
When~$\by$ is not generic, ties may occur and instead of an $r$-tuple we obtain an ordered collection $(J_1, \ldots, J_k)$ of non-empty disjoints subsets of $\bV$ with $\sum_{i=1}^{k-1} |J_i| < r \leq \sum_{i=1}^{k} |J_i|$ and such that for all $1\leq i \leq k$ we have:
\begin{equation}\label{def:ranking}
\begin{cases}
\np{\by,\bv}=\np{\by,\bw}& \text{ for all  $\bv,\bw\in J_i$,}\\
\np{\by,\bv}>\np{\by,\bw}& \text{ for all $\bv\in J_i$ and $\bw\in \bV\setminus\bigcup_{k=1}^i J_k$.}
\end{cases}
\end{equation}
We call such an ordered collection of subsets an \defn{$r$-ranking} of $\bV$. 
We use the same notation $\ell_{\bV,r}(\by)$ to denote the $r$-ranking of~$\bV$ associated to~$\by$. 
When every part of the ranking is a singleton, which occurs if $\by$ is generic, we recover an $r$-lineup.

\begin{example}\label{ex:hypersimplex36}
Consider the hypersimplex $\pol[H](3,6)$ and its vertices $\bV=\FerPC[3][6]$, which are identified with $3$-subsets in $\FerPoset[3][6]$.
The vector $\by=(3,1,1,1,0,0)$ induces the following ranking on $\bV$. 
If $r\in\{18,19,20\}$, it is also an $r$-ranking.
\begin{center}
\begin{tabular}{c|c}
Subsets	& $\np{(3,1,1,1,0,0),\cdot}$ \\\hline
123, 124, 134                & 5 \\
125, 126, 135, 136, 145, 146 & 4 \\
156, 234                     & 3 \\
235, 236, 245, 246, 345, 346 & 2 \\
256, 356, 456                & 1 
\end{tabular}
\end{center}

\end{example}

Two different vectors can define the same $r$-ranking, so for each $r$-ranking $\ell$ we define the following relatively-open polyhedral cone
\[
\pol[K]_\bV^\circ(\ell):=\{\by\in\RR^n: \ell(\by)=\ell\}.
\]
Equation~\eqref{def:ranking} gives an $H$-representation defining this cone. We denote its closure by $\pol[K]_{\bV}(\ell)$, which is obtained by replacing the strict inequalities by the corresponding non-strict inequalities. 

\begin{remark}
For $\{i,j\}\in\binom{[m]}{2}$, we define the hyperplanes $\mathsf{H}_{ij}:=\{\bz\in\RR^d~:~\np{\bz,\bv_i}=\np{\bz,\bv_j}\}$. 
A vector $\by$ is generic with respect to $\bV$ if and only if it is contained in $\RR^d\setminus \bigcup \mathsf{H}_{ij}$, where the union is over all pairs in~$\binom{[m]}{2}$.
The union of $\pol[K]_\bV^\circ(\ell)$ for   all lineups contains the set of generic vectors with respect to~$\bV$. 
The complement of the relatively-open cones of $r$-lineups is contained in the union of the hyperplanes:
\begin{equation}\label{eq:complement}
\RR^d~\mathbin{\Big\backslash}\bigcup_{\ell\in\lu} \pol[K]_\bV^\circ(\ell) \subset \bigcup_{\{i,j\}\in\binom{[m]}{2}} \mathsf{H}_{ij}.
\end{equation}
\end{remark}

Let $\Rg:= \{\pol[K]_\bV(\ell)\ : \ \ell \text{ is an $r$-ranking of }\bV\}$ be the collection of cones given by all the $r$-rankings. 
Its maximal cones are given by the $r$-lineups, and its rays are given by the coarsest non-trivial $r$-rankings. 
By non-trivial we mean different from $([m])$, and coarsest here refers to the refinement order: we say that $(I_1, \ldots , I_k)$ coarsens $(J_1, \ldots , J_l)$ if each $I_i$ is the union of some consecutive $J_j$'s for $i=1,\dots,k-1$ and the last set in the partition $I_k$ contains the remaining blocks of $J$, but it can be larger.

\begin{example}[{Example~\ref{ex:hypersimplex36} continued}]
The $r$-ranking induced by $\by=(3,1,1,1,0,0)$ cannot be non-trivially coarsened for $r > 10$ but for $r\leq 9$ it is coarsened by the ranking induced by $\bef_1=(1,0,0,0,0,0)$.
For example, for $r=9$ we get
\[
\begin{split}
\ell_{9}(3,1,1,1,0,0)&=(\{123,124,134\},\{125,126,135,136,145,146\}),\\
\ell_{9}(1,0,0,0,0,0)&=(\{123,124,134,125,126,135,136,145,146,156\}).
\end{split}
\]
According to the definition, the ranking $\ell_{9}(1,0,0,0,0,0)$ coarsens the ranking $\ell_{9}(3,1,1,1,0,0)$.
\end{example}

\begin{mainthm}\label{thm:main}
The collection of cones $\Rg$ is the normal fan of a polytope.
\end{mainthm}

\noindent 
In the next section, we present four explicit constructions proving the above theorem:
\begin{compactenum}
\item via weighted vectors,
\item via projections of partial permutahedra,
\item via Minkowski sums of $k$-set polytopes, and
\item via fiber polytopes.
\end{compactenum}

\subsection{Constructions}
Recall from Equation~\eqref{eq:pauli_simplex} that the Pauli simplex $\Pauli_{r-1}$ is the set of points in~$\RR^r$ with non-negative decreasing coordinates summing to~$1$.
Throughout this section, we assume that $\bw$ is a weight vector in $\Pauli_{r-1}^{\circ}$ with \emph{strictly} decreasing entries.

\subsubsection{Weighted vectors}

\begin{definition}[Weighted vector of a subconfiguration, $\OccVec_{\bw}(\ell)$]
Let $\bw\in\Pauli_{r-1}$, and $\ell=(\bv_{i_1},\dots,\bv_{i_r})\subset\bV$ be an ordered subconfiguration of $\bV$ of cardinality $r$.
We associate a vector $\OccVec_{\bw}(\ell)$ to $\ell$ with respect to~$\bw$ as follows
\[
\OccVec_{\bw}(\ell) := \sum_{k=1}^r w_k\bv_{i_k}.
\]
\end{definition}

\begin{remark}
Occupation vectors defined in Section~\ref{ssec:expansion} are of the form $\OccVec_{\bw}(\ell)$ with $\bV=\FerPC$, $r=\binom{d}{N}$ and the ordering of $\ell$ results from the expansion of $h$ to $H$.
\end{remark}

\begin{definition}[Lineup polytope of a point configuration $\bV$]\label{def:lineup_pol}
The \defn{r-lineup polytope} $\Lineup[r,\bw]$ of~$\bV$ with respect to $\bw$ is the convex hull of the weighted vectors
\[
\Lineup[r,\bw]:=\conv \left(\{\OccVec_{\bw}(\ell) ~:~\ell\text{ an ordered subconfiguration of $\bV$ of cardinality $r$}\}\right).
\]
\end{definition}

As the weight vector $\bw$ is often understood from the context, we omit it and write $\Lineup$ for the lineup polytope of $\bV$.
This omission in the notation and the name lineup polytope is motivated by the following theorem.

\begin{mainthm}[$V$-representation of lineup polytopes]
\label{thm:vrepr_lineup}
Let $\bV=\{\bv_1,\dots,\bv_{m}\}\subset\RR^d$, $r\in[m]$, and~$\bw$ be a strictly decreasing sequence of $r$ positive real numbers summing to one.
\begin{enumerate}
\item The set of vertices of $\Lineup[r,\bw]$ is $\{\OccVec_{\bw}(\ell)~:~\ell\in\lu\}$. 
\item The normal cone of the vertex $\OccVec_{\bw}(\ell)$ is $\pol[K]_{\bV}(\ell)$.
\item The normal fan of~$\Lineup[r,\bw]$ is $\Rg$ (which is independent of $\bw$).
\end{enumerate}
\end{mainthm}

\begin{proof}
Let $\by$ be a generic vector with respect to $\bV$ and let $\ell_1=\ell(\by)=(\ba_i)_{i=1}^r$ be the $r$-lineup induced by $\by$.
Pick a vertex $\bv$ of the face $\pol^{\by}$.
By construction, there exists an $r$-tuple $\ell_2=(\bb_i)_{i=1}^r$ such that $\bv=\OccVec_{\bw}(\ell_2)$.
We claim that $\ell_2=\ell_1$.
Since $\OccVec_{\bw}(\ell_2)\in\pol^{\by}$,  we have
\[
\langle \by,\OccVec_{\bw}(\ell_2)\rangle= \sum_{i=1}^r w_i\langle \by,\bb_i\rangle \geq  \sum_{i=1}^r w_i\langle \by,\ba_i\rangle=\langle \by,\OccVec_{\bw}(\ell_1)\rangle.
\]
Equivalently, 
\begin{equation}\label{eq:silly}
 \sum_{i=1}^r w_i\left( \np{\by,\bb_i\rangle- \langle \by,\ba_i}\right)\geq 0.
\end{equation} 
The definition of $\ell_1$ via Equation \eqref{def:lineup} implies
\begin{equation}\label{eq:silly2}
\langle \by,\ba_i\rangle \geq \langle \by,\bb_i\rangle, \text{ for }i\in[r].
\end{equation} 
Since $\bw$ is positive,  in order to satisfy Equations \eqref{eq:silly} and \eqref{eq:silly2}, we must have $\np{\by,\ba_i}=\np{\by,\bb_i}$ for all $i\in[r]$.
Since $\by$ is generic, we conclude that $\ell_1=\ell_2$ as claimed.
This implies that for any generic~$\by$, the face $\pol^{\by}$ is equal to the vertex $\OccVec_{\bw}(\ell)$, where $\ell$ is the $r$-lineup induced by $\by$. 
This shows that $\{\OccVec_{\bw}(\ell)~:~\ell\in\lu\}\subseteq \{\text{vertices of }	\Lineup[r,\bw]\}$. 

For the reverse inclusion, observe that the above shows that, for all $\ell\in\lu$,
\[
\pol[K]_{\bV}^\circ(\ell) \subseteq \ncone_{\pol}^\circ\left(\OccVec_{\bw}(\ell)\right).
\]
By Equation \eqref{eq:complement}, $\RR^d\setminus \bigcup_{\ell\in\lu}\pol[K]^\circ_{\bV}(\ell)$ is contained in a union of hyperplanes, so we have 
\[
\{\text{vertices of }\Lineup[r,\bw]\} = \{\OccVec_{\bw}(\ell)~:~\ell\in\lu\},
\]
as there is no more room for an open normal cone of another vertex.
We now use the following claim whose proof is left to the reader.

\noindent
\textbf{Claim.} Let $\{\pol[K]_1,\dots,\pol[K]_m\}$ and $\{\pol[K]'_1,\dots,\pol[K]'_m\}$ be two sets of $n$-dimensional cones in $\RR^n$ such that 
\begin{enumerate}
\item $\pol[K]^\circ_i\cap\pol[K]^\circ_j=\pol[K]'^\circ_i\cap\pol[K]'^\circ_j=\varnothing$, for $1\leq i<j\leq m$.
\item $\bigcup \pol[K]_i = \bigcup \pol[K]'_i = \RR^n$.
\item $\pol[K]_i\subseteq \pol[K]'_i$ for all $i$.
\end{enumerate}
Then both collections are equal.

Finally, we apply this claim on the collections $\{\pol[K]_\bV(\ell)~:~\ell\in\lu \}$ and ${\{\ncone_{\pol}(\OccVec_{\bw}(\ell)):\ell\in\lu\}}$ to conclude that $\pol[K]_{\bV}(\ell) = \ncone_{\pol}(\OccVec_{\bw}(\ell))$ for every lineup $\ell\in\lu$. 
Since the normal cones of all faces of $\Lineup$ can be recovered from the normal cones of the vertices, we have shown that $\fan[N]({\Lineup})=\Rg$.
\end{proof}

\begin{corollary}
\label{cor:dep_w}
For two distinct vectors $\bw,\bw'\in\Pauli^\circ_{r-1}$, the polytopes $\Lineup[r,\bw]$ and $\Lineup[r,\bw']$ have the same normal fan.
As a consequence, the combinatorial type of $\Lineup[r,\bw]$ does not depend on $\bw$.
\end{corollary}

\begin{example}\label{ex:redefinition}
The fermionic and bosonic spectral polytopes in Definitions~\ref{def:ferm_specpo} and \ref{def:bos_specpo} are lineup polytopes.
Indeed 
\[
\FerPoly[r][\bw,N,d] = \Lineup[r,\bw][\FerPC],\qquad \BosPoly[r][\bw,N,d] = \Lineup[r,\bw][\BosPC].
\]
\end{example}

\begin{example}[Sweep polytopes]
If we consider lineups whose length coincides with the number of points in the configuration, we recover all its linear orderings, which are known under the name of \emph{sweeps}, see~\cite{PadrolPhilippe2021}. 
In this case, the fan $\Rg[m][\bV]$ is induced by a hyperplane arrangement, studied under the name of \emph{valid order arrangement} or \emph{sweep hyperplane arrangement}~\cite{Edelman2000}\cite{Stan15}\cite{PadrolPhilippe2021}. The associated polytopes are called \emph{sweep polytopes} in~\cite{PadrolPhilippe2021}.
\end{example}

\begin{example}
Consider the sweep polytope of the hypersimplex $\pol[H](N,d)$.
Its normal fan is the hyperplane arrangement given by the hyperplanes
\[
\sum_{i\in S_1} \bx_i = \sum_{i\in S_2} \bx_i ,
\]
for any two subsets $S_1,S_2\in\binom{[d]}{N}$.
The (finitely many) rays (see Convention \ref{convention}) of this arrangement generate the finite list of inequalities defining the polytope $\spec^{\downarrow}(\Density^1_N(\bw)))$, see  \cite[Theorem 2]{klyachko_2006}.
\end{example}

\subsubsection{Projection of partial permutohedra}\label{ssec:proj_permutohedra}
Let $\simplexconfig :=\{\be_1, \ldots , \be_m\} \subset \RR^m$ be the elementary basis of $\RR^m$.
If $r=m$, the possible $m$-lineups consist of all the tuples $\ell_{\sigma}:=(\be_{\sigma(1)} , \ldots , \be_{\sigma(m)})$ for $\sigma \in \Sym{m}$. 
Indeed, by taking $\by \in \RR^m$ such that $y_i = m+1-\sigma^{-1}(i)$ for $i=1, \ldots, m$ we obtain $\ell_{\sigma}=\ell(\by)$. 
The rankings are the ordered set partitions of $[m]$. 
The fan $\fan[R]_m(\simplexconfig)$ consists of the cells of the braid arrangement $\{\bz\in \RR^m~:~\np{\bz, \be_i}=\np{\bz, \be_j} \}$ for $1\leq i < j \leq m$. 
The lineup polytope is a usual permutohedron: $\Lineup[m,\bw][\simplexconfig]=\Perm(\bw)$, and if we choose $\bw = \tfrac{2}{m(m+1)} (m, m-1, \ldots, 2, 1)$, we recover a dilation of the standard permutohedron.
Similarly, for $r\leq m$, we call \emph{partial permutohedron} the lineup polytope of the $(m-1)$-simplex of length $r$ with weight $\bw$: $\Lineup[r,\bw][\simplexconfig]= \Perm(\sum_{i=1}^r w_i\be_i)$.

The $r$-lineups are the tuples $(\be_{\sigma(1)}, \ldots, \be_{\sigma(r)})$ for all injective functions $\sigma : [r] \rightarrow [m]$. 
The rankings are the ordered collections $(J_1, \ldots, J_k)$ of non-empty disjoints subsets of $[m]$ such that $\sum_{i=1}^{k-1} |J_i| < r \leq \sum_{i=1}^{k} |J_i|$. 
A face of the partial permutohedron corresponding to such a ranking has dimension $\sum_{i=1}^k |J_i| -k$ and is combinatorially isomorphic to the product of permutohedra and partial permutohedron 
$\Lineup[|J_1|][{\simplexconfig[|J_1|]}]\times \ldots \times\Lineup[|J_{k-1}|][{\simplexconfig[|J_{k-1}|]}]\times \Lineup[r'][{\simplexconfig[|J_k|]}],$ 
where $r':= r - \sum_{i=1}^{k-1} |J_{i}|$.

\begin{proposition}
Let $\bV=\{\bv_1,\dots,\bv_{m}\}\subset\RR^d$, and let $M_\bV$ be the matrix whose columns are given by $\bV$.
The lineup polytope $\Lineup$ is a projection of a partial permutahedron, that is
\[
\Lineup = M_\bV\cdot\Lineup[r,\bw][\simplexconfig].
\]
Conversely, every affine image of a partial permutohedron is a lineup polytope, up to translation.
\end{proposition}
\begin{proof}
This follows from Theorem~\ref{thm:vrepr_lineup}, and the fact that linear transformations commute with convex hulls.
\end{proof}

\subsubsection{Minkowski sum of $k$-set polytopes}
Given $k\in[m]$, the \defn{$k$-set polytope} $\ksetpol$ of the point configuration $\bV$ is the convex hull of the $\sum_{i\in I} \bv_i$ for all subsets $I\subseteq[m]$ with~$k$ elements, see~\cite{EdelsbrunnerValtrWelzl1997}\cite{AndrzejakWelzl2003}\cite{MartinezSandovalPadrol2020}. 
The vertices of the $k$-set polytope correspond to the \defn{$k$-sets} of $\bV$: the subsets of cardinality $k$ of $\bV$ that can be separated from the $n-k$ other elements by an affine hyperplane. 
It is a classical result that the $(n-1)$-dimensional standard permutohedron can be written as the Minkowski sum of the hypersimplices~$\pol[H](k,n)$ with $k\in[n]$, see e.g.~\cite{postnikov_permutahedra_2009}. 
Similarly, the $r$-lineup polytope of $\bV$ with respect to $\bw$ can be described as a weighted sum of the $k$-set polytopes of $\bV$ with $k\in[r]$, see Figure \ref{fig:ksetpolytopes}.

\begin{figure}[!ht]
\input{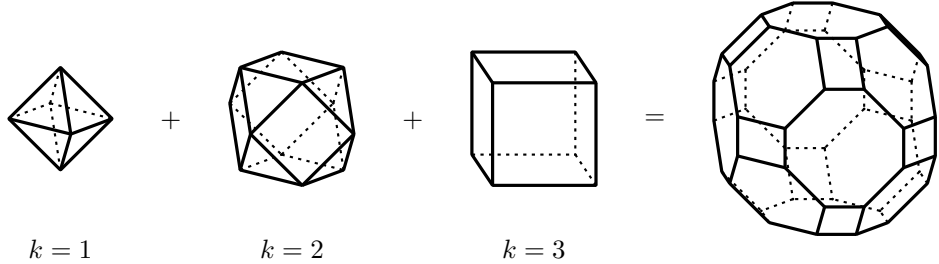}
\caption{The lineup polytope of the vertices of the octahedron for $r=3$, expressed as a Minkowski sum of a $1$-set, $2$-set and $3$-set polytope.}
\label{fig:ksetpolytopes}
\end{figure}

\begin{proposition}\label{prop:construction3}
Lineup polytopes are Minkowski sums of $k$-set polytopes given by
\[
\Lineup = \sum_{k=1}^r (w_k-w_{k+1}) \ksetpol,
\]
with the convention that $w_{r+1}=0$. 
\end{proposition}

\begin{proof}
This follows from the fact that for any $\ell=(\bv_{i_1}, \ldots, \bv_{i_r})\in \bV^r$, we have $\OccVec_{\bw}(\ell)=\sum_{t=1}^r w_t \bv_{i_t} = \sum_{k=1}^r (w_k-w_{k+1})\sum_{t=1}^k \bv_{i_t}$.
\end{proof}

\begin{corollary}[{see e.g.\cite[Proposition~7.12]{ziegler_lectures_1995}}]
	\label{cor:degenerate}
Let $1\leq r<r'$, the polytope $\Lineup$ is a weak Minkowski summand of $\Lineup[r']$, and in particular $\Rg[r']$ refines $\Rg$.
\end{corollary}

With this representation, we see that if $\bw\in \Pauli_{r-1}$ is non-strictly decreasing, the convex hull of weighted vectors gives a polytope that is a Minkowski summand (up to dilation) of a lineup polytope obtained with a strictly decreasing weight in $\Pauli_{r-1}^{\circ}$. 
The normal fan of this polytope coarsens $\Rg$.

\subsubsection{Monotone path polytopes}
Since permutohedra and sweep polytopes are monotone path polytopes, we can describe lineup polytopes as monotone path polytopes.
Let $\pol[P]$ be a $d$-polytope, and $\phi:\pol[P] \rightarrow \RR$ a linear functional given by $\phi(\bx)=\np{\bu, \bx}$ for a certain vector~$\bu\in \RR^d$. 
We denote $a_{\min}:=\min_{\bx\in \pol[P]}\phi(\bx)$ and $a_{\max}:=\max_{\bx\in \pol[P]}\phi(\bx)$. 
A $\phi$-\emph{monotone path} is a sequence $(\bx_1, \ldots, \bx_s)$ of vertices of $\pol[P]$ such that $a_{\min}=\phi(\bx_1)<\phi(\bx_2)<\ldots <\phi(\bx_s)=a_{\max}$.  Similarly, a $\phi$-\emph{cellular string} is a sequence of faces $\pol[F]_1,\dots,\pol[F]_k$ of~$\pol$ of dimension at least~$1$ such that $a_{\min}\in \phi(\pol[F]_1)$, $a_{\max}\in \phi(\pol[F]_k)$, and every pair of adjacent faces $\pol[F]_i,\pol[F]_{i+1}$ meet at a vertex $\bv$ such that $\phi(\bx)\leq \phi(\bv)\leq \phi(\by)$ for each $\bx\in \pol[F]_i$ and $\by\in \pol[F]_{i+1}$. Cellular strings are ordered by refinement, and the finest cellular strings are the edges of monotone paths.
One way to obtain $\phi$-cellular strings is to consider some vector $\bw$ orthogonal to~$\bu$ and consider the sequence of faces of $\pol[P]$ that are extreme in the direction $\bw+\alpha\bu$ as $\alpha$ ranges from $-\infty$ to~$\infty$. 
If $\bw$ is generic, this gives rise to a $\phi$-monotone path. 
The $\phi$-cellular strings that can be obtained this way are called \emph{coherent}.

The \defn{monotone path polytope} $\mathrm{Fiber}(\pol[P], \phi)$ is a polytope of dimension $d-1$ in $\RR^d$ whose faces are in bijection with coherent $\phi$-cellular strings of $\pol[P]$. 
In particular, its vertices are indexed by the $\phi$-monotone paths. 
As the notation indicates, it is indeed a fiber polytope.
Fiber polytopes are defined for any projection between two polytopes and their face posets are isomorphic to the posets of coherent subdivisions induced by the projection, see \cite{BS92}\cite[Section 9]{ziegler_lectures_1995}.
In the monotone path case $\mathrm{Fiber}(\pol[P],\phi)$ can be described as the Minkowski sum:
\begin{equation}\label{eq:monotonepathpolytope}
\mathrm{Fiber}(\pol[P], \phi)=\frac{1}{a_{max}-a_{min}} \sum_{j=1}^t \frac{a_j-a_{j-1}}{2}\left(\phi^{-1}(a_{j-1})+\phi^{-1}(a_j)\right),
\end{equation}
where $a_{min}=a_0 < a_1 < \dots < a_t=a_{max}$ are the ordered values taken by $\phi$ on the vertices of $\pol[P]$.

To make the connexion between monotone path polytopes and lineup polytopes we denote by $\bar{\bV}=\{\bar{\bv}_1,\dots,\bar{\bv}_m\}\subset\RR^{d+1}$ the \emph{homogenization} of $\bV$, consisting of the vectors $\bar{\bv}_i=(\bv_i,1)$, and define $\pol[Z]_r(\bar{\bV})$ the truncated zonotope associated to $\bV$ as the truncated Minkowski sum of segments:
\[\pol[Z]_r(\bar{\bV}) = \left( \sum_{i=1}^m [\mathbf{0}, \bar{\bv}_i] \right) \cap \left\{ \bx\in \RR^{d+1}: x_{d+1}\leq r\right\}.\]
Let $h:\pol[Z]_r(\bV)\to\RR$ denote the projection onto the last coordinate.

\begin{proposition}
Let $\bw \in \Pauli_r^{\circ}$.
The lineup polytope $\Lineup$ and the monotone path polytope $\mathrm{Fiber}(\pol[Z]_r(\bar{\bV}),h)$ have the same normal fan. 
In particular, if $\bw=\tfrac{2}{r(r+1)} (r, r-1, \ldots, 1)$, then
\[\Lineup \times \{1\} = \tfrac{2}{r+1} \mathrm{Fiber}(\pol[Z]_r(\bar{\bV}),h) + \tfrac{1}{r(r+1)} \ksetpol[r][\bV]\times \left\{\tfrac{1}{r+1}\right\}.\]
\end{proposition}

\begin{proof}
The description in Equation~\eqref{eq:monotonepathpolytope} of a monotone path polytope as a Minkowski sum of fibers gives
\[
\resizebox{\textwidth}{!}{$
\begin{array}{rl}
\mathrm{Fiber}(\pol[Z]_r(\bar{\bV}),h) &= \tfrac{1}{r} \sum_{k=1}^r \tfrac{1}{2}\left(h^{-1}(k-1) + h^{-1}(k)\right)= \tfrac{1}{r} \sum_{k=1}^r \tfrac{1}{2}\left(\ksetpol[k-1] \times \{k-1\} + \ksetpol \times \{k\} \right)\\[0.5em]
&= \tfrac{1}{r} \Big( \sum_{k=1}^{r-1} \ksetpol\times \{k\} + \tfrac{1}{2} \ksetpol[r]\times\{\tfrac{r}{2}\}\Big)= \tfrac{1}{r} \Big( \sum_{k=1}^r \ksetpol \Big)\times\{\tfrac{r+1}{2}\} - \tfrac{1}{2r}\ksetpol[r]\times\{\tfrac{1}{2}\}.
\end{array}$}
\]

On the other hand, for the weight $\bw=\tfrac{2}{r(r+1)} (r, r-1, \ldots, 1)$, the description of the lineup polytope as a sum of $k$-set polytopes gives
\[\Lineup=\sum_{k=1}^r \tfrac{2}{r(r+1)} \ksetpol.\]
Hence, $\Lineup \times \{1\} = \tfrac{2}{r+1} \mathrm{Fiber}(\pol[Z]_r(\bar{\bV}),h) + \tfrac{1}{r(r+1)} \ksetpol[r][\bV]\times \{\tfrac{1}{r+1}\}$. 
Moreover, as $\ksetpol[r]$ is a Minkowski summand of $\Lineup$, this equality implies that $\Lineup$ and $\mathrm{Fiber}(\pol[Z]_r(\bar{\bV}),h)$ have the same normal fan. 
\end{proof}

In particular, this construction gives the interpretation of $r$-lineups and $r$-rankings of $\bV$ as coherent monotone paths and coherent cellular strings of the truncated zonotope $\pol[Z]_r(\bar{\bV})$, respectively.
This polytope may also be interpreted using non-coherent monotone paths, as done in~\cite{PadrolPhilippe2021} for the case of sweeps ($r=m$), under the name of pseudo-sweeps. 
A \defn{$r$-pseudo-lineup} of~$\bV$ is an $r$-tuple $\ell=(\bv_{i_1},\dots,\bv_{i_r})\in\bV^r$ such that $\{\bv_{i_1},\dots,\bv_{i_k}\}$ is a $k$-set of $\bV$ for all $1\leq k\leq r$. 

\begin{remark}
\label{rem:pseudosweep}
Every $r$-lineup is an $r$-pseudo-lineup, but the converse is not true in general.
In Example \ref{ex:pseudosweep} we show an ordering that is not a lineup, but it is a pseudo-lineup.
We invite the reader to examine Figure~\ref{fig:bosonic_non_coherent} on page \pageref{fig:bosonic_non_coherent} to see how can one change the sweeping line slightly to achieve the desired order.
Pseudo-lineups are called \emph{broken line shellings} in \cite{heaton_dual_2020}.
\end{remark}

\subsection{Back to the challenge}

In Example \ref{ex:redefinition} it was shown that $\FerPoly[r]$ and $\BosPoly[r]$ are lineup polytopes, hence by Theorem~\ref{thm:vrepr_lineup} we can identify their normal fans as ranking fans 
\begin{equation*}
\fan[N](\FerPoly[r])= \FerFan, \text{ and}\quad \fan[N](\BosPoly[r])= \BosFan.
\end{equation*}

Since increasing $r$ further refines the fan we have the following hierarchy, where $\succeq$ denotes refinement of fans (each cone of the first is the union of a collection of cones of the second).

\begin{proposition}\label{prop:hierarchy}
	For fixed parameters $N,d$ we have
	\[
	\FerFan[1]\succeq \FerFan[2]\succeq \FerFan[3]\succeq \cdots \succeq\FerFan[D],
	\]
	and similarly for the bosonic case.
\end{proposition}

In particular, this means that as we increase $r$, we gain new inequalities while keeping the (normals of the) old ones.
We can be explicit about the inequality induced by each ray generator.

\begin{proposition}\label{prop:rhs}
	Let $\bw\in\Delta^{\circ}_{r-1}$. The facet defining inequality induced on $\FerPoly[r]$ by a  ray generator $\by\in\RR^d$ of $\FerFan$ is
	\[
	\np{\by,\bx}\leq \np{s_r(\by),\bw}\quad \text{for all } \bx\in\RR^d,
	\]
	where $s_r(\by)$ is the vector consisting of the $r$ largest $N$-sums of entries of $\by$ ordered decreasingly and the inner product takes place in $\RR^{r}$.
\end{proposition}

\begin{proof}
	Recall from Section \ref{sec:nfans} that the facet inequality on $\pol$ induced from an normal ray is given by $\np{\by,\bx}\leq \mathrm{supp}_{\pol[P]}(\by)$.
	We are in the case $\bV=\FerPC$ and $\pol[P]=\FerPoly[r]$,
	\[
	\mathrm{supp}_{\pol[P]}(\by)=\max_{\ell\in\lu}\np{\by,\OccVec_{\bw}(\ell)}=\max_{\substack{S_1,\dots,S_r\subset[d]\\|S_i|=N}}\sum_{i=1}^r w_i\np{\by,\bchi(S_i)},
	\]
	but $\np{\by,\bchi(S)}$ is the sum of the entries of $\by$ in the coordinates indexed by~$S$.
	By the rearrangement inequality~\cite[Theorem~368]{hardy_inequalities_1988}, since the entries of $\bw$ are ordered decreasingly, the maximum on the right-hand side is attained when the partial sums are also ordered decreasingly.
\end{proof}

We apply Proposition~\ref{prop:rhs} to obtain an explicit $H$-representation in Equation \eqref{eq:rhs_example} on page~\pageref{eq:rhs_example}.
Since the point configurations $\FerPC$ and $\BosPC$ are $\Sym{d}$-invariant, so are the polytopes $\FerPoly[r]$ and $\BosPoly[r]$.
We call a \defn{fundamental lineup} a lineup induced by a fundamental linear functional, and denote
\[
\begin{split}
\non     &:=\{\OccVec_{\bw}(\ell)~:~\ell\in\lu[\FerPC]\text{ and } \ell=\ell(\by)\text{ for some } \by\in\Phi_d\},\\
\non[b]  &:=\{\OccVec_{\bw}(\ell)~:~\ell\in\lu[\BosPC]\text{ and } \ell=\ell(\by)\text{ for some } \by\in\Phi_d\},
\end{split}
\]
the set of fundamental occupation vectors.
In Proposition~\ref{prop:non_fundamental} we prove that vectors in $\non$ and $\non$ are indeed fundamental.

\begin{example}\label{ex:4-10}
The set 
$\non[f][3][4][10]$ consists of only two vectors corresponding to the two possible lineups:
\[
\resizebox{\textwidth}{!}{$
\begin{array}{r@{\hspace{1pt}}l@{\hspace{0.25cm}}r@{\hspace{1pt}}l}
\ell_1&=(1,2,3,4),(1,2,3,5),(1,2,3,6),& \OccVec_{(w_1,w_2,w_3)}(\ell_1)&=(1,1,1,w_1,w_2,w_3,0,0,0,0),\\
\ell_2&=(1,2,3,4),(1,2,3,5),(1,2,4,5),& \OccVec_{(w_1,w_2,w_3)}(\ell_2)&=(1,1,w_1+w_2,w_1+w_3,w_2+w_3,0,0,0,0,0).
\end{array}$}
\]
\end{example}	

We end by describing some fundamental normal rays.
By definition, the fundamental fan is a subdivision of $\Phi_d$, so it always contains the rays spanned by $\bef_1,\dots,\bef_{d-1}$ (recall Convention~\ref{convention}).
We now determine which elements in the fundamental basis are rays in  $\FerFan$ and $\BosFan$, for some fixed parameters $N,d$.
In the fermionic case, these rays are called \emph{Grassmannian inequalities} in \cite{altunbulak_pauli_2008}.

\begin{proposition}\label{prop:basic_cleanup_2}
	Let $r\geq1$, $N\geq r-1$, $d\geq r+N-1$.
	Among the rays spanned by $\bef_1,\dots,\bef_{d-1}$ only the rays spanned by $\bef_1,\bef_{N},\bef_{d-1}$ are in $\FerFan$, and only the ray spanned by $\bef_{d-1}$ is in $\BosFan$.
\end{proposition}

Before embarking into the proof we remark that the conditions on $N,d$ are there to guarantee the existence of the two special lineups used in the proof.
If we ignore the conditions and let $r$ be as large as possible, then all elements of fundamental basis span a normal ray of $\FerFan$.
The restriction on $N$ and $d$ relative ro $r$ are also important in Section \ref{ssec:stability}.

\begin{proof}[Proof of Proposition \ref{prop:basic_cleanup_2}]
	We do first the fermionic case. Let $\pol=\FerPoly[r]$. 
	By Lemma~\ref{lem:homogeneous}, $\dim \pol=d-1$, so its facets have dimension $d-2$.
	We need to consider the subset $\bV_{\bef_i}$ of $\non$ that maximizes $\np{\bef_i,\cdot}$.
	Then we can use Theorem~\ref{thm:dimension} to determine the dimension of the face $\pol^{\bef_i}$.
	There are three cases.
	
	\emph{Case $i<N$}. 
Using Proposition~\ref{prop:rhs}, we see that the maximum value of $\np{\bef_i, \cdot}$ on $\non$ is~$i$ and it is achieved for occupation vectors such that all elements of the corresponding lineup are of the form $\bchi([i]\cup \{i+j\, ,\, j\in S\})$ with $S\in \binom{[d-i]}{N-i}$. 
	Hence the set $\bV_{\bef_i}$ is an affine embedding of $\non[f][r][N-i][d-i]$ and so the face $\pol^{\bef_i}$ is isomorphic to $\FerPoly[r][N-i,d-i]$.
	Since $\dim \FerPoly[r][N-i,d-i]=d-i-1$, the vector $\bef_i$ defines a facet only when $i=1$.
	
	\emph{Case $i=N$}.
	The maximum value of $\np{\bef_N, \cdot}$ on $\non$ is $Nw_1 + (N-1)\sum_{j=2}^r w_j$ and it is achieved for occupation vectors such that the corresponding lineup is of the form $(\bchi([N]), \bchi(S_2), \ldots, \bchi(S_r))$ with $|S_j\cap [N] |=N-1$ for all $2\leq j \leq r$. 
We use Theorem~\ref{thm:dimension} to compute the dimension of $\pol^{\bef_N}$. The Young subgroup corresponding to $\bef_N$ is $\Sym{N}\times\Sym{d-N}$. We see that $\Sym{N}$ acts non-trivially on the
occupation vector of $\bV_{\bef_N}$ corresponding to the lineup $(\bchi([N-1]\cup\{N+i-1\}))_{i=1}^{r}$,
and $\Sym{d-N}$ acts non-trivially on the occupation vector of $\bV_{\bef_N}$ corresponding to the lineup $(\bchi([N+1]\setminus\{N+2-i\}))_{i=1}^{r}$. 
Hence $\pol^{\bef_N}$ has dimension $d-2$ and it is a facet.
	
	\emph{Case $i>N$}.
	The linear functional is maximized by lineups all whose elements are contained in $[i]$.
	Therefore the set $\bV_{\bef_i}$ consists of the occupation vectors whose last $d-i$ coordinates are equal to 0.
	Similarly to the case $i<N$, $\bV_{\bef_i}$ is identified with $\non[f][r][N][d-(d-i)]$ and the face $\pol^{\bef_i}$ is isomorphic to the polytope $\FerPoly[r][N,i]$.
	Since $\dim \FerPoly[r][N,i]=i-1$, the vector $\bef_i$, defines a facet only when $i=d-1$.
	
	The bosonic case is similar and there is no need to distinguish the values of $i$. Let $\pol=\BosPoly[r]$. 
	 The maximum value of $\np{\bef_i, \cdot}$ on $\non[b]$ is $N$ and it is achieved for occupation vectors such that all elements of the corresponding lineup are multisubsets in $\multiset{[i]}{N}$. 
	Hence the set $\bV_{\bef_i}$ is an affine embedding of $\non[b][r][N][i]$ and the face $\pol^{\bef_i}$ is isomorphic to $\BosPoly[r][N,i]$.
	Since $\dim \BosPoly[r][N,i]=i-1$, the vector $\bef_i$ defines a facet only when $i=d-1$.
\end{proof}

\section{Gale orders}
\label{sec:gale}

The point configurations $\FerPC$ and $\BosPC$ have some extra combinatorial structure that we now exploit to describe the vertices of the spectral polytopes $\FerPoly[r]$ and $\BosPoly[r]$.
In this section we review the Gale order, a partial order on these configurations, and use it to study fundamental occupation vectors.

\subsection{Gale order on subsets}

To obtain all vectors in $\non$, we use the \defn{Gale order} on~$\binom{[d]}{N}$ \cite[Section~1.3]{borovik_coxeter_2003}.
Given two $N$-subsets $S=\{s_1,\dots,s_N\}$ and $T=\{t_1,\dots,t_N\}$, ordered from smallest to largest, we say that $ S\leq T$ if and only if  $s_k\leq t_k,$ for all $1\leq k\leq N.$
This order defines the \emph{Gale poset} $\FerPoset$ on $\binom{[d]}{N}$.
We refer the reader to \cite[Chapter~3]{stanley_enumerative_2012} for background on posets.
An \defn{order ideal} of a poset~$P$ is a subposet ${Q \subseteq P}$ that satisfies $x \in Q$ and $y \leq x  \Rightarrow y \in Q$.
The order ideals of a poset~$P$ can be ordered by containment to get the distributive \emph{lattice of order ideals} $\J(P)$.
A \emph{saturated chain} in $\J(P)$ is a sequence $(Q_0,Q_1,\dots,Q_r)$ of ideals of $\J(P)$ such that $|Q_i|=i$, for each $0\leq i\leq r$, and $Q_{i-1}\subset Q_{i}$, for $i\in[r]$.

\begin{lemma}
\label{lem:Fthreshold}
Let $r\geq 1$ and $S_i\in\binom{[d]}{N}$, for $i\in[r]$.
If $(\bchi(S_1),\dots,\bchi(S_r))$ is a fundamental lineup of length $r$ of $\FerPC$, then 
\begin{enumerate}[(i)]
\item for each $k\in[r]$, the set $\{S_1,\dots,S_k\}$ is an order ideal of~$\FerPoset$, and
\item $\{S_1\}\subset \{S_1,S_2\}\subset \cdots \subset\{S_1,\dots,S_r\}$ is a saturated chain of ideals of $\J(\FerPoset)$.
\end{enumerate}
\end{lemma}

The vectors in $\non$ are indeed fundamental, conveniently fitting the convention of expressing a $\Sym{d}$-invariant polytope using generators in the fundamental chamber.

\begin{proposition}\label{prop:non_fundamental}
Let $d\geq N\geq 1$ and $r\in[\binom{d}{N}]$.
The occupation vectors in $\non$ are contained in $\Phi_d$.
\end{proposition}
\begin{proof}
Let $(\bchi(S_1),\dots,\bchi(S_r))$ be a fundamental lineup.
The $k$-th coefficient of its occupation vector is $\sum_{S_j\ni k} w_j$.
Let $I_k:=\{j\in[r]~:~k\in S_j\}$. It is enough to show that $|I_k|\geq |I_{k+1}|$ and that the $i$-th element of $I_{k+1}$ is larger or equal than the $i$-th element of $I_k$. 
This is a consequence of Lemma~\ref{lem:Fthreshold}:
For every set $S_j$ involved in the lineup such that $k+1\in S_j$ and $k\notin S_j$, the order ideal property implies that there must exist an index $j'<j$ such that $S_{j'}=S_j\setminus\{k+1\}\cup\{k\}$.
The relation $j\to {j'}$ is a decreasing injection between $I_{k+1}\setminus I_{k}$ and $I_{k}\setminus I_{k+1}$.
\end{proof}

However, as the examples below show, the converse of Lemma~\ref{lem:Fthreshold} is not true: not every order ideal of $\FerPoset$ comes from a lineup, and not every saturated chain of order ideals arises from a lineup (even if all the involved ideals do).
So further analysis is required.

\begin{definition}[Threshold fermionic ideals and coherent saturated chains]
\label{def:ferm_threshold}
Let $r\geq 1$ and $S_i\in\binom{[d]}{N}$, for $i\in[r]$.
If $(\bchi(S_1),\dots,\bchi(S_r))$ is a fundamental lineup of length $r$ of $\FerPC$, then 
\begin{enumerate}[(i)]
\item for each $k\in[r]$, the set $\{S_1,\dots,S_k\}$ is called a \defn{threshold (fermionic) ideal}  of~$\FerPoset$, and
\item the saturated chain $\{S_1\}\subset \{S_1,S_2\}\subset \cdots \subset\{S_1,\dots,S_r\}$ is called \emph{coherent}.
\end{enumerate}
The collection of threshold ideals of $\FerPoset$ ordered by inclusion is a subposet $\Thres{\FerPoset}$ of $\J(\FerPoset)$.
\end{definition}

See Figure~\ref{fig:big_gale} for an illustration of threshold ideals of the Gale poset with $N=5$.

\begin{figure}[!ht]

\begin{tikzpicture}%
[scale=0.400000,
x={(0.8, 0)},
y={(0,4.0)},
edge/.style={thick},
vertex/.style={inner sep=1pt,circle,draw=black,fill=black,thick},
vertex2/.style={inner sep=2pt,circle,draw=black,fill=white,thick},
rotate=-90]

\coordinate (12345) at (0,0);
\coordinate (12346) at (0,1);
\coordinate (12347) at (-1,2);
\coordinate (12348) at (-2,3);
\coordinate (12349) at (-4,4);
\coordinate (123410) at (-6,5);
\coordinate (12356) at (1,2);
\coordinate (12357) at (0,3);
\coordinate (12358) at (-2,4);
\coordinate (12359) at (-4,5);
\coordinate (12367) at (0,4);
\coordinate (12368) at (-2,5);
\coordinate (12456) at (2,3);
\coordinate (12457) at (2,4);
\coordinate (12458) at (0,5);
\coordinate (12467) at (2,5);
\coordinate (13456) at (4,4);
\coordinate (13457) at (4,5);
\coordinate (23456) at (6,5);

\draw[edge] (12345) -- (12346);
\draw[edge] (12346) -- (12347);
\draw[edge] (12346) -- (12356);
\draw[edge] (12347) -- (12348);
\draw[edge] (12347) -- (12357);
\draw[edge] (12348) -- (12349);
\draw[edge] (12348) -- (12358);
\draw[edge] (12349) -- (123410);
\draw[edge] (12349) -- (12359);
\draw[edge] (12356) -- (12357);
\draw[edge] (12356) -- (12456);
\draw[edge] (12357) -- (12358);
\draw[edge] (12357) -- (12367);
\draw[edge] (12357) -- (12457);
\draw[edge] (12358) -- (12359);
\draw[edge] (12358) -- (12368);
\draw[edge] (12358) -- (12458);
\draw[edge] (12367) -- (12368);
\draw[edge] (12367) -- (12467);
\draw[edge] (12456) -- (12457);
\draw[edge] (12456) -- (13456);
\draw[edge] (12457) -- (12458);
\draw[edge] (12457) -- (12467);
\draw[edge] (12457) -- (13457);
\draw[edge] (13456) -- (13457);
\draw[edge] (13456) -- (23456);

\node[vertex2,label=below:$12345$] at (12345) {};
\node[vertex2,label=above:$12346$] at (12346) {};
\node[vertex2,label=above:$12347$] at (12347) {};
\node[vertex2,label=above:$12348$] at (12348) {};
\node[vertex2,label=above:$12349$] at (12349) {};
\node[vertex,label=right:${123410}$] at (123410) {};
\node[vertex2,label=above:$12356$] at (12356) {};
\node[vertex2,label=above:$12357$] at (12357) {};
\node[vertex,label=above:$12358$] at (12358) {};
\node[vertex,label=right:$12359$] at (12359) {};
\node[vertex,label=above:$12367$] at (12367) {};
\node[vertex,label=right:$12368$] at (12368) {};
\node[vertex2,label=above:$12456$] at (12456) {};
\node[vertex,label=above:$12457$] at (12457) {};
\node[vertex,label=right:$12458$] at (12458) {};
\node[vertex,label=right:$12467$] at (12467) {};
\node[vertex2,label=above:$13456$] at (13456) {};
\node[vertex,label=right:$13457$] at (13457) {};
\node[vertex,label=right:$23456$] at (23456) {};

\end{tikzpicture}

\caption{The Hasse diagram (depicted from left to right) of the Gale order for $N=5$ and elements of rank at most $6$. 
Elements in white are contained in some ideal of cardinality $5$, hence are potentially part of a threshold ideal with five $5$-subsets.}
\label{fig:big_gale}
\end{figure}
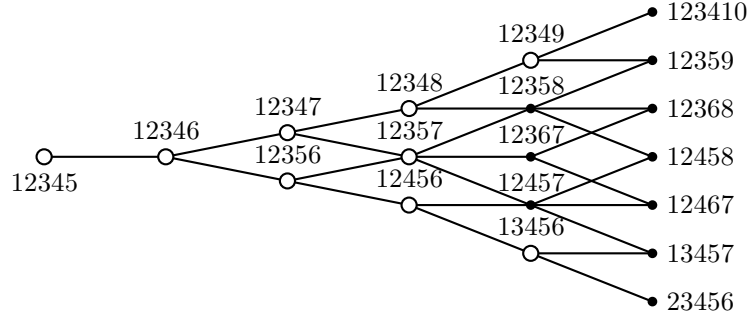
To illustrate the level of complexity involved in obtaining the occupation vectors, we exhibit  their intrinsic geometric nature in the next two examples.
First, we observe that there are order ideals that are not threshold (that is, that the associated points do not form a $k$-set of $\FerPC$).
Second, there are also saturated chains of threshold ideals in $\Thres{\FerPoset}$ that are not coherent.

\begin{example}[Non-threshold fermionic ideal]
\label{ex:notferthres}
Consider $d=9$ and $N=3$, the ideal $Q$ taken from \cite[Theorem~3.1]{klivans_shifted_2008} generated by the elements $178,239$, and $456$ has $36$ elements and is not threshold.
The convex hull ${\conv\{\bchi(S)~:~S\in Q\}}$ intersects the convex hull $\conv\left\{\bchi(S)~:~S\in\binom{[9]}{3}\setminus Q\right\}$ at the point $\frac{1}{3}\bef_9$.
This shows that no (strictly) separating hyperplane for $Q$ exists, hence it can not be a threshold fermionic complex.
\end{example}

\begin{example}[Saturated chain not giving a fermionic lineup]\label{ex:noncoherent_fermionic}
Consider the sequence of elements $(12,13,14,15,23,24,25,34,16,26)\in\FerPoset[2][6]$.
When $N=2$ all order ideals are threshold \cite[Theorem~3.1]{klivans_shifted_2008} so order ideals formed by the initial segments of the sequence are all threshold.
We claim that the sequence
\[
(\bchi(12),\bchi(13),\bchi(14),\bchi(15),\bchi(23),\bchi(24),\bchi(25),\bchi(34),\bchi(16),\bchi(26))
\]
is not a lineup.
For the sake of contradiction, assume that $\by$ is a linear functional inducing~$\ell$.
We must have
\[
\begin{array}{l@{\hspace{1cm}}c@{\hspace{1cm}}r}
\np{\by,\bchi(15)}>	\np{\by,\bchi(24)} & \Longleftrightarrow & y_1-y_2 >y_4-y_5,\\
\np{\by,\bchi(34)}>	\np{\by,\bchi(16)} & \Longleftrightarrow & y_4-y_6 >y_1-y_3,\\
\np{\by,\bchi(26)}> \np{\by,\bchi(35)} & \Longleftrightarrow & y_2-y_3 >y_5-y_6.
\end{array}
\]
Adding the three inequalities on the right-hand side, we arrive at $y_1+y_4-y_3-y_6>y_1+y_4-y_3-y_6$ which is plainly false.
According to Lemma \ref{lem:Fthreshold}, after the first nine elements, there are only two possibilities for the tenth one $(26)$ or $(34)$, but $(26)$ is not possible, leaving $(34)$ as the only possibility to extend it to a threshold ideal.
\end{example}

\begin{remark}
In addition to Example \ref{ex:noncoherent_fermionic} we point to \cite[Example 10]{heaton_dual_2020} for an example of a pseudo-lineup of $\FerPC[2][4]$ that is not a lineup.
The importance of pseudo-lineups is emphasized in \cite[Conjecture 2]{heaton_dual_2020} where Heaton and Samper conjecture that the ordering coming from a pseudo-lineup might help to prove Stanley's pure O-sequence conjecture \cite[Page 59]{stanley_cohen_1976}.
\end{remark}

Example~\ref{ex:noncoherent_fermionic} shows that lineups alone do not cover all saturated chains of threshold fermionic ideals.
However, pseudo-lineups do, see Remark~\ref{rem:pseudosweep}.

\begin{proposition}\label{prop:pseudo_lineups}
	Saturated chains of threshold fermionic ideals are equivalent to fundamental pseudo-lineups of~$\FerPC$.
\end{proposition}

\begin{proof}
Each threshold ideal in the sequence is obtained from a linear functional that splits the ideal from the other elements.
Two consecutive splittings differ exactly by one element, making it possible to choose \emph{distinct} linear functionals at each step and connect them through a rotation along the codimension-$2$ intersection of the two corresponding hyperplanes.
\end{proof}
We summarize our observations in the following chain of strict inclusions
\[
\left\{\parbox[c][1cm]{3cm}{Coherent fermionic sequences.}\right\}\subsetneq
\left\{\parbox[c][1cm]{4cm}{Saturated chains of threshold fermionic ideals.}\right\}\subsetneq
\left\{\parbox[c][1cm]{4cm}{Saturated chains of fermionic order ideals.}\right\}.
\]

\begin{remark}
Ideals and threshold ideals are studied, for instance, in combinatorial commutative algebra and the topology of finite simplicial complexes.
In these contexts, order ideals of $\FerPoset$ are called \emph{pure shifted complexes} \cite{klivans_threshold_2007,klivans_shifted_2008}.
Here, a lineup of length $r$ gives rise to a pure shifted complex containing $r$ facets, that are totally ordered by a linear functional, therefore yielding a shelling order of the complex.
As such, they are called \emph{threshold complexes} \cite{klivans_threshold_2007}\cite{edelman_simplicial_2013}.
In Section~\ref{ssec:fer_vs_bos}, we show a fundamental difference between the classical threshold complexes and the threshold complexes arising from \emph{bosonic} threshold ideals, discussed in the next section.
\end{remark}

\begin{remark}
The poset of $r$-rankings is isomorphic to the face-lattice of the $r$-lineup polytope. Even if not all $r$-pseudo-rankings are rankings, a particular instance of the \emph{generalized Baues problem}~\cite{BS92,BilleraKapranovSturmfels1994}\cite{Reiner1999} states that the boundary of the $r$-lineup polytope is a strong deformation retract of the order complex of the poset of $r$-pseudo-rankings.
\end{remark}

\subsection{Gale order on multisubsets}
The Gale order $\BosPoset$ on multisubsets in $\multiset{[d]}{N}$ is defined using the same order relation of $\FerPoset$ for subsets.
The analogue version of Lemma~\ref{lem:Fthreshold} for $\BosPC$ is valid, i.e. the underlying sets of lineups of $\BosPC$ form \defn{threshold bosonic ideals}.

\begin{lemma}
\label{lem:Bthreshold}
Let $r\geq 1$ and $S_i\in\multiset{[N]}{d}$, for $i\in[r]$.
If $(\bchi(S_1),\dots,\bchi(S_r))$ is a fundamental lineup of length~$r$ of $\BosPC$, then 
\begin{enumerate}[(i)]
\item for each $k\in[r]$, the set $\{S_1,\dots,S_k\}$ is an order ideal of~$\BosPoset$, and
\item $\{S_1\}\subset \{S_1,S_2\}\subset \cdots \subset\{S_1,\dots,S_r\}$ is a saturated chain of ideals of $\J(\BosPoset)$.
\end{enumerate}
\end{lemma}

The notion of coherent saturated chains of threshold fermionic ideals transfers naturally to obtain \emph{coherent saturated chains of threshold bosonic ideals}.
To obtain the lineups of $\BosPC$, we make use of a common trick to translate between subsets and multisubsets.

\begin{definition}[Natural map between $\FerPoset$ and $\BosPoset$]
Let $d\geq N\geq 1$.
The \emph{natural map}\footnote{It should not be confused with the natural inclusion of $\FerPoset$ into $\BosPoset$.} sends an element $S\in\NN^d$ of the Gale poset $\FerPoset$---considered as an increasing vector $(S_1,S_2,\dots,S_N)$---to the multisubset with $N$ elements in $[d-N+1]$ given by the vector $(S_1,S_2-1,\dots,S_N-N+1)\in\BosPoset[N][d-N+1]$.
\end{definition}

For example the element $1234$ is mapped to $1111$ and the element $1256$ is mapped to $1133$.
The natural map induces a poset isomorphism between $\FerPoset$ and $\BosPoset[N][d-N+1]$, hence a bijection between their order ideals.
However, under the natural map threshold fermionic ideals $\Thres{\FerPoset}$ \emph{do not} correspond to threshold bosonic ideals
$\Thres{\BosPoset[N][d-N+1]}$, see Example~\ref{ex:contrast} on page~\pageref{ex:contrast}.

\part{The Solution}
\label{part3}

We can finally face the challenge stated in Section \ref{ssec:challenge}: provide an $H$-representation of the polytopes $\FerPoly[r]$ and $\BosPoly[r]$.
We use the combinatorial methods of Part \ref{part2} to provide a general algorithm that computes this $H$-representation.
Combining the algorithm with stability results we provide a \emph{complete} solution for small values of $r$.
Recall that to simplify the notation, we omit the mention to $\bw$ in $\FerPoly[r][\bw,N,d]$ and denote it $\FerPoly[r]$, when $\bw$ is some fixed vector with strictly decreasing entries that is clear from the context and similarly for $\BosPoly[r]$.
In the case where $\bw$ has repeated entries, all the inequalities remain valid but some of them may become redundant (see Corollary \ref{cor:degenerate}).

\section{$H$-representation of fermionic and bosonic spectral polytopes}
\label{sec:Hrepr_spectral}

\subsection{Recursive generation}
\label{ssec:recursive}

We describe how to use Lemma~\ref{lem:Fthreshold} to recursively compute all possible fundamental lineups in Algorithm~\ref{algo:lineup} while simultaneously computing the corresponding fundamental cones.
We focus on the fermionic case. 
The bosonic case works analogously.
Let~$\ell$ be a fundamental lineup of $\FerPC$ and
\[
\mathdutchcal{R}(\ell):=\{S\in\FerPoset~:~{\bchi(S)\not\in\ell\text{ and }\bchi^{-1}(\ell)\cup\{S\}} \text{ is an ideal of }\FerPoset\}
\]
be the set of \defn{runner-ups} of $\ell$.
See Figure~\ref{fig:gale} for an example of a fundamental lineup with three runner-ups.

\begin{example}
If $N=4$, $d=6$ and $\by$ is a fundamental vector, then having $1246$ in a lineup implies that $1234$, $1235$, $1236$ and $1245$ should also be in the lineup, see Figure~\ref{fig:gale}.
Then, any lineup of length $6$ containing these five elements should contain exactly one element in $\{1237,1256,1345\}$.
\end{example}

\begin{figure}[H]
\centering
\begin{tikzpicture}%
[scale=1.000000,
	x={(0.4, 0)},
	y={(0,1.5)},
	edge/.style={thick},
	edge2/.style={dotted,thick,shorten >=1cm},
	vertex/.style={inner sep=1pt,circle,draw=black,fill=black,thick},
	vertex2/.style={inner sep=1.5pt,circle,draw=black,fill=white,thick},
	vertex3/.style={inner sep=2pt,rectangle,draw=black,fill=white,thick},
	rotate=-90]

	\coordinate (12345) at (0,0);
	\coordinate (12346) at (0,1);
	\coordinate (12347) at (-1,2);
	\coordinate (12348) at (-2,3);
	\coordinate (12349) at (-4,4);
	\coordinate (123410) at (-6,5);
	\coordinate (12356) at (1,2);
	\coordinate (12357) at (0,3);
	\coordinate (12358) at (-2,4);
	\coordinate (12359) at (-4,5);
	\coordinate (12367) at (0,4);
	\coordinate (12368) at (-2,5);
	\coordinate (12456) at (2,3);
	\coordinate (12457) at (2,4);
	\coordinate (12458) at (0,5);
	\coordinate (12467) at (2,5);
	\coordinate (13456) at (4,4);
	\coordinate (13457) at (4,5);
	\coordinate (23456) at (6,5);

	\draw[edge] (12345) -- (12346);
	\draw[edge] (12346) -- (12347);
	\draw[edge] (12346) -- (12356);
	\draw[edge] (12347) -- (12348);
	\draw[edge] (12347) -- (12357);
	\draw[edge] (12348) -- (12349);
	\draw[edge] (12348) -- (12358);
	\draw[edge2] (12349) -- (123410);
	\draw[edge2] (12349) -- (12359);
	\draw[edge] (12356) -- (12357);
	\draw[edge] (12356) -- (12456);
	\draw[edge] (12357) -- (12358);
	\draw[edge] (12357) -- (12367);
	\draw[edge] (12357) -- (12457);
	\draw[edge2] (12358) -- (12359);
	\draw[edge2] (12358) -- (12368);
	\draw[edge2] (12358) -- (12458);
	\draw[edge2] (12367) -- (12368);
	\draw[edge2] (12367) -- (12467);
	\draw[edge] (12456) -- (12457);
	\draw[edge] (12456) -- (13456);
	\draw[edge2] (12457) -- (12458);
	\draw[edge2] (12457) -- (12467);
	\draw[edge2] (12457) -- (13457);
	\draw[edge2] (13456) -- (13457);
	\draw[edge2] (13456) -- (23456);

	\node[vertex2,label=below:$1234$] at (12345) {};
\node[vertex2,label=above:$1235$] at (12346) {};
\node[vertex2,label=above:$1236$] at (12347) {};
\node[vertex3,label=above:$1237$] at (12348) {};
\node[vertex,label=above:$1238$] at (12349) {};
\node[vertex2,label=below:$1245$] at (12356) {};
\node[vertex2,label=above:$1246$] at (12357) {};
\node[vertex,label=above:$1247$] at (12358) {};
\node[vertex3,label=above:$1256$] at (12367) {};
\node[vertex3,label=below:$1345$] at (12456) {};
\node[vertex,label=below:$1346$] at (12457) {};
\node[vertex,label=below:$2345$] at (13456) {};

\end{tikzpicture}

\caption{The Hasse diagram of the Gale order for $N=4$ and $d=8$ for elements of rank at most $5$.
The elements in white circles form an ideal of cardinality $5$, hence are potentially part of a lineup of length $5$.
Only the three elements represented by white squares may be added to the ideal while remaining an ideal.}
\label{fig:gale}
\end{figure}
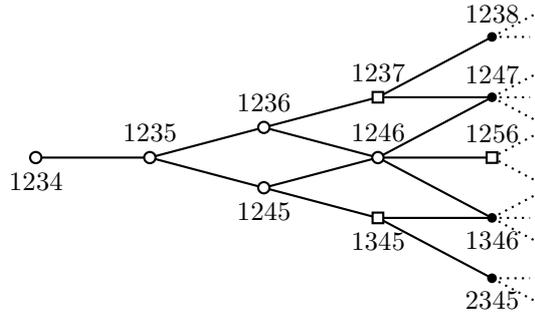

As shown in Example \ref{ex:noncoherent_fermionic}, we must verify, for each runner-up, whether the ideal obtained by it remains induced by a linear functional.
So every time we append an element to a lineup, we certify that it is coherent by computing all fundamental linear functions that induce it.
The following algorithm recursively computes a minimal $H$-representation of $\FerPoly[r]$, providing an effective solution to the convex $1$-body $N$-representability problem.
First, it recursively enumerates all fundamental lineups together with their fundamental cones.
Then, for each ray generating a fundamental cone, it uses Theorem~\ref{thm:dimension} to assert that it defines a facets of the polytope $\FerPoly[r]$.
If the ray defines a facet, the right-hand side of the corresponding inequality is determined by examination on the fundamental occupation vectors corresponding to the fundamental lineups.

\begin{algorithmplain}[{$H$-representation of $\FerPoly[r]$}]
Given $d\geq N\geq 1$ and $r\in[\binom{d}{N}]$, the following procedure gives a non-redundant $H$-representation of $\FerPoly[r]$.
\label{algo:lineup}
\begin{description}
\item[Step 1. Determine candidates for normal rays]
\item[If $r=1$] There is exactly one possible fundamental lineup, namely the subset $\{1,\dots,N\}$. 
The procedure returns $(\bchi(\{1,\dots,N\}))$ and the fundamental cone $\Phi_d$.
\item[Else ($r>1$)] Consider every possible lineup $\ell$ of length $r-1$ with corresponding fundamental cone $\pol[K](\ell)$.
For every such $\ell$, determine the set of runner-ups $\mathdutchcal{R}(\ell)$, and for each $S\in \mathdutchcal{R}(\ell)$, we construct
\begin{equation}\label{eq:keystep}
\pol[K](\ell) \cap\left( \bigcap_{T\in\mathdutchcal{R}(\ell)}\{\by\in\RR^d~:~\langle\by,\bchi(S)-\bchi(T)\rangle\geq 0\}\right).
\end{equation}
If the cone in Equation \eqref{eq:keystep} is $d$-dimensional, then $\ell$ appended with $\bchi(S)$ is a fundamental lineup of length $r$, and Equation \eqref{eq:keystep} gives an $H$-representation of its fundamental cone.
\end{description}
\begin{description}
\item[Step 2. $(V\to H)$-translation]
\item[Rays] Having obtained the fundamental fan with an $H$-representation of each maximal cone, obtain the fundamental rays generators.
\end{description}
\begin{description}
\item[Step 3. Assert normal rays]
\item[Facet check] Use Theorem~\ref{thm:dimension} to discard the fundamental rays that are not normal rays.
\item[Right-hand side] Evaluate each normal ray generator on the occupation vectors found to determine the right-hand side of the corresponding inequality. 
\end{description}
\end{algorithmplain}

This bottom-up dynamic programming algorithm computes (an $H$-representation of) the inclu\-sion-maximal cones of the fundamental fan of $\FerPoly[r]$ starting with $r=1$ and increasing $r$ by $1$ at each step.
Determining the runner-ups $\mathdutchcal{R}(\ell)$ can be done in $O(r)$.
Assume that $m$ fundamental lineups were obtained at step $1$ for some $r$.
Step 2 translates the $H$-representation of every fundamental cone to a $V$-representation.
Fortunately, the $H$-representation has a minimal number of inequalities, i.e. $m-1$, one for each comparison between a fixed fundamental cone and the others.
In total, there are only $\binom{m}{2}$ inequalities to keep track of.
Adding the $d-1$ inequalities stemming from the fundamental cone $\Phi_d$, there are always at most $d+m-2$ inequalities to handle for each cone.
The fact that there is at most $d+m-2$ inequalities is a great advantage.
This signifies that the number of inequalities increases only linearly with the size of the previous output.
The symmetry allows the complete description of the spectral polytopes using only fundamental occupation vectors and fundamental defining inequalities.
For instance, the $7$-dimensional $\FerPoly[5][8,4]$ has $154\ 560$ vertices and $7\ 366$ facets, but only 10 orbits of vertices and 9 orbits of facets.

\subsection{Case study}
\label{ssec:hrepr_spec}

\subsubsection{The hypersimplex $\pol[H](N,d)$}

Let $(N,d)=(3,6)$.
The polytope $\pol[H](3,6)$ has $20$ vertices and is centrally symmetric.
Thus, for $r>10$ each lineup is extended to exactly 1 new lineup, i.e. lineups of length $10$ determine all lineups of length $20$.
Furthermore, as the fan $\Rg[r+1]$ is a refinement of the fan $\Rg$, we have the following chain of refinements of the normal fans (see Proposition \ref{prop:hierarchy}), which induces a chain of inclusions on the corresponding sets of normal rays:
$\Rg[1]\succeq\Rg[2]\succeq\cdots\succeq\Rg[10]$, in particular for $\bV=\FerPC[3][6]$, the vertices of $\pol[H](3,6)$.
The number of lineups at each step are shown in Table~\ref{tab:lineups_h36} along with the number of \emph{new} normal rays in the fundamental cone.

\begin{table}[!htbp]
\begin{center}
\begin{tabular}{ c|c|c|c|c|c|c|c|c|c|c} 
r &1&2&3&4&5&6&7&8&9&10 \\\hline
\# lineups          & 1 & 1 & 2 & 4 & 8 & 18 & 40 & 90 & 168 & 324\\
\# new inequalities & 2 & 1 & 1 & 2 & 3 & 3  &  5 &  9 &  14 & 32
\end{tabular}
\end{center}
\caption{Number of lineups of length $r$ for the hypersimplex $\pol[H](3,6)$ and corresponding number of new inequalities}
\label{tab:lineups_h36}
\end{table}

Summing give $72$ normal rays for the lineup polytope of the hypersimplex $\pol[H](3,6)$ spanning the normal cones of $324$ vertices.
A minimal $H$-representation of the lineup polytope is presented in Appendix~\ref{app:hyper36}.
The normal fan of $\Rg[10]$ has $233\ 280$ full-dimensional cones and $29\ 582$ rays.
For $(N,d)=(3,7)$, the polytope $\pol[H](3,7)$ has $35$ vertices.
The lineup polytope $\FerPoly[18][3,7]$ has 95 941 440 vertices and 5 910 198 facets.
The fundamental chamber contains 19036 occupation vectors and 1501 rays.

\subsubsection{Generation of generalized exclusion inequalities for $r\leq 13$}

We implemented the above algorithm to obtain the $H$-representation of $\FerPoly[r]$ and $\BosPoly[r]$ for $r\leq 13$ and the minimal choices for $N$ and $d$. 
Table~\ref{tab:ferbos} gathers the results.
The numbers show the increasing hierarchy of \emph{new} inequalities appearing after increasing the value of $r$ by $1$.

\begin{table}[!htbp]
\begin{center}
\resizebox{\columnwidth}{!}{$
\begin{array}{l|r|rrrrrrrrrrrrr}
& r                                        & 1 & 2 & 3 & 4 &  5 &  6 &  7 &   8 &    9 &   10 &    11 &    12 &   13\\\hline\hline
\multirow{2}{*}{Fermions} & \text{\#OVs}  & 1 & 1 & 2 & 4 & 10 & 28 & 90 & 312 & 1160 & 4518 & 18008 & 73224 & 300692\\
                          & \text{\#ineqs} & 2 & 1 & 1 & 2 &  3 &  5 & 10 &  19 &   46 &  115 &   283 &   771 & 2132\\\hline\hline
\multirow{2}{*}{Bosons}   & \text{\#OVs}  & 1 & 1 & 2 & 4 &  8 & 17 & 37 &  82 &  184 &  418 &   967 &  2278 & 5456\\
                          & \text{\#ineqs} & 1 & 1 & 1 & 2 &  3 &  5 &  9 &  14 &   23 &   40 &    72 &   128 &  241\\
\end{array}$}
\end{center}
\caption{Number of fundamental occupation vectors (\#OVs) for small $r$ and the number of new fundamental inequalities (\#ineqs) for fermions and bosons.}
\label{tab:ferbos}
\end{table}

The ratio between the total number of vertices and facets versus the fundamental ones are significant in these two cases.
The number of vertices of the fermionic spectral polytope increases by a factor of $\approx100$ for each $r$ while the number of fundamental occupation vectors increases by roughly $\approx2.5$.
The number of facets of the fermionic spectral polytope increases by a factor of $\approx36$ for each $r$ while the number of fundamental rays increases by roughly $\approx1.76$.
The number of vertices of the bosonic spectral polytope increases by a factor of $\approx10$ for each $r$ while the number of fundamental occupation vectors increases by roughly $\approx2.1$.
The number of facets of the fermionic spectral polytope increases by a factor of $\approx5.8$ for each $r$ while the number of fundamental rays increases by roughly $\approx1.66$.

The computation for $r=13$ for fermions used a parallel depth-first search algorithm distributed on $16$ cores to obtain all occupation vectors along with their fundamental cones as in Algorithm~\ref{algo:lineup} and required about $58$Gb of RAM.
As a post-processing, the union of the rays of the fundamental cones that consisted of $3\ 410$ rays were filtered using Theorem~\ref{thm:dimension} to preserve only the normal rays.
The generation of the occupation vectors and their cones took 5h30min and the post-precessing 1h15min on a AMD FX6274(@2.2GHz) processor using 16 cores.
The spectral polytope $\FerPoly[13][12,24]$ of dimension $23$ has $25\ 762\ 023\ 560\ 117\ 406\ 481\ 920$ vertices and $1\ 766\ 398\ 153\ 945\ 819\ 988$ facets.
The implementation was done in using \texttt{SageMath} \cite{sagemath} (using in particular combinations of \texttt{ppl} \cite{bagnara_parma_2008} and \texttt{normaliz} \cite{bruns_power_2016}) and exploiting multiprocessing tools available in \texttt{python3}.
The implementation used many high-level objects and stored them on runtime for verification purposes, making it much less memory efficient, leaving room for much improvement.
The bottleneck of the algorithm did not seem to have been reached at $r=13$.
We found it unnecessary to consume the energy required to obtain the next step.
Once deemed relevant for use, it is well within the realm of possibility to obtain the next $H$-representations.

\subsection{Stabilization for unbounded $N$ and $d$}
\label{ssec:stability}
\subsubsection{For $\FerPC$}
In this section, we examine the relationship between lineups of the point configurations $\FerPC$ and those obtained for larger values of $N$ and $d$.
The following lemma shows that fundamental lineups of length $r$ are intrinsically independent of the values of $N$ and $d$ when these are large enough. 

\begin{lemma}[Stabilization of lineups]
\label{lem:gale}
Let $r\geq2$ and $\ell=(\bchi(S_1),\bchi(S_2),\dots,\bchi(S_r))$ be a fundamental lineup of $\FerPC$, where $N\geq r-1$.
\begin{enumerate}[label=\roman{enumi}),ref=\roman{enumi})]
\item The largest element of $S_k$ is at most $N+r-1$, for all $1\leq k\leq r$.
\label{lem:gale_i}
\item If $N>r-1$, then $\{1,\dots,N-r+1\}\subset S_k$, for all $1\leq k\leq r$.
\label{lem:gale_ii}
\end{enumerate}
\end{lemma}

\begin{proof}
\ref{lem:gale_i}
Consider a $N$-subset $T$ with largest entry $N+r$.
Any chain saturated chain in $\FerPoset$ going from $[N]$ to $T$ involves at least $(r+1)$ subsets (when including $T$).
Thus $T$ can not be involved in any lineup $\ell$ of length~$r$.

\ref{lem:gale_ii}
The Gale poset is ranked with rank function given by $\operatorname{rank}( T)=\sum_{i\in T} i - \binom{N+1}{2}+1$.
Every $N$-subset $S_k \in \ell$ has rank at most $r$.
Consider a $N$-subset $ T$ such that $1\not\in T$.
The rank of~$ T$ is at least $N+1>r$, hence $ T\not\in\ell$.
Repeating the argument by subtracting one from every element of~$T$ and decreasing $N$ by~$1$ leads to the result.
\end{proof}

Properties~\ref{lem:gale_i} and~\ref{lem:gale_ii} show that increasing $N$ beyond $r-1$ and $d$ beyond $N+r-1$ barely has an effect on the structure and complexity of possible fundamental lineups. 

\begin{example}\label{ex:different_gales}
We illustrate Lemma \ref{lem:gale} in Figure~\ref{fig:differente_gale}, where we depict the elements of rank 4 of the Gale orders with parameters $(N,d)=(2,5),(3,6)$, and $(4,7)$ respectively.
The reader is invited to verify the conclusions of the lemma for various lineups.
\begin{figure}[H]
\centering
\begin{tikzpicture}%
[scale=1.000000,
x={(0.4, 0)},
y={(0,1.5)},
edge/.style={thick},
edge2/.style={dotted,thick,shorten >=1cm},
vertex/.style={inner sep=1pt,circle,draw=black,fill=black,thick},
vertex2/.style={inner sep=1.5pt,circle,draw=black,fill=white,thick},
vertex3/.style={inner sep=2pt,rectangle,draw=black,fill=white,thick},
rotate=-90]

\begin{scope}[yshift=-5cm,scale=0.8]
\coordinate (12345) at (0,0);
\coordinate (12346) at (0,1);
\coordinate (12347) at (-1,2);
\coordinate (12348) at (-2,3);
\coordinate (12349) at (-4,4);
\coordinate (123410) at (-6,5);
\coordinate (12356) at (1,2);
\coordinate (12357) at (0,3);
\coordinate (12358) at (-2,4);
\coordinate (12359) at (-4,5);
\coordinate (12367) at (0,4);
\coordinate (12368) at (-2,5);
\coordinate (12456) at (2,3);
\coordinate (12457) at (2,4);
\coordinate (12458) at (0,5);
\coordinate (12467) at (2,5);
\coordinate (13456) at (4,4);
\coordinate (13457) at (4,5);
\coordinate (23456) at (6,5);

\draw[edge] (12345) -- (12346);
\draw[edge] (12346) -- (12347);
\draw[edge] (12346) -- (12356);
\draw[edge] (12347) -- (12348);
\draw[edge] (12347) -- (12357);
\draw[edge] (12356) -- (12357);

\node[vertex2,label=below:$12$] at (12345) {};
\node[vertex2,label=above:$13$] at (12346) {};
\node[vertex2,label=above:$14$] at (12347) {};
\node[vertex2,label=above:$15$] at (12348) {};
\node[vertex2,label=below:$23$] at (12356) {};
\node[vertex2,label=above:$24$] at (12357) {};
\end{scope}

\begin{scope}[yshift=0cm,scale=0.8]
\coordinate (12345) at (0,0);
\coordinate (12346) at (0,1);
\coordinate (12347) at (-1,2);
\coordinate (12348) at (-2,3);
\coordinate (12349) at (-4,4);
\coordinate (123410) at (-6,5);
\coordinate (12356) at (1,2);
\coordinate (12357) at (0,3);
\coordinate (12358) at (-2,4);
\coordinate (12359) at (-4,5);
\coordinate (12367) at (0,4);
\coordinate (12368) at (-2,5);
\coordinate (12456) at (2,3);
\coordinate (12457) at (2,4);
\coordinate (12458) at (0,5);
\coordinate (12467) at (2,5);
\coordinate (13456) at (4,4);
\coordinate (13457) at (4,5);
\coordinate (23456) at (6,5);

\draw[edge] (12345) -- (12346);
\draw[edge] (12346) -- (12347);
\draw[edge] (12346) -- (12356);
\draw[edge] (12347) -- (12348);
\draw[edge] (12347) -- (12357);
\draw[edge] (12356) -- (12357);
\draw[edge] (12356) -- (12456);

\node[vertex2,label=below:$123$] at (12345) {};
\node[vertex2,label=above:$124$] at (12346) {};
\node[vertex2,label=above:$125$] at (12347) {};
\node[vertex2,label=above:$126$] at (12348) {};
\node[vertex2,label=below:$134$] at (12356) {};
\node[vertex2,label=above:$135$] at (12357) {};
\node[vertex2,label=below:$234$] at (12456) {};
\end{scope}

\begin{scope}[yshift=5cm,scale=0.8]
\coordinate (12345) at (0,0);
\coordinate (12346) at (0,1);
\coordinate (12347) at (-1,2);
\coordinate (12348) at (-2,3);
\coordinate (12349) at (-4,4);
\coordinate (123410) at (-6,5);
\coordinate (12356) at (1,2);
\coordinate (12357) at (0,3);
\coordinate (12358) at (-2,4);
\coordinate (12359) at (-4,5);
\coordinate (12367) at (0,4);
\coordinate (12368) at (-2,5);
\coordinate (12456) at (2,3);
\coordinate (12457) at (2,4);
\coordinate (12458) at (0,5);
\coordinate (12467) at (2,5);
\coordinate (13456) at (4,4);
\coordinate (13457) at (4,5);
\coordinate (23456) at (6,5);

\draw[edge] (12345) -- (12346);
\draw[edge] (12346) -- (12347);
\draw[edge] (12346) -- (12356);
\draw[edge] (12347) -- (12348);
\draw[edge] (12347) -- (12357);
\draw[edge] (12356) -- (12357);
\draw[edge] (12356) -- (12456);

\node[vertex2,label=below:$1234$] at (12345) {};
\node[vertex2,label=above:$1235$] at (12346) {};
\node[vertex2,label=above:$1236$] at (12347) {};
\node[vertex2,label=above:$1237$] at (12348) {};
\node[vertex2,label=below:$1245$] at (12356) {};
\node[vertex2,label=above:$1246$] at (12357) {};
\node[vertex2,label=below:$1345$] at (12456) {};
\end{scope}
\end{tikzpicture}
\caption{The Hasse diagram of the Gale order for different parameters}
\label{fig:differente_gale}
\end{figure}
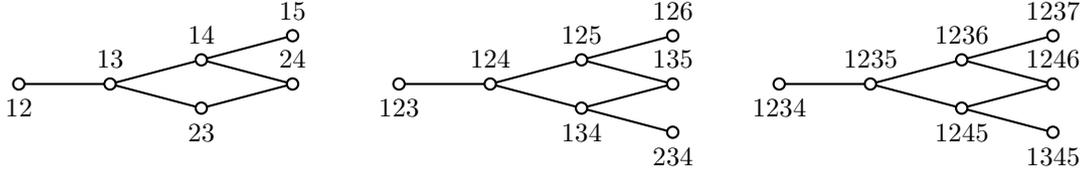
\end{example}

As a consequence of Lemma \ref{lem:gale}, the parameters $N=r-1$ and $d=2r-2$ represent a base case.
For the rest of this section, we consider pairs $(N,d)$ such that $N\geq r-1$ and $d\geq N+r-1$ for a fixed value of $r\geq 2$.
In this case, two lineups play a special role.

\begin{definition}[Breakaway and peloton lineups]
\label{def:special}
Let $r\geq 2$, $N=r-1$ and $d=2r-2$.
The \emph{breakaway} lineup is $(\bchi([N-1]\cup\{N+i-1\}))_{i=1}^{r}$ and the \emph{peloton} lineup is $(\bchi([N+1]\setminus\{N+2-i\}))_{i=1}^{r}$. 
Their occupation vectors are called breakaway and peloton occupation vectors, respectively.
\end{definition}

\begin{lemma}\label{lem:breakawaypeloton}
If $r\geq 2$, $N=r-1$ and $d=2r-2$, then the breakaway occupation vector is the unique occupation vector in $\non[f]$ whose last coordinate is non-zero and the peloton occupation vector is the unique occupation vector in $\non[f]$ whose first coordinate is strictly less than one.
\end{lemma}

\begin{proof}
The ideal of $\FerPoset$ given by the breakaway lineup $\ell_b$ is the unique ideal of cardinality~$r$ with an element that contains $N+r-1$.
Similarly, the ideal of $\FerPoset$ given by the peloton lineup $\ell_p$ is the unique ideal of cardinality $r$ with an element that does not contain $1$.
These two ideals are chains, hence they lead to exactly one saturated chain of ideals.
These two ideals are threshold: the appropriate linear functionals are directly given by the corresponding occupation vectors:
\[
\resizebox{0.9\textwidth}{!}{$
\OccVec_{\bw}(\ell_b)_j=
\begin{cases}
1 & \text{ for } j\in [N-1],\\
w_{j+1-N} & \text{ for } j \in \{N,\dots,N+r-1\},\\
0 & \text{ for } j \in \{N+r,\dots,d\},
\end{cases}\quad
\OccVec_{\bw}(\ell_p)_j=
\begin{cases}
1 & \text{ for } j\in [N-r+1],\\
1-w_{N+2-j} & \text{ for } j\in \{N-r+2,\dots,N+1\},\\
0 & \text{ for } j \in \{N+2,\dots,d\}.
\end{cases}$}\qedhere
\]
\end{proof}

We now define three maps allowing us to increase parameters. 
The first map is an order-preserving map.
\[
\begin{array}{rlcl}
\upsilon: & \FerPoset[r-1][2r-2]& \to & \FerPoset[N][d]\\
& \{s_1,\dots,s_{r-1}\} & \mapsto & [N-r+1]\cup \{s_i+N-r+1\}_{i=1}^{r-1}.
\end{array}
\]
The map $\upsilon:\FerPoset[r-1][2r-2]\to\FerPoset[N][d]$ is injective and its image consists of the $N$-subsets $S$ of $[d]$ such that $[N-r+1]\subset S \subset [N+r-1]$.
The poset morphism $\upsilon$ induces an inclusion of the point configurations $\FerPC[r-1][2r-2]$ into $\FerPC$, and furthermore on their respective sets of lineups.
By abusing notation, we refer to these two inclusions also using the letter $\upsilon$.
Namely, we write $\upsilon\bchi(S)=\bchi(\upsilon S)$ to lessen the amount of parentheses.
Finally, we define the linear maps $\Upsilon$ and $\Psi$.
They allow the comparison of the sets of fundamental $r$-lineups of $\FerPC[r-1][2r-2]$ and $\FerPC$ as formulated in Lemma~\ref{lem:UpDown}.
\[
\begin{array}{rcl}
\Upsilon: \RR^{2r-2} & \longrightarrow & \RR^{d},\\
(y_1,\dots,y_{2r-2})& \mapsto & (\underbrace{y_1,\dots,y_1}_{N-r+2},y_2,\dots,y_{2r-3},\underbrace{y_{2r-2},\dots,y_{2r-2}}_{d-N-r+2}).
\\\\
\Psi: \RR^{d} & \longrightarrow & \RR^{2r-2},\\
(z_1,\dots,z_{d})& \mapsto & (z_{N-r+2},\dots,z_{N+r-1}).
\end{array}
\]
They verify
\[
\resizebox{\textwidth}{!}{$
\Upsilon( \bef_i)= \begin{cases}
\bef_{i+N-r+1}&\text{ for }i\in[2r-3],\\
\bef_{d} &\text{ if } i=2r-2,
\end{cases}
\text{ and }
\Psi(\bef_j)=\begin{cases}
(0,\dots,0) &\text{ if } j\in [1,\, N-r+1],\\
\bef_{j-N+r-1} &\text{ if } j\in [N-r+2,\, N+r-1],\\
\bef_{2r-2}& \text{ if } j \in [N+r,\, d].
\end{cases}$}
\]
For any vectors $\by\in \RR^{2r-2}$ and $\bz\in\RR^d$, the inner product and the maps $\Upsilon$ and $\Psi$ behave according to the two equalities
\begin{align}
\np{\Upsilon(\by),\upsilon\bchi(S)}&=\np{\by,\bchi(S)}+ y_{1}(N-r+1) + y_{2r-2}(d-N-r+1),\label{eq:rising}\\
\np{\Psi(\bz),\upsilon^{-1}\bchi(T)}&=\np{\bz,\bchi(T)}-\left(\sum_{i=1}^{N-r+1}z_i\right),\label{eq:lowering}
\end{align}
where $S\subset[2r-2]$ and $[N-r+1]\subset T\subset [N+r-1]$ are of cardinality $r-1$ and $N$ respectively.

\begin{lemma}\label{lem:UpDown}
Let $1 \leq r-1\leq N \leq d-r+1$.
\begin{compactenum}[\roman{enumi})]
\item If $\by\in\RR^{2r-2}$ is a fundamental vector and $\ell(\by)=(\bchi(S_1),\dots,\bchi(S_r))$ is the corresponding $r$-lineup on $\FerPC[r-1][2r-2]$, then the fundamental vector $\Upsilon(\by)$ leads to the $r$-lineup $\upsilon\ell(\by)=(\upsilon\bchi(S_1),\dots,\upsilon\bchi(S_r))$ on $\FerPC$.
\item Conversely, if $\bz\in\RR^{d}$ is a fundamental vector and $\ell(\bz)=(\bchi(T_1),\dots,\bchi(T_r))$ is the corresponding $r$-lineup on $\FerPC$, then the fundamental vector $\Psi(\bz)$ leads to the $r$-lineup $\upsilon^{-1}\ell(\bz) =(\upsilon^{-1}\bchi(T_1),\dots,\upsilon^{-1}\bchi(T_r))$ on $\FerPC[r-1][2r-2]$.
\end{compactenum}
\end{lemma}

\begin{proof}
This is a consequence of Lemma~\ref{lem:gale} together with Equations \eqref{eq:rising} and \eqref{eq:lowering}.
\end{proof}

\begin{proposition}\label{prop:lift_lineup}
Let $1\leq r-1\leq N \leq d-r+1$.
The map $\upsilon$ induces a bijection between the sets of fundamental occupation vectors $\non[f][r][r-1][2r-2]$ and $\non[f][r][N][d]$.
Furthermore, if $\bo\in\non[f][r][r-1][2r-2]$ and $\pol[C]$ is its fundamental cone in $\FerPoly[r][r-1,2r-2]$, then the fundamental cone of $\upsilon(\bo)$ in $\FerPoly[r]$ is
\[
\cone\{\bef_1,\dots,\bef_{N+1-r}\}+\Upsilon(\pol[C])+\cone\{\bef_{N+r},\dots,\bef_{d-1}, \bef_d,-\bef_d\}.
\]
\end{proposition}

\begin{proof}	
The bijection follows from Lemma~\ref{lem:UpDown}.
Let $\pol[C]'$ be the fundamental cone of $\upsilon(\bo)$.
Again by Lemma~\ref{lem:UpDown}, we have that $\pol[C]'=\Psi^{-1}(\pol[C])\cap \Phi_d$, that is, a vector $\bz=(z_1,\dots, z_d)\in \RR^d$ belongs to the fundamental cone of $\upsilon(\bo)$ if and only if $\bz\in \Phi_d$ and $(z_{N-r+2},\dots,z_{N+r-1})\in \pol[C]$. 
This is equivalent to $\bz\in\cone\{\bef_1,\dots,\bef_{N+1-r}\}+\Upsilon(\pol[C])+\cone\{\bef_{N+r},\dots,\bef_{d-1}, \bef_d,-\bef_d\}$.
\end{proof}

\begin{example}[Example~\ref{ex:4-10} continued]
\label{ex:Upsilon}
The two fundamental occupation vectors in
\[
\non[f][3][4][10]=\{(1,1,1,w_1,w_2,w_3,0,0,0,0),(1,1,w_1+w_2,w_1+w_3,w_2+w_3,0,0,0,0,0)\}
\]
are obtained by applying the map $\upsilon$ on the two fundamental occupation vectors in $\non[f][3][2][4]$:
\[
\non[f][3][2][4]=\{(1,w_1,w_2,w_3),(w_1+w_2,w_1+w_3,w_2+w_3,0)\}.
\]
They are therefore obtained by adding two 1's at the beginning and four 0's at the end.
\end{example}

The following theorem provides the stability result for fermions.

\begin{mainthm}[Stability Theorem]\label{thm:stability}
Let $1\leq r-1\leq N \leq d-r+1$.
The set of uncoarsenable fundamental $r$-rankings of $\FerPC[r-1][2r-2]$ and that of $\FerPC$ are in bijection.
Phrased using Convention~\ref{convention}, the bijection is as follows:
\begin{compactenum}
\item The two fundamental normal ray generators $\bef_1$ and $\bef_{2r-3}$ of $\FerFan[r][r-1][2r-2]$ correspond to the fundamental normal ray generators $\bef_1$ and $\bef_{d-1}$ of $\FerFan$.
\item Let $\by\not\in(\RR\bef_1\cup\RR\bef_{2r-3})$ be a fundamental ray generator of $\FerFan[r][r-1][2r-2]\cap\Phi_{2r-2}$. The vector $\by$ is a normal ray generator of $\FerFan[r][r-1][2r-2]$ if and only if the vector $\Upsilon(\by)$ is a fundamental normal ray generator of $\FerFan$.
\end{compactenum}
\end{mainthm}

\begin{proof}
(1).
By Proposition~\ref{prop:lift_lineup}, the fundamental ray generators of $\FerFan\cap\Phi_d$ are either of the form $\Upsilon(\by)$ for a fundamental ray generator $\by$ of $\FerFan[r][r-1][2r-2]$, or of the form $\bef_i$, for some $i\in[N+1-r]\cup\{N+r,\dots,d,-d\}$. 
Furthermore, by Proposition~\ref{prop:basic_cleanup_2}, among the set $\{\bef_i\}_{i\in[d]}\setminus\{\bef_N=\Upsilon(\bef_{r-1})\}$, only $\bef_1$ and~$\bef_{d-1}$ give fundamental normal ray generators. 

\noindent (2). By the proof of part (1), we may assume that $\by$ is not a scalar multiple of $\bef_i$, for any ${i\in[2r-2]}$.
Let $\by'=\Upsilon(\by)$.
We use Theorem \ref{thm:dimension} as follows to compute the dimension of the face $\FerPoly[r][N,d]^{\by'}$ from the dimension of the face ${\FerPoly[r][r-1,2r-2]^{\by}}$.
As illustrated in Example~\ref{ex:Upsilon}, the occupation vectors $\non$ are obtained by adding $1$'s at the beginning and $0$'s at the end to the vectors in $\non[f][r][r-1][2r-2]$.
This induces a bijection between $\bV_{\by}$ and~$\bV_{\by'}$ (as defined in Theorem~\ref{thm:dimension}).
If $\bc_{\by}=(c_1,c_2,\dots,c_{k-1},c_k)$ denotes the composition associated to $\by$ (see Section~\ref{ssec:faces_sym}), then the one associated to $\by'$ is $(c_1+(N-r+1),c_2,\dots,c_{k-1},c_k+(d-N-r+1))$.
The projections $\Proj{\by}(\bV_{\by})$ and $\Proj{\by'}(\bV_{\by'})$ only differ by an affine transformation, and hence the first term in Equation~\eqref{eq:dimension} remains the same. 
In fact, the only difference in determining the dimension happens in the first and last blocks of the composition involved in the sum.
We start by analyzing the difference in the first block.

\textbf{Case 1) $1\not\in\operatorname{Fix}(\by)$.}
The vectors in $\bV_{\by'}$ are obtained from those of $\bV_{\by}$ by appending $1$'s at the beginning.
Therefore, we must have $1\notin \operatorname{Fix}(\by')$ and the contributed increase of dimension of the first block is $N-r+1$. 

\textbf{Case 2) $1\in\operatorname{Fix}(\by)$.}
If $c_1>1$, then the first $c_1$ coordinates of the points in $\bV_{\by}$ are equal to~$1$.
This follows from Lemma~\ref{lem:breakawaypeloton}, i.e. there is exactly one occupation vector in $\non[f][r][r-1][2r-2]$ whose first coordinate is different from~$1$ and its second coordinate is strictly smaller because $\bw$ has distinct entries.
Hence, $1\in\operatorname{Fix}(\by')$ and there is no contribution to the increase in dimension in this case.
Else if $c_1=1$, there are two situations to verify.
If the peloton occupation vector is in~$\bV_{\by}$, then adding 1's at the beginning makes the corresponding vector not inert and so $1\not\in\operatorname{Fix}(\by')$.
The contribution to the increase in dimension in this case is $N-r+1$.
Finally, if the peloton occupation vector is not in $\bV_{\by}$, then all vectors in $\bV_{\by}$ start with a $1$ by Lemma~\ref{lem:breakawaypeloton} and $1\in\operatorname{Fix}(\by')$.
In this case, there is no contribution to the increase in dimension.

The analysis for the last block is obtained similarly by using the breakaway occupation vector in place of the peloton while using Lemma~\ref{lem:breakawaypeloton}.
So we have the following increase in dimension by applying $\Upsilon$:
\begin{equation}\label{eq:dimension_jump}
\resizebox{0.93\textwidth}{!}{$
\dim (\FerPoly[r][N,d])^{\by'}- \dim (\FerPoly[r][r-1,2r-2])^{\by} 
=
{\begin{cases}
    d-2r+2 &\text{ if }1\notin \operatorname{Fix}(\by')\text{ and }k\notin \operatorname{Fix}(\by'),\\
    N-r+1 &\text{ if }1\notin \operatorname{Fix}(\by')\text{ and }k\in \operatorname{Fix}(\by'),\\
    d-N-r+1 &\text{ if } 1\in \operatorname{Fix}(\by')\text{ and }k\notin \operatorname{Fix}(\by'),\\
    0 &\text{ if } 1\in \operatorname{Fix}(\by')\text{ and }k\in \operatorname{Fix}(\by').
\end{cases}}$}
\end{equation}

\noindent
$\Leftarrow$)
If $\by$ does not induce a facet, then $\dim (\FerPoly[r][r-1,2r-2])^{\by} <  2r-4$ and Equation~\eqref{eq:dimension_jump} gives
\[\dim (\FerPoly[r][N,d])^{\by'} \leq \dim (\FerPoly[r][r-1,2r-2])^{\by} +d-2r+2 < d-2;\]
showing that $\by'$ does not induce a facet either.

\noindent
$\Rightarrow$)
If $\by$ is a normal ray generator, then $\dim (\FerPoly[r][r-1,2r-2])^{\by}=2r-4$.
We need to show that the last three cases in Equation~\eqref{eq:dimension_jump} do not occur, i.e. $1\in\operatorname{Fix}(\by')$ or $k\in\operatorname{Fix}(\by')$.
If $1\in\operatorname{Fix}(\by')$, then by the above case analysis, we must have $1\in\operatorname{Fix}(\by)$.
There are two possibilities: $c_1=1$ or $c_1>1$.
If $c_1>1$, the two summands of Theorem~\ref{thm:dimension} decrease by at least $1$: the dimension of the image of the projection decreases by at least $1$ since it satisfies one more equality $x_1=c_1$ (because the first coordinate of points in $\bV_{\by}$ is $1$), and the summand $(c_1-1)\geq1$ does not appear in the sum.
So,~$\by$ does not induce a facet contrary to the assumption.
Since the maximal value of the formula for the dimension is $d-1$, a decrease of at least $2$ can not lead to a facet.
If $c_1=1$, by the above case analysis, the first coordinate of the points in $\bV_{\by}$ is equal to $1$, which means that~$1$ belongs to every element in the order ideal given by the $r$-lineup induced by~$\by$.
This allows to reduce both parameters $r-1$ and $2r-2$ by $1$ as follows.
Consider the associated saturated chain of ideals given by $\by$, remove~$1$ from all the $(r-1)$-subsets in the ideals and substract $1$ to all the other elements, to end up with the set $\bV_{\tilde\by}$, where $\tilde\by\in\RR^{2r-3}$ is equal to last $2r-3$ entries of $\by$.
This does not affect the count in Equation \eqref{eq:dimension}: the first coordinate of $\Proj{\by}(\bv)$ is~$1$ for every $\bv\in \bV_{\by}$, so the dimension does not drop when it is omitted, and in the second term we have $c_1-1=0$.
Hence, the symmetrization of $\bV_{\tilde\by}$ by the Young subgroup gives rise to a face of dimension $2r-4$ in the polytope $\FerPoly[r][r-2,2r-3]$, which has itself dimension~$2r-4$.
The only way this can happen is if $\tilde\by$ is a multiple of $(1,\dots,1)\in\RR^{2r-3}$, but then this implies that $\by\in \RR\bef_1$ by Convention \ref{convention}.
This case was excluded at the beginning.
Excluding the case $k\in\operatorname{Fix}(\by')$ is obtained similarly by reducing only the dimension from $2r-2$ to $2r-3$ and keeing $r-1$ fixed.

Finally, if $1\notin \operatorname{Fix}(\by')$ when $N>r-1$ and $k\notin \operatorname{Fix}(\by')$ when $d>N+r-1$, then by Equation~\eqref{eq:dimension_jump}
\[
\dim (\FerPoly[r][N,d])^{\by'}= \dim (\FerPoly[r][r-1,2r-2])^{\by} +d-2r+2=d-2,
\]
which is the dimension of a facet.
\end{proof}

Theorem \ref{thm:stability} implies that the knowledge of the fan $\FerFan[r][r-1][2r-2]$ for some $r$, suffices to construct every fan $\FerFan[r][N][d]$, for $N>r-1$ and $d>N+r-1$.

\subsubsection{For $\BosPC$}

Similar results hold in the case of bosons. 
We state the translated definitions, lemmas and propositions without proofs as they are obtained similarly as for fermions.

\begin{lemma}
\label{lem:gale_bosonic}
Let $r\geq2$ and $\ell=(\bchi(S_1),\bchi(S_2),\dots,\bchi(S_r))$ be a fundamental lineup of $\BosPC$, where $N\geq r-1$.
\begin{enumerate}[label=\roman{enumi}),ref=\roman{enumi})]
\item The largest element of $S_k$ with non-zero multiplicity is at most $r$, for all $1\leq k\leq r$.
\item $1$ has multiplicity greater than $N-r+1$ in $S_k$, for all $1\leq k\leq r$.
\end{enumerate}
\end{lemma}

For the rest of this section, we consider pairs $(N,d)$ such that $N\geq r-1$ and $d\geq r$ for a fixed value of $r\geq 2$.
The order-preserving map for bosons is
\[
\begin{array}{rlcl}
\upsilon: & \BosPoset[r-1][r]& \to & \BosPoset[N][d]\\
& (S(1), S(2), \ldots, S(r)) & \mapsto & (S(1)+N-r+1, S(2), \ldots, S(r)).
\end{array}
\]
\begin{lemma}
The map $\upsilon:\BosPoset[r-1][r]\to\BosPoset[N][d]$ is injective and its image consists of the $N$-multisubsets $S$ of $[d]$ that have $S(1)\geq N-r+1$ and $S(t)=0$ for all $t\in [r+1, d]$.
\end{lemma}

The linear maps $\Upsilon$ and $\Psi$ become
\[
\begin{array}{rcl}
\Upsilon: \RR^{r} & \longrightarrow & \RR^{d},\\
(y_1,\dots,y_{r})& \mapsto & (y_1,\dots,y_{r-1},\underbrace{y_{r},\dots,y_{r}}_{d-r+1}).
\\\\
\Psi: \RR^{d} & \longrightarrow & \RR^{r},\\
(z_1,\dots,z_{d})& \mapsto & (z_1,\dots,z_r).
\end{array}
\]
Finally, the equalities satisfied by inner products are
\[
\begin{split}
\np{\Upsilon(\by),\upsilon\bchi(S)}&=\np{\by,\bchi(S)}+ y_{1}(N-r+1),\\
\np{\Psi(\bz),\upsilon^{-1}\bchi(T)}&=\np{\bz,\bchi(T)}-z_1(N-r+1),
\end{split}
\]
where $S\in \multiset{[r]}{r-1}$ and $T\in \multiset{[d]}{N}$ is such that $T(1)\geq N-r+1$ and $T(t)=0$ for $t\in [r+1,d]$.

\begin{proposition}
Let $r\geq2$, $N\geq r-1$ and $d\geq r$.
The map $\upsilon$ induces a bijection between the sets of fundamental occupation vectors $\non[b][r][r-1][r]$ and $\non[b][r][N][d]$.
Furthermore, if $\bo\in\non[f][r][r-1][r]$ and $\pol[C]$ is its fundamental cone in $\BosPoly[r][r-1,r]$,
then the fundamental cone of $\upsilon(\bo)$ in $\BosPoly[r]$ is
\[
\Upsilon(\pol[C])+\cone\{\bef_{r+1},\dots,\bef_{d-1}, \bef_d,-\bef_d\}.
\]
\end{proposition}

We use the following lemma which distinguishes bosonic from fermionic occupation vectors.

\begin{lemma}\label{lem:last_lemma_omg_finally}
Let $\bw\in\Pauli_{N}^\circ$ and $\ell$ be a bosonic $(N+1)$-lineup.
For any $d\geq 2$, the first coordinate of the occupation vector of $\ell$ is larger than its second.
\end{lemma}

\begin{proof}
The first multiset is $1^N$ and the number of $1$'s in every subsequent multiset in a lineup decreases by at most one.
Moreover, the first multiset does not contain any $2$'s and the number of $2$'s in every subsequent multiset increases by at most one.
Every occurrence of a $2$ in a multiset is therefore preceded by an occurrence of a $1$ in a previous multiset.
\end{proof}

\begin{mainthm}[Stability Theorem]\label{thm:stability_bosonic}
Let $r\geq2$, $N\geq r-1$ and $d\geq r$.
The set of uncoarsenable fundamental $r$-rankings of $\BosPC[r-1][r]$ and that of $\BosPC$ are in bijection.
Phrased using Convention~\ref{convention}, the bijection is as follows:
\begin{enumerate}
\item The fundamental normal ray generator $\bef_{r-1}$ corresponds to the fundamental normal ray generator $\bef_{d-1}$.
\item Let $\by\not\in\RR\bef_{r-1}$ be a fundamental ray generator of $\BosFan[r][r-1][r]\cap\Phi_r$.
The vector $\by$ is a normal ray generator of $\BosFan[r][r-1][r]$ if and only if the vector $\Upsilon(\by)$ is a fundamental normal ray generator of $\BosFan$.
\end{enumerate}
\end{mainthm}

\begin{proof}
Let $\by\in\Phi_r$ and $\by'=\Upsilon(\by)\in\Phi_d$ and denote their associated compositions $\bc$ and $\bc'$ respectively.
They are the same except for $c'_k=c_k+d-r$. Therefore, to study the dimensions of the faces $\dim(\BosPoly)^{\by'}$ and $\dim(\BosPoly[r][r-1,r])^{\by}$ via Theorem~\ref{thm:dimension}, we can focus on the behavior of the first and last blocks.

In this case, whenever $c_1>1$ we have $1\notin \operatorname{Fix}(\by)$ by Lemma \ref{lem:last_lemma_omg_finally}.
For $c_k>1$, $k\notin \operatorname{Fix}(\by)$ implies $k\in \operatorname{Fix}(\by')$, as there is exactly one occupation vector whose last coordinate is non-zero, and its second to last coordinate is strictly larger than the last.
We have the following increase in dimension by applying $\Upsilon$:
\begin{equation}\label{eq:dimension_count_bosonic}
	\dim(\BosPoly)^{\by'}-\dim(\BosPoly[r][r-1,r])^{\by}= 
	{\begin{cases}
			d-r &\text{ if } k\notin \operatorname{Fix}(\by'),\\
			0 &\text{ if } k\in \operatorname{Fix}(\by').
	\end{cases}}
\end{equation}

\noindent
$\Rightarrow$) The same argument as in the proof of Theorem~\ref{thm:stability} is valid.

\noindent
$\Rightarrow$) We must prove that if $\by\not\in\RR\bef_{r-1}$ induces a facet, then the first case of Equation~\eqref{eq:dimension_count_bosonic} occurs.
It is enough to prove that there is a point in $\bV_{\by}$ with a coordinate different form zero in the last $c_k$ entries.
As in the proof of Theorem~\ref{thm:stability}, for the sake of contradiction assume that they all satisfy that the last $c_k$ coordinates are~$0$. 
If $c_k>1$, then the dimension count of \eqref{eq:dimension} shows that $\by$ does not induce a facet. And if $c_k=1$, then this shows that $\by=\bef_{r-1}$.	
\end{proof}

\subsection{A worked-out example: $r=4$}\label{ssec:ex4}

\subsubsection{For fermions}
We illustrate the technique by giving a non-redundant $H$-representation of the polytope $\FerPoly[4][3,6]$, which can be done without resorting to computational tools.
We begin by finding the $V$-representation.
There are four fundamental lineups in $\FerPC[3][6]$, giving rise to four occupation vectors:
\begin{equation}\label{eq:NON_r4}
\resizebox{0.93\textwidth}{!}{$
\begin{array}{rlcl}
\ell_1:&\bchi\left(123\right)\rightarrow\bchi\left(124\right)\rightarrow\bchi\left(125\right)\rightarrow\bchi\left(126\right)\leadsto\\
\ell_2:&\bchi\left(123\right)\rightarrow\bchi\left(124\right)\rightarrow\bchi\left(125\right)\rightarrow\bchi\left(134\right)\leadsto\\
\ell_3:&\bchi\left(123\right)\rightarrow\bchi\left(124\right)\rightarrow\bchi\left(134\right)\rightarrow\bchi\left(125\right)\leadsto\\
\ell_4:&\bchi\left(123\right)\rightarrow\bchi\left(124\right)\rightarrow\bchi\left(134\right)\rightarrow\bchi\left(234\right)\leadsto
\end{array}
\begin{array}{r@{}rrrrrr}
\OccVec_{\bw}(\ell_1) = (&1,&1,&w_{1},&w_{2},&w_{3},&w_{4}),\\
\OccVec_{\bw}(\ell_2) = (&1,&w_1+w_2+w_3,&w_{1} + w_{4},&w_{2} + w_{4},&w_{3},&0),\\
\OccVec_{\bw}(\ell_3) = (&1,&w_1+w_2+w_4,&w_{1} + w_{3},&w_{2} + w_{3},&w_{4},&0),\\
\OccVec_{\bw}(\ell_4) = (& w_1+w_2+w_3,&w_1+w_2+w_4,&w_1+w_3+w_4,&w_2+w_3+w_4,&0,&0).
\end{array}$}
\end{equation}

For each lineup we determine its fundamental cone.
By Theorem \ref{thm:main}, this is the set of fundamental vectors that induce the given lineup.

\noindent
\textbf{Lineup $\ell_1$}.
Let $\by$ be a fundamental vector inducing $\ell_1$.
Any fundamental lineup begins with $\bchi(123)$ and $\bchi(124)$ as the first two elements.
For the third position there are two options, it is either $\bchi(125)$ or $\bchi(134)$, so in order for $\by$ to induce the first lineup, we must have
\[
\np{\by,\bchi(125)}>\np{\by,\bchi(134)},\quad\text{ equivalently }\quad y_2-y_3 > y_4-y_5.
\]
For the fourth element there are two options, either $\bchi(126)$ or $\bchi(134)$, so in addition to the previous inequality, the vector $\by$ must satisfy
\[
\np{\by,\bchi(126)}>\np{\by,\bchi(134)},\quad\text{ equivalently }\quad y_2-y_3 > (y_4-y_5)+(y_5-y_6).
\]
Strict inequalities are necessary to induce the lineup uniquely.
The closed fundamental cone is given by relaxing the inequalities to allow equality.
In conclusion, the fundamental cone of the vertex $\OccVec_{\bw}(\ell_1)$ in $\FerPoly[4][3,6]$ is given by all $\by\in\RR^6$ such that
\[
y_1\geq y_2\geq y_3\geq y_4\geq y_5\geq y_6=0,\quad y_2-y_3 \geq y_4-y_5, \quad y_2-y_3 \geq (y_4-y_5)+(y_5-y_6).
\]
Notice that the inequality $ y_2-y_3 \geq y_4-y_5$ is redundant.

\noindent
\textbf{Lineup $\ell_2$}.
This is analogous to the previous case, we need only to reverse the second inequality since now we have $\bchi(134)$ before $\bchi(126)$ in the lineup. The fundamental cone is described by
\[
y_1\geq y_2\geq y_3\geq y_4\geq y_5\geq y_6=0,\quad y_4-y_5\geq y_2-y_3 \geq (y_4-y_5)+(y_5-y_6).
\]

\noindent
\textbf{Lineup $\ell_3$}.
Similar considerations lead to the cone described by
\[
y_1\geq y_2\geq y_3\geq y_4\geq y_5\geq y_6=0,\quad (y_1-y_2)+(y_2-y_3) \geq y_4-y_5 \geq y_1-y_2.
\]

\noindent
\textbf{Lineup $\ell_4$}.
Similar considerations lead to the cone described by
\[
	y_1\geq y_2\geq y_3\geq y_4\geq y_5\geq y_6=0, \quad y_4-y_5\geq (y_1-y_2)+(y_2-y_3) .
\]
The defining inequalities become simpler in the fundamental basis. 
We write $(\tilde{y})_i$ for the coordinates of a vector $\by$ in the fundamental basis, that is, $\by=\sum_{i=1}^{5} \tilde{y}_i\bef_i$.
Equivalently, we have $\tilde{y}_i=y_i-y_{i+1}$ for all $i\in[5]$.
With this notation in mind, the inequalities $y_1\geq y_2\geq y_3\geq y_4\geq y_5\geq y_6=0$ read as $\tilde{y}_i\geq 0$ for $i=1,\dots,5$.
So $\Phi_6=\{\by\in\RR^6~:~(\tilde{y}_i)_{i=1}^5\in\RR^5_{\geq 0},\}$.
We can rewrite the four fundamental cones as
\[
\begin{array}{cccccccllc}
\pol[K]_1 &=&\{(\tilde{y}_i)_{i=1}^5\in\RR^5_{\geq 0}~:~&                       &    &\tilde{y}_2&\geq&\tilde{y}_4+\tilde{y}_5&\},\\[0.25em]
\pol[K]_2 &=&\{(\tilde{y}_i)_{i=1}^5\in\RR^5_{\geq 0}~:~&\tilde{y}_4+\tilde{y}_5&\geq&\tilde{y}_2&\geq&\tilde{y}_4            &\},\\[0.25em]
\pol[K]_3 &=&\{(\tilde{y}_i)_{i=1}^5\in\RR^5_{\geq 0}~:~&\tilde{y}_1+\tilde{y}_2&\geq&\tilde{y}_4&\geq&\tilde{y}_2            &\},\\[0.25em]
\pol[K]_4 &=&\{(\tilde{y}_i)_{i=1}^5\in\RR^5_{\geq 0}~:~&                       &    &\tilde{y}_4&\geq&\tilde{y}_1+\tilde{y}_2&\}.
\end{array}
\]

The coordinate corresponding to $\bef_3$ is never used, so these four cones can be obtained by restricting to  $\pol[Q]=\cone\{\bef_1,\bef_2,\bef_4,\bef_5\}$ and then taking the pyramid over the ray $\bef_3$. 
By taking an affine slice of $\pol[Q]$ (say normalizing the sum to 1), we can represent the polyhedral subdivision induced by the cones $\pol[K]_1,\pol[K]_2,\pol[K]_3,$ and $\pol[K]_4$ in 3-D, see Figure~\ref{fig:fan_r4}.

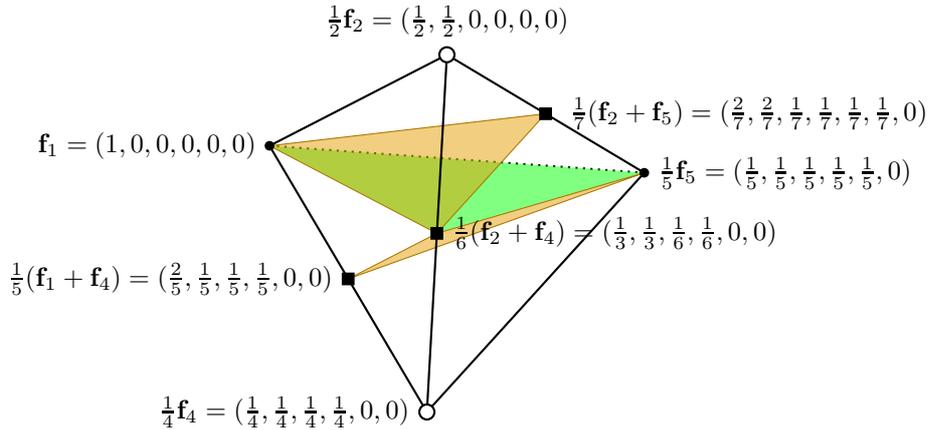
\begin{figure}[H]
\begin{center}
\begin{tikzpicture}%
	[x={(-0.995755cm, 0.072284cm)},
	y={(-0.092046cm, -0.781701cm)},
	z={(0.000039cm, 0.619451cm)},
	scale=3.500000,
	back/.style={dotted},
	edge/.style={thick},
	facet/.style={fill=green,fill opacity=0.500000},
	facet2/.style={fill=orange,draw=orange!75!black,fill opacity=0.500000},
	vertex/.style={inner sep=1pt,circle,draw=black,fill=black,thick},
	vertex2/.style={inner sep=2pt,circle,draw=black,fill=white,thick},
	vertex3/.style={inner sep=2pt,rectangle,draw=black,fill=black,thick}]
%
%

\coordinate (0.70711, 0.40825, 1.15470) at (0.70711, 0.40825, 1.15470); 
\coordinate (0.70711, 1.22474, 0.00000) at (0.70711, 1.22474, 0.00000); 
\coordinate (1.41421, 0.00000, 0.00000) at (1.41421, 0.00000, 0.00000); 
\coordinate (0.00000, 0.00000, 0.00000) at (0.00000, 0.00000, 0.00000); 
\coordinate (24) at ($(0.70711, 0.40825, 1.15470)!0.5!(0.70711, 1.22474, 0.00000)$); 
\coordinate (25) at ($(0.70711, 0.40825, 1.15470)!0.5!(0.00000, 0.00000, 0.00000)$); 
\coordinate (14) at ($(1.41421, 0.00000, 0.00000)!0.5!(0.70711, 1.22474, 0.00000)$); 
\draw[edge,back] (1.41421, 0.00000, 0.00000) -- (0.00000, 0.00000, 0.00000);
\fill[facet] (1.41421, 0.00000, 0.00000) -- (0.00000, 0.00000, 0.00000) -- (24) -- cycle {};
\fill[facet2] (1.41421, 0.00000, 0.00000) -- (24) -- (25) -- cycle {};
\fill[facet2] (0.00000, 0.00000, 0.00000) -- (24) -- (14) -- cycle {};
\draw[edge] (0.70711, 0.40825, 1.15470) -- (0.70711, 1.22474, 0.00000);
\draw[edge] (0.70711, 0.40825, 1.15470) -- (1.41421, 0.00000, 0.00000);
\draw[edge] (0.70711, 0.40825, 1.15470) -- (0.00000, 0.00000, 0.00000);
\draw[edge] (0.70711, 1.22474, 0.00000) -- (1.41421, 0.00000, 0.00000);
\draw[edge] (0.70711, 1.22474, 0.00000) -- (0.00000, 0.00000, 0.00000);
\node[vertex2,label=above:{$\frac{1}{2}\mathbf{f}_2=(\frac{1}{2},\frac{1}{2},0,0,0,0)$}] at (0.70711, 0.40825, 1.15470)     {};
\node[vertex2,label=left:{$\frac{1}{4}\mathbf{f}_4=(\frac{1}{4},\frac{1}{4},\frac{1}{4},\frac{1}{4},0,0)$}] at (0.70711, 1.22474, 0.00000)     {};
\node[vertex,label=left:{$\mathbf{f}_1=(1,0,0,0,0,0)$}] at (1.41421, 0.00000, 0.00000)     {};
\node[vertex,label=right:{$\frac{1}{5}\mathbf{f}_5=(\frac{1}{5},\frac{1}{5},\frac{1}{5},\frac{1}{5},\frac{1}{5},0)$}] at (0.00000, 0.00000, 0.00000)     {};
\node[vertex3,label=right:{$\frac{1}{6}(\mathbf{f}_2+\mathbf{f}_4)=(\frac{1}{3},\frac{1}{3},\frac{1}{6},\frac{1}{6},0,0)$}] at (24)     {};
\node[vertex3,label=right:{$\ \frac{1}{7}(\mathbf{f}_2+\mathbf{f}_5)=(\frac{2}{7},\frac{2}{7},\frac{1}{7},\frac{1}{7},\frac{1}{7},\frac{1}{7},0)$}] at (25)     {};
\node[vertex3,label=left:{$\frac{1}{5}(\mathbf{f}_1+\mathbf{f}_4)=(\frac{2}{5},\frac{1}{5},\frac{1}{5},\frac{1}{5},0,0)$}] at (14)     {};
\end{tikzpicture}
\end{center}
\caption{The subdivision of the fundamental simplex for $r=4$ with $N=3$ and $d=6$ ignoring~$\bef_3$ and normalizing the sum of coordinates to 1.}
\label{fig:fan_r4}
\end{figure}
There are four tetrahedra, corresponding to the four lineups.
They are $\pol[K]_1,\pol[K]_2,\pol[K]_3,$ and $\pol[K]_4$ from top to bottom in Figure \ref{fig:fan_r4}.
The triangular region in the middle, the one between $\bef_1$, $\frac{1}{5}\bef_5$ and $\frac{1}{6}(\bef_2+\bef_4)$, represents the points $\by$ such that $\tilde{y}_2=\tilde{y}_4$. 
The two other triangular regions represent the other two hyperplanes where $\tilde{y}_2=\tilde{y}_4+\tilde{y}_5$ and $\tilde{y}_4=\tilde{y}_1+\tilde{y}_2$.
There are three vertices in the subdivision which are not vertices of the tetrahedron; they are depicted using squares in Figure~\ref{fig:fan_r4}. 
The fundamental fan has eight rays respectively spanned by the following vectors:
\[
\{\bef_1,\bef_2,\bef_3,\bef_4,\bef_5,\bef_2+\bef_4,\bef_2+\bef_5, \bef_1+\bef_4\}.
\]
However, by Proposition~\ref{prop:basic_cleanup_2}, the rays spanned by $\bef_1,\bef_3,$ and $\bef_5$ are normal rays but the rays spanned by $\bef_2$ and $\bef_4$ are \emph{not}; their appearance is a consequence of intersecting with the fundamental chamber.
The second and third occupation vectors have all distinct coordinates, hence their normal cone is equal to the fundamental cone, so all of their rays are normal.
Since the rays spanned by $\bef_2+\bef_4,\bef_2+\bef_5, \bef_1+\bef_4$ are fundamental rays for them, we conclude that they are normal.
In conclusion, there are six fundamental normal rays, each spanned by one of the following vectors:
\[
\{\bef_1,\bef_3,\bef_5,\bef_2+\bef_4,\bef_2+\bef_5, \bef_1+\bef_4\}.
\]

Finally, we evaluate each of this generators on the four vectors obtained in Equation~\eqref{eq:NON_r4}, to obtain the right-hand side of the following non-redundant $H$-representation of $\FerPoly[4][\bw,3,6]$. 
Alternatively, the right-hand side can be determined using Proposition~\ref{prop:rhs}.
\begin{equation}\label{eq:rhs_example}
\FerPoly[4][\bw,3,6]=\left\{\bx\in\RR^6:\quad
\begin{blockarray}{rrrrrrrcrrrrl}
\begin{block}{(rrrrrr)lcrrrrl}
1 & 1 & 1 & 1 & 1 & 1 & \bx^{\downarrow} & = & 3 \\
\end{block}
\begin{block}{(rrrrrr)rc(rrrr)l}
1 & 0 & 0 & 0 & 0 & 0 & \multirow{6}{*}{$\bx^{\downarrow}$} & \multirow{6}{*}{$\leq$} & 1 & 1 & 1 & 1 & \multirow{6}{*}{$\bw$}\\
1 & 1 & 1 & 1 & 1 & 0 & & & 3 & 3 & 3 & 3 \\
1 & 1 & 1 & 0 & 0 & 0 & & & 3 & 2 & 2 & 2 \\
2 & 2 & 1 & 1 & 0 & 0 & & & 5 & 5 & 4 & 4 \\
2 & 1 & 1 & 1 & 0 & 0 & & & 4 & 4 & 4 & 3 \\
2 & 2 & 1 & 1 & 1 & 0 & & & 5 & 5 & 5 & 4 \\
\end{block}
\end{blockarray}\right\}.
\end{equation}

By Theorem~\ref{thm:stability}, we obtain the following characterization of the $H$-representation of $\FerPoly[4][\bw,N,d]$ for $N\geq3$, $d\geq 6$ and $d-N\geq3$.

\begin{theorem}
\label{thm:r4}
Let $N\geq3$, $d\geq 6$ and $d-N\geq3$.
The polytope $\FerPoly[4][\bw,N,d]$ is the subset of $\RR^d$ defined by the equality $\sum_{i=1}^d x_i = N$ and the inequalities
\[
\begin{array}{rcl}
x_1^{\downarrow} & \leq & 1, \\
\sum_{i=1}^{d-1} x_i^{\downarrow} & \leq & N,\\
\sum_{i=1}^{N} x_i^{\downarrow}   & \leq & N-1+w_1,\\
2\sum_{i=1}^{N-1}x_i^{\downarrow} + (x_N^{\downarrow} + x_{N+1}^{\downarrow}) & \leq & 2N-2 + w_1 + w_2,\\
2\sum_{i=1}^{N-2}x_i^{\downarrow} +(x_{N-1}^{\downarrow}+x_{N}^{\downarrow}+x_{N+1}^{\downarrow}) & \leq & 2N-3+w_1+w_2+w_3,\\
2\sum_{i=1}^{N-1}x_i^{\downarrow} +(x_{N}^{\downarrow}+x_{N+1}^{\downarrow}+x_{N+2}^{\downarrow}) & \leq & 2N-2 +w_1+w_2 + w_3.\\
\end{array}
\]
Furthermore, this $H$-representation is non-redundant.
\end{theorem}

\begin{proof}
We take the particular case $\FerPoly[4][\bw,3,6]$ and raise the parameters $N,d$ using Theorem~\ref{thm:stability}.
\end{proof}

Since we have computed the base case for $r\leq 13$, we have a minimal description for small $r$ and \emph{arbitrarily large} $N$ and $d$.
The next theorem gives the case $r=5$, and the cases $r\in\{6,7,8\}$ are given in Appendix~\ref{app:ferm}.

\begin{theorem}
\label{thm:r5}
Let $N\geq4$, $d\geq 8$ and $d-N\geq4$.
The polytope $\FerPoly[5][\bw,N,d]$ is the subset of $\RR^d$ defined by the equality $\sum_{i=1}^d x_i = N$, the inequalities in Theorem~\ref{thm:r4} and the three inequalities
\[
\begin{array}{rcl}
2\sum_{i=1}^{N-3}x_i^{\downarrow} +(x_{N-2}^{\downarrow}+x_{N-1}^{\downarrow}x_{N}^{\downarrow}+x_{N+1}^{\downarrow}) & \leq & 2N-4+w_1+w_2+w_3+w_4,\\
2\sum_{i=1}^{N-1}x_i^{\downarrow} +(x_{N}^{\downarrow}+x_{N+1}^{\downarrow}+x_{N+2}^{\downarrow}+x_{N+3}^{\downarrow}) & \leq & 2N-2 +w_1+w_2 + w_3 + w_4,\\
4\sum_{i=1}^{N-2}x_i^{\downarrow} +3x_{N-1}^{\downarrow}+2(x_{N}^{\downarrow}+x_{N+1}^{\downarrow})+x_{N+2}^{\downarrow} & \leq & 4N-5 +2w_1+2w_2 + w_3 + w_4.\\
\end{array}
\]
Furthermore, this $H$-representation is non-redundant.
\end{theorem}

\subsubsection{For bosons}

We now determine a non-redundant $H$-representation of the polytope $\BosPoly[4][\bw,3,4]$.
We begin by finding the $V$-representation.
There are four fundamental lineups in $\BosPC[3][4]$, hence four occupation vectors:
\begin{equation}\label{eq:NON_r4_bos}
\resizebox{0.9\textwidth}{!}{$
\begin{array}{rlcl}
\ell_1:&\bchi\left(111\right)\rightarrow\bchi\left(112\right)\rightarrow\bchi\left(113\right)\rightarrow\bchi\left(114\right)\leadsto\\
\ell_2:&\bchi\left(111\right)\rightarrow\bchi\left(112\right)\rightarrow\bchi\left(113\right)\rightarrow\bchi\left(122\right)\leadsto\\
\ell_3:&\bchi\left(111\right)\rightarrow\bchi\left(112\right)\rightarrow\bchi\left(122\right)\rightarrow\bchi\left(113\right)\leadsto\\
\ell_4:&\bchi\left(111\right)\rightarrow\bchi\left(112\right)\rightarrow\bchi\left(122\right)\rightarrow\bchi\left(222\right)\leadsto
\end{array}
\begin{array}{r@{}rrrr}
\OccVec_{\bw}(\ell_1) = (&2+w_1,&w_2,&w_3,&w_4), \\
\OccVec_{\bw}(\ell_2) = (&1+2w_1+w_2+w_3,&w_2+w_4,&w_3,&0),\\
\OccVec_{\bw}(\ell_3) = (&1+2w_1+w_2+w_4,&w_2+w_3,&w_4,&0),\\
\OccVec_{\bw}(\ell_4) = ( & 3w_1+2w_2+w_3,&w_2+2w_3+3w_4,&0,&0).
\end{array}$}
\end{equation}
For each lineup we determine its fundamental cone as in the last example.
Also, as before we use coordinates on the fundamental basis to write the four cones and by taking an affine slice, it is possible to represent them, see~Figure~\ref{fig:fan_r4_bos}.

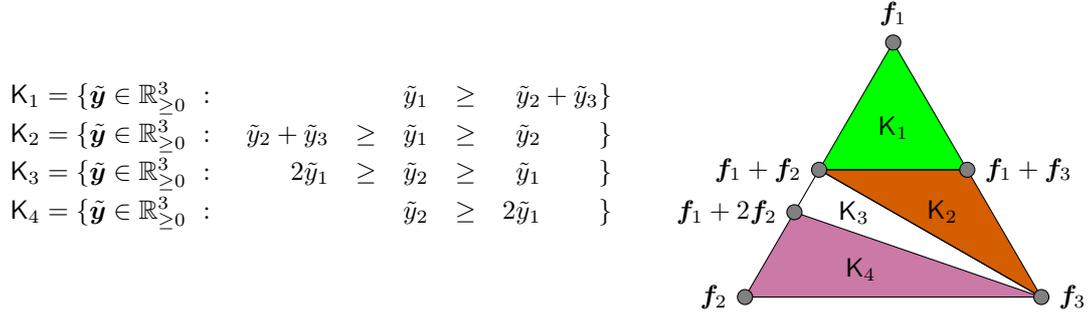
\begin{figure}[!ht]
\begin{center}
\begin{tikzpicture}

\node at (-3.75,0) {$\begin{array}{crcccl@{}l}
\pol[K]_1 = \{\tilde{\by}\in\RR^3_{\geq0}~:~ & & & \tilde{y}_1 & \geq & \phantom{1}\tilde{y}_2+\tilde{y}_3 &\} \\[0.25em]
\pol[K]_2 = \{\tilde{\by}\in\RR^3_{\geq0}~:~ & \tilde{y}_2+\tilde{y}_3 & \geq & \tilde{y}_1 & \geq & \phantom{1}\tilde{y}_2 & \}\\[0.25em]
\pol[K]_3 = \{\tilde{\by}\in\RR^3_{\geq0}~:~ & 2\tilde{y}_1 & \geq & \tilde{y}_2 & \geq & \phantom{1}\tilde{y}_1 & \}\\[0.25em]
\pol[K]_4 = \{\tilde{\by}\in\RR^3_{\geq0}~:~  & & & \tilde{y}_2 & \geq & 2\tilde{y}_1 &\}
\end{array}$};
\node at  (3.75,0) {	\begin{tikzpicture}%
	[x={(-30:1)},
	y={(0,1)},
	z={(-1,-1)},
	scale=0.75,
	point/.style={inner sep=2pt,circle,draw=black,fill=white!50!black},
	rotate=120]
	
	\coordinate (003) at (0, 0, 3) {};
	\coordinate (012) at (0, 1, 2) {};
	\coordinate (021) at (0, 2, 1) {};
	\coordinate (030) at (0, 3, 0) {};
	\coordinate (102) at (1, 0, 2) {};
	\coordinate (111) at (1, 1, 1) {};
	\coordinate (120) at (1, 2, 0) {};
	\coordinate (201) at (2, 0, 1) {};
	\coordinate (210) at (2, 1, 0) {};
	\coordinate (300) at (3, 0, 0) {};
	\coordinate (110) at (1.5,1.5,0) {};
	\coordinate (101) at (1.5,0,1.5) {};
	
	\fill[purple] (030) -- (003) -- (120) -- cycle;
	\fill[red] (003) -- (110) -- (101) -- cycle;
	\fill[green] (300) -- (110) -- (101) -- cycle;
	
	\draw (300) -- (110)--(101)--(300);
	\draw (003) -- (110)--(101)--(003);
	\draw (003) -- (110)--(120);
	\draw (003) -- (030)--(120)--(003);
	
	\node[point,label=above:{$\bef_1$}] (n300) at (300) {};	
	\node[point,label=right:{$\bef_3$}] (n003) at (003) {};
	\node[point,label=left:{$\bef_2$}] (n030) at (030) {};	
	\node[point,label=left:{$\bef_1+\bef_2$}] (n110) at (110) {};
	\node[point,label=right:{$\bef_1+\bef_3$}] (n101) at (101) {};
	\node[point,label=left:{$\bef_1+2\bef_2$}] (n120) at (120) {};
	
	\node[label=center:{$\mathsf{K}_1$}] (nk1) at ($(110)!0.5!(101)!0.33!(300)$) {};
	\node[label=center:{$\mathsf{K}_2$}] (nk2) at ($(110)!0.5!(101)!0.33!(003)$) {};	
	\node[label=center:{$\mathsf{K}_3$}] (nk3) at ($(110)!0.5!(120)!0.2!(003)$){};
	\node[label=center:{$\mathsf{K}_4$}] (nk4) at ($(030)!0.5!(003)!0.33!(120)$) {};

\end{tikzpicture}};

\end{tikzpicture}

\end{center}
\caption{Affine slice of the subdivision of $\Phi_3$ given by the four fundamental cones.}
\label{fig:fan_r4_bos}
\end{figure}

There are four triangles, corresponding to the four lineups. 
The fundamental fan has six rays, each spanned by one of the following vectors:
\begin{equation}\label{eq:rays_boson_4}
\{\bef_1,\bef_2,\bef_3,\bef_1+\bef_2,\bef_1+\bef_3, \bef_1+2\bef_2\}.
\end{equation}
In this case, the first three occupation vectors have distinct entries so that their normal cones are contained in the fundamental chamber.
This verifies that all vectors in Equation \eqref{eq:rays_boson_4} except $\bef_2$ generate normal rays.
A simple calculation shows that the ray spanned by $\bef_2$ is not a normal ray.

Finally, we evaluate each of this generators on the four vectors obtained in Equation \eqref{eq:NON_r4_bos}, to obtain the right-hand side of following non-redundant $H$-representation.
\[
\
\BosPoly[4][\bw,3,4]=\left\{\bx\in\RR^4:
\begin{array}{lcl}
	x^\downarrow_1+x^\downarrow_2+x^\downarrow_3&\leq& 3\\
	x^\downarrow_1&\leq& 2+w_1 \\ 
	2x^\downarrow_1+x^\downarrow_2&\leq& 4+2w_1+w_2\\
	2x^\downarrow_1+x^\downarrow_2+x^\downarrow_3&\leq& 4+2w_1+w_2+w_3\\
	3x^\downarrow_1+2x^\downarrow_2&\leq& 6+3w_1+2w_2+w_3\\
\end{array}
\right\}.
\]

Finally we use Theorem \ref{thm:stability_bosonic} to produce the following general result for the case $r=4$.
The cases $r\in\{5,6,7,8\}$ are given in Appendix~\ref{app:bos}.
\begin{theorem}
Let $N\geq3$ and $d\geq 4$.
The polytope $\BosPoly[4][\bw,N,d]$ is the subset of $\RR^d$ defined by the equality $\sum_{i=1}^d x_i = N$ and the inequalities
\[
\begin{array}{lcl}
\sum_{i=1}^{d-1}x^\downarrow_i&\leq& N,\\
x^\downarrow_1&\leq& N-1+w_1,\\ 
2x^\downarrow_1+x^\downarrow_2&\leq& 2N-2+2w_1+w_2,\\
2x^\downarrow_1+x^\downarrow_2+x^\downarrow_3&\leq& 2N-2+2w_1+w_2+w_3,\\
3x^\downarrow_1+2x^\downarrow_2&\leq& 3N-3+3w_1+2w_2+w_3.
\end{array}
\]
Furthermore, this $H$-representation is non-redundant.
\end{theorem}

\section{Outro}
\label{sec:outro}

In this section we conclude with discussions that emerged from the effective solution to convex $1$-body $N$-representability.

\subsection{Comparison with Klyachko's polytope}
\label{ssec:klyachko}

We review Klyachko's $H$-representation of $\klypo=\spec^{\downarrow}(\Density^1_N(\bw))$ using our approach.
There are certain cohomological computations that play a key role in the formulation of \cite[Theorem 4.2.1]{klyachko_2006}.
Following \cite[Remark 4.2.2]{klyachko_2006} they can be described combinatorially as follows.
For $\by\in\RR^d$, $\pi\in\Sym{d}$ and $\sigma\in\Sym{D}$ we define
\[
c^\pi_\sigma(\by) = \partial_\pi S_\sigma(z)|_{z_k=s(\by)_k},
\]
where $s(\by) := \left(\sum_{i\in S}y_i~:~S\in\binom{[d]}{N}\right)^\downarrow\in\RR^D$,
$S_\sigma\in\mathbb{Z}[z_1,\dots,z_D]$ is the \emph{Schubert polynomial} of $\sigma$, and $\partial_\pi$ is a differential operator acting on $\mathbb{Z}[y_1,\dots,y_d]$ see \cite[Section~10.3]{fulton_young_1997} for further details.
The most relevant case for us is when $\sigma$ is the identity.
In this case, the Schubert polynomial is equal to 1, so we have $c^\pi_e(\by)=1$ when $\pi$ is equal to the identity and zero in all other cases.

Now we state the linear inequalities described in \cite[Theorem 4.2.1]{klyachko_2006}.
For each fundamental normal ray generator $\by$ of the fan $\FerFan[D]$ (in \cite[Page 10]{klyachko_pauli_2009}, these are referred to as the extremal edges of the principal cubicle), and permutations $\pi\in\Sym{d}$ and $\sigma\in\Sym{D}$, there is an inequality
\begin{equation}\label{eq:kly_inequality}
\np{\pi\cdot \by,\bx}\leq \np{\sigma\cdot s(\by),\bw},\text{ whenever }c^\pi_\sigma(\by)\neq 0.
\end{equation} 
We refer to the last condition as the \emph{topological condition}.
By contrast, our inequalities for the convex symmetrization of $\klypo$ are of the form (see Proposition \ref{prop:rhs})
\begin{equation}\label{eq:our_inequality}
\np{\pi\cdot \by,\bx}\leq \np{s(\by),\bw}.
\end{equation} 
We take $\by$ to be any normal ray, but we write it as $\pi$ acting on a fundamental normal ray to make our formulation closer to that of Equation \eqref{eq:kly_inequality}.
Comparing the inequalities in \eqref{eq:our_inequality} to those in \eqref{eq:kly_inequality}, we see that in \eqref{eq:our_inequality} we are ignoring the topological condition always setting $\sigma$ to be the identity.
However, when both permutations are the identity, we have the topological condition satisfied and we recover his inequalities (these are called basic inequalities in \cite[Section 3.2.1]{klyachko_pauli_2009}).
Another difference is that the inequalities in \eqref{eq:kly_inequality} have redundancy, whereas every inequality in~\eqref{eq:our_inequality} is facet defining.
Interestingly, the left-hand sides are the same.
This means that the normals of the facet defining inequalities for $\klypo$ are a subset of that of $\FerPoly[][\bw,N,d]$.
In fact, we conjecture:

\begin{conjecture}
If $\bw$ has distinct entries, then the polytope $\klypo$ is a Minkowski summand of $\FerPoly[][\bw,N,d]$.
\end{conjecture}

In \cite[Section~4.3.4]{klyachko_2006} and \cite[Section~5.1]{altunbulak_pauli_2008}, the authors propose a numerical procedure to produce all \emph{pure} $N$-representability constraints (i.e. when $\bw=(1,0,\dots)$).
The utilized technique mixes $V$-representation (obtaining certain points that should belong to $\Density^1_N((1,0,\dots))$ through representation theory procedures) and $H$-representation (by recognizing inequalities for pure $N$-representability through Schubert calculus), and repeating with increasingly larger parameters until the process stabilizes once all inequalities have been recognized as instances of Theorem~2 of \cite{altunbulak_pauli_2008}.
The authors also use a criterion for unspecified weight vector $\bw$, see \cite[Theorem~3]{altunbulak_pauli_2008}.
Algorithm~\ref{algo:lineup} provides an enhanced version of this result in which a fixed weight $\bw$ is fixed.

\begin{example}
	Let $N=2$, $d=4$, $\bw_1=(\frac{1}{2},\frac{1}{3},\frac{1}{6},0,0,0)$, and $\bw_2=\frac{1}{21}(6,5,4,3,2,1)$.
	Figure~\ref{fig:ex_42_3weights} illustrates $\FerPoly[][\bw_1,2,4]$.
	The symmetric polytope $\FerPoly[][\bw_1,2,4]$ contains $\spec^{\downarrow}\left(\Density^1_2(\bw_1)\right)$, obtained from \cite[Section~4.2.3]{klyachko_2006}.
	Figure~\ref{fig:hierarchy} illustrates the difference between the polytopes $\pol[\Sigma]^{\mathrm{f}\downarrow}_3(\bw,2,4)$ and $\spec^{\downarrow}\left(\Density^1_2(\bw)\right)$ for $\bw\in\{\bw_1,\bw_2\}$.
\end{example}

\begin{figure}[!h]
	\begin{center}
		\begin{tikzpicture}[auto]
		\node (step1) at (-3.5,0) {\input{ordered_spec423}};
		\node (step2) at  (3.5,0) {\begin{tikzpicture}%
	[x={(0.420026cm, -0.294365cm)},
	y={(0.907512cm, 0.136282cm)},
	z={(-0.000039cm, 0.945926cm)},
	scale=10.000000,
	back/.style={densely dotted},
	edge/.style={color=red,cap=round,thick},
	edgekly/.style={color=forestgreen,thick,cap=round},
	facet/.style={fill=forestgreen,fill opacity=0.600000},
	baseline=(a)]
%
%

\coordinate (a) at (-0.11785, 0.00000, 0.08333);
\coordinate (-0.11785, 0.00000, 0.08333) at (-0.11785, 0.00000, 0.08333);
\coordinate (0.05893, 0.17678, 0.08333) at (0.05893, 0.17678, 0.08333);
\coordinate (0.11785, 0.00000, -0.08333) at (0.11785, 0.00000, -0.08333);
\coordinate (0.11785, 0.00000, 0.25000) at (0.11785, 0.00000, 0.25000);
\coordinate (0.14731, 0.08839, -0.04167) at (0.14731, 0.08839, -0.04167);
\coordinate (0.14731, 0.08839, 0.20833) at (0.14731, 0.08839, 0.20833);
\coordinate (0.17678, 0.00000, 0.00000) at (0.17678, 0.00000, 0.00000);
\coordinate (0.17678, 0.00000, 0.16667) at (0.17678, 0.00000, 0.16667);
\coordinate (0.17678, 0.05893, 0.00000) at (0.17678, 0.05893, 0.00000);
\coordinate (0.17678, 0.05893, 0.16667) at (0.17678, 0.05893, 0.16667);

\coordinate (0.17678, 0.05893, 0.00000) at (0.17678, 0.05893, 0.00000);
\coordinate (0.17678, 0.05893, 0.16667) at (0.17678, 0.05893, 0.16667);
\coordinate (0.11785, 0.00000, 0.16667) at (0.11785, 0.00000, 0.16667);
\coordinate (0.11785, 0.11785, 0.00000) at (0.11785, 0.11785, 0.00000);
\coordinate (0.00000, 0.00000, 0.00000) at (0.00000, 0.00000, 0.00000);
\coordinate (0.11785, 0.00000, 0.00000) at (0.11785, 0.00000, 0.00000);
\coordinate (-0.11785, 0.00000, 0.08333) at (-0.11785, 0.00000, 0.08333);
\coordinate (0.05893, 0.17678, 0.08333) at (0.05893, 0.17678, 0.08333);
\coordinate (0.11785, 0.11785, 0.16667) at (0.11785, 0.11785, 0.16667);
\coordinate (0.00000, 0.00000, 0.16667) at (0.00000, 0.00000, 0.16667);
\draw[edge,back] (-0.11785, 0.00000, 0.08333) -- (0.05893, 0.17678, 0.08333);
\draw[edgekly] (0.11785, 0.11785, 0.00000) -- (0.00000, 0.00000, 0.00000);
\draw[edgekly] (-0.11785, 0.00000, 0.08333) -- (0.05893, 0.17678, 0.08333);
\fill[facet] (0.11785, 0.11785, 0.16667) -- (0.17678, 0.05893, 0.16667) -- (0.17678, 0.05893, 0.00000) -- (0.11785, 0.11785, 0.00000) -- (0.05893, 0.17678, 0.08333) -- cycle {};
\fill[facet] (0.11785, 0.00000, 0.00000) -- (0.17678, 0.05893, 0.00000) -- (0.17678, 0.05893, 0.16667) -- (0.11785, 0.00000, 0.16667) -- cycle {};
\fill[facet] (0.00000, 0.00000, 0.16667) -- (0.11785, 0.00000, 0.16667) -- (0.17678, 0.05893, 0.16667) -- (0.11785, 0.11785, 0.16667) -- cycle {};
\fill[facet] (0.00000, 0.00000, 0.16667) -- (0.11785, 0.00000, 0.16667) -- (0.11785, 0.00000, 0.00000) -- (0.00000, 0.00000, 0.00000) -- (-0.11785, 0.00000, 0.08333) -- cycle {};
\draw[edge] (-0.11785, 0.00000, 0.08333) -- (0.11785, 0.00000, -0.08333);
\draw[edge] (-0.11785, 0.00000, 0.08333) -- (0.11785, 0.00000, 0.25000);
\draw[edge] (0.05893, 0.17678, 0.08333) -- (0.14731, 0.08839, -0.04167);
\draw[edge] (0.05893, 0.17678, 0.08333) -- (0.14731, 0.08839, 0.20833);
\draw[edge] (0.11785, 0.00000, -0.08333) -- (0.14731, 0.08839, -0.04167);
\draw[edge] (0.11785, 0.00000, -0.08333) -- (0.17678, 0.00000, 0.00000);
\draw[edge] (0.11785, 0.00000, 0.25000) -- (0.14731, 0.08839, 0.20833);
\draw[edge] (0.11785, 0.00000, 0.25000) -- (0.17678, 0.00000, 0.16667);
\draw[edge] (0.14731, 0.08839, -0.04167) -- (0.17678, 0.05893, 0.00000);
\draw[edge] (0.14731, 0.08839, 0.20833) -- (0.17678, 0.05893, 0.16667);
\draw[edge] (0.17678, 0.00000, 0.00000) -- (0.17678, 0.00000, 0.16667);
\draw[edge] (0.17678, 0.00000, 0.00000) -- (0.17678, 0.05893, 0.00000);
\draw[edge] (0.17678, 0.00000, 0.16667) -- (0.17678, 0.05893, 0.16667);
\draw[edge] (0.17678, 0.05893, 0.00000) -- (0.17678, 0.05893, 0.16667);
\draw[edgekly] (0.17678, 0.05893, 0.00000) -- (0.17678, 0.05893, 0.16667);
\draw[edgekly] (0.17678, 0.05893, 0.00000) -- (0.11785, 0.11785, 0.00000);
\draw[edgekly] (0.17678, 0.05893, 0.00000) -- (0.11785, 0.00000, 0.00000);
\draw[edgekly] (0.17678, 0.05893, 0.16667) -- (0.11785, 0.00000, 0.16667);
\draw[edgekly] (0.17678, 0.05893, 0.16667) -- (0.11785, 0.11785, 0.16667);
\draw[edgekly] (0.11785, 0.00000, 0.16667) -- (0.11785, 0.00000, 0.00000);
\draw[edgekly] (0.11785, 0.00000, 0.16667) -- (0.00000, 0.00000, 0.16667);
\draw[edgekly] (0.11785, 0.11785, 0.00000) -- (0.05893, 0.17678, 0.08333);
\draw[edgekly] (0.00000, 0.00000, 0.00000) -- (0.11785, 0.00000, 0.00000);
\draw[edgekly] (0.00000, 0.00000, 0.00000) -- (-0.11785, 0.00000, 0.08333);
\draw[edgekly] (-0.11785, 0.00000, 0.08333) -- (0.00000, 0.00000, 0.16667);
\draw[edgekly] (0.05893, 0.17678, 0.08333) -- (0.11785, 0.11785, 0.16667);
\draw[edgekly] (0.11785, 0.11785, 0.16667) -- (0.00000, 0.00000, 0.16667);
\end{tikzpicture}};
		\end{tikzpicture}
	\end{center}
	\caption{On the left, the polytope $\FerPoly[][\bw_1,2,4]$ that includes $\spec^{\downarrow}\left(\Density^1_2(\bw_1)\right)$.
		On the right, the polytope $\pol[\Sigma]^{\mathrm{f}\downarrow}_3(\bw_1,2,4)$ includes $\spec^{\downarrow}\left(\Density^1_2(\bw_1)\right)$.}
	\label{fig:ex_42_3weights}
\end{figure}

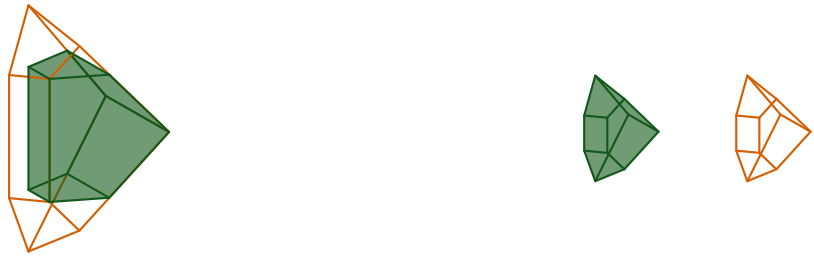
\begin{figure}[!h]
	\begin{center}
		\begin{tikzpicture}[auto]
		\node (step1) at (-5,0) {\begin{tikzpicture}%
	[x={(-0.431539cm, -0.182926cm)},
	y={(0.902094cm, -0.087614cm)},
	z={(0.000099cm, 0.979215cm)},
	scale=10.000000,
	edge/.style={color=red,cap=round,thick},
	edgekly/.style={color=black,very thick,cap=round},
	edgekly/.style={color=forestgreen,thick,cap=round},
	facet/.style={fill=forestgreen,fill opacity=0.600000},
	baseline=(a)]
%
%

\coordinate (a) at (-0.11785, 0.00000, 0.08333);
\coordinate (-0.11785, 0.00000, 0.08333) at (-0.11785, 0.00000, 0.08333);
\coordinate (0.05893, 0.17678, 0.08333) at (0.05893, 0.17678, 0.08333);
\coordinate (0.11785, 0.00000, -0.08333) at (0.11785, 0.00000, -0.08333);
\coordinate (0.11785, 0.00000, 0.25000) at (0.11785, 0.00000, 0.25000);
\coordinate (0.14731, 0.08839, -0.04167) at (0.14731, 0.08839, -0.04167);
\coordinate (0.14731, 0.08839, 0.20833) at (0.14731, 0.08839, 0.20833);
\coordinate (0.17678, 0.00000, 0.00000) at (0.17678, 0.00000, 0.00000);
\coordinate (0.17678, 0.00000, 0.16667) at (0.17678, 0.00000, 0.16667);
\coordinate (0.17678, 0.05893, 0.00000) at (0.17678, 0.05893, 0.00000);
\coordinate (0.17678, 0.05893, 0.16667) at (0.17678, 0.05893, 0.16667);
\coordinate (0.17678, 0.05893, 0.00000) at (0.17678, 0.05893, 0.00000);
\coordinate (0.17678, 0.05893, 0.16667) at (0.17678, 0.05893, 0.16667);
\coordinate (0.11785, 0.00000, 0.16667) at (0.11785, 0.00000, 0.16667);
\coordinate (0.11785, 0.11785, 0.00000) at (0.11785, 0.11785, 0.00000);
\coordinate (0.00000, 0.00000, 0.00000) at (0.00000, 0.00000, 0.00000);
\coordinate (0.11785, 0.00000, 0.00000) at (0.11785, 0.00000, 0.00000);
\coordinate (-0.11785, 0.00000, 0.08333) at (-0.11785, 0.00000, 0.08333);
\coordinate (0.05893, 0.17678, 0.08333) at (0.05893, 0.17678, 0.08333);
\coordinate (0.11785, 0.11785, 0.16667) at (0.11785, 0.11785, 0.16667);
\coordinate (0.00000, 0.00000, 0.16667) at (0.00000, 0.00000, 0.16667);
\draw[edge] (-0.11785, 0.00000, 0.08333) -- (0.05893, 0.17678, 0.08333);
\draw[edge] (-0.11785, 0.00000, 0.08333) -- (0.11785, 0.00000, -0.08333);
\draw[edge] (-0.11785, 0.00000, 0.08333) -- (0.11785, 0.00000, 0.25000);
\draw[edgekly] (0.11785, 0.11785, 0.00000) -- (0.00000, 0.00000, 0.00000);
\draw[edgekly] (0.00000, 0.00000, 0.00000) -- (0.11785, 0.00000, 0.00000);
\draw[edgekly] (0.00000, 0.00000, 0.00000) -- (-0.11785, 0.00000, 0.08333);
\draw[edgekly] (-0.11785, 0.00000, 0.08333) -- (0.05893, 0.17678, 0.08333);
\draw[edgekly] (-0.11785, 0.00000, 0.08333) -- (0.00000, 0.00000, 0.16667);
\fill[facet] (0.11785, 0.11785, 0.16667) -- (0.17678, 0.05893, 0.16667) -- (0.17678, 0.05893, 0.00000) -- (0.11785, 0.11785, 0.00000) -- (0.05893, 0.17678, 0.08333) -- cycle {};
\fill[facet] (0.11785, 0.00000, 0.00000) -- (0.17678, 0.05893, 0.00000) -- (0.17678, 0.05893, 0.16667) -- (0.11785, 0.00000, 0.16667) -- cycle {};
\fill[facet] (0.00000, 0.00000, 0.16667) -- (0.11785, 0.00000, 0.16667) -- (0.17678, 0.05893, 0.16667) -- (0.11785, 0.11785, 0.16667) -- cycle {};
\draw[edge] (0.05893, 0.17678, 0.08333) -- (0.14731, 0.08839, -0.04167);
\draw[edge] (0.05893, 0.17678, 0.08333) -- (0.14731, 0.08839, 0.20833);
\draw[edge] (0.11785, 0.00000, -0.08333) -- (0.14731, 0.08839, -0.04167);
\draw[edge] (0.11785, 0.00000, -0.08333) -- (0.17678, 0.00000, 0.00000);
\draw[edge] (0.11785, 0.00000, 0.25000) -- (0.14731, 0.08839, 0.20833);
\draw[edge] (0.11785, 0.00000, 0.25000) -- (0.17678, 0.00000, 0.16667);
\draw[edge] (0.14731, 0.08839, -0.04167) -- (0.17678, 0.05893, 0.00000);
\draw[edge] (0.14731, 0.08839, 0.20833) -- (0.17678, 0.05893, 0.16667);
\draw[edge] (0.17678, 0.00000, 0.00000) -- (0.17678, 0.00000, 0.16667);
\draw[edge] (0.17678, 0.00000, 0.00000) -- (0.17678, 0.05893, 0.00000);
\draw[edge] (0.17678, 0.00000, 0.16667) -- (0.17678, 0.05893, 0.16667);
\draw[edge] (0.17678, 0.05893, 0.00000) -- (0.17678, 0.05893, 0.16667);
\draw[edgekly] (0.17678, 0.05893, 0.00000) -- (0.17678, 0.05893, 0.16667);
\draw[edgekly] (0.17678, 0.05893, 0.00000) -- (0.11785, 0.11785, 0.00000);
\draw[edgekly] (0.17678, 0.05893, 0.00000) -- (0.11785, 0.00000, 0.00000);
\draw[edgekly] (0.17678, 0.05893, 0.16667) -- (0.11785, 0.00000, 0.16667);
\draw[edgekly] (0.17678, 0.05893, 0.16667) -- (0.11785, 0.11785, 0.16667);
\draw[edgekly] (0.11785, 0.00000, 0.16667) -- (0.11785, 0.00000, 0.00000);
\draw[edgekly] (0.11785, 0.00000, 0.16667) -- (0.00000, 0.00000, 0.16667);
\draw[edgekly] (0.11785, 0.11785, 0.00000) -- (0.05893, 0.17678, 0.08333);
\draw[edgekly] (0.05893, 0.17678, 0.08333) -- (0.11785, 0.11785, 0.16667);
\draw[edgekly] (0.11785, 0.11785, 0.16667) -- (0.00000, 0.00000, 0.16667);
\end{tikzpicture}};
		\node (step2) at  (2,0) {\begin{tikzpicture}%
	[x={(-0.431539cm, -0.182926cm)},
	y={(0.902094cm, -0.087614cm)},
	z={(0.000099cm, 0.979215cm)},
	scale=10.000000,
	edge/.style={color=purple,cap=round,thick},
	edgekly/.style={color=forestgreen,thick,cap=round},
	facet/.style={fill=forestgreen,fill opacity=0.600000},
	baseline=(a)]
%
%

\coordinate (a) at (-0.11785, 0.00000, 0.08333);
\coordinate (-0.11785, 0.00000, 0.08333) at (-0.11785, 0.00000, 0.08333);
\coordinate (0.05893, 0.17678, 0.08333) at (0.05893, 0.17678, 0.08333);
\coordinate (0.11785, 0.00000, -0.08333) at (0.11785, 0.00000, -0.08333);
\coordinate (0.11785, 0.00000, 0.25000) at (0.11785, 0.00000, 0.25000);
\coordinate (0.14731, 0.08839, -0.04167) at (0.14731, 0.08839, -0.04167);
\coordinate (0.14731, 0.08839, 0.20833) at (0.14731, 0.08839, 0.20833);
\coordinate (0.17678, 0.00000, 0.00000) at (0.17678, 0.00000, 0.00000);
\coordinate (0.17678, 0.00000, 0.16667) at (0.17678, 0.00000, 0.16667);
\coordinate (0.17678, 0.05893, 0.00000) at (0.17678, 0.05893, 0.00000);
\coordinate (0.17678, 0.05893, 0.16667) at (0.17678, 0.05893, 0.16667);
\coordinate (-0.01684, 0.00000, 0.15476) at (-0.01684, 0.00000, 0.15476);
\coordinate (0.00000, 0.05051, 0.13095) at (0.00000, 0.05051, 0.13095);
\coordinate (-0.01684, 0.00000, 0.01190) at (-0.01684, 0.00000, 0.01190);
\coordinate (0.01684, 0.00000, 0.05952) at (0.01684, 0.00000, 0.05952);
\coordinate (0.01684, 0.03367, 0.10714) at (0.01684, 0.03367, 0.10714);
\coordinate (0.01684, 0.00000, 0.10714) at (0.01684, 0.00000, 0.10714);
\coordinate (0.00000, 0.05051, 0.03571) at (0.00000, 0.05051, 0.03571);
\coordinate (0.01684, 0.03367, 0.05952) at (0.01684, 0.03367, 0.05952);
\coordinate (-0.03367, 0.08418, 0.08333) at (-0.03367, 0.08418, 0.08333);
\coordinate (-0.11785, 0.00000, 0.08333) at (-0.11785, 0.00000, 0.08333);
\draw[edgekly] (-0.01684, 0.00000, 0.15476) -- (-0.11785, 0.00000, 0.08333);
\draw[edgekly] (-0.01684, 0.00000, 0.01190) -- (-0.11785, 0.00000, 0.08333);
\draw[edgekly] (-0.03367, 0.08418, 0.08333) -- (-0.11785, 0.00000, 0.08333);
\fill[facet] (0.01684, 0.03367, 0.10714) -- (0.01684, 0.03367, 0.05952) -- (0.01684, 0.00000, 0.05952) -- (0.01684, 0.00000, 0.10714) -- cycle {};
\fill[facet] (0.00000, 0.05051, 0.03571) -- (-0.03367, 0.08418, 0.08333) -- (0.00000, 0.05051, 0.13095) -- (0.01684, 0.03367, 0.10714) -- (0.01684, 0.03367, 0.05952) -- cycle {};
\fill[facet] (0.01684, 0.03367, 0.05952) -- (0.01684, 0.00000, 0.05952) -- (-0.01684, 0.00000, 0.01190) -- (0.00000, 0.05051, 0.03571) -- cycle {};
\fill[facet] (0.01684, 0.00000, 0.10714) -- (-0.01684, 0.00000, 0.15476) -- (0.00000, 0.05051, 0.13095) -- (0.01684, 0.03367, 0.10714) -- cycle {};
\draw[edgekly] (-0.01684, 0.00000, 0.15476) -- (0.00000, 0.05051, 0.13095);
\draw[edgekly] (-0.01684, 0.00000, 0.15476) -- (0.01684, 0.00000, 0.10714);
\draw[edgekly] (0.00000, 0.05051, 0.13095) -- (0.01684, 0.03367, 0.10714);
\draw[edgekly] (0.00000, 0.05051, 0.13095) -- (-0.03367, 0.08418, 0.08333);
\draw[edgekly] (-0.01684, 0.00000, 0.01190) -- (0.01684, 0.00000, 0.05952);
\draw[edgekly] (-0.01684, 0.00000, 0.01190) -- (0.00000, 0.05051, 0.03571);
\draw[edgekly] (0.01684, 0.00000, 0.05952) -- (0.01684, 0.00000, 0.10714);
\draw[edgekly] (0.01684, 0.00000, 0.05952) -- (0.01684, 0.03367, 0.05952);
\draw[edgekly] (0.01684, 0.03367, 0.10714) -- (0.01684, 0.00000, 0.10714);
\draw[edgekly] (0.01684, 0.03367, 0.10714) -- (0.01684, 0.03367, 0.05952);
\draw[edgekly] (0.00000, 0.05051, 0.03571) -- (0.01684, 0.03367, 0.05952);
\draw[edgekly] (0.00000, 0.05051, 0.03571) -- (-0.03367, 0.08418, 0.08333);
\end{tikzpicture}};
		\node (step2) at  (4,0) {\begin{tikzpicture}%
	[x={(-0.431539cm, -0.182926cm)},
	y={(0.902094cm, -0.087614cm)},
	z={(0.000099cm, 0.979215cm)},
	scale=10.000000,
	edgekly/.style={color=red,thick,cap=round},
	facet/.style={fill=forestgreen,fill opacity=0.600000},
	baseline=(a)]
%
%

\coordinate (a) at (-0.11785, 0.00000, 0.08333);
\coordinate (-0.11785, 0.00000, 0.08333) at (-0.11785, 0.00000, 0.08333);
\coordinate (0.05893, 0.17678, 0.08333) at (0.05893, 0.17678, 0.08333);
\coordinate (0.11785, 0.00000, -0.08333) at (0.11785, 0.00000, -0.08333);
\coordinate (0.11785, 0.00000, 0.25000) at (0.11785, 0.00000, 0.25000);
\coordinate (0.14731, 0.08839, -0.04167) at (0.14731, 0.08839, -0.04167);
\coordinate (0.14731, 0.08839, 0.20833) at (0.14731, 0.08839, 0.20833);
\coordinate (0.17678, 0.00000, 0.00000) at (0.17678, 0.00000, 0.00000);
\coordinate (0.17678, 0.00000, 0.16667) at (0.17678, 0.00000, 0.16667);
\coordinate (0.17678, 0.05893, 0.00000) at (0.17678, 0.05893, 0.00000);
\coordinate (0.17678, 0.05893, 0.16667) at (0.17678, 0.05893, 0.16667);
\coordinate (-0.01684, 0.00000, 0.15476) at (-0.01684, 0.00000, 0.15476);
\coordinate (0.00000, 0.05051, 0.13095) at (0.00000, 0.05051, 0.13095);
\coordinate (-0.01684, 0.00000, 0.01190) at (-0.01684, 0.00000, 0.01190);
\coordinate (0.01684, 0.00000, 0.05952) at (0.01684, 0.00000, 0.05952);
\coordinate (0.01684, 0.03367, 0.10714) at (0.01684, 0.03367, 0.10714);
\coordinate (0.01684, 0.00000, 0.10714) at (0.01684, 0.00000, 0.10714);
\coordinate (0.00000, 0.05051, 0.03571) at (0.00000, 0.05051, 0.03571);
\coordinate (0.01684, 0.03367, 0.05952) at (0.01684, 0.03367, 0.05952);
\coordinate (-0.03367, 0.08418, 0.08333) at (-0.03367, 0.08418, 0.08333);
\coordinate (-0.11785, 0.00000, 0.08333) at (-0.11785, 0.00000, 0.08333);
\draw[edgekly] (-0.01684, 0.00000, 0.15476) -- (-0.11785, 0.00000, 0.08333);
\draw[edgekly] (-0.01684, 0.00000, 0.01190) -- (-0.11785, 0.00000, 0.08333);
\draw[edgekly] (-0.03367, 0.08418, 0.08333) -- (-0.11785, 0.00000, 0.08333);
\draw[edgekly] (-0.01684, 0.00000, 0.15476) -- (0.00000, 0.05051, 0.13095);
\draw[edgekly] (-0.01684, 0.00000, 0.15476) -- (0.01684, 0.00000, 0.10714);
\draw[edgekly] (0.00000, 0.05051, 0.13095) -- (0.01684, 0.03367, 0.10714);
\draw[edgekly] (0.00000, 0.05051, 0.13095) -- (-0.03367, 0.08418, 0.08333);
\draw[edgekly] (-0.01684, 0.00000, 0.01190) -- (0.01684, 0.00000, 0.05952);
\draw[edgekly] (-0.01684, 0.00000, 0.01190) -- (0.00000, 0.05051, 0.03571);
\draw[edgekly] (0.01684, 0.00000, 0.05952) -- (0.01684, 0.00000, 0.10714);
\draw[edgekly] (0.01684, 0.00000, 0.05952) -- (0.01684, 0.03367, 0.05952);
\draw[edgekly] (0.01684, 0.03367, 0.10714) -- (0.01684, 0.00000, 0.10714);
\draw[edgekly] (0.01684, 0.03367, 0.10714) -- (0.01684, 0.03367, 0.05952);
\draw[edgekly] (0.00000, 0.05051, 0.03571) -- (0.01684, 0.03367, 0.05952);
\draw[edgekly] (0.00000, 0.05051, 0.03571) -- (-0.03367, 0.08418, 0.08333);
\end{tikzpicture}};
		\end{tikzpicture}
	\end{center}
	\caption{On the left, the polytope $\pol[\Sigma]^{\mathrm{f}\downarrow}_3(\bw_1,2,4)$ includes $\spec^{\downarrow}\left(\Density^1_2(\bw_1)\right)$.
		On the right, the polytopes $\pol[\Sigma]^{\mathrm{f}\downarrow}_3(\bw_2,2,4)$ and $\spec^{\downarrow}\left(\Density^1_2(\bw_2)\right)$ are equal.
		The two polytopes on the right are included in $\pol[\Sigma]^{\mathrm{f}\downarrow}_3(\bw_1,2,4)$ due to the majorization $\bw_2\prec\bw_1$.}
	\label{fig:hierarchy}
\end{figure}

\noindent
This example raises the question: 

\begin{question}
	When are $\pol[\Sigma]^{\mathrm{f}\downarrow}_r(\bw,N,d)$ and $\spec^{\downarrow}\left(\Density^1_N(\bw)\right)$ equal?
\end{question}

\subsection{Beyond the permutation action}
After reading Section~\ref{sec:perm_inv_poly}, an informed reader will certainly have had a propensity to attempt to generalize the discussion in a number of ways:

\begin{enumerate}[label=\Roman*)]
\item Change the action of $\Sym{d}$
\item Consider a proper subgroup $G\subseteq\Sym{d}$, see \cite{mirsky_results_1962,onn_geometry_1993,cruickshank_rearrangement_2006}
\item Consider another group entirely (say, a finite Coxeter group), see \cite{ardila_coxeter_2020}
\item Consider a combination of I), II), and III)
\item Consider an infinite group, see \cite{sanyal_orbitopes_2011}.
\end{enumerate}

The above extensions were considered for orbit polytopes as the non-exhaustive references indicate.
The consideration of these possibilities for $\Sym{d}$-invariant polytopes certainly seems to be fertile ground for further experimentations and investigation.
However, these investigations go beyond the scope of the present work.
Nevertheless, let us perhaps emphasize two key aspects of the $\Sym{d}$-invariant polytopes which are critical in the present circumstance and should maintain a certain importance more generally.

\begin{itemize}
\item The \emph{existence} and \emph{simple description} of the fundamental chamber for the action of the symmetric group
\item The \emph{simple description} of fundamental representatives in Definition~\ref{def:orb_repr} (which are exactly the vertices maximized by generic linear functional in the fundamental chamber).
\end{itemize}
Both of these facts were crucial and extend to the case of finite Coxeter groups.
In Section~\ref{ssec:hrepr_spec}, we describe a method to translate a $V$-representation into an $H$-representation for certain $\Sym{d}$-invariant polytopes which relies heavily on the above two facts.
The method presented can be extended to a more general setting, as long as one has a fundamental cone and a description of the fundamental representatives maximized by fundamental linear functionals.

The survey \cite{schuermann_exploiting_2013} offers a fairly recent overview of the methods developed to exploit symmetry in polyhedral computations.
We refer the reader to this document and the references therein for further details on some existing techniques.
There are, however, two observations made by the author that deserve to be highlighted here.
First he observes that ``all methods known so far do not use geometric insights and still rely on subproblem conversions that do not exploit available symmetry'' (p. 269).
Furthermore, he concludes that ``for polyhedral representation conversions we see potential in enhancing decomposition methods through the use of geometric information like fundamental domains, classical invariant theory and symmetric polyhedral decompositions'' (p.~275).
Theorem~\ref{thm:dimension} and Algorithm~\ref{algo:lineup} confirm these anticipated algorithmic improvements by exploiting fundamental domains (from Definition~\ref{def:fund_cham}) and normal cones (from Equation~\eqref{eq:ncone_vertex}) to obtain fundamental cones (in Equation~\eqref{eq:keystep}) which are then used within a variation of the so-called ``incidence  decomposition  method''.

\subsection{Fermionic vs bosonic threshold ideals}
\label{ssec:fer_vs_bos}

Order ideals of $\FerPoset$ are called \emph{shifted complexes} \cite{klivans_threshold_2007,klivans_shifted_2008}.
Lineups of the hypersimplex lead to a slightly more specialized type of simplicial complex.
That is, a lineup of length $r$ gives rise to a shifted complex containing $r$ facets \emph{with} a total order provided by a fundamental linear functional $\by$.
Such shifted complexes are called \emph{threshold complexes} \cite{klivans_threshold_2007,edelman_simplicial_2013}.
We determined in Theorem~\ref{thm:vrepr_lineup} that the generating vertices of $\FerPoly[r]$ correspond to (shellings of) pure threshold complexes of dimension $N-1$ with exactly~$r$ facets and at most $d$ vertices.
The following example shows that the vertices of the two types of spectral polytopes are structurally different.

\begin{example}\label{ex:contrast}
	\hfill\\
\underline{\smash{Fermionic threshold ideal which is not bosonic threshold:}}\hfill\\
Let $d=5$ and $N=3$.
The bosonic ideal $Q^\mathrm{b}=\langle135,234\rangle$ is not a bosonic threshold ideal, i.e. $Q^\mathrm{b}\not\in\Thres{\BosPoset[3][5]}$.
There are thirty-five $3$-multisets of $[5]$ of which $17$ are in the ideal $Q^\mathrm{b}$.
The convex hull ${\conv\{\bchi(S)~:~S\in Q^\mathrm{b}\}}$ intersect the convex hull $\conv\left\{\bchi(S)~:~S\in\multiset{[5]}{3}\setminus Q^\mathrm{b}\right\}$ at the point $\frac{1}{3}(1, 2, 3, 2, 1)$.
This shows that no (strictly) separating hyperplane for $Q^\mathrm{b}$ exists, hence it can not be a bosonic threshold complex.
However, doing the same computation for the fermionic ideal $Q^\mathrm{f}=\langle147,246\rangle$ corresponding to $Q^\mathrm{b}$, the two convex hulls do not intersect, certifying the existence of a separating hyperplane making it a fermionic threshold ideal, i.e. $Q^\mathrm{f}\in\Thres{\FerPoset[3][7]}$.

\noindent
\underline{\smash{Bosonic threshold ideal which is not fermionic threshold:}}\hfill\\
Let $d=9$ and $N=3$.
The fermionic ideal $R^\mathrm{f}=\langle159,178,239,456\rangle$ is not a fermionic threshold ideal, i.e. $R^\mathrm{f}\not\in\Thres{\FerPoset[3][9]}$.
There are eighty-four $3$-subsets of which thirty-eight are in $R^\mathrm{f}$ (the smallest non-fermionic threshold has thirty-six $3$-subsets, see Example~\ref{ex:notferthres}).
Repeating the convex hull computation shows that $\frac{1}{3}\bef_9$ is in the intersection and therefore $R^\mathrm{f}$ is not a fermionic threshold ideal.
However, doing the same computation for the bosonic ideal $R^\mathrm{b}=\langle147,166,227,444\rangle$ corresponding to $R^\mathrm{f}$, the two convex hulls do not intersect certifying the existence of a separating hyperplane making it a bosonic threshold ideal, i.e. $R^\mathrm{b}\in\Thres{\BosPoset[3][7]}$.  \end{example}

\begin{example}[Saturated chains not giving a bosonic lineup]\label{ex:pseudosweep}
Consider the sequence of elements $(111,112,122,113,123)\in\BosPoset[3][3]$.
In this case it is easy to see that all the order ideals formed by the initial segments of the sequence are all threshold.
We claim that the sequence
\[
\ell= ((3,0,0),(2,1,0),(1,2,0),(2,0,1),(1,1,1)),
\]
is not a lineup.
For the sake of contradiction assume that there exists a vector $\by$ inducing $\ell$.
We must have
\[
\resizebox{\textwidth}{!}{$
\begin{array}{l@{\hspace{0.5cm}}c@{\hspace{0.5cm}}r}
	\np{\by,(1,2,0)}>	\np{\by,(2,0,1)},\text{ and } \np{\by,(1,1,1)}>	\np{\by,(0,3,0)} & \Longleftrightarrow & 2y_2 >y_1-y_2,\text{ and } y_1+y_3>2y_2,
\end{array}$}
\]
which is nonsensical. 
The line segments $[(1,2,0),(2,0,1)]$ and $[(0,3,0),(1,1,1)]$ are parallel to each other, see Figure \ref{fig:bosonic_non_coherent}. 
So any linear functional with a larger value on $(1,2,0)$ than on $(2,0,1)$ must have a larger value on $(0,3,0)$ than on $(1,1,1)$. 
However, the sequence $\ell$ is a pseudo-lineup as it is possible to rotate the sweeping line and obtain $\ell$.
\end{example}

\begin{figure}[H]
\begin{tikzpicture}%
[x={(-30:1)},
y={(0,1)},
z={(-1,-1)},
scale=0.85,
point/.style={inner sep=2pt,circle,draw=black,fill=white!50!red},
rotate=120]

\coordinate (003) at (0, 0, 3) {};
\coordinate (012) at (0, 1, 2) {};
\coordinate (021) at (0, 2, 1) {};
\coordinate (030) at (0, 3, 0) {};
\coordinate (102) at (1, 0, 2) {};
\coordinate (111) at (1, 1, 1) {};
\coordinate (120) at (1, 2, 0) {};
\coordinate (201) at (2, 0, 1) {};
\coordinate (210) at (2, 1, 0) {};
\coordinate (300) at (3, 0, 0) {};

\node[point,label=above:{$(0, 0, 3)$}] (n003) at (003) {};
\node[point,label=above:{$(0, 1, 2)$}] (n012) at (012) {};
\node[point,label=above:{$(0, 2, 1)$}] (n021) at (021) {};
\node[point,label=above:{$(0, 3, 0)$}] (n030) at (030) {};
\node[point,label=above:{$(1, 0, 2)$}] (n102) at (102) {};
\node[point,label=above:{$(1, 1, 1)$}] (n111) at (111) {};
\node[point,label=above:{$(1, 2, 0)$}] (n120) at (120) {};
\node[point,label=above:{$(2, 0, 1)$}] (n201) at (201) {};
\node[point,label=above:{$(2, 1, 0)$}] (n210) at (210) {};
\node[point,label=above:{$(3, 0, 0)$}] (n300) at (300) {};

\draw[dotted] (n201) -- (n120);
\draw[dotted] (n111) -- (n030);

\end{tikzpicture}
\caption{An example of an obstruction to coherence of bosonic lineup.}
\label{fig:bosonic_non_coherent}
\end{figure}
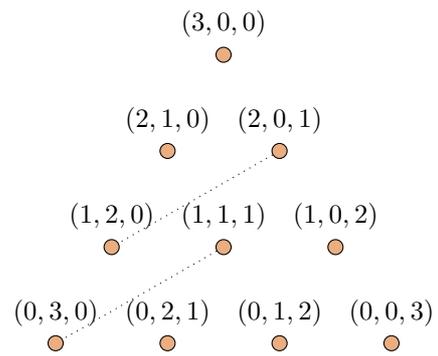

The notion of bosonic threshold ideal leads to the notion of multiset threshold complex.
What properties do these complexes have? 
Are they related to other combinatorial structures?
We leave these questions open.

\newpage

\addtocontents{toc}{\vspace{1em}}
\section*{Table of notations}
\label{notations}
\begin{longtable}{l@{\hspace{1cm}}l}
$[d]$ & the set $\{1,2,\dots,d\}$, for $d\in \NN\setminus\!\{0\}$\\
$|S|$ & the cardinality of a set $S$ \\
$\Sym{d}$ & the symmetric group on~$d$ objects\\
$\cH,\cH_1,\cH_N$  & finite-dimensional Hilbert spaces (see p.\pageref{ssec:hilbert}, p.\pageref{ssec:manyparticle}, and p.\pageref{ssec:manyparticle})\\
$\Herm(\cH)$ & the space of Hermitian operators on $\cH$ (see p.\pageref{ssec:hilbert})\\
$\spec^{\downarrow}(H)$ & the vector of decreasing eigenvalues of $H\in\Herm(\cH)$ (see p.\pageref{ssec:hilbert})\\
$\spec(H)$ & the set of vectors $\Sym{d}\cdot\spec^{\downarrow}(H)$, for an operator $H\in\Herm(\cH)$ (see p.\pageref{ssec:hilbert}) \\ 
$\Density(\cH)$ & the set of density operators on $\cH$ (see p.\pageref{def:density})\\
$\FerPoset$ & $\{(i_1,\cdots,i_N)\in [d]^N~:~1\leq i_1<\dots< i_N\leq d\}$ (see p.\pageref{ssec:manyparticle})\\
$\bi$ & an element of $\FerPoset$ or $\BosPoset$\\
$D$ & the dimension of $\cH_N$, fermionic case: $D=\binom{d}{N}$, bosonic case: $D=\binom{d+N-1}{N}$\\
$\Density^1$, $\Density^N$ & the set of density operators on $\cH_1$ and $\cH_N$ (see p.\pageref{ssec:manyparticle})\\
$\Unitary(\cH)$ & the unitary group of $\cH$ (see p.\pageref{eq:pauli_simplex})\\
$\Pauli_{D-1}$ & $\left\{\bw\in\RR^D:1\geq w_1\geq w_2\geq \cdots\geq w_D \geq 0,\quad \sum_{i=1}^D w_i=1\right\}$ (see p.\pageref{eq:pauli_simplex})\\
$\Density(\bw)$ & $\{\rho\in\Density(\cH)~:~\spec^\downarrow(\rho)=\bw\}$ (see p.\pageref{eq:pauli_simplex})\\
$L^N_M$ & the partial trace operator from $\Density^N$ to $\Density^M$ (see p.\pageref{def:partial_trace})\\
$\Density^M_N$ & $L^N_M(\Density^N)$, the $M$-reduced density operators on $\cH_M$ (see p.\pageref{rem:partial trace})\\
$\pol[H](N,d)$ & $\left\{\bx\in\RR^d~:~\sum_{i=1}^d x_i = N,\quad 0\leq x_i\leq 1\text{ for all }i=1,\dots,d\right\}$ (see p.\pageref{eq:def_hypersimplex})\\
$\klypo$ & the polytope $\spec^{\downarrow}(\Density^1_N(\bw))$ (see p.\pageref{eq:rosette_klypo})\\
$\Lambda(\bw,N,d)$, & $\spec(\Density^1_N(\bw))$, the symmetrization of $\klypo$ (see p.\pageref{eq:rosette_klypo})\\
$\overline{\Density}{}^1_N(d,\bw)$ & $\conv\left\{L^N_1(\tau)~:~\tau\in\Density^N(\bw)\right\}$ (see p.\pageref{pb:relaxedNrepr})\\
$\bw'\prec\bw$ & $\sum_{i=1}^kw'_i \leq \sum_{i=1}^kw_i$ (see p.\pageref{fig:relaxed_ensemble})\\
$\Gamma^1_N(h)$ & expansion of the operator $h\in \Herm(\cH_1)$ to $\Herm(\cH_N)$ (see p.\pageref{eq:lift})\\
$\OccVec_{\bw}(\ell(\bm{\lambda}))$ & occupation vector associated to $\bm{\lambda} \in \RR^d$ (see p.\pageref{eq:NON})\\
$\Lineups(\bw)$ & $\{\OccVec_{\bw}(\ell(\bm{\lambda}))~:~\bm{\lambda}\in\RR^d\}$, the set of occupation vectors (see p.\pageref{eq:NON})\\
$\FerPoly[][\bw,N,d]$ & fermionic spectral polytope (see p.\pageref{def:ferm_specpo})\\
$\BosPoset$ & $\{(i_1,\dots,i_N)\in [d]^N~:~1\leq i_1\leq\dots\leq i_N\leq d\}$ (see p.\pageref{ssec:repr_boson})\\
$\SymP^N\cH_1$ & the $N$-boson Hilbert space(see p.\pageref{ssec:repr_boson})\\
$\BosPoly[][\bw,N,d]$ & bosonic spectral polytope (see p.\pageref{def:bos_specpo})\\
$\binom{J}{k},\multiset{J}{k}$ & the collection of $k$-elements subsets and multisubsets of the set $J$\\
$\bchi$ & multiplicity function $S\in \NN^d \mapsto \sum_{j\in [d]}S(j)\be_j \in \RR^d$ (see p.\pageref{def:FBPC})\\
$\FerPC$ & $\left\{\bchi(S)~:~S\in \binom{[d]}{N} \right\}$ (see p.\pageref{def:FBPC})\\
$\BosPC$ & $\left\{\bchi(S)~:~S\in \multiset{[d]}{N} \right\}$ (see p.\pageref{def:FBPC})\\
$\mathrm{supp}_{\pol}$ & support function of polytope $\pol$ (see p.\pageref{sec:nfans})\\
$\pol^{\by}$ & face of $\pol$ maximized by the functional $\np{\by, \cdot}$ (see p.\pageref{sec:nfans})\\
$\ncone_{\pol}(\bv)$ & normal cone of $\pol$ at vertex $\bv$ (see p.\pageref{sec:nfans})\\
$\fan[N](\pol)$ & normal fan of $\pol$ (see p.\pageref{sec:nfans})\\
$\bef_k$ & $\sum_{1\leq i\leq k} \be_i$ (see p.\pageref{sec:perm_inv_poly})\\
$\Perm(\bV)$ & $\Sym{d}$-invariant polytope of the point configuration $\bV$ (see p.\pageref{def:invariant_pol})\\
$\Phi_d$ & $\{\by\in\RR^d~:~y_{1}\geq y_{2}\geq\cdots\geq y_{d}\}$, the fundamental chamber (see p.\pageref{def:fund_cham})\\
$\bx^\downarrow$ & fundamental representative of a vector $\bx$ (see p.\pageref{def:orb_repr})\\
$\ell_{\bV,r}(\by)$ & $r$-lineup or $r$-ranking of the point configuration $\bV$ induced by $\by$ (see p.\pageref{def:lineup})\\
$\lu$ & the set of $r$-lineups of point configuration $\bV$ (see p.\pageref{def:lineup})\\
$\pol[K]_\bV^\circ(\ell)$ & $\{\by\in\RR^n: \ell(\by)=\ell\}$, for $\ell$ an $r$-ranking of $\bV$ (see p.\pageref{eq:complement})\\
$\Rg$ & $\{\pol[K]_\bV(\ell)\ : \ \ell \text{ is an $r$-ranking of }\bV\}$ (see p.\pageref{eq:complement})\\
$\pol[L]_{r,\bw}(\bV)$ & the $r$-lineup polytope of~$\bV$ (see p.\pageref{def:lineup_pol})\\
$\non$ & $\{\OccVec_{\bw}(\ell)~:~\ell\in\lu[\FerPC]\text{ and } \ell=\ell(\by)\text{ for some } \by\in\Phi^d\}$ (see p.\pageref{prop:rhs})\\
$\non[b]$ & $\{\OccVec_{\bw}(\ell)~:~\ell\in\lu[\BosPC]\text{ and } \ell=\ell(\by)\text{ for some } \by\in\Phi^d\}$ (see p.\pageref{prop:rhs})\\
$\Thres{P}$  & poset of threshold ideals of $P\in\{\FerPoset,\BosPoset\}$ (see p.\pageref{def:ferm_threshold} and p.\pageref{lem:Bthreshold})\\
\end{longtable}

\newpage
\newcommand{\etalchar}[1]{$^{#1}$}
\providecommand{\bysame}{\leavevmode\hbox to3em{\hrulefill}\thinspace}
\providecommand{\MR}{\relax\ifhmode\unskip\space\fi MR }
\providecommand{\MRhref}[2]{%
  \href{http://www.ams.org/mathscinet-getitem?mr=#1}{#2}
}
\providecommand{\href}[2]{#2}

\newpage
\appendix

\section{$H$-representation of the lineup polytope of the hypersimplex $\pol[H](3,6)$}
\label{app:hyper36}

The $H$-representation of the $20$-lineup polytope of the $5$-dimensional hypersimplex $\pol[H](3,6)$ has $72$ inequalities represented below.
The inequalities arise in a hierarchy while increasing $r$ from $1$ to $10$.
The weights are taken in the Pauli simplex $\Pauli_{9}$, so that $1\geq w_1 \geq w_2 \geq \cdots \geq w_{10} \geq 0$ and $\sum_{i=1}^{10}w_i=1$.
Since $\sum_{i=1}^6x_i=3$ and $\sum_{i=1}^{10}w_i=1$, there are many ways to represent the matrices if one does not use the orthogonal complement of these hyperplanes.
We chose the representation where the last coefficient of the rays is zero, as a consequence, all coefficients are non-negative.

\[
\resizebox{\textwidth}{!}{$
\begin{blockarray}{rrrrrrrcrrrrrrrrrrl}
\begin{block}{rrrrrrrcrrrrrrrrrrl}
& \BAmulticolumn{6}{c}{\text{Ray coefficients}} & & \BAmulticolumn{10}{c}{\text{Right-hand side}} \\
\end{block}
\begin{block}{r(rrrrrr)c(rrrrrrrrrr)l}
\multirow{2}{*}{$r=1$}   & 1 & 0 & 0 & 0 & 0 & 0 & \multirow{44}{*}{$\times\left(\begin{array}{c} x_{1}^{\downarrow} \\ x_{2}^{\downarrow}  \\ x_{3}^{\downarrow}  \\ x_{4}^{\downarrow}  \\ x_{5}^{\downarrow}  \\ x_{6}^{\downarrow} \end{array}\right)\leq$} & 1 & 1 & 1 & 1 & 1 & 1 & 1 & 1 & 1 & 1 & \multirow{44}{*}{$\times\left(\begin{array}{c} w_{1} \\ w_{2} \\ w_{3} \\ w_{4} \\ w_{5} \\ w_{6} \\ w_{7} \\ w_{8} \\ w_{9} \\ w_{10} \end{array}\right)$}\\
                         & 1 & 1 & 1 & 1 & 1 & 0 &  & 3 & 3 & 3 & 3 & 3 & 3 & 3 & 3 & 3 & 3 & \\\cline{2-7}\cline{9-18}
                   r=2   & 1 & 1 & 1 & 0 & 0 & 0 & & 3 & 2 & 2 & 2 & 2 & 2 & 2 & 2 & 2 & 2 \\\cline{2-7}\cline{9-18}
                   r=3   & 2 & 2 & 1 & 1 & 0 & 0 & & 5 & 5 & 4 & 4 & 4 & 4 & 3 & 3 & 3 & 3 \\\cline{2-7}\cline{9-18}
\multirow{2}{*}{$r=4$}   & 2 & 1 & 1 & 1 & 0 & 0 & & 4 & 4 & 4 & 3 & 3 & 3 & 3 & 3 & 3 & 3 \\
                         & 2 & 2 & 1 & 1 & 1 & 0 & & 5 & 5 & 5 & 4 & 4 & 4 & 4 & 4 & 4 & 4 \\\cline{2-7}\cline{9-18}
\multirow{3}{*}{$r=5$}   & 1 & 1 & 0 & 0 & 0 & 0 & & 2 & 2 & 2 & 2 & 1 & 1 & 1 & 1 & 1 & 1 \\
                         & 1 & 1 & 1 & 1 & 0 & 0 & & 3 & 3 & 3 & 3 & 2 & 2 & 2 & 2 & 2 & 2 \\
                         & 4 & 3 & 2 & 2 & 1 & 0 & & 9 & 9 & 8 & 8 & 7 & 7 & 7 & 7 & 6 & 6 \\\cline{2-7}\cline{9-18}
\multirow{3}{*}{$r=6$}   & 3 & 2 & 1 & 1 & 0 & 0 & & 6 & 6 & 5 & 5 & 5 & 4 & 4 & 4 & 4 & 4 \\
                         & 3 & 2 & 2 & 1 & 1 & 0 & & 7 & 6 & 6 & 6 & 6 & 5 & 5 & 5 & 5 & 5 \\
                         & 3 & 3 & 2 & 2 & 1 & 0 & & 8 & 8 & 7 & 7 & 7 & 6 & 6 & 6 & 6 & 6 \\\cline{2-7}\cline{9-18}
\multirow{5}{*}{$r=7$}   & 2 & 1 & 1 & 1 & 1 & 0 & & 4 & 4 & 4 & 4 & 4 & 4 & 3 & 3 & 3 & 3 \\
                         & 4 & 3 & 2 & 1 & 1 & 0 & & 9 & 8 & 8 & 7 & 7 & 7 & 6 & 6 & 6 & 6 \\
                         & 4 & 3 & 3 & 2 & 1 & 0 & & 10 & 9 & 9 & 8 & 8 & 8 & 7 & 7 & 7 & 7 \\
                         & 6 & 4 & 3 & 2 & 1 & 0 & & 13 & 12 & 11 & 11 & 10 & 10 & 9 & 9 & 9 & 8 \\
                         & 6 & 5 & 4 & 3 & 2 & 0 & & 15 & 14 & 13 & 13 & 12 & 12 & 11 & 11 & 11 & 10 \\\cline{2-7}\cline{9-18}
\multirow{9}{*}{$r=8$}   & 2 & 1 & 1 & 0 & 0 & 0 & & 4 & 3 & 3 & 3 & 3 & 3 & 3 & 2 & 2 & 2 \\
                         & 2 & 2 & 2 & 1 & 1 & 0 & & 6 & 5 & 5 & 5 & 5 & 5 & 5 & 4 & 4 & 4 \\
                         & 3 & 2 & 1 & 1 & 1 & 0 & & 6 & 6 & 6 & 5 & 5 & 5 & 5 & 4 & 4 & 4 \\
                         & 3 & 2 & 2 & 2 & 1 & 0 & & 7 & 7 & 7 & 6 & 6 & 6 & 6 & 5 & 5 & 5 \\
                         & 5 & 3 & 2 & 2 & 1 & 0 & & 10 & 10 & 9 & 9 & 8 & 8 & 8 & 7 & 7 & 7 \\
                         & 5 & 4 & 3 & 2 & 1 & 0 & & 12 & 11 & 10 & 10 & 9 & 9 & 9 & 8 & 8 & 8 \\
                         & 5 & 4 & 3 & 3 & 2 & 0 & & 12 & 12 & 11 & 11 & 10 & 10 & 10 & 9 & 9 & 9 \\
                         & 6 & 5 & 3 & 2 & 1 & 0 & & 14 & 13 & 12 & 11 & 11 & 10 & 10 & 9 & 9 & 9 \\
                         & 6 & 5 & 4 & 3 & 1 & 0 & & 15 & 14 & 13 & 12 & 12 & 11 & 11 & 10 & 10 & 10 \\\cline{2-7}\cline{9-18}
\multirow{14}{*}{$r=9$}  & 3 & 2 & 2 & 1 & 0 & 0 & & 7 & 6 & 6 & 5 & 5 & 5 & 5 & 5 & 4 & 4 \\
                         & 3 & 3 & 2 & 1 & 1 & 0 & & 8 & 7 & 7 & 6 & 6 & 6 & 6 & 6 & 5 & 5 \\
                         & 4 & 2 & 2 & 1 & 1 & 0 & & 8 & 7 & 7 & 7 & 7 & 6 & 6 & 6 & 5 & 5 \\
                         & 4 & 3 & 2 & 1 & 0 & 0 & & 9 & 8 & 7 & 7 & 7 & 6 & 6 & 6 & 5 & 5 \\
                         & 4 & 3 & 3 & 2 & 2 & 0 & & 10 & 9 & 9 & 9 & 9 & 8 & 8 & 8 & 7 & 7 \\
                         & 4 & 4 & 3 & 2 & 1 & 0 & & 11 & 10 & 9 & 9 & 9 & 8 & 8 & 8 & 7 & 7 \\
                         & 5 & 3 & 2 & 1 & 1 & 0 & & 10 & 9 & 9 & 8 & 8 & 8 & 7 & 7 & 6 & 6 \\
                         & 5 & 4 & 2 & 2 & 1 & 0 & & 11 & 11 & 10 & 9 & 9 & 8 & 8 & 8 & 7 & 7 \\
                         & 5 & 4 & 3 & 3 & 1 & 0 & & 12 & 12 & 11 & 10 & 10 & 9 & 9 & 9 & 8 & 8 \\
                         & 5 & 4 & 4 & 3 & 2 & 0 & & 13 & 12 & 12 & 11 & 11 & 11 & 10 & 10 & 9 & 9 \\
                         & 7 & 4 & 3 & 2 & 1 & 0 & & 14 & 13 & 12 & 12 & 11 & 11 & 10 & 10 & 9 & 9 \\
                         & 7 & 5 & 4 & 2 & 1 & 0 & & 16 & 14 & 13 & 13 & 12 & 12 & 11 & 11 & 10 & 10 \\
                         & 7 & 6 & 5 & 3 & 2 & 0 & & 18 & 16 & 15 & 15 & 14 & 14 & 13 & 13 & 12 & 12 \\
                         & 7 & 6 & 5 & 4 & 3 & 0 & & 18 & 17 & 16 & 16 & 15 & 15 & 14 & 14 & 13 & 13 \\\cline{2-7}\cline{9-18}
\multirow{4}{*}{$r=10$} & 3 & 1 & 1 & 1 & 0 & 0 & & 5 & 5 & 5 & 4 & 4 & 4 & 4 & 4 & 4 & 3 & \\
                         & 3 & 3 & 2 & 2 & 2 & 0 & & 8 & 8 & 8 & 7 & 7 & 7 & 7 & 7 & 7 & 6 \\
                         & 4 & 2 & 1 & 1 & 0 & 0 & & 7 & 7 & 6 & 6 & 6 & 5 & 5 & 5 & 5 & 4 \\
                         & 4 & 3 & 3 & 1 & 1 & 0 & & 10 & 8 & 8 & 8 & 8 & 7 & 7 & 7 & 7 & 6 \\
\end{block}
\begin{block}{rrrrrrrcrrrrrrrrrrl}
& \BAmulticolumn{6}{c}{\vdots} & & \BAmulticolumn{10}{c}{\vdots} \\
& \BAmulticolumn{17}{c}{\text{continued on page~\pageref{H_h36}}} \\
\end{block}
\end{blockarray}$}\]
\[
\resizebox{\textwidth}{!}{$
\label{H_h36}
\begin{blockarray}{rrrrrrrcrrrrrrrrrrl}
\begin{block}{rrrrrrrcrrrrrrrrrrl}
& \BAmulticolumn{6}{c}{\vdots} & & \BAmulticolumn{10}{c}{\vdots} \\
\end{block}
\begin{block}{r(rrrrrr)c(rrrrrrrrrr)l}
\multirow{28}{*}{$r=10$}  & 4 & 4 & 3 & 3 & 2 & 0 & \multirow{28}{*}{$\times\left(\begin{array}{c} x_{1}^{\downarrow} \\ x_{2}^{\downarrow} \\ x_{3}^{\downarrow} \\ x_{4}^{\downarrow} \\ x_{5}^{\downarrow} \\ x_{6}^{\downarrow} \end{array}\right)\leq$} & 11 & 11 & 10 & 10 & 10 & 9 & 9 & 9 & 9 & 8 & \multirow{28}{*}{$\times\left(\begin{array}{c} w_{1} \\ w_{2} \\ w_{3} \\ w_{4} \\ w_{5} \\ w_{6} \\ w_{7} \\ w_{8} \\ w_{9} \\ w_{10} \end{array}\right)$}\\
                         & 5 & 3 & 3 & 2 & 1 & 0 & & 11 & 10 & 10 & 9 & 9 & 8 & 8 & 8 & 8 & 7 \\
                         & 5 & 4 & 3 & 1 & 1 & 0 & & 12 & 10 & 10 & 9 & 9 & 9 & 8 & 8 & 8 & 7 \\
                         & 5 & 4 & 3 & 2 & 2 & 0 & & 12 & 11 & 11 & 10 & 10 & 9 & 9 & 9 & 9 & 8 \\
                         & 5 & 4 & 4 & 2 & 1 & 0 & & 13 & 11 & 11 & 10 & 10 & 10 & 9 & 9 & 9 & 8 \\
                         & 6 & 3 & 2 & 2 & 1 & 0 & & 11 & 11 & 10 & 10 & 9 & 9 & 9 & 8 & 8 & 7 \\
                         & 6 & 4 & 4 & 3 & 1 & 0 & & 14 & 13 & 13 & 11 & 11 & 11 & 10 & 10 & 10 & 9 \\
                         & 6 & 5 & 3 & 2 & 2 & 0 & & 14 & 13 & 13 & 11 & 11 & 11 & 10 & 10 & 10 & 9 \\
                         & 6 & 5 & 4 & 2 & 1 & 0 & & 15 & 13 & 12 & 12 & 11 & 11 & 11 & 10 & 10 & 9 \\
                         & 6 & 5 & 4 & 4 & 3 & 0 & & 15 & 15 & 14 & 14 & 13 & 13 & 13 & 12 & 12 & 11 \\
                         & 7 & 5 & 3 & 2 & 1 & 0 & & 15 & 14 & 13 & 12 & 12 & 11 & 10 & 10 & 10 & 9 \\
                         & 7 & 5 & 4 & 3 & 1 & 0 & & 16 & 15 & 14 & 13 & 12 & 12 & 12 & 11 & 11 & 10 \\
                         & 7 & 6 & 4 & 2 & 1 & 0 & & 17 & 15 & 14 & 13 & 13 & 12 & 12 & 11 & 11 & 10 \\
                         & 7 & 6 & 4 & 3 & 2 & 0 & & 17 & 16 & 15 & 14 & 13 & 13 & 13 & 12 & 12 & 11 \\
                         & 7 & 6 & 5 & 3 & 1 & 0 & & 18 & 16 & 15 & 14 & 14 & 13 & 13 & 12 & 12 & 11 \\
                         & 7 & 6 & 5 & 4 & 2 & 0 & & 18 & 17 & 16 & 15 & 15 & 14 & 13 & 13 & 13 & 12 \\
                         & 8 & 4 & 3 & 2 & 1 & 0 & & 15 & 14 & 13 & 13 & 12 & 12 & 11 & 11 & 10 & 9 \\
                         & 8 & 5 & 4 & 3 & 2 & 0 & & 17 & 16 & 15 & 15 & 14 & 13 & 13 & 12 & 12 & 11 \\
                         & 8 & 6 & 4 & 3 & 1 & 0 & & 18 & 17 & 15 & 15 & 14 & 13 & 13 & 12 & 12 & 11 \\
                         & 8 & 6 & 5 & 2 & 1 & 0 & & 19 & 16 & 15 & 15 & 14 & 14 & 13 & 13 & 12 & 11 \\
                         & 8 & 6 & 5 & 4 & 1 & 0 & & 19 & 18 & 17 & 15 & 15 & 14 & 14 & 13 & 13 & 12 \\
                         & 8 & 6 & 5 & 4 & 3 & 0 & & 19 & 18 & 17 & 17 & 16 & 15 & 15 & 14 & 14 & 13 \\
                         & 8 & 7 & 4 & 3 & 2 & 0 & & 19 & 18 & 17 & 15 & 15 & 14 & 14 & 13 & 13 & 12 \\
                         & 8 & 7 & 5 & 4 & 2 & 0 & & 20 & 19 & 17 & 17 & 16 & 15 & 15 & 14 & 14 & 13 \\
                         & 8 & 7 & 6 & 3 & 2 & 0 & & 21 & 18 & 17 & 17 & 16 & 16 & 15 & 15 & 14 & 13 \\
                         & 8 & 7 & 6 & 5 & 4 & 0 & & 21 & 20 & 19 & 19 & 18 & 18 & 17 & 17 & 16 & 15 \\
                         & 9 & 7 & 4 & 3 & 1 & 0 & & 20 & 19 & 17 & 16 & 16 & 14 & 14 & 13 & 13 & 12 \\
                         & 9 & 8 & 6 & 5 & 2 & 0 & & 23 & 22 & 20 & 19 & 19 & 17 & 17 & 16 & 16 & 15 \\
\end{block}
\end{blockarray}$}
\]

\section{Generalized exclusion inequalities for fermions}
\label{app:ferm}

The Pauli exclusion principle is equivalent to the first of the following two inequalities which describe the case $r=1$:
\[
x_1^{\downarrow} \leq 1,\hspace{3cm} \sum_{i=1}^{d-1} x_i^{\downarrow} \leq N.
\]
The second equation is equivalent to $x_d^{\downarrow}\geq 0$, which implies that all coordinates should be indeed non-negative, which is inherently true in the physical context.
To illustrate larger values, it is more compact to express them in a matrix.
Following Theorem~\ref{thm:stability}, after solving the case $(r,N,d)$, when increasing the value of $r$ by $1$, the minimal case $(r+1,N',d')$ to consider is such that $N'=N+1$ and $d'=d+2$.
Below, we represent this minimal case by the coefficients located between the two vertical bars.
The matrix gives the result for $r=8$, so in dimension $d=14$.
The right-hand side term involving $N$ is determined using Proposition~\ref{prop:rhs}.

\[
\begin{blockarray}{lrrrrrrrrrrrrrr@{\hspace{0.25cm}}|@{\hspace{0.25cm}}ll}
\begin{block}{l@{\hspace{0.2cm}}rrrrrrrrrrrrrr@{\hspace{0.25cm}}|@{\hspace{0.25cm}}ll}
& \BAmulticolumn{14}{c}{\text{Ray coefficients}} & \BAmulticolumn{2}{c}{\text{Right-hand side}} \\
\end{block}
\begin{block}{l@{\hspace{0.2cm}}(@{\hspace{0.2cm}}rrrrrrrrrrrrrr@{\hspace{0.25cm}}|@{\hspace{0.25cm}}ll)}
\multirow{2}{*}{$r=1$} & 1 & 0 & 0 & 0 & 0 & 0 & 0 & 0 & 0 & 0 & 0 & 0 & 0 & 0 & 1    & \\
 & 1 & 1 & 1 & 1 & 1 & 1 & 1 & 1 & 1 & 1 & 1 & 1 & 1 & 0 & N    & \\\cline{2-17}
\multirow{1}{*}{$r=2$} & 1 & 1 & 1 & 1 & 1 & \BAmulticolumn{1}{r|}{1} & 1 & \BAmulticolumn{1}{r|}{0} & 0 & 0 & 0 & 0 & 0 & 0 & N-1  & +w_1 \\\cline{2-17}
\multirow{1}{*}{$r=3$} & 2 & 2 & 2 & 2 & \BAmulticolumn{1}{r|}{2} & 2 & 1 & 1 & \BAmulticolumn{1}{r|}{0} & 0 & 0 & 0 & 0 & 0 & 2N-2 & + w_{1} + w_{2}\\\cline{2-17}
\multirow{2}{*}{$r=4$} & 2 & 2 & 2 & \BAmulticolumn{1}{r|}{2} & 2 & 1 & 1 & 1 & 0 & \BAmulticolumn{1}{r|}{0} & 0 & 0 & 0 & 0 & 2N-3 & + w_{1} + w_{2} + w_{3}\\
 & 2 & 2 & 2 & \BAmulticolumn{1}{r|}{2} & 2 & 2 & 1 & 1 & 1 & \BAmulticolumn{1}{r|}{0} & 0 & 0 & 0 & 0 & 2N-2 & + w_{1} + w_{2} + w_{3}\\\cline{2-17}
\multirow{3}{*}{$r=5$} & 2 & 2 & \BAmulticolumn{1}{r|}{2} & 2 & 1 & 1 & 1 & 1 & 0 & 0 & \BAmulticolumn{1}{r|}{0} & 0 & 0 & 0 & 2N-4 & + w_{1} + w_{2} + w_{3} + w_{4}\\
 & 2 & 2 & \BAmulticolumn{1}{r|}{2} & 2 & 2 & 2 & 1 & 1 & 1 & 1 & \BAmulticolumn{1}{r|}{0} & 0 & 0 & 0 & 2N-2 & + w_{1} + w_{2} + w_{3} + w_{4}\\
 & 4 & 4 & \BAmulticolumn{1}{r|}{4} & 4 & 4 & 3 & 2 & 2 & 1 & 0 & \BAmulticolumn{1}{r|}{0} & 0 & 0 & 0 & 4N-5 & + 2w_{1} + 2w_{2} + w_{3} + w_{4}\\
\end{block}
\end{blockarray}
\]

\[
\resizebox{\textwidth}{!}{$
\begin{blockarray}{lrrrrrrrrrrrrrr@{\hspace{0.25cm}}|@{\hspace{0.25cm}}ll}
\begin{block}{l@{\hspace{0.2cm}}rrrrrrrrrrrrrr@{\hspace{0.25cm}}|@{\hspace{0.25cm}}ll}
& \BAmulticolumn{14}{c}{\text{Ray coefficients}} & \BAmulticolumn{2}{c}{\text{Right-hand side}} \\
\end{block}
\begin{block}{l@{\hspace{0.2cm}}(@{\hspace{0.2cm}}rrrrrrrrrrrrrr@{\hspace{0.25cm}}|@{\hspace{0.25cm}}ll)}
\multirow{5}{*}{$r=6$} & 2 & \BAmulticolumn{1}{r|}{2} & 2 & 1 & 1 & 1 & 1 & 1 & 0 & 0 & 0 & \BAmulticolumn{1}{r|}{0} & 0 & 0 & 2N-5 & + w_{1} + w_{2} + w_{3} + w_{4} + w_{5}     \\
 & 2 & \BAmulticolumn{1}{r|}{2} & 2 & 2 & 2 & 2 & 1 & 1 & 1 & 1 & 1 & \BAmulticolumn{1}{r|}{0} & 0 & 0 & 2N-2 & + w_{1} + w_{2} + w_{3} + w_{4} + w_{5}    \\
 & 3 & \BAmulticolumn{1}{r|}{3} & 3 & 3 & 3 & 2 & 2 & 1 & 1 & 0 & 0 & \BAmulticolumn{1}{r|}{0} & 0 & 0 & 3N-4 & + 2 w_{1} + w_{2} + w_{3} + w_{4} + w_{5}  \\
 & 4 & \BAmulticolumn{1}{r|}{4} & 4 & 4 & 3 & 3 & 2 & 2 & 1 & 0 & 0 & \BAmulticolumn{1}{r|}{0} & 0 & 0 & 4N-6 & + 2 w_{1} + 2 w_{2} + w_{3} + w_{4} + w_{5}\\
 & 4 & \BAmulticolumn{1}{r|}{4} & 4 & 4 & 4 & 3 & 2 & 2 & 1 & 1 & 0 & \BAmulticolumn{1}{r|}{0} & 0 & 0 & 4N-5 & + 2 w_{1} + 2 w_{2} + w_{3} + w_{4} + w_{5}\\\cline{2-17}
\multirow{10}{*}{$r=7$}  & \BAmulticolumn{1}{r|}{2} & 2 & 1 & 1 & 1 & 1 & 1 & 1 & 0 & 0 & 0 & 0 & \BAmulticolumn{1}{r|}{0} & 0 & 2N-6 & + w_{1} + w_{2} + w_{3} + w_{4} + w_{5} + w_{6}         \\
 &\BAmulticolumn{1}{r|}{2} & 2 & 2 & 2 & 2 & 1 & 1 & 1 & 1 & 0 & 0 & 0 & \BAmulticolumn{1}{r|}{0} & 0 & 2N-3 & + w_{1} + w_{2} + w_{3} + w_{4} + w_{5} + w_{6}        \\
 & \BAmulticolumn{1}{r|}{2} & 2 & 2 & 2 & 2 & 2 & 1 & 1 & 1 & 1 & 1 & 1 & \BAmulticolumn{1}{r|}{0} & 0 & 2N-2 & + w_{1} + w_{2} + w_{3} + w_{4} + w_{5} + w_{6}        \\
 & \BAmulticolumn{1}{r|}{4} & 4 & 4 & 3 & 3 & 3 & 2 & 2 & 1 & 0 & 0 & 0 & \BAmulticolumn{1}{r|}{0} & 0 & 4N-7 & + 2 w_{1} + 2 w_{2} + w_{3} + w_{4} + w_{5} + w_{6}    \\
 & \BAmulticolumn{1}{r|}{4} & 4 & 4 & 4 & 3 & 3 & 2 & 2 & 1 & 1 & 0 & 0 & \BAmulticolumn{1}{r|}{0} & 0 & 4N-6 & + 2 w_{1} + 2 w_{2} + w_{3} + w_{4} + w_{5} + w_{6}    \\
 & \BAmulticolumn{1}{r|}{4} & 4 & 4 & 4 & 4 & 3 & 2 & 2 & 1 & 1 & 1 & 0 & \BAmulticolumn{1}{r|}{0} & 0 & 4N-5 & + 2 w_{1} + 2 w_{2} + w_{3} + w_{4} + w_{5} + w_{6}    \\
 & \BAmulticolumn{1}{r|}{5} & 5 & 5 & 5 & 4 & 3 & 3 & 2 & 1 & 0 & 0 & 0 & \BAmulticolumn{1}{r|}{0} & 0 & 5N-8 & + 3 w_{1} + 2 w_{2} + 2 w_{3} + w_{4} + w_{5} + w_{6}  \\
 & \BAmulticolumn{1}{r|}{5} & 5 & 5 & 5 & 5 & 4 & 3 & 2 & 2 & 1 & 0 & 0 & \BAmulticolumn{1}{r|}{0} & 0 & 5N-6 & + 3 w_{1} + 2 w_{2} + 2 w_{3} + w_{4} + w_{5} + w_{6}  \\
 & \BAmulticolumn{1}{r|}{7} & 7 & 7 & 7 & 6 & 5 & 4 & 3 & 2 & 0 & 0 & 0 & \BAmulticolumn{1}{r|}{0} & 0 & 7N-10 & + 4 w_{1} + 3 w_{2} + 2 w_{3} + 2 w_{4} + w_{5} + w_{6}\\
 & \BAmulticolumn{1}{r|}{7} & 7 & 7 & 7 & 7 & 5 & 4 & 3 & 2 & 1 & 0 & 0 & \BAmulticolumn{1}{r|}{0} & 0 & 7N-9 & + 4 w_{1} + 3 w_{2} + 2 w_{3} + 2 w_{4} + w_{5} + w_{6}\\\cline{2-17}
\multirow{19}{*}{$r=8$} & 2 & 1 & 1 & 1 & 1 & 1 & 1 & 1 & 0 & 0 & 0 & 0 & 0 & 0 & 2N-7    & + w_{1} + w_{2} + w_{3} + w_{4} + w_{5} + w_{6} + w_{7}            \\
 & 2 & 2 & 2 & 2 & 2 & 2 & 1 & 1 & 1 & 1 & 1 & 1 & 1 & 0 & 2N-2 & + w_{1} + w_{2} + w_{3} + w_{4} + w_{5} + w_{6} + w_{7}           \\
 & 3 & 3 & 3 & 3 & 2 & 2 & 2 & 1 & 1 & 0 & 0 & 0 & 0 & 0 & 3N-5 & + 2 w_{1} + w_{2} + w_{3} + w_{4} + w_{5} + w_{6} + w_{7}         \\
 & 3 & 3 & 3 & 3 & 3 & 2 & 2 & 1 & 1 & 1 & 0 & 0 & 0 & 0 & 3N-4 & + 2 w_{1} + w_{2} + w_{3} + w_{4} + w_{5} + w_{6} + w_{7}         \\
 & 4 & 4 & 3 & 3 & 3 & 3 & 2 & 2 & 1 & 0 & 0 & 0 & 0 & 0 & 4N-8 & + 2 w_{1} + 2 w_{2} + w_{3} + w_{4} + w_{5} + w_{6} + w_{7}       \\
 & 4 & 4 & 4 & 3 & 3 & 3 & 2 & 2 & 1 & 1 & 0 & 0 & 0 & 0 & 4N-7 & + 2 w_{1} + 2 w_{2} + w_{3} + w_{4} + w_{5} + w_{6} + w_{7}       \\
 & 4 & 4 & 4 & 4 & 3 & 2 & 2 & 2 & 1 & 0 & 0 & 0 & 0 & 0 & 4N-7 & + 2 w_{1} + 2 w_{2} + 2 w_{3} + w_{4} + w_{5} + w_{6} + w_{7}     \\
 & 4 & 4 & 4 & 4 & 3 & 3 & 2 & 2 & 1 & 1 & 1 & 0 & 0 & 0 & 4N-6 & + 2 w_{1} + 2 w_{2} + w_{3} + w_{4} + w_{5} + w_{6} + w_{7}       \\
 & 4 & 4 & 4 & 4 & 4 & 3 & 2 & 2 & 1 & 1 & 1 & 1 & 0 & 0 & 4N-5 & + 2 w_{1} + 2 w_{2} + w_{3} + w_{4} + w_{5} + w_{6} + w_{7}       \\
 & 4 & 4 & 4 & 4 & 4 & 3 & 2 & 2 & 2 & 1 & 0 & 0 & 0 & 0 & 4N-5 & + 2 w_{1} + 2 w_{2} + 2 w_{3} + w_{4} + w_{5} + w_{6} + w_{7}     \\
 & 5 & 5 & 5 & 4 & 4 & 3 & 3 & 2 & 1 & 0 & 0 & 0 & 0 & 0 & 5N-9 & + 3 w_{1} + 2 w_{2} + 2 w_{3} + w_{4} + w_{5} + w_{6} + w_{7}     \\
 & 5 & 5 & 5 & 5 & 5 & 4 & 3 & 2 & 2 & 1 & 1 & 0 & 0 & 0 & 5N-6 & + 3 w_{1} + 2 w_{2} + 2 w_{3} + w_{4} + w_{5} + w_{6} + w_{7}     \\
 & 6 & 6 & 6 & 6 & 5 & 4 & 3 & 3 & 2 & 0 & 0 & 0 & 0 & 0 & 6N-9 & + 3 w_{1} + 3 w_{2} + 2 w_{3} + 2 w_{4} + w_{5} + w_{6} + w_{7}   \\
 & 6 & 6 & 6 & 6 & 6 & 4 & 3 & 3 & 2 & 1 & 0 & 0 & 0 & 0 & 6N-8 & + 3 w_{1} + 3 w_{2} + 2 w_{3} + 2 w_{4} + w_{5} + w_{6} + w_{7}   \\
 & 7 & 7 & 7 & 6 & 6 & 5 & 4 & 3 & 2 & 0 & 0 & 0 & 0 & 0 & 7N-11 & + 4 w_{1} + 3 w_{2} + 2 w_{3} + 2 w_{4} + w_{5} + w_{6} + w_{7}   \\
 & 7 & 7 & 7 & 7 & 6 & 5 & 4 & 3 & 2 & 1 & 0 & 0 & 0 & 0 & 7N-10 & + 4 w_{1} + 3 w_{2} + 2 w_{3} + 2 w_{4} + w_{5} + w_{6} + w_{7}   \\
 & 7 & 7 & 7 & 7 & 7 & 5 & 4 & 3 & 2 & 1 & 1 & 0 & 0 & 0 & 7N-9 & + 4 w_{1} + 3 w_{2} + 2 w_{3} + 2 w_{4} + w_{5} + w_{6} + w_{7}   \\
 & 9 & 9 & 9 & 9 & 7 & 6 & 5 & 4 & 2 & 1 & 0 & 0 & 0 & 0 & 9N-14& + 5 w_{1} + 4 w_{2} + 3 w_{3} + 2 w_{4} + 2 w_{5} + w_{6} + w_{7} \\
 & 9 & 9 & 9 & 9 & 8 & 7 & 5 & 4 & 3 & 2 & 0 & 0 & 0 & 0 & 9N-12 & + 5 w_{1} + 4 w_{2} + 3 w_{3} + 2 w_{4} + 2 w_{5} + w_{6} + w_{7} \\
\end{block}
\end{blockarray}$}
\]

\section{Generalized exclusion inequalities for bosons}
\label{app:bos}

The case $r=1$ has one inequality
\[
\sum_{i=1}^{d-1} x_i^{\downarrow} \leq N.
\]
It is equivalent to $x_d^{\downarrow}\geq 0$, which implies that all coordinates should be indeed non-negative, which is inherently true in the physical context.
To illustrate larger values, we again express the new inequalities in a matrix.
Using Theorem~\ref{thm:stability_bosonic}, after solving the case $(r,N,d)$, when increasing the value of $r$ by $1$, the minimal case $(r+1,N',d')$ to consider is such that $N'=N+1$ and $d'=d+1$.
Below, we represent this minimal case by the coefficients located on the left of the vertical bar.
The matrix gives the result for $r=8$, so in dimension $d=8$.
The right-hand side term involving $N$ is determined by adapting Proposition~\ref{prop:rhs}.

\[
\begin{blockarray}{lrrrrrrrr@{\hspace{0.25cm}}|@{\hspace{0.25cm}}ll}
\begin{block}{l@{\hspace{0.2cm}}rrrrrrrr@{\hspace{0.25cm}}|@{\hspace{0.25cm}}ll}
& \BAmulticolumn{8}{c}{\text{Ray coefficients}} & \BAmulticolumn{2}{c}{\text{Right-hand side}} \\
\end{block}
\begin{block}{l@{\hspace{0.2cm}}(@{\hspace{0.2cm}}rrrrrrrr@{\hspace{0.25cm}}|@{\hspace{0.25cm}}l@{\hspace{0cm}}l)}
r=1 & 1 & 1 & 1 & 1 & 1 & 1 & 1 & 0 & N    & \\\cline{2-11}
r=2 & 1 & 0 & 0 & 0 & 0 & 0 & 0 & 0 & N-1  & {}+w_{1} \\\cline{2-11}
r=3 & 2 & 1 & \BAmulticolumn{1}{r|}{0} & 0 & 0 & 0 & 0 & 0 & 2(N-1) & {}+2 w_{1} + w_{2}\\\cline{2-11}
\multirow{2}{*}{$r=4$} & 2 & 1 & 1 & \BAmulticolumn{1}{r|}{0} & 0 & 0 & 0 & 0 & 2(N-1) & {}+2 w_{1} + w_{2} + w_{3} \\
& 3 & 2 & 0 & \BAmulticolumn{1}{r|}{0} & 0 & 0 & 0 & 0 & 3(N-1) & {}+3 w_{1} + 2 w_{2} + w_{3}\\
\end{block}
\end{blockarray}
\]

\[
\begin{blockarray}{lrrrrrrrr@{\hspace{0.25cm}}|@{\hspace{0.25cm}}ll}
\begin{block}{l@{\hspace{0.2cm}}rrrrrrrr@{\hspace{0.25cm}}|@{\hspace{0.25cm}}ll}
& \BAmulticolumn{8}{c}{\text{Ray coefficients}} & \BAmulticolumn{2}{c}{\text{Right-hand side}} \\
\end{block}
\begin{block}{l@{\hspace{0.2cm}}(@{\hspace{0.2cm}}rrrrrrrr@{\hspace{0.25cm}}|@{\hspace{0.25cm}}l@{\hspace{0cm}}l)}
\multirow{3}{*}{$r=5$} & 2 & 1 & 1 & 1 & \BAmulticolumn{1}{r|}{0} & 0 & 0 & 0 & 2(N-1) & {}+2 w_{1} + w_{2} + w_{3} + w_{4}\\
& 3 & 2 & 1 & 0 & \BAmulticolumn{1}{r|}{0} & 0 & 0 & 0 & 3(N-1) & {}+3 w_{1} + 2 w_{2} + w_{3} + w_{4}\\
& 4 & 3 & 0 & 0 & \BAmulticolumn{1}{r|}{0} & 0 & 0 & 0 & 4(N-1) & {}+4 w_{1} + 3 w_{2} + 2 w_{3} + w_{4}\\\cline{2-11}
\multirow{5}{*}{$r=6$} & 2 & 1 & 1 & 1 & 1 & \BAmulticolumn{1}{r|}{0} & 0 & 0 & 2(N-1) & {}+2 w_{1} + w_{2} + w_{3} + w_{4} + w_{5}\\
& 3 & 2 & 1 & 1 & 0 & \BAmulticolumn{1}{r|}{0} & 0 & 0 & 3(N-1) & {}+3 w_{1} + 2 w_{2} + w_{3} + w_{4} + w_{5}\\
& 4 & 3 & 1 & 0 & 0 & \BAmulticolumn{1}{r|}{0} & 0 & 0 & 4(N-1) & {}+4 w_{1} + 3 w_{2} + 2 w_{3} + w_{4} + w_{5}\\
& 5 & 4 & 0 & 0 & 0 & \BAmulticolumn{1}{r|}{0} & 0 & 0 & 5(N-1) & {}+5 w_{1} + 4 w_{2} + 3 w_{3} + 2 w_{4} + w_{5}\\
& 6 & 4 & 3 & 0 & 0 & \BAmulticolumn{1}{r|}{0} & 0 & 0 & 6(N-1) & {}+6 w_{1} + 4 w_{2} + 3 w_{3} + 2 w_{4} + w_{5}\\\cline{2-11}
\multirow{9}{*}{$r=7$} & 2 & 1 & 1 & 1 & 1 & 1 & \BAmulticolumn{1}{r|}{0} & 0 & 2(N-1) & {}+2 w_{1} + w_{2} + w_{3} + w_{4} + w_{5} + w_{6}\\
& 3 & 2 & 1 & 1 & 1 & 0 & \BAmulticolumn{1}{r|}{0} & 0 & 3(N-1) & {}+3 w_{1} + 2 w_{2} + w_{3} + w_{4} + w_{5} + w_{6}\\
& 3 & 2 & 2 & 0 & 0 & 0 & \BAmulticolumn{1}{r|}{0} & 0 & 3(N-1) & {}+3 w_{1} + 2 w_{2} + 2 w_{3} + w_{4} + w_{5} + w_{6}\\
& 4 & 3 & 1 & 1 & 0 & 0 & \BAmulticolumn{1}{r|}{0} & 0 & 4(N-1) & {}+4 w_{1} + 3 w_{2} + 2 w_{3} + w_{4} + w_{5} + w_{6}\\
& 4 & 3 & 2 & 0 & 0 & 0 & \BAmulticolumn{1}{r|}{0} & 0 & 4(N-1) & {}+4 w_{1} + 3 w_{2} + 2 w_{3} + 2 w_{4} + w_{5} + w_{6}\\
& 5 & 4 & 1 & 0 & 0 & 0 & \BAmulticolumn{1}{r|}{0} & 0 & 5(N-1) & {}+5 w_{1} + 4 w_{2} + 3 w_{3} + 2 w_{4} + w_{5} + w_{6}\\
& 6 & 4 & 3 & 1 & 0 & 0 & \BAmulticolumn{1}{r|}{0} & 0 & 6(N-1) & {}+6 w_{1} + 4 w_{2} + 3 w_{3} + 2 w_{4} + w_{5} + w_{6}\\
& 6 & 4 & 3 & 2 & 0 & 0 & \BAmulticolumn{1}{r|}{0} & 0 & 6(N-1) & {}+6 w_{1} + 4 w_{2} + 3 w_{3} + 2 w_{4} + 2 w_{5} + w_{6}\\
& 6 & 5 & 0 & 0 & 0 & 0 & \BAmulticolumn{1}{r|}{0} & 0 & 6(N-1) & {}+6 w_{1} + 5 w_{2} + 4 w_{3} + 3 w_{4} + 2 w_{5} + w_{6}\\\cline{2-11}
\multirow{14}{*}{$r=8$} & 2 & 1 & 1 & 1 & 1 & 1 & 1 & 0 & 2(N-1) & {}+2 w_{1} + w_{2} + w_{3} + w_{4} + w_{5} + w_{6} + w_{7}\\
& 3 & 2 & 1 & 1 & 1 & 1 & 0 & 0 & 3(N-1) & {}+3 w_{1} + 2 w_{2} + w_{3} + w_{4} + w_{5} + w_{6} + w_{7}\\
& 3 & 2 & 2 & 1 & 0 & 0 & 0 & 0 & 3(N-1) & {}+3 w_{1} + 2 w_{2} + 2 w_{3} + w_{4} + w_{5} + w_{6} + w_{7}\\
& 4 & 3 & 1 & 1 & 1 & 0 & 0 & 0 & 4(N-1) & {}+4 w_{1} + 3 w_{2} + 2 w_{3} + w_{4} + w_{5} + w_{6} + w_{7}\\
& 4 & 3 & 2 & 1 & 0 & 0 & 0 & 0 & 4(N-1) & {}+4 w_{1} + 3 w_{2} + 2 w_{3} + 2 w_{4} + w_{5} + w_{6} + w_{7}\\
& 5 & 4 & 1 & 1 & 0 & 0 & 0 & 0 & 5(N-1) & {}+5 w_{1} + 4 w_{2} + 3 w_{3} + 2 w_{4} + w_{5} + w_{6} + w_{7}\\
& 5 & 4 & 2 & 0 & 0 & 0 & 0 & 0 & 5(N-1) & {}+5 w_{1} + 4 w_{2} + 3 w_{3} + 2 w_{4} + 2 w_{5} + w_{6} + w_{7}\\
& 6 & 4 & 3 & 1 & 1 & 0 & 0 & 0 & 6(N-1) & {}+6 w_{1} + 4 w_{2} + 3 w_{3} + 2 w_{4} + w_{5} + w_{6} + w_{7}\\
& 6 & 4 & 3 & 2 & 1 & 0 & 0 & 0 & 6(N-1) & {}+6 w_{1} + 4 w_{2} + 3 w_{3} + 2 w_{4} + 2 w_{5} + w_{6} + w_{7}\\
& 6 & 4 & 3 & 2 & 2 & 0 & 0 & 0 & 6(N-1) & {}+6 w_{1} + 4 w_{2} + 3 w_{3} + 2 w_{4} + 2 w_{5} + 2 w_{6} + w_{7}\\
& 6 & 4 & 3 & 3 & 0 & 0 & 0 & 0 & 6(N-1) & {}+6 w_{1} + 4 w_{2} + 3 w_{3} + 3 w_{4} + 2 w_{5} + w_{6} + w_{7}\\
& 6 & 5 & 1 & 0 & 0 & 0 & 0 & 0 & 6(N-1) & {}+6 w_{1} + 5 w_{2} + 4 w_{3} + 3 w_{4} + 2 w_{5} + w_{6} + w_{7}\\
& 7 & 5 & 4 & 0 & 0 & 0 & 0 & 0 & 7(N-1) & {}+7 w_{1} + 5 w_{2} + 4 w_{3} + 3 w_{4} + 2 w_{5} + w_{6} + w_{7}\\
& 7 & 6 & 0 & 0 & 0 & 0 & 0 & 0 & 7(N-1) & {}+7 w_{1} + 6 w_{2} + 5 w_{3} + 4 w_{4} + 3 w_{5} + 2 w_{6} + w_{7}\\
\end{block}
\end{blockarray}
\]


\begin{thebibliography}{AHBC{\etalchar{+}}16}

\bibitem[AK08]{altunbulak_pauli_2008}
Murat Altunbulak and Alexander Klyachko, \emph{The {P}auli principle
  revisited}, Comm. Math. Phys. \textbf{282} (2008), no.~2, 287--322.

\bibitem[AW03]{AndrzejakWelzl2003}
Artur Andrzejak and Emo Welzl, \emph{In between {$k$}-sets, {$j$}-facets, and
  {$i$}-faces: {$(i,j)$}-partitions}, Discrete Comput. Geom. \textbf{29}
  (2003), no.~1, 105--131.

\bibitem[ACEP20]{ardila_coxeter_2020}
Federico Ardila, Federico Castillo, Christopher Eur, and Alexander Postnikov,
  \emph{Coxeter submodular functions and deformations of {C}oxeter
  permutahedra}, Adv. Math. \textbf{365} (2020), 107039, 36.

\bibitem[AHBC{\etalchar{+}}16]{arkani_grassmannian_2016}
Nima Arkani-Hamed, Jacob Bourjaily, Freddy Cachazo, Alexander Goncharov,
  Alexander Postnikov, and Jaroslav Trnka, \emph{Grassmannian geometry of
  scattering amplitudes}, Cambridge University Press, Cambridge, 2016.

\bibitem[AHLM21]{arkani_positive_2021}
Nima Arkani-Hamed, Thomas Lam, and Spradlin Marcus, \emph{Positive
  configuration space}, Comm. Math. Phys. (2021), to appear.

\bibitem[Ati82]{atiyah_convexity_1982}
Michael~F. Atiyah, \emph{Convexity and commuting {H}amiltonians}, Bull. London
  Math. Soc. \textbf{14} (1982), no.~1, 1--15.

\bibitem[AF92]{avis_pivoting_1992}
David Avis and Komei Fukuda, \emph{A pivoting algorithm for convex hulls and
  vertex enumeration of arrangements and polyhedra}, Discrete Comput. Geom.
  \textbf{8} (1992), no.~3, 295--313.

\bibitem[Ayr58]{ayres_variational_1958}
Robert~U. Ayres, \emph{Variational approach to the many-body problem}, Phys.
  Rev. \textbf{111} (1958), 1453--1460.

\bibitem[Bab77]{babai_symmetry_1977}
L\'aszl\'o Babai, \emph{Symmetry groups of vertex-transitive polytopes}, Geom.
  Dedicata \textbf{6} (1977), no.~3, 331--337.

\bibitem[BR19]{bach_orthogonalization_2019}
Volker Bach and Robert Rauch, \emph{Orthogonalization of fermion $k$-body
  operators and representability}, Phys. Rev. A \textbf{99} (2019), 042109, 10.

\bibitem[BHZ08]{bagnara_parma_2008}
Roberto Bagnara, Patricia~M. Hill, and Enea Zaffanella, \emph{The {P}arma
  {P}olyhedra {L}ibrary: toward a complete set of numerical abstractions for
  the analysis and verification of hardware and software systems}, Sci. Comput.
  Programming \textbf{72} (2008), no.~1-2, 3--21.

\bibitem[BS00]{berenstein_coadjoint_2000}
Arkady Berenstein and Reyer Sjamaar, \emph{Coadjoint orbits, moment polytopes,
  and the {H}ilbert-{M}umford criterion}, J. Amer. Math. Soc. \textbf{13}
  (2000), no.~2, 433--466.

\bibitem[BKS94]{BilleraKapranovSturmfels1994}
Louis~J. Billera, Mikhail~M. Kapranov, and Bernd Sturmfels, \emph{Cellular
  strings on polytopes}, Proc. Amer. Math. Soc. \textbf{122} (1994), no.~2,
  549--555.

\bibitem[BS92]{BS92}
Louis~J. Billera and Bernd Sturmfels, \emph{Fiber polytopes}, Ann. of Math. (2)
  \textbf{135} (1992), no.~3, 527--549.

\bibitem[BD72]{borland_conditions_1972}
R.~E. Borland and K.~Dennis, \emph{The conditions on the one-matrix for
  three-body fermion wavefunctions with one-rank equal to six}, J. Phys. B
  \textbf{5} (1972), no.~1, 7--15.

\bibitem[BGW03]{borovik_coxeter_2003}
Alexandre~V. Borovik, Israil~M. Gelfand, and Neil White, \emph{Coxeter
  matroids}, Progress in Mathematics, vol. 216, Birkh\"{a}user Boston, Inc.,
  Boston, MA, 2003.

\bibitem[Bri99]{brion_general_1999}
Michel Brion, \emph{On the general faces of the moment polytope}, Int. Math.
  Res. Not. IMRN \textbf{1999} (1999), no.~4, 185--201.

\bibitem[BIS16]{bruns_power_2016}
Winfried Bruns, Bogdan Ichim, and Christof S\"{o}ger, \emph{The power of
  pyramid decomposition in {N}ormaliz}, J. Symbolic Comput. \textbf{74} (2016),
  513--536.

\bibitem[CCP03]{cassam_hopf_2003}
Patrick Cassam-Chena\"{\i} and Fr\'{e}d\'{e}ric Patras, \emph{The {H}opf
  algebra of identical, fermionic particle systems---fundamental concepts and
  properties}, J. Math. Phys. \textbf{44} (2003), no.~11, 4884--4906.

\bibitem[CJR{\etalchar{+}}12]{chen_comment_2012}
Jianxin Chen, Zhengfeng Ji, Mary~Beth Ruskai, Bei Zeng, and Duan-Lu Zhou,
  \emph{Comment on some results of {E}rdahl and the convex structure of reduced
  density matrices}, J. Math. Phys. \textbf{53} (2012), no.~7, 072203, 11.

\bibitem[CM06]{christandl_spectra_2006}
Matthias Christandl and Graeme Mitchison, \emph{The spectra of quantum states
  and the {K}ronecker coefficients of the symmetric group}, Comm. Math. Phys.
  \textbf{261} (2006), no.~3, 789--797.

\bibitem[Cio00]{cioslowski_many_2000}
Jerzy Cioslowski (ed.), \emph{Many-electron densities and reduced density
  matrices}, Mathematical and Computational Chemistry, Springer, US, 2000.

\bibitem[Col63]{coleman_structure_1963}
Albert~John Coleman, \emph{Structure of fermion density matrices}, Rev. Modern
  Phys. \textbf{35} (1963), 668--689.

\bibitem[Col72]{coleman_necessary_1972}
\bysame, \emph{Necessary conditions for {$N$}-representability of reduced
  density matrices}, J. Math. Phys. \textbf{13} (1972), 214--222.

\bibitem[Col77]{coleman_convex_1977}
\bysame, \emph{Convex structure of electrons}, Int. J. Quant. Chem. \textbf{11}
  (1977), no.~6, 907--916.

\bibitem[Col01]{coleman_reduced_2001}
\bysame, \emph{Reduced density matrices--then and now}, Int. J. Quant. Chem.
  \textbf{85} (2001), no.~4-5, 196--203.

\bibitem[Col02]{coleman_kummer_2002}
\bysame, \emph{Kummer variety, geometry of {$N$}-representability, and phase
  transitions}, Phys. Rev. A (3) \textbf{66} (2002), no.~2, 022503, 8.

\bibitem[CY00]{coleman_reduced_2000}
Albert~John Coleman and Vyacheslav~I. Yukalov, \emph{Reduced density matrices},
  Lecture Notes in Chemistry, vol.~72, Springer-Verlag, Berlin, 2000, Coulson's
  challenge.

\bibitem[Cou60]{coulson_present_1960}
Charles~A. Coulson, \emph{Present state of molecular structure calculations},
  Rev. Mod. Phys. \textbf{32} (1960), 170--177.

\bibitem[CK06]{cruickshank_rearrangement_2006}
James Cruickshank and S\'{e}amus Kelly, \emph{Rearrangement inequalities and
  the alternahedron}, Discrete Comput. Geom. \textbf{35} (2006), no.~2,
  241--254.

\bibitem[DH05]{daftuar_quantum_2005}
Sumit Daftuar and Patrick Hayden, \emph{Quantum state transformations and the
  {S}chubert calculus}, Ann. Physics \textbf{315} (2005), no.~1, 80--122.

\bibitem[Dir30]{dirac_note_1930}
Paul A.~M. Dirac, \emph{Note on exchange phenomena in the thomas atom}, Math.
  Proc. Cambridge Philos. Soc. \textbf{26} (1930), no.~3, 376--385.

\bibitem[Ede00]{Edelman2000}
Paul~H. Edelman, \emph{Ordering points by linear functionals}, European J.
  Combin. \textbf{21} (2000), no.~1, 145--152.

\bibitem[EGS13]{edelman_simplicial_2013}
Paul~H. Edelman, Tatiana Gvozdeva, and Arkadii Slinko, \emph{Simplicial
  complexes obtained from qualitative probability orders}, SIAM J. Discrete
  Math. \textbf{27} (2013), no.~4, 1820--1843.

\bibitem[EVW97]{EdelsbrunnerValtrWelzl1997}
Herbert Edelsbrunner, Pavel Valtr, and Emo Welzl, \emph{Cutting dense point
  sets in half}, Discrete Comput. Geom. \textbf{17} (1997), no.~3, 243--255.

\bibitem[ES87]{coleman_sympo_1987}
Robert Erdahl and Vedene~H. Smith (eds.), \emph{Density matrices and density
  functionals: proceedings of the a. john coleman symposium}, Dordrecht,
  Springer, Dordrecht, 1987.

\bibitem[FL16]{friese_affine_2016}
Erik Friese and Frieder Ladisch, \emph{Affine symmetries of orbit polytopes},
  Adv. Math. \textbf{288} (2016), 386--425.

\bibitem[Fuk08]{fukuda_exact_2008}
Komei Fukuda, \emph{Exact algorithms and software in optimization and
  polyhedral computation}, I{SSAC} 2008, ACM, New York, 2008, pp.~333--334.

\bibitem[Ful97]{fulton_young_1997}
William Fulton, \emph{Young tableaux}, London Mathematical Society Student
  Texts, vol.~35, Cambridge University Press, Cambridge, 1997.

\bibitem[GGL75]{gabrielov_combinatorial_1975}
Andrei~M. Gabri\`elov, Israil~M. Gelfand, and Mark~V. Losik,
  \emph{Combinatorial computation of characteristic classes. {I}, {II}},
  Funkcional. Anal. i Prilo\v{z}en. \textbf{9} (1975), no.~2, 12--28; ibid. 9
  (1975), no. 3, 5--26.

\bibitem[Gal68]{gale_optimal_1968}
David Gale, \emph{Optimal assignments in an ordered set: {A}n application of
  matroid theory}, J. Combinatorial Theory \textbf{4} (1968), 176--180.

\bibitem[GP64]{garrod_reduction_1964}
Claude Garrod and Jerome~K. Percus, \emph{Reduction of the {$N$}-particle
  variational problem}, J. Math. Phys. \textbf{5} (1964), 1756--1776.

\bibitem[GJ00]{polymake}
Ewgenij Gawrilow and Michael Joswig, \emph{{\tt polymake}: a framework for
  analyzing convex polytopes}, Polytopes---combinatorics and computation
  ({O}berwolfach, 1997), DMV Sem., vol.~29, Birkh\"auser, Basel, 2000,
  pp.~43--73.

\bibitem[GGMS87]{gelfand_combinatorial_1987}
Israil~M. Gelfand, Robert~Mark Goresky, Robert~D. MacPherson, and Vera~V.
  Serganova, \emph{Combinatorial geometries, convex polyhedra, and {S}chubert
  cells}, Adv. Math. \textbf{63} (1987), no.~3, 301--316.

\bibitem[Gro03]{groetsch_functional_2003}
Charles~W. Groetsch, \emph{Functional analysis}, Encyclopedia of Physical
  Science and Technology (Third Edition) (Robert~A. Meyers, ed.), Academic
  Press, New York, third edition ed., 2003, pp.~337--353.

\bibitem[GOK88]{gross_rayleigh_1988}
Eberhard K.~U. Gross, Luiz~N. Oliveira, and Walter Kohn, \emph{Rayleigh-{R}itz
  variational principle for ensembles of fractionally occupied states}, Phys.
  Rev. A (3) \textbf{37} (1988), no.~8, 2805--2808.

\bibitem[Gr{\"{u}}03]{grunbaum_convex_2003}
Branko Gr{\"{u}}nbaum, \emph{Convex polytopes}, second ed., GTM, vol. 221,
  Springer-Verlag, New York, 2003.

\bibitem[GS82]{guillemin_convexity_1982}
Victor Guillemin and Shlomo Sternberg, \emph{Convexity properties of the moment
  mapping}, Invent. Math. \textbf{67} (1982), no.~3, 491--513.

\bibitem[Hal13]{hall_quantum_2013}
Brian~C. Hall, \emph{Quantum theory for mathematicians}, GTM, vol. 267,
  Springer, New York, 2013.

\bibitem[HLP88]{hardy_inequalities_1988}
Godfrey~H. Hardy, John~E. Littlewood, and George P\'{o}lya,
  \emph{Inequalities}, Cambridge Mathematical Library, Cambridge University
  Press, Cambridge, 1988, Reprint of the 1952 edition.

\bibitem[HS20]{heaton_dual_2020}
Alexander Heaton and Jose~Alejandro Samper, \emph{Dual matroid polytopes and
  internal activity of independence complexes}, preprint,
  \href{http://arxiv.org/abs/2005.04252}{\tt arXiv:2005.04252} (May 2020),
  34~pp.

\bibitem[Hum90]{humphreys_reflection_1990}
James~E. Humphreys, \emph{Reflection groups and {C}oxeter groups}, Cambridge
  Studies in Advanced Mathematics, vol.~29, Cambridge University Press,
  Cambridge, 1990.

\bibitem[Hus40]{husimi_formal_1940}
K\^odi Husimi, \emph{Some formal properties of the density matrix}, Proc. Phys.
  Math. Soc. Japan \textbf{22} (1940), no.~4, 264--314.

\bibitem[Kir04]{kirillov_lectures_2004}
Alexandre~A. Kirillov, \emph{Lectures on the orbit method}, Graduate Studies in
  Mathematics, vol.~64, American Mathematical Society, Providence, RI, 2004.

\bibitem[Kir84]{kirwan_convexity_1984}
Frances Kirwan, \emph{Convexity properties of the moment mapping. {III}},
  Invent. Math. \textbf{77} (1984), no.~3, 547--552. \MR{759257}

\bibitem[Kli07]{klivans_threshold_2007}
Caroline~J. Klivans, \emph{Threshold graphs, shifted complexes, and graphical
  complexes}, Discrete Math. \textbf{307} (2007), no.~21, 2591--2597.

\bibitem[KR08]{klivans_shifted_2008}
Caroline~J. Klivans and Vic Reiner, \emph{Shifted set families, degree
  sequences, and plethysm}, Electron. J. Combin. \textbf{15} (2008), no.~1,
  Research Paper 14, 35.

\bibitem[Kly98]{klyachko_stable_1998}
Alexander~A. Klyachko, \emph{Stable bundles, representation theory and
  {H}ermitian operators}, Selecta Math. (N.S.) \textbf{4} (1998), no.~3,
  419--445.

\bibitem[Kly06]{klyachko_2006}
Alexander~A Klyachko, \emph{Quantum marginal problem and n-representability},
  Journal of Physics: Conference Series \textbf{36} (2006), 72--86.

\bibitem[Kly09]{klyachko_pauli_2009}
Alexander~A. Klyachko, \emph{The {P}auli exclusion principle and beyond},
  preprint, \href{http://arxiv.org/abs/0904.2009}{\tt arXiv:0904.2009} (April
  2009), 4~pp.

\bibitem[Knu00]{knutson_symplectic_2000}
Allen Knutson, \emph{The symplectic and algebraic geometry of {H}orn's
  problem}, Linear Algebra Appl. \textbf{319} (2000), no.~1-3, 61--81.

\bibitem[Kos73]{kostant_convexity_1973}
Bertram Kostant, \emph{On convexity, the {W}eyl group and the {I}wasawa
  decomposition}, Ann. Sci. \'{E}cole Norm. Sup. (4) \textbf{6} (1973),
  413--455 (1974).

\bibitem[Kuh60]{kuhn_linear_1960}
Harold~W. Kuhn, \emph{Linear inequalities and the {P}auli principle}, Proc.
  {S}ympos. {A}ppl. {M}ath., {V}ol. 10, American Mathematical Society,
  Providence, R.I., 1960, pp.~141--147.

\bibitem[Kum67]{kummer_repr_1967}
Hans Kummer, \emph{{$n$}-representability problem for reduced density
  matrices}, J. Math. Phys. \textbf{8} (1967), 2063--2081.

\bibitem[Lan17]{landsman_foundations_2017}
Klaas Landsman, \emph{Foundations of quantum theory}, Fundamental Theories of
  Physics, vol. 188, Springer, Cham, 2017.

\bibitem[Lan02]{serge_lang_2002}
Serge Lang, \emph{Algebra}, third ed., GTM, vol. 211, Springer-Verlag, New
  York, 2002.

\bibitem[Lev79]{Levy79}
Mel Levy, \emph{Universal variational functionals of electron densities,
  first-order density matrices, and natural spin-orbitals and solution of the
  v-representability problem}, Proc. Natl. Acad. Sci. U.S.A \textbf{76} (1979),
  no.~12, 6062.

\bibitem[Lie83]{Lieb83}
Elliot~H. Lieb, \emph{Density functionals for coulomb systems}, Int. J. Quantum
  Chem. \textbf{24} (1983), no.~3, 243.

\bibitem[LCLS21]{liebert_foundation_2021}
Julia Liebert, Federico Castillo, Jean-Philippe Labb{\'e}, and Christian
  Schilling, \emph{Foundation of one-particle reduced density matrix functional
  theory for excited states}, preprint,
  \href{http://arxiv.org/abs/2106.03918}{\tt arXiv:2106.03918} (June 2021),
  20~pp.

\bibitem[LCV07]{liu_quantum_2007}
Yi-Kai Liu, Matthias Christandl, and Frank Verstraete, \emph{Quantum
  computational complexity of the $n$-representability problem: Qma complete},
  Phys. Rev. Lett. \textbf{98} (2007), 110503, 4.

\bibitem[L{\"{o}}w55]{loewdin_quantum_1955}
Per-Olov L{\"{o}}wdin, \emph{Quantum theory of many-particle systems. {I}.
  {P}hysical interpretations by means of density matrices, natural
  spin-orbitals, and convergence problems in the method of configurational
  interaction}, Phys. Rev. (2) \textbf{97} (1955), 1474--1489.

\bibitem[LPW20]{lukowski_positive_2020}
Tomasz Lukowski, Matteo Parisi, and Lauren Williams, \emph{The positive
  tropical grassmannian, the hypersimplex, and the $m=2$ amplituhedron},
  preprint, \href{http://arxiv.org/abs/2002.06164}{\tt arXiv:2002.06164}
  (February 2020), 50~pp.

\bibitem[MT17]{maciazek_quantum_2007}
Tomasz Maci\k{a}\.{z}ek and Valdemar Tsanov, \emph{Quantum marginals from pure
  doubly excited states}, J. Phys. A \textbf{50} (2017), no.~46, 465304, 39.

\bibitem[MOA11]{marshall_inequalities_2011}
Albert~W. Marshall, Ingram Olkin, and Barry~C. Arnold, \emph{Inequalities:
  theory of majorization and its applications}, second ed., Springer Series in
  Statistics, Springer, New York, 2011.

\bibitem[MSP20]{MartinezSandovalPadrol2020}
Leonardo Mart\'inez-Sandoval and Arnau Padrol, \emph{The convex dimension of
  hypergraphs and the hypersimplicial {V}an {K}ampen-{F}lores theorem},
  Preprint, \href{http://arxiv.org/abs/1909.01189}{arXiv:1909.01189}, 2020.

\bibitem[Mat16]{matteo_combinatorially_2016}
Nicholas Matteo, \emph{Combinatorially two-orbit convex polytopes}, Discrete
  Comput. Geom. \textbf{55} (2016), no.~3, 662--680.

\bibitem[Maz06]{mazziotti_anti_2006}
David~A. Mazziotti, \emph{Anti-hermitian contracted schr\"odinger equation:
  Direct determination of the two-electron reduced density matrices of
  many-electron molecules}, Phys. Rev. Lett. \textbf{97} (2006), 143002, 4.

\bibitem[Maz07]{mazziotti_reduced_2007}
David~A. Mazziotti (ed.), \emph{Reduced-density-matrix mechanics: With
  application to many-electron atoms and molecules}, Advances in Chemical
  Physics, vol. 134, John Wiley \& Sons, 2007.

\bibitem[Maz12]{mazziotti_structure_2012}
\bysame, \emph{Structure of fermionic density matrices: Complete
  $n$-representability conditions}, Phys. Rev. Lett. \textbf{108} (2012),
  263002, 5.

\bibitem[Maz16]{mazziotti_pure_2016}
\bysame, \emph{Pure-$n$-representability conditions of two-fermion reduced
  density matrices}, Phys. Rev. A \textbf{94} (2016), 032516, 5.

\bibitem[Mir63]{mirsky_results_1962}
Leonid Mirsky, \emph{Results and problems in the theory of doubly-stochastic
  matrices}, Z. Wahrscheinlichkeitstheorie und Verw. Gebiete \textbf{1}
  (1962/63), 319--334.

\bibitem[NRC95]{national_1995}
{National Research Council}, \emph{Mathematical challenges from
  theoretical/computational chemistry}, The National Academies Press,
  Washington, DC, 1995.

\bibitem[NC00]{nielsen_quantum_2000}
Michael~A. Nielsen and Isaac~L. Chuang, \emph{Quantum computation and quantum
  information}, Cambridge University Press, Cambridge, 2000.

\bibitem[Onn93]{onn_geometry_1993}
Shmuel Onn, \emph{Geometry, complexity, and combinatorics of permutation
  polytopes}, J. Combin. Theory Ser. A \textbf{64} (1993), no.~1, 31--49.

\bibitem[PP21]{PadrolPhilippe2021}
Arnau Padrol and Eva Philippe, \emph{Sweeps, polytopes, oriented matroids, and
  allowable graphs of permutations}, preprint,
  \href{http://arxiv.org/abs/2102.06134}{\tt arXiv:2102.06134} (February 2021),
  41~pp.

\bibitem[PSBW21]{parisi_amplituhedron_2021}
Matteo Parisi, Melissa Sherman-Bennett, and Lauren Williams, \emph{The $m=2$
  amplituhedron and the hypersimplex: signs, clusters, triangulations,
  {E}ulerian numbers}, preprint, \href{http://arxiv.org/abs/2104.08254}{\tt
  arXiv:2104.08254} (April 2021), 74~pp.

\bibitem[{Pau}25]{pauli_zusammenhang_1925}
Wolfgang {Pauli}, \emph{{{\"U}ber den Zusammenhang des Abschlusses der
  Elektronengruppen im Atom mit der Komplexstruktur der Spektren}}, Z. Phys.
  \textbf{31} (1925), no.~1, 765--783.

\bibitem[Pos09]{postnikov_permutahedra_2009}
Alexander Postnikov, \emph{Permutohedra, associahedra, and beyond}, Int. Math.
  Res. Not. IMRN (2009), no.~6, 1026--1106.

\bibitem[Rad52]{rado_inequality_1952}
Richard Rado, \emph{An inequality}, J. London Math. Soc. \textbf{27} (1952),
  1--6.

\bibitem[Rei99]{Reiner1999}
Victor Reiner, \emph{The generalized {B}aues problem}, New perspectives in
  algebraic combinatorics ({B}erkeley, {CA}, 1996--97), Math. Sci. Res. Inst.
  Publ., vol.~38, Cambridge Univ. Press, Cambridge, 1999, pp.~293--336.

\bibitem[Res10]{ressayre_geometric_2010}
Nicolas Ressayre, \emph{Geometric invariant theory and the generalized
  eigenvalue problem}, Invent. Math. \textbf{180} (2010), no.~2, 389--441.

\bibitem[Roc97]{rockafellar_convex_1997}
Ralph~T. Rockafellar, \emph{Convex analysis}, Princeton Landmarks in
  Mathematics, Princeton University Press, Princeton, NJ, 1997.

\bibitem[Rus07]{ruskai_2007}
Mary~Beth Ruskai, \emph{Connecting {$N$}-representability to {W}eyl's problem:
  the one-particle density matrix for {$N=3$} and {$R=6$}}, J. Phys. A
  \textbf{40} (2007), no.~45, F961--F967.

\bibitem[SS20]{sanyal_spectral_2020}
Raman Sanyal and James Saunderson, \emph{Spectral polyhedra}, preprint,
  \href{http://arxiv.org/abs/2001.04361}{\tt arXiv:2001.04361} (January 2020),
  13~pp.

\bibitem[SSS11]{sanyal_orbitopes_2011}
Raman Sanyal, Frank Sottile, and Bernd Sturmfels, \emph{Orbitopes}, Mathematika
  \textbf{57} (2011), no.~2, 275--314.

\bibitem[Sch15]{schilling_quantum_2015}
Christian Schilling, \emph{The quantum marginal problem}, Mathematical results
  in quantum mechanics, World Sci. Publ., Hackensack, NJ, 2015, pp.~165--176.

\bibitem[SAK{\etalchar{+}}18]{schilling_generalized_2018}
Christian Schilling, Murat Altunbulak, Stefan Knecht, Alexandre Lopes, James~D.
  Whitfield, Matthias Christandl, David Gross, and Markus Reiher,
  \emph{Generalized {P}auli constraints in small atoms}, Phys. Rev. A
  \textbf{97} (2018), 052503, 12.

\bibitem[SGC13]{schilling_pinning_2013}
Christian Schilling, David Gross, and Matthias Christandl, \emph{Pinning of
  fermionic occupation numbers}, Phys. Rev. Lett. \textbf{110} (2013), 040404.

\bibitem[SP21]{schilling_ensemble_2021}
Christian Schilling and Stefano Pittalis, \emph{Ensemble reduced density matrix
  functional theory for excited states and hierarchical generalization of
  {P}auli’s exclusion principle}, Phys. Rev. Lett. (2021), to appear.

\bibitem[Sch86]{schrijver_theory_1986}
Alexander Schrijver, \emph{Theory of linear and integer programming},
  Wiley-Interscience Series in Discrete Mathematics, John Wiley \& Sons, 1986.

\bibitem[Sch97]{schulte_symmetry_1997}
Egon Schulte, \emph{Symmetry of polytopes and polyhedra}, Handbook of discrete
  and computational geometry, CRC Press Ser. Discrete Math. Appl., CRC, Boca
  Raton, FL, 1997, pp.~311--330.

\bibitem[Sch13]{schuermann_exploiting_2013}
Achill Sch\"{u}rmann, \emph{Exploiting symmetries in polyhedral computations},
  Discrete geometry and optimization, Fields Inst. Commun., vol.~69, Springer,
  New York, 2013, pp.~265--278.

\bibitem[Ser77]{serre_1977}
Jean-Pierre Serre, \emph{Linear representations of finite groups}, vol.~42,
  Springer, 1977.

\bibitem[Sta77]{stanley_cohen_1976}
Richard~P. Stanley, \emph{Cohen-{M}acaulay complexes}, Higher combinatorics
  ({P}roc. {NATO} {A}dvanced {S}tudy {I}nst., {B}erlin, 1976), D. Reidel
  Publishing Co., Dordrecht-Boston, Mass., 1977, pp.~51--62.

\bibitem[Sta12]{stanley_enumerative_2012}
\bysame, \emph{Enumerative combinatorics. {V}olume 1}, second ed., Cambridge
  Studies in Advanced Mathematics, vol.~49, Cambridge University Press,
  Cambridge, 2012.

\bibitem[Sta15]{Stan15}
\bysame, \emph{Valid orderings of real hyperplane arrangements}, Discrete
  Comput. Geom. \textbf{53} (2015), no.~4, 951--964.

\bibitem[Sage]{sagemath}
{The Sage Developers}, \emph{{S}agemath, the {S}age {M}athematics {S}oftware
  {S}ystem ({V}ersion 9.3)}, 2021, {\tt https://www.sagemath.org}.

\bibitem[Val80]{V80}
Steven~M. Valone, \emph{Consequences of extending 1-matrix energy functionals
  from pure–state representable to all ensemble representable 1-matrices}, J.
  Chem. Phys. \textbf{73} (1980), no.~3, 1344.

\bibitem[Wal14]{walter_multipartite_2014}
Michael Walter, \emph{Multipartite quantum states and their marginals}, Ph.D.
  thesis, ETH Z\"urich, Z\"urich, 2014, pp.~xi+201.

\bibitem[Wat39]{watanabe_anwendung_1939}
Satosi Watanabe, \emph{{\"{U}}ber die anwendung thermodynamischer begriffe auf
  den normalzustand des atomkern}, Z. Physik \textbf{113} (1939), no.~7-8,
  482--513.

\bibitem[Yan62]{yang_concept_1962}
Chen~Ning Yang, \emph{Concept of off-diagonal long-range order and the quantum
  phases of liquid {H}e and of superconductors}, Rev. Modern Phys. \textbf{34}
  (1962), 694--704.

\bibitem[Zie95]{ziegler_lectures_1995}
G\"unter~M. Ziegler, \emph{Lectures on polytopes}, GTM, vol. 152,
  Springer-Verlag, New York, 1995.

\end{thebibliography}
\end{document}